%% file: Erignoux2016_v7.4.tex
\documentclass[leqno, 10pt]{smfart}
\usepackage{amssymb, 
mathrsfs, 
dsfont, 
yfonts, 
amsmath, 
bbm, 
gensymb, 
stmaryrd, 
}
\usepackage{smfthm, smfenum}
\usepackage[dvipsnames]{xcolor}

\usepackage[T1]{fontenc}
\usepackage{graphicx}
\usepackage{enumerate}
\usepackage[english,french]{babel}
\usepackage{geometry}
\usepackage{subcaption}
\usepackage[pdftex, breaklinks, colorlinks]{hyperref} 

\allowdisplaybreaks

\makeatletter 

\xdefinecolor{definition}{named}{OliveGreen}
\xdefinecolor{proposition}{named}{BrickRed}
\definecolor{dkgreen}{rgb}{0,0.6,0}
\definecolor{gray}{rgb}{0.5,0.5,0.5}
\definecolor{header}{gray}{0.3}

\numberwithin{equation}{section}

\newcommand{\R}{\mathbb{R}}

\newcommand{\Z}{\mathbb{Z}}
\newcommand{\N}{\mathbb{N}}

\newcommand{\F}{{\mathcal F}}
\newcommand{\dir}{{\mathscr{D}}}
\newcommand{\gene}{{\mathcal L}}
\newcommand{\V}{{\mathcal V}}
\newcommand{\W}{{\mathcal W}}
\newcommand{\ctoruspi}{\mathds{S}}

\newcommand{\Prob}{\mathbb{P}}
\newcommand{\1}{\mathbbm{1}}

\newcommand{\abs}[1]{\;\left\vert\;#1 \;\right\vert\;}
\newcommand{\abss}[1]{ \vert \,#1 \,\vert}
\newcommand{\norm}[1]{\left\vert\left\vert #1 \right\vert\right\vert}
\newcommand{\normm}[1]{{\left\vert\kern-0.1ex\left\vert\kern-0.1ex\left\vert\; #1 \; \right\vert\kern-0.1ex\right\vert\kern-0.1ex\right\vert}}
    
\newcommand{\cro}[1]{\left[#1\right]}
\newcommand{\pa}[1]{\left(#1\right)}
\newcommand{\egal}{ \; =\; }

\newcommand{\intro}[1]{{\textit{#1}} \bigskip}

\newcommand{\eqand }{\qquad \mbox{and} \qquad}

\newcommand{\func}[5]{\begin{array}{ccccc}
#1 &~: & #2 & \longrightarrow & #3 \\
& & #4 & \mapsto & #5 \\
\end{array}}

\newcommand{\limep}{\limsup_{\varepsilon\to0}}
\newcommand{\liml}{\limsup_{l\to\infty}}
\newcommand{\limN}{\limsup_{N\to\infty}}

\newcommand{\proofthm}[2]{\begin{proof}[Proof of #1]#2\end{proof}}

\newcommand{\step}[2]{
\bigskip
\noindent \emph{{\large\underline{#1 step~: #2}}}\\
\medskip
}

\newcommand{\aom}{{\alpha_{\omega}}}
\newcommand{\am}{{\alpha}}

\newcommand{\cm}{\mu_{B,\K}}
\newcommand{\cml}{\mu_{l,\K_l}}
\newcommand{\cmlk}{\mu_{l,\K}}
\newcommand{\conf}{\eta}
\newcommand{\confhat}{\widehat{\conf}}
\newcommand{\com}{\conf^{\omega}}
\newcommand{\conftilde}{\conf^{\widehat{\omega}}}

\newcommand{\comp}{\conf^{\omega, p}}
\newcommand{\confrond}{\mathring{\conf}}

\newcommand{\confe}{\xi}
\newcommand{\confehat}{\widehat{\confe}}

\newcommand{\confi}{\zeta}
\newcommand{\confihat}{\widehat{\confi}}

\newcommand{\ctorus}{\mathbb{T}^2}
\newcommand{\cur}{j}
\newcommand{\curom}{\cur^{\omega}}

\newcommand{\curg}{\boldsymbol{\mathfrak{j}}}

\newcommand{\curomz}{\cur^{\widehat{\omega}}}
\newcommand{\comz}{\conf^{\widehat{\omega}}}
\newcommand{\cura}{r}
\newcommand{\curaom}{\cura^{\omega}}

\newcommand{\czero}{{\mathcal C}_0}
\newcommand{\D}{{\mathcal D}}
\newcommand{\ddi}{\boldsymbol \delta_i}

\newcommand{\diffom}{\mathfrak{d}^{\omega}}
\newcommand{\diff}{\mathfrak{d}}
\newcommand{\diffh}{\widehat{\mathfrak{d}}}

\newcommand{\dens}{\widehat{\rho}}
\newcommand{\densep}{\dens_{\varepsilon N}}
\newcommand{\densom}{\rho^{\omega}}
\newcommand{\densl}{\dens_l}\newcommand{\E}{{\mathbb E}}

\newcommand{\Ecm}{\E_{B,\K}}
\newcommand{\Ecml}{\E_{l,\K_l}}
\newcommand{\Ecmlk}{\E_{l,\K}}
\newcommand{\Eref}{\E^*_{\alpha}}
\newcommand{\Egcm}{\E_{\param}}
\newcommand{\epx}{E_{p,x}}
\newcommand{\epxc}{E_{p,x}^c}

\newcommand{\function}{g}
\newcommand{\G}{{\mathscr G}}
\newcommand{\gcm}{\mu_{\widehat{\alpha}}}
\newcommand{\genas}{{\mathcal L}}
\newcommand{\genex}{{\mathcal L}^{\!\!\mbox{\begin{tiny} D\end{tiny}}}}

\newcommand{\genis}{{\mathcal L}^{\!\!\mbox{\begin{tiny} G\end{tiny}}}}
\newcommand{\geniz}{{\mathcal L}^{\!\!\mbox{\begin{tiny} G\end{tiny}}, \beta=0}}
\newcommand{\genizad}{{\mathcal L}^{\!\!\mbox{\begin{tiny} G\end{tiny}}, \beta=0, *}}

\newcommand{\genwa}{{\mathcal L}^{\!\!\mbox{\begin{tiny} W\!\!A\end{tiny}}}}
\newcommand{\genwad}{{\mathcal L}^{\!\!\mbox{\begin{tiny} W\!\!A\end{tiny}},*}}
\newcommand{\grad}{\nabla_{a}}
\newcommand{\halpha}{{\mathcal{H}}^{\omega}_{\param}}
\newcommand{\halphaz}{{\mathcal{H}}^{\omega,0}_{\param}}

\newcommand{\scalstar}[1]{\scal{#1}}
\newcommand{\hx}{H_{x/N}}
\newcommand{\rn}[1]{\rho_{#1}^{N,H}}

\newcommand{\Kset}{{\mathbb{ K}}}
\newcommand{\Ksett}{\widetilde{\mathbb{K}}}
\newcommand{\K}{\widehat{K}}
\newcommand{\cdf}{{\mathfrak{F}}}
\newcommand{\M}{\cdf}
\newcommand{\Mhat}{\meas^{[0,T]}}
\newcommand{\meas}{{\mathcal M}}

\newcommand{\mesinv}{\mu_{\widehat{\alpha}}}
\newcommand{\mesref}{\mu_{\alpha}^*}
\newcommand{\modphi}{\widetilde{\varphi}}

\newcommand{\param}{\widehat{\alpha}}

\newcommand{\pset}{\mathcal{M}_1(\ctoruspi)}

\newcommand{\rhoep}{\rho_{\varepsilon N}}
\newcommand{\rdir}{D}

\newcommand{\rb}{\overline{\rho}}

\newcommand{\Sp}{{\mathcal S}}

\newcommand{\drift}{\mathfrak{s}}
\newcommand{\drifth}{\widehat{\mathfrak{s}}}

\newcommand{\scal}[1]{\ll #1\gg_{\widehat{\alpha}}}

\newcommand{\statespace}{{\Sigma}_N}
\newcommand{\statespaceinf}{\Sigma_{\infty}}
\newcommand{\subspace}{\Sigma^{\K}_l}
\newcommand{\torus}{\mathbb{T}_N^2}
\newcommand{\tzero}{{\mathcal{ T}}_0^\omega}
\newcommand{\tzerm}{\tzero}
\newcommand{\tzn}{{\color{purple}\mathcal{T}_{\param}^\omega}}

\newcommand{\zetahat}{\widehat{\zeta}}

\newcommand{\sdc}{d_s}
\newcommand{\ind}[1]{\1_{\{#1\}}}

\newcommand{\ohat}{\widehat{\omega}}
\newcommand{\curohat}{j^{\ohat}}
\newcommand{\tcal}{{T^\omega}}
\newcommand{\tcall}{T_l^\omega}

\newcommand{\bff}[1]{\boldsymbol{\mathfrak{#1}}}
\newcommand{\uf}{\bff{u}}
\newcommand{\ufbar}{\overline{\uf}}

\newcommand{\Ecmnk}{\E_{n,\K}}
\newcommand{\cmnk}{\mu_{n,\K}}
\newcommand{\Ecmnkt}{\widetilde{\E}_{n,K}}

\newcommand{\Dnk}{{\mathscr{D}}_{n,\K}}
\newcommand{\Dnkt}{\widetilde{{\mathscr{D}}}_{n,K}}
\newcommand{\pbar}{\overline{\psi}}
\newcommand{\ptilde}{\widetilde{\psi}}
\newcommand{\Dnkxt}{\widetilde{{\mathscr{D}}}_{n, K}^{x}}
\newcommand{\Dnkxzt}{\widetilde{{\mathscr{D}}}_{n,K}^{x+z}}

\theoremstyle{plain}
\author[C.Erignoux]{Cl\'ement Erignoux}
\address{Universit\`a degli studi Roma Tre, Largo San Leonardo Murialdo 1, 00146 Roma, RM, Italy}
\email{clement[dot]erignoux[at]gmail[dot]com}
\title{Hydrodynamic Limit for an Active Exclusion Process}
\alttitle{Limite hydrodynamique pour un processus d'exclusion actif}
\thanks{\textbf{Acknowledgments. } I would first like to warmly thank Thierry Bodineau, my PhD advisor, for his unwavering support and help writing this article. I would also like to thank Jeremy Quastel for his help navigating his original article and solving the spectral gap issue, Claudio Landim for our numerous discussions on the non-gradient techniques, as well as Julien Tailleur for his insight on active matter and MIPS. I would like to thank the anonymous referee, for both the attention given to my work, and for many insightful comments that significantly improved this article. Finally, I gratefully acknowledge funding from the European Research Council under the European Unions Horizon 2020 Programme, ERC Consolidator GrantUniCoSM (grant agreement no 724939)}
\altkeywords{Physique statistique, Limites hydrodynamiques, Gaz sur r\'eseau, Syst\`emes hors-equilibre, Syst\`emes non-gradients, Processus d'exclusion}

\begin{document}
\selectlanguage{english}
\frontmatter

\selectlanguage{french}
\begin{altabstract}
L'\'etude des dynamiques collectives, observables chez de nombreuses esp\`eces animales, a motiv\'e dans les derni\`eres d\'ecennies un champ de recherche actif et transdisciplinaire. De tels comportements sont souvent mod\'elis\'es par de la mati\`ere active, c'est-\`a-dire  par des mod\`eles dans lesquels chaque individu est caract\'eris\'e par une vitesse propre qui tend \`a s'ajuster selon celle de ses voisins. 

De nombreux mod\`eles de mati\`ere active sont li\'es \`a un mod\`ele fondateur propos\'e en 1995 par Vicsek~\and~al.. Ce dernier, ainsi que de nombreux mod\`eles proches, pr\'esentent une transition de phase entre un comportement chaotique \`a haute temp\'erature, et un comportement global et coh\'erent \`a faible temp\'erature. De  nombreuses preuves num\'eriques de telles transitions de phase ont \'et\'e obtenues dans le cadre des dynamiques collectives. D'un point de vue math\'ematique, toutefois, ces syst\`emes actifs sont encore mal compris. Plusieurs r\'esultats  ont \'et\'e obtenus r\'ecemment sous une approximation de champ moyen, mais il n'y a encore \`a ce jour que peu d'\'etudes math\'ematiques de mod\`eles actifs faisant intervenir des interactions purement microscopiques. 

Dans cet article, nous d\'ecrivons un syst\`eme de particules actives sur r\'eseau interagissant localement pour aligner leurs vitesses. Comme premi\`ere \'etape afin d'atteindre une meilleure compr\'ehension des mod\`eles microscopiques de mati\`ere active, nous obtenons rigoureusement, \`a l'aide du formalisme des limites hydrodynamiques pour les gaz sur r\'eseau, la limite macroscopique de ce syst\`eme hors-\'equilibre. 
Nous d\'eveloppons le travail r\'ealis\'e par Quastel \cite{Quastel1992}, en apportant une preuve plus d\'etaill\'ee et en incorporant plusieurs g\'en\'eralisations posant de nombreuses difficult\'es techniques et ph\'enom\'enologiques. 

\end{altabstract}

\selectlanguage{english}
\begin{abstract}
Collective dynamics can be observed among many animal species, and have given rise in the last decades to an active and interdisciplinary field of study. Such behaviors are often modeled by active matter, in which each individual is self-driven and tends to update its velocity depending on the one of its neighbors.

In a classical model introduced by Vicsek \and al., as well as in numerous related active matter models, a phase transition between chaotic behavior at high temperature and global order at low temperature can be observed. Even though ample evidence of these phase transitions has been obtained for collective dynamics, from a mathematical standpoint, such active systems are not fully understood yet. Significant progress has been achieved in the recent years under an assumption of mean-field interactions, however to this day, few rigorous results have been obtained for models involving purely local interactions.

In this paper, as a first step towards the mathematical understanding of active microscopic dynamics, we describe a lattice active particle system, in which particles interact locally to align their velocities. We obtain rigorously, using the formalism developed for hydrodynamic limits of lattice gases, the scaling limit of this out-of-equilibrium system. 
This article builds on the multi-type exclusion model introduced by Quastel \cite{Quastel1992} by detailing his proof and incorporating several generalizations, adding significant technical and phenomenological difficulties.

\end{abstract}

\maketitle
\tableofcontents
\mainmatter

\section{Introduction}
\subsection{Active matter and active exclusion process}

\emph{Active matter} systems, i.e. microscopic interacting particles models in which each particle consumes energy to self-propel, 
have been the subject of intense scrutiny in physics in the recent years. As explained thoroughly in Appendix \ref{sec:contexte}, 
active matter exhibits a rich phenomenology. Its two most studied features are the emergence of global polarization, first 
discovered with Vicsek's seminal model \cite{Vicsek1995}, and the so-called Motility Induced Phase Separation (MIPS, cf. \cite{CT2015}), 
which can be roughly described as the particle's tendency to cluster where they move more slowly. 
As detailed in Appendix \ref{sec:contexte}, these two phenomena have been extensively studied by the physics community in the last decade  (e.g. \cite{SCB2015} \cite{SChT2015} \cite{ST2015} for alignment phase transition , \cite{CT2013} \cite{CT2015} for MIPS).

By essence, active matter models are driven out-of-equilibrium at a microscopic level, and although many are now well-understood from a physics standpoint, 
their mathematical understanding to this day remains partial. Inspired by Vicsek's original model \cite{Vicsek1995}, significant mathematical progress 
has been achieved using analytical tools for  active alignment models submitted to \emph{mean-field} or \emph{local-field} interactions, i.e. for which the particle's interactions are 
locally averaged out over a large number of their neighbors (e.g. \cite{BCC2011}, \cite{DFL2011}, \cite{DM2007}). However, in some cases,
the local-field approximation is not mathematically justified, and deriving exact results on models with purely microscopic interactions can provide
welcome insight for their phenomenological study \cite{BEKH2017}.

\medskip

Let us start by briefly describing a simplified version of the active exclusion process studied in this article before giving some mathematical context.  
On a two-dimensional  periodic lattice, consider two-types of particles, denoted ''$+$'' and ''$-$'', which move and update their type according to their neighbors.  
\begin{itemize}\item Each particle's type is randomly updated by a Glauber dynamics depending on its nearest neighbors. 
\item The motion of any particle is a random walk, weakly biased in one direction depending on its type : the ''$+$'' particles will tend to move to the right, whereas the ''$-$'' particles will tend to move to the left. 
\item The vertical displacement is symmetric regardless of the particle's type.  
\end{itemize}To model hard-core interactions, an \emph{exclusion rule} is imposed, i.e. two particles cannot be present on the same site : a particle jump towards an occupied site will be canceled. This induces the congestion effects which can lead to MIPS, and one can therefore hope that this model 
encompasses both the alignment phase transition and MIPS which are characteristic of many of the active models described in Appendix \ref{sec:contexte}. However, mathematically proving such phenomenology for our microscopic active model is still out of reach.

\medskip

In this article, as a first step towards this goal, we derive the hydrodynamic limit for an extension of the model briefly described above. 
From a mathematical standpoint, a first microscopic dynamics combining alignment and stirring was introduced in \cite{DMFL1986}, where De Masi et al. considered a lattice gas with two types of particles, in which two neighboring particles can swap their positions, and can change type according to the neighboring particles. They derived the hydrodynamic limit, as well as the fluctuations, when the stirring dynamics is accelerated by a diffusive scaling, w.r.t. the alignment dynamics. 
This scale separation is crucial to have both alignment and stirring present in the hydrodynamic limit. 
Generally, the strategy 
to obtain the hydrodynamic limit for a lattice gas depends significantly on the microscopic features of the model, and must be adapted 
on a case-by-case basis to the considered dynamics. For example, the exclusion rule in the active exclusion process makes it non-gradient, thus the proof of 
its hydrodynamic limit is significantly more elaborate. The end of this introduction is dedicated to putting in context the mathematical 
contributions of this article and  describing the difficulties 
occurring  in the derivation of the hydrodynamic limit of our model.

\subsection{Hydrodynamics limits for non-gradients systems} The active exclusion process presented above belongs to a broad class of microscopic lattice dynamics for which the instantaneous particle currents along any edge cannot be written as a discrete gradient. This difficulty appears naturally in exclusion systems, in particular for systems with multiple particle types, or for generalized exclusion processes where only a fixed number $\kappa$ $ (\kappa\geq 2)$ of particles can be present at the same site.   Such systems are called \emph{non-gradients}. A considerable part of this article is dedicated to solving the difficulties posed by the non-gradient nature the active exclusion process.

The first proof for a non-gradient hydrodynamic limit was obtained by Varadhan in  \cite{Varadhan1994b}, and Quastel \cite{Quastel1992} (cf. below). To illustrate the difficulty let us consider a general diffusive particle system of size $N$ in $1$ dimension, evolving according to a Markov generator $\gene_N$. Such a diffusive system must be rescaled in time by a factor $N^2$, therefore each jump in $\gene_N$ should occur at rate $N^2$. Denoting by $\conf_x$ the state of the system at the site $x$ (e.g. number of particles, energy of the site), $\gene_N \conf_x$  is a microscopic gradient,
\[\gene_N \conf_x=N^2(j_{x-1,x}- j_{x,x+1}),\]
where $j_{x,x+1}$ is the instantaneous current along the edge $(x,x+1)$, and the $N^2$ comes from the time-rescaling. This microscopic gradient balances out a first factor $N$, and acts as a spatial derivative on a macroscopic level. In order to obtain a diffusive equation similar to the heat equation, one needs to absorb the second factor $N$ in a second spatial derivative.  This is the main difficulty for non-gradient systems, for which the instantaneous current $j_{x,x+1}$ does not take the form of a microscopic gradient. The purpose of the non-gradient method developed by Varadhan is to establish a so-called  \emph{microscopic fluctuation-dissipation relation}
\[j_{x,x+1}\simeq -D  (\conf_{x+1}-\conf_x)+\gene_N g_x,\]
where $\gene_N g_x$  is a small fluctuation which usually disappears in the macroscopic limit according to Fick's law for diffusive systems. Although the link to the macroscopic fluctuation-dissipation relation   (cf.  Section 8.8, p140-141 in \cite{SpohnB1991} for more detail on this relation) is not apparent, the latter is indeed a consequence of the microscopic identification above.

\subsection{Multi-type lattice gases, and contributions of this article}

The difficulties to derive the hydrodynamic limit of multi-type particle models vary significantly depending on the  specificities of each microscopic dynamics.
Active matter provides natural examples of multi-type particle systems, since each possible velocity can be interpreted as a different type. When the particles evolve in a continuous space domains, (e.g. \cite{DFL2011}, \cite{DFL2014}) and in the absence of hard-core interactions, the density of each type of particles can essentially be considered independently regarding displacement, and the scaling limit usually decouples the velocity variable and the space variable.
 
In the case of lattice gases, however, it becomes necessary to specify the way particles interact when they are on the same site. Dynamically speaking,  multi-type models often allow either
\begin{itemize}\item  swapping particles with different types, as in \cite{Sasada2010} for a totally asymmetric system with velocity flips.
 \item The coexistence on a same site of particles with different velocities, as in \cite{DMF2015} or \cite{Simas2010} for a model closely related to the one investigated in this article with weak driving forces, or in \cite{DSZ2017} for a zero-range model exhibiting MIPS-like behavior. 
\end{itemize}
These simplifications allow to bypass the specific issues arising for diffusive systems with complete exclusion between particles, since the latter often require the non-gradient tools mentioned previously. 

\medskip

The first hydrodynamic limits for non-gradient microscopic systems were studied by Varadhan and Quastel. They developed in \cite{Varadhan1994b} and \cite{Quastel1992} a general method to derive the hydrodynamic limit for non-gradient systems with main requirement a sharp estimate for the Markov generator's spectral gap. Quastel also notably obtained in \cite{Quastel1992} an explicit expression for the diffusion and conductivity matrices for the multi-type exclusion process, as a function of the various particle densities and of the self-diffusion coefficient $d_s(\rho)$ of a tagged particle for the equilibrium symmetric simple exclusion process with density $\rho$. This result was then partially extended to the weakly asymmetric case (in \cite{QRV1999} as a step to obtain a  large deviation principle for the empirical measure of the symmetric simple exclusion process, and where the asymmetry does not depend on the configuration, and in \cite{GLM2000} for a weak asymmetry with a mean-field dependency in the configuration), as well as a more elaborate dynamics with creation and annihilation of particles \cite{Sasada2009}. 

\bigskip

In this article, we derive the hydrodynamic limit for an active matter lattice gas with purely microscopic interactions. 
To do so, we generalize the results obtained by Quastel \cite{Quastel1992} by incorporating many natural extensions, 
and apply in great detail the non-gradient method for multi-type exclusion with a weak drift. 

There are several reasons behind our choice to detail this difficult proof.
First, Quastel's original article suffers from typos which are fixed in this paper, in particular the spectral gap for the multi-type 
exclusion process is not uniform with respect to the density and this required an adaptation of the original proof.
Second, Quastel's proof relied significantly on the structure of the microscopic dynamics which could be controlled by the symmetric 
exclusion. This played a crucial role in \cite{Quastel1992} to ensure that the particle density does not reach 1, because when this is the case, 
the system loses its mixing properties as represented by the decay of the spectral gap.
 When the considered dynamics is a multi-type symmetric exclusion (identical for any particle type, as in \cite{Quastel1992}), 
 the macroscopic density 
 for the total number of particles evolves according to the heat equation, and density control at any given time is ensured 
 by the maximum principle. In our case, the limiting equation is not diffusive and a priori estimates on the density 
 are much harder to derive.
Finally, \cite{Quastel1992} was one of the first examples of hydrodynamic limit for non-gradient systems, 
and to make the proof more accessible, we used the more recent formalism developed in \cite{KLB1999}, 
in which an important upside is the clear identification of the orders of the estimates in the scaling parameter $N$.

We extend the proof of the hydrodynamic limit for the multi-type exclusion process \cite{Quastel1992} to the weakly asymmetric case when the particle types depend on a continuous parameter. The hydrodynamic limit for lattice gases with $K$ particle types takes the form of $K$ coupled partial differential equations. Extending it to a continuum of particle types therefore poses the issue of the well-posedness of the system. To solve this issue, we therefore introduce an angular variable joint to the space variable. Although the global outline of the proof remains similar, this induced numerous technical difficulties. In particular, as opposed to the previous examples, local equilibrium is not characterized by a finite number of real-valued parameters (e.g. density, local magnetization), which required significant adaptation of the proof of the hydrodynamic limit.

\subsection{Active exclusion process and main result}
\label{subsec:intromodel}

The remainder of this section is dedicated to a short description of our model and its hydrodynamic limit. 
For clarity's sake, we first describe in more details the simplified model with only two types of particles briefly presented above, and then introduce the more general active exclusion process studied in this article. Precisely describing the complete model, and rigorously stating its hydrodynamic limit, will be the purpose of Section \ref{sec:2}.

\subsubsection*{Description of a simplified process with two particle types}

For the clarity of notations, we describe and study our model in dimension $d=2$. The simplified version of the model can be considered as an active Ising model \cite{ST2015} with an \emph{exclusion rule} : each site $x$ of the periodic lattice $\torus$  of size $N$ is either 
\begin{itemize}
\item occupied by a particle of type ``$+$'' ($\conf^+_x=1$),
\item occupied by a particle of type ``$-$'' ($\conf^-_x=1$),
\item empty if  $\conf^+_x=\conf^-_x=0$.
\end{itemize}
Each site contains at most one particle, thus the pair $(\conf^+_x,\conf^-_x)$ entirely determines the state of any site $x$, and is either $(1,0)$, $(0,1)$ or $(0,0)$. The initial configuration for our particle system is chosen at local equilibrium and close to a smooth macroscopic profile $\zeta_0=\zeta^+_0+\zeta^-_0:\ctorus\to [0,1]$, where $\ctorus$ is the continuous domain $[0,1]^2$ with periodic boundary conditions, and $\zeta^+_0(x/N)$ (resp. $\zeta^-_0(x/N)$) is the initial probability that the site $x$ contains a ``$+$'' particle (resp. ``$-$''). We denote by $\confhat$ the collection $((\conf^+_x, \conf^-_x))_{x\in \torus}.$

Each particle performs a random walk, which is symmetric in the direction $i=2$, and weakly asymmetric in the direction $i=1$. The asymmetry is tuned via a positive parameter $\lambda$, thus a ``$+$''  (resp. ``$-$'') particle at site $x$ jumps towards $x+e_1$ at rate $1+\lambda/N$ (resp. $1-\lambda/N$) and towards $x-e_1$ at rate $1-\lambda/N$ (resp. $1+\lambda/N$). If a particle tries to jumps to an occupied site, the jump is canceled. In order to obtain a macroscopic contribution of this displacement dynamics, it must be accelerated by a factor $N^2$.

Moreover, the type of the particle at site $x$ is updated at random times,  depending on its nearest neighbors. Typically, to model collective motion, a ``$-$'' particle surrounded by ``$+$'' particles will change type quickly, whereas a ``$-$'' particle surrounded by ``$-$'' particles will change type slowly, to model the tendency of each individual to mimic the behavior of its neighbors. 
Although they determine the shape of the last term of the hydrodynamic limit, the microscopic details of this update dynamics are technically not crucial to the proof of the hydrodynamic limit (in the scaling considered here), we therefore choose general, bounded flip rates $c_{x,\beta}(\confhat)$ parametrized by an inverse temperature $\beta \geq 0$ and depending only on the local configuration around $x$. 

\bigskip

The complete dynamics can be split into three parts, namely the symmetric and asymmetric contributions of the exclusion process, and the Glauber dynamics, evolving on different time scales. For this reason, each corresponding part in the Markov generator has a different scaling in the parameter $N$ : the two-type process is driven by the generator 
\[L_N=N^2 \cro{ \genas+\frac{1}{N}\genwa}+\genis,\]
whose three elements we now define. Fix a function $f$ of the configuration, we denote by  \[\conf_x=\conf^+_x+\conf^-_x\in \{0,1\}\] the total occupation state of the site $x$.
The nearest-neighbor simple symmetric exclusion process generator $\genas$ is
\begin{equation*}\genas f(\confhat)=\sum_{x\in \torus}\sum_{|z|=1}\conf_x\left(1-\conf_{x+z}\right)\left(f(\confhat^{x,x+z})-f(\confhat)\right),\end{equation*}
$\genwa$ encompasses the weakly asymmetric  part of the displacement process,  
\begin{equation*}\genwa f(\confhat)=\sum_{x\in \torus}\sum_{\delta=\pm1}\delta\lambda(\conf^+_x-\conf^-_x)\left(1-\conf_{x+\delta e_
1}\right)\left(f(\confhat^{x,x+\delta e_1})-f(\confhat)\right),\end{equation*}
which is not a Markov generator because of its negative jump rates, but is well-defined once added to the symmetric part of the exclusion process.
Finally,  $\genis$ is the generator which rules the local alignment of the angles
\begin{equation*}\genis f(\confhat)=\sum_{x\in \torus}\conf_x c_{x,\beta}( \confhat)\pa{f(\confhat^{x})-f(\confhat)}.\end{equation*}
In the identities above, $\confhat^{x,x+z}$ is the configuration where the states of $x$ and $x+z$ have been swapped in $\confhat$,  and $\confhat^{x}$ is the configuration where the type of the particle at site $x$ has been changed.

\subsubsection*{Hydrodynamic limit}

Let us denote by $\rho^{+}_t(u)$ (resp. $\rho^{-}_t(u)$) the macroscopic density of ``$+$''  (resp.``$-$'') particles, and by $\rho_t(u)=\rho^{+}_t(u)+\rho^{-}_t(u)$ the total density at any point $u$ in  $\ctorus$. Let us also denote by $m_t(u)=\rho^{+}_t(u)-\rho^{-}_t(u)$ the local average asymmetry.

Then, as a special case of our main result the pair  $(\rho^{+}_t,\rho^{-}_t)$ is solution, in a weak sense, to the partial differential system
\begin{equation}\label{twotypesequadiff}\left\{\begin{matrix}\partial_t \rho^+_t=\nabla\cdot\cro{\diff(\rho_t,\rho^+_t)\nabla\rho_t+\sdc(\rho_t)\nabla\rho^{+}_t}-2\lambda \partial_{u_1}\cro{m_t\drift(\rho_t,\rho^+_t)+\sdc(\rho_t)\rho^+_t}+\Gamma_t,\\
\partial_t \rho^-_t=\nabla\cdot\cro{\diff(\rho_t,\rho^-_t)\nabla\rho_t+\sdc(\rho_t)\nabla\rho^{-}_t}+2\lambda \partial_{u_1}\cro{m_t\drift(\rho_t,\rho^-_t)-\sdc(\rho_t)\rho^-_t}-\Gamma_t
\end{matrix}\right.\end{equation}
with initial profile \begin{equation}\label{twotypesinitialcond}\rho_0^{\pm}(u)=\zeta^{\pm}(u).\end{equation}
In the PDE \eqref{eqcintro}, $\partial_{u_1}$ denotes the partial derivative in the first space variable, $d_s$ is the self-diffusion coefficient for the SSEP in dimension $2$ mentioned in the introduction, the coefficients $\diff$ and $\drift$ are given by
\begin{equation}
\label{diffconddef}  
\diff(\rho,\rho^*)=\frac{\rho^*}{\rho}(1-\sdc(\rho)) \eqand   \drift(\rho,\rho^*)=\frac{\rho^*}{\rho}(1-\rho-\sdc(\rho)),
\end{equation}
and $\Gamma_t$ is the local creation rate of particles with type ``$+$'', which can be written as the expectation under a product measure of the microscopic creation rate. Although it is not apparent, the coefficients $\diff$, $\drift$, and $d_s$ satisfy a Stokes-Einstein relation in a matrix form when the differential equation is written for the vector $(\rho^{+}_t, \rho^{-}_t)$, in the sense that
\begin{multline*}
\pa{\begin{matrix}
\diff(\rho,\rho^+)+d_s(\rho)& \diff(\rho,\rho^+)\\
\diff(\rho,\rho^-)&\diff(\rho,\rho^-)+d_s(\rho)
\end{matrix}}\pa{\begin{matrix}
\rho^+(1-\rho^+)& -\rho^+\rho^-\\
-\rho^+\rho^-&\rho^-(1-\rho^-)
\end{matrix}}\\
=\pa{\begin{matrix}
\rho^+[\drift(\rho,\rho^+)+d_s(\rho)]& \rho^-\drift(\rho,\rho^+)\\
\rho^+\drift(\rho,\rho^-)& \rho^-[\drift(\rho,\rho^-)+d_s(\rho)]
\end{matrix}}.
\end{multline*}
The second matrix above is the compressibility matrix, whose components are $Cov_{\rho^+, \rho^-}(\conf_0^{s_1}, \conf_0^{s_2})$, where both $s_1$ and $s_2$ take value in $\{+, -\}$.

This simplified model is very close to the active Ising model (cf. Appendix \ref{sec:contexte}, and \cite{ST2015}) with a weak driving force. The main difference is the exclusion rule~: in the active Ising model, there is no limit to the number of particles per site, and each particle's type is updated depending on the other particles present at the same site. In our two-type model, the exclusion rule creates a strong constraint on the displacement and therefore changes the form of the hydrodynamic limit, which is no longer the one derived in \cite{ST2015}.

\subsubsection*{Description of the active exclusion process} 

We now describe the \emph{active exclusion process} considered in this article, which is in some form a generalization of the model presented above. 
Indeed, although for technical reasons the proof of our main result cannot be applied verbatim to a finite number of particle types, the overwhole scheme is exremely similar, and under suitable assumptions on the initial profile, one can state an analogous result in the case of a finite number of particle types as well.
Since the active exclusion process is thoroughly introduced in Section \ref{sec:2}, we briefly describe it here, and only give a 
heuristic formulation for our main result. Denoting 
\[\ctoruspi:=[0,2\pi[,\]
the \emph{periodic} set of possible angles, the type of any particle is now a parameter $\theta\in\ctoruspi$ representing the angular direction of its weak driving force. To compare with the simplified model, the ``$+$'' particles correspond to the angle $\theta=0$, whereas the ``$-$'' particles correspond to the angular direction $\theta=\pi$. 
\index{$\ctoruspi$\dotfill  set of angles $[0,2\pi[$}

Any site is now either occupied by a particle with angle $\theta$ ($\conf_x=1$, $\theta_x=\theta$), or empty ($\conf_x=0$, $\theta_x=0$ by default). 
The initial configuration $\confhat(0)$ of the system is chosen at local equilibrium, close to a smooth macroscopic profile 
$\widehat{\zeta}: \ctorus\times\ctoruspi\to \R_+$,  where each site $x$ is occupied by a particle with angle $\theta_x\in[\theta, \theta+d\theta[$ 
with probability $\widehat{\zeta}( x/N,\theta)d\theta$, and the site remains empty w.p. $1-\int_{\ctoruspi}\widehat{\zeta}(x/N, \theta)d\theta$.

Our active exclusion process is driven by the Markov generator 
\[L_N=N^2 \cro{ \genas+\frac{1}{N}\genwa}+\genis,\]
with three parts described below. Fix a function $f$ of the configuration.
The nearest-neighbor simple symmetric exclusion process generator $\genas$ is unchanged with respect to the two-type case, whereas $\genwa$ is now given by
\begin{equation*}\genwa f(\confhat)=\sum_{x\in \torus}\sum_{\substack{|z|=1\\z=\delta e_i}}\delta\lambda_i(\theta_x)\conf_x\left(1-\conf_{x+\delta e_
i}\right)\left(f(\confhat^{x,x+\delta e_i})-f(\confhat)\right),\end{equation*}
where the asymmetry in the direction $i$ for a particle  with angle $\theta$ is encoded by the functions $\lambda_i(\theta)$, \[\lambda_1(\theta)=\lambda \cos(\theta)\eqand \lambda_2(\theta)=\lambda\sin(\theta).\]
To fix ideas, the Glauber generator will be taken of the form
\begin{equation*}\genis f(\confhat)=\sum_{x\in \torus}\conf_x\int_{\ctoruspi}c_{x,\beta}(\theta, \confhat)\pa{f(\confhat^{x,\theta})-f(\confhat)}d\theta,\end{equation*}
where $\confhat^{x,\theta}$ is the configuration where  $\theta_x$ has been set to $\theta$, and we choose alignment rates similar to the Glauber dynamics of the XY model (cf. Appendix \ref{sec:contexte}). More precisely, we consider
\[c_{x,\beta}(\theta, \confhat)=\frac{\exp\pa{\beta\sum_{y\sim x}\conf_y\cos(\theta_y-\theta)}}{\int_{\ctoruspi}\exp\pa{\beta\sum_{y\sim x}\conf_y\cos(\theta_y-\theta')}d\theta'},\]
which tends to align $\theta_x$ with the $\theta_y$'s, for $y$ a neighbor site of $x$. In the jump rates above, we take the value in $[-\pi,\pi]$ of the angle $\theta_y-\theta$. The intensity $\lambda$ and the inverse temperature $\beta\geq 0$ still tune the strength of the drift and the alignment. 

As mentioned before, we settle for now for a heuristic formulation of the hydrodynamic limit. Let us denote by $\rho^{\theta}_t(u)$ the macroscopic density of particles with angle $\theta$, and by $\rho_t(u)=\int_{\theta} \rho^{\theta}_t(u) d\theta$ the total density at any point $u$ in the \emph{periodic domain} $\ctorus :=[0,1]^2$. Let us also denote by $\overset{\rightarrow}{\Omega}_t$ the  direction of the local average asymmetry
\[\overset{\rightarrow}{\Omega}_t(u)=\int_{\ctoruspi}\rho^{\theta}_t(u)\pa{\begin{matrix}\cos(\theta)\\
\sin(\theta)\end{matrix}}d\theta.\]
As expected from \eqref{twotypesequadiff}, the main result (cf. Theorem \ref{thm:mainthm}) of this article is that $\rho^{\theta}_t$ is solution, in a weak sense, to the partial differential equation
\begin{equation}\label{eqcintro}\partial_t \rho^{\theta}_t=\nabla\cdot\cro{\diff(\rho_t,\rho^{\theta}_t)\nabla\rho_t+\sdc(\rho_t)\nabla\rho^{\theta}_t}-2 \nabla\cdot\cro{\drift(\rho_t,\rho^{\theta}_t)\lambda\overset{\rightarrow}{\Omega}_t+\sdc(\rho_t)\rho^{\theta}_t\pa{\begin{matrix}\lambda_1(\theta)\\
\lambda_2(\theta)\end{matrix}}}+\Gamma_t,
\end{equation}
with initial profile \[\rho_0^{\theta}(u)=\widehat{\zeta}(u, \theta).\]
In the PDE \eqref{eqcintro}, $d_s$ is the self-diffusion coefficient for the SSEP in dimension $2$ mentioned previously, the coefficients $\diff$ and $\drift$ are given by \eqref{diffconddef} as in the two-type case, and $\Gamma_t$ is the local creation rate of particles with angles $\theta$, which can be written as the expectation under a product measure of the microscopic creation rate. 

Before properly stating the hydrodynamic limit, let us recall the major difficulties of the proof. The main challenge is the non-gradient nature of the model~: the instantaneous  current of particles with angle $\theta$ between two neighboring sites  $x$ and $x+e_i$ can be written
\[\cur^{\theta}_{x,x+e_i}=\1_{\{\theta_x=\theta\}}\conf_{x}(1-\conf_{x+e_i})-\1_{\{\theta_{x+e_i}=\theta\}}\conf_{x+e_i}(1-\conf_{x}),\] 
which is not a discrete gradient.
One also has to deal with the loss of ergodicity at high densities, and with the asymmetry affecting the displacement of each particle, which drives the system out-of-equilibrium, and complicates the non-gradient method. Finally, the non-linearity of the limiting equation also induces several difficulties throughout the proof.

\subsubsection*{Model extensions}

Several design choices for the model have been made either to simplify the notations, or to be coherent with the collective dynamics motivations (cf. Appendix \ref{sec:contexte}). However, we present now some of the possible changes for which our proof still holds with minimal adaptations.
\begin{itemize}
\item The model can easily be adapted to dimensions $d\geq 2$. The dimension $1$, however, exhibits very different behavior, since neighboring particles with opposite drifts have pathological behavior and freeze the system due to the exclusion rule.
\item {The nearest neighbor jumps dynamics can be replaced by one {with local and irreducible} transition function $p(\cdot)$. This involves minor adjustments of the limiting equation, as solved by Quastel \cite{Quastel1992}. In this case, the total jump generator must be split between a symmetric part scaled as $N^2$, and an asymmetric part scaled as $N$ whose jumps can be decomposed as a succession of jumps from the symmetric part. 
However, providing exact criteria for the validity of the extension to a more general jump kernel would be rather difficult, and such extensions are best checked on a case-by-case basis.
In the case of nearest-neighbor exclusion, the drift functions} can be replaced by any bounded function, and can also involve a spatial dependency, 
as soon as $\lambda_i(u,\theta)$ is a smooth {$C^{1,1}$ function of its two variables $u$ and $\theta$. 
}
\item We chose for our alignment dynamics a jump process, however analogous results would hold for diffusive alignment. 
The jump rates can also be changed to any local and bounded rates, provided they are smooth in the $ \theta_x$'s, and that the overall realignment rate $\int_{\ctoruspi}c_{x,\beta}(\theta, \confhat)d\theta$ only depends on the configuration $\confhat$ through the occupational variable $\eta_x$. The smoothness assumption in the last two comments is there to make sure that the expectation of their microscopic contribution under the grand-canonical measures 
is a Lipschitz-continuous function in the grand-canonical parameter.

\end{itemize}

\subsection{Structure of the article}

Section \ref{sec:2} is dedicated to the full description of the model, to introducing the main notations, and the proper formulation of the hydrodynamic limit for the active exclusion process.

Section \ref{sec:3} is composed of three distinct parts. In Subsection \ref{subsec:canonicalmeasures} we characterize local equilibrium for our process by introducing the set $\pset$ of parameters for the grand-canonical measures of our process. We also give a topological setup for $\pset$, for which some elementary properties are given in Appendix \ref{sec:B}. In Subsection  \ref{subsec:entropy}, we prove using classical tools that the entropy of the measure of our process with respect to a reference product measure is of order $N^2$. The last Subsection \ref{subsec:irreducibility} tackles the problem of irreducibility, which is specific to our model and is one of its major difficulties. Its main result, Proposition \ref{prop:fullclusters}, relies on a-priori density estimates, and states that on a microscopic scale, large local clusters are seldom completely full, which is necessary to ensure irreducibility on a microscopic level.

Section \ref{sec:4} proves a law of large numbers for our process. The so-called Replacement Lemma stated in Subsection \ref{gradientreplacement} relies on the usual one block (Subsection \ref{lem:OBE}) and two blocks (Subsection \ref{subsec:TBE}) estimates. However, even though we use the classical strategy to prove both estimates, some technical adaptations are necessary to account for the specificities of our model.

Section \ref{sec:5} acts as a preliminary to the non-gradient method. The first result of this section is the comparison of the active exclusion process's 
measure to that of an equilibrium process without drift nor alignment 
(Subsection \ref{subsec:FeynmanKac}). We also prove, adapting the classical methods, a  compactness result for the sequence of 
measures of our process, (Subsection \ref{subsec:compactnessQN}) as well as an energy estimate (Subsection \ref{subsec:regularity}) necessary 
to prove our main result.

The non-gradient estimates are obtained in Section \ref{sec:6}. It is composed of a large number of intermediate results which we 
do not describe in this introduction. The application of the non-gradient method to the active exclusion process, however, requires to overcome several issues which are 
specific to our model. 
One such difficulty is solved in Subsection \ref{subsec:k2}, where we estimate the contributions of microscopic full clusters.
In  Subsections \ref{subsec:halpha} and \ref{subsec:drift}, we prove that for our well chosen diffusion and conductivity coefficients, 
the total displacement currents can be replaced by the sum of a gradient quantity and the drift term. For the sake of clarity, we use 
to do so the modern formalism for hydrodynamic limits as presented in \cite{KLB1999} rather than the one used in \cite{Quastel1992}. 
We state in this section a convergence result at the core of the 
non-gradient method (Theorem \ref{thm:limcovariance}) whose proof is intricate and is postponed to the last section.

All these results come together in Section \ref{sec:7}, where we conclude the proof of the hydrodynamic limit for our process. 
Some more specific work is necessary in order to perform the second integration by parts, due to the delicate shape of the diffusive 
part of our limiting differential equation.

Finally, Section \ref{sec8} is dedicated to proving Theorem \ref{thm:limcovariance}, following similar steps as in \cite{KLB1999}. To do so, we estimate in Subsection \ref{subsec:spectralgap} the spectral gap of the active exclusion process on a subclass of functions. We then describe in Subsection \ref{subsec:differentialforms} the notion of 
germs of closed forms for the active exclusion process, and prove using the spectral gap estimate a decomposition theorem for the set of germs of closed forms.  A difficulty of this model is that the spectral gap is not uniform in the density, and decays faster as the density goes to $1$. This issue is solved by cutting off large densities (cf. equation \eqref{def:phin} and Lemma \ref{lem:bulkconvergence}). 
Using the decomposition of closed forms, Theorem \ref{thm:limcovariance} is derived in Subsection \ref{sec:C}.

\section{Notations and Main theorem}
\label{sec:2}
{\it We describe an interacting particle system, where a particle follows an exclusion dynamics with a weak bias depending on an angle associated with this particle. At the same time, each particle updates its angle according to the angles of the neighboring particle. We study the macroscopic behavior of the corresponding 2-dimensional system with a periodic boundary condition. }

\subsection{Main notations and introduction of the Markov generator}
\index{$\conf_x$\dotfill occupation state of the site $x$}
\index{$\theta_x$\dotfill angle of the particle in $x$}
\index{$\confhat_x$\dotfill  the pair $(\conf_x, \theta_x)$}
On the two dimensional discrete set \index{$\torus$\dotfill discrete torus of size $N$}\[\torus=\{1,\ldots ,N\}^2\]
with \emph{periodic boundary conditions,} we define the occupation configuration $\conf=(\conf_x)_{x\in \torus}\in\{0,1\}^{\torus}$ where $\conf_x\in\{0,1\}$ is the number of particles at site $x$. With any \emph{occupied} site $x\in \torus$, we associate an angle $\theta_x\in\ctoruspi$ representing the mean direction of the  velocity in the plane of the particle occupying the site. 
When the site $x$ is empty, we set the angle of the site to  $\theta_x =0$ by default.
\begin{defi}[Configurations, cylinder \& angle-blind functions]\label{def:conf}
For any site $x\in \torus$, we denote by $\confhat_x$ \index{$\confhat$\dotfill family of the $\confhat_x$, $x\in \torus$} the pair $(\eta_x, \theta_x)$, and by  $\confhat=(\confhat_x)_{x\in \torus}$ the complete configuration. The set of all configurations will be denoted by \index{$\statespace$ \dotfill set of configurations on $\torus$} 
\[\statespace=\left\{(\conf_x, \theta_x)_{x\in \torus}\in \pa{ \{0,1\}\times \ctoruspi}^{\torus} \; \Big|\; \theta_x=0 \mbox{ if }\conf_x=0\right\}.\]
Denote by $\Sigma_{\infty}$ the set of infinite configurations above, where $\torus$ is replaced by $\Z^2$. We will call \emph{cylinder function} any function $f$ depending on the configuration only through a finite set of vertices $B_f\subset\Z^2$, and $C^1$ w.r.t. each $\theta_x$, for any $x\in B_f$. The set of cylinder functions on $\Z^2$ will be denoted {by} ${\mathcal C}$\index{$ {\mathcal C}$ \dotfill set of cylinder functions}. Note that a cylinder function is always bounded, and that any function $f\in{\mathcal C}$ admits a natural image as a function on $\statespace$ for any $N$ large enough. This is always the latter that we will consider, and we therefore abuse the notation and denote in the same way both $f$ and its counterpart on $\statespace$.

We will call \emph{angle-blind function} any function depending on $\confhat$ only through the occupation variables ${\conf=(\conf_x)_{x\in\torus}}$. In other words, an angle-blind function depends on the position of particles, but not on their angles. We denote by $\Sp$ the set of angle-blind functions\index{$\Sp $ \dotfill set of angle-blind functions}. 
 \end{defi}
We will use on the discrete torus the notations $\abss{\cdot}$ for the norm $\abss{x}={\sum_{i=1}^2} \abss{x_i}$.
\index{$|z| $ \dotfill $\sum_i \abss{z_i}$ }
\index{$\lambda $ \dotfill real parameter tuning the asymmetry}

\bigskip

Let $T$ be a fixed time, we now introduce the process $(\confhat(t))_{t\in [0,T]}$ on $\statespace$ which is central to our work.  
Our goal is to combine the two dynamics present in Viscek's model \cite{Vicsek1995}~: 
The first part of the process is the  \emph{displacement dynamic{s}}, which rules the \emph{motion of each particle}. The moves occur at rates biased by the angle of the particle, and follows the exclusion rule. 
Thus, for $\delta=\pm1$ the rate $p_x(\delta e_i, \confhat)$ at which the particle at site $x$ moves to an \emph{empty site} $x+\delta e_i$, letting $e_1=(1,0)$, $e_2=(0,1)$ be the canonical basis in $\Z^2$, is given by
\[p_x(\delta e_i, \confhat)=\left\{\begin{array}{ccc}
1 +\lambda \delta\cos(\theta_x)/N &\mbox{  if  } &i=1\\
1 +\lambda \delta\sin(\theta_x)/N &\mbox{  if  } &i=2\\          
\end{array}\right.
,\]
where $\lambda\in \R$ is a positive parameter which characterizes the strength of the asymmetry. For convenience, we will denote throughout the proof
\index{$\lambda_i(\theta) $ \dotfill strength of the asymmetry in the direction $i$  \
 {\color{white}aaaaaaaaaaaaaaa }\hfill on a particle with angle $\theta$ }
\begin{equation}\label{lidef}\lambda_1(\theta)=\lambda\cos(\theta)\eqand \lambda_2(\theta)=\lambda\sin(\theta).\end{equation} 
 
The previous rates indicate that the motion of each particle is biased in a  direction given by its angle. The motion follows an exclusion rule, which means that if the target site is already occupied, the jump is canceled. Note that in order to see the symmetric and asymmetric  contributions in the diffusive scaling limit, we must indeed choose an asymmetry scaling as $1/N$. Furthermore, in order for the system to exhibit a macroscopic behavior in the limit $N\to \infty$, we need to accelerate the whole exclusion process by $N^2$, as discussed further later on.

The second part of the dynamic is the angle update process, which will be from now on referred to as the \emph{Glauber part of the dynamic{s}}. A wide variety of choices is available among  discontinuous angle dynamics (jump process) and  continuous angle dynamics (diffusion). We choose here a Glauber jump process with inverse temperature $\beta\geq0$ described more precisely below. 

The generator of the complete Markov process is given by
\index{$L_N$ \dotfill complete generator of the active exclusion process}
\begin{equation}\label{defgenecomplet}L_N =N^2 \genex+\genis,\end{equation} 
where 
\index{$\genex$ \dotfill displacement part of $L_N$}
\index{$\genas$ \dotfill symmetric part of $L_N$}
\begin{equation}\label{defgenex}\genex= \genas+\frac{1}{N}\genwa\end{equation} is the generator for the displacement process (which two parts are defined below) and $\genis$ is the generator of the Glauber dynamics. The process can therefore be decomposed into three distinct parts, with different scalings in $N$, namely the symmetric part of the motion, with generator $N^2\gene$, the asymmetric contribution to the displacement generator $N\genwa$ with parameter $\lambda\geq0$, and finally the angle-alignment with generator $\genis$ and inverse temperature $\beta\geq 0$, which are defined for any {cylinder (and therefore $C^1$ in the angular variables, cf. Definition \ref{def:conf})} function $f:\statespace\to \R$, by
\index{$\genwa$ \dotfill weakly asymmetric part of $L_N$}
\begin{equation}\label{gensymdef}\genas f(\confhat)=\sum_{x\in \torus}\sum_{|z|=1}\conf_x\left(1-\conf_{x+z}\right)\left(f(\confhat^{x,x+z})-f(\confhat)\right),\end{equation}
\index{$\genis$ \dotfill Glauber part of $L_N$}
\begin{equation*}\label{genwadef}\genwa f(\confhat)=\sum_{x\in \torus}{\sum_{\substack{\delta=\pm 1\\i=1,2}}}\delta\lambda_i(\theta_x)\conf_x\left(1-\conf_{x+\delta e_i}\right)\left(f(\confhat^{x,x+\delta e_i})-f(\confhat)\right),\end{equation*}
\begin{equation}\label{genisdef}\genis f(\confhat)=\sum_{x\in \torus}\conf_x\int_{\ctoruspi}c_{x,\beta}(\theta, \confhat)\pa{f(\confhat^{x,\theta})-f(\confhat)}d\theta.\end{equation}
Note that $\genwa$ alone is not a Markov generator due to the negative jump rates, but considering the complete displacement generator $\genas+N^{-1}\genwa$ solves this issue for any $N$ large enough. 
In the expressions above, we denoted $\confhat^{x,x+z}$ the configuration where the occupation variables $\confhat_x$ and $\confhat_{x+z}$ at sites $x$ and  $x+z$ have been exchanged in $\confhat$
\index{$\beta$ \dotfill inverse temperature for $\genis$}
\index{$\confhat^{x,y}$ \dotfill $\confhat$ after inversion of $\confhat^{x}$ and $\confhat^{y} $}
\[\confhat^{x,x+z}_y=\left\{\begin{array}{ccc}
\confhat_{x+z} &\mbox{  if  } &y=x,\\
\confhat_x &\mbox{  if  } &y=x+z,\\
\confhat_y &\mbox{  otherwise,}& \\            
\end{array}\right.
\]
and $\confhat^{x,\theta}$ the configuration where the angle $\theta_x$ in $\confhat$ has been updated to $\theta$
\index{$\confhat^{x,\theta}$ \dotfill $\confhat$ after setting $\theta_x=\theta$}
\[\confhat^{x,\theta}_y=\left\{\begin{array}{ccc}
(\conf_y, \theta) &\mbox{  if  } &y=x,\\
\confhat_y &\mbox{  otherwise.}& \\            
\end{array}\right.
\]
\index{$c_{x,\beta}$ \dotfill jump rates for $\genis$}
For $x,y\in \torus $, we write $x\sim y$ iff $|x-y|=1$.  We choose for $c_{x,\beta}$ the jump rates 
\[c_{x,\beta}(\theta, \confhat)=\frac{\exp\pa{\beta\sum_{y\sim x}\conf_y\cos(\theta_y-\theta)}}{\int_{\ctoruspi}\exp\pa{\beta\sum_{y\sim x}\conf_y\cos(\theta_y-\theta')}d\theta'},\]
which tend to align the angle in $x$ with the neighboring particles according to XY-like jump rates (cf. Appendix \ref{sec:contexte}) with inverse temperature {$\beta\geq 0$}. Note that by construction, for any non-negative $\beta$, $\int_{\ctoruspi}c_{x,\beta}(\theta, \confhat)d\theta=1$ and that the jump rates $c_{x,\beta}(\theta, \confhat)$ can be uniformly bounded from above and below by two positive constants depending only on $\beta$.

The process defined above will be referred to as active exclusion process.

\subsection{Measures associated with a smooth profile and definition of the Markov process}
\label{subsec:processdef}
We now introduce the important measures and macroscopic quantities appearing in the expression of the hydrodynamic limit. 
Let us denote by $\ctorus$ the continuous periodic domain in dimension $2$\index{$\ctorus$ \dotfill continuous $2$-dimensional torus}, 
\[\ctorus={[0,1)}^2.\]
\begin{defi}[Density profile on $\ctorus$]\label{def:densityprofile}{We denote by $\pset$ the set of non-negative measures $\param$ on $\ctoruspi$ with total mass $\param(\ctoruspi)$ in $[0,1]$}. We call
\index{$\boldsymbol \dens$ \dotfill density profile on the torus} 
\index{$\rho(u)$ \dotfill mass of the measure $\boldsymbol \dens(u,\cdot)$}
{\emph{density profile on the torus} any function 
\[\boldsymbol \dens :(u, d\theta)\mapsto \boldsymbol \dens(u, d\theta)\] 
such that $\boldsymbol \dens(u, .)\in \pset$ $\forall u\in \ctorus$.}
For any density profile $\boldsymbol\dens$ on the torus, $ \boldsymbol\dens(u, d\theta)$ represents the local density in $u$ of particles with angle in $d\theta$, and $\rho(u)$ represents the total density of particles in $u$.
\end{defi}
\begin{defi}[Measure associated with a density profile on the torus]
\label{defi:measuresforaDPT}\index{$\mu_{\boldsymbol \dens}^N$ \dotfill product measure on $\statespace$ associated with $\boldsymbol \dens$}
To any density profile on the torus $\boldsymbol \dens$, we associate $\mu^N_{\boldsymbol \dens}$, the product measure on $\statespace$ such that the distribution of $\confhat_x$ is given for any $x\in \torus$ by 
\begin{equation}\label{refmeasures}\begin{cases}\mu^N_{\boldsymbol\dens}(\conf_x=0 )=1-\rho(x/N),\\
\mu^N_{\boldsymbol\dens}(\conf_x=1 )=\rho(x/N),\\
\mu^N_{\boldsymbol\dens}(\theta_x\in d\theta\mid \conf_x=1 )=\boldsymbol\dens(x/N,d\theta)/\rho(x/N),\end{cases}\end{equation}
and such that $\confhat_x, \confhat_y$ are independent as soon as $x\neq y$. 
\end{defi}
In other words, under $\mu^N_{\boldsymbol\dens}$, the probability that a site $x\in \torus$ is occupied is $\rho(x/N)=\int_{\ctoruspi}\boldsymbol\dens(x/N,\theta)d\theta\in [0,1]$. Furthermore,  the angle of an empty site is set to $0$ by default, and the angle of an occupied site $x$ is distributed according to the probability distribution $\boldsymbol\dens(x/N,\cdot)/\rho(x/N)$. 

\bigskip

\paragraph{Definition of the process}
Let \index{$\statespace^{[0,T]}$ \dotfill space of c\`adl\`ag trajectories on $\statespace$}  $\statespace^{[0,T]}:=D([0,T], \statespace)$ 
denote the space of right-continuous and left-limited (c\`adl\`ag) trajectories $\confhat:t\to \confhat(t)$. 
We will denote by $\confhat^{[0,T]}$ the elements of $\statespace^{[0,T]}$. For any initial measure $\nu$ on ${\statespace}$, any non-negative drift $\lambda\leq N$ (to make the displacement operator $\gene+N^{-1}\genwa$ a Markov generator), and any $\beta\geq0$, we write $\Prob^{\lambda, \beta}_{\nu}$ for the measure on $\statespace^{[0,T]}$ starting from the measure $\confhat(0)\sim\nu$, and driven by the Markov generator $L_N=L_N(\lambda, \beta)$ described earlier. We denote by $\E^{\lambda,\beta}_{\nu}$ the expectation w.r.t. $\Prob^{\lambda,\beta}_{\nu}$. In the case  $\lambda=\beta=0$, there is no drift and the angle of the particles are chosen uniformly in $ \ctoruspi$. In this case, we will omit $\lambda$ and $\beta$ in the previous notation and write $\Prob_{\nu}$ for the measure and $\E_{\nu}$ for the corresponding expectation.
\index{$\confhat^{[0,T]}$ \dotfill element of $\statespace^{[0,T]}$}
\index{$\Prob_{\nu}^{\lambda,\beta}$ \dotfill measure of active exclusion process $(\lambda,\beta)$ started from $\nu$}
Let us now define the initial measure from which we start our process. Let $\widehat{\zeta} \in C(\ctorus\times \ctoruspi)$ be a continuous 
{non-negative} function on $\ctorus\times\ctoruspi$, which will define the initial macroscopic state of our particle system. We assume that for any $u\in \ctorus$,
\index{$\boldsymbol \dens_0$ \dotfill initial density profile on the torus}
\begin{equation}\label{assumption0} \zeta(u):=\int_{\ctoruspi}\widehat{\zeta}(u,\theta)d\theta<1, \end{equation}
i.e. that the initial density is less than one initially everywhere on  $\ctorus$.
This assumption is crucial, because when the local density hits one, because of the exclusion rule, the system loses most of its mixing properties. At density $1$, mixing only comes from the (slow, because of the scaling) Glauber dynamics, which is not sufficient to ensure that local equilibrium is preserved. 
\index{$\mu^N$ \dotfill initial measure of the active exclusion process, fitting $\boldsymbol\dens_0$}

We can now define the initial density profile on the torus $\boldsymbol \dens_0$ by 
\index{$\widehat{\zeta}$ \dotfill initial macroscopic profile}
\begin{equation}\label{initdens}\boldsymbol \dens_0(u, d\theta)=\widehat{\zeta}(u,\theta)d\theta.\end{equation}
We start our process from a random configuration 
\begin{equation}\label{mesinit}\confhat(0)\sim \mu^N:=\mu_{\boldsymbol\dens_0}^N\end{equation}
 fitting the profile $\boldsymbol\dens_0$, according to Definition \ref{defi:measuresforaDPT}.  Given this initial configuration, we define the Markov process $\confhat^{[0,T]}\in \statespace^{[0,T]}\sim\Prob_{\mu^{N}}^{\lambda, \beta}$ driven by the generator $L_N$ introduced in \eqref{defgenecomplet},  starting from $\mu^N$.

\bigskip

\paragraph{Topological setup}
Let  us denote by $\meas(\ctorus\times\ctoruspi)$ the space of {non-negative} measures on the continuous configuration space {endowed with the weak topology}, and \begin{equation}\label{Mhatdef}\Mhat=D\pa{[0,T],\meas(\ctorus\times\ctoruspi)}\end{equation} the space of right-continuous and left-limited trajectories of measures on $\ctorus\times \ctoruspi$. Each trajectory $\confhat^{[0,T]}$ of the process admits a natural image in $\Mhat$ through its empirical measure 
\begin{equation*}
\pi_t^N\pa{\confhat^{[0,T]}}=\frac{1}{N^2}\sum_{x\in \torus}\conf_x(t)\delta_{x/N, \theta_x(t)}.
\end{equation*}
{We further define the projection $\pi^N$, which associates to $\confhat^{[0,T]}$ the trajectory $t\mapsto \pi_t^N\pa{\confhat^{[0,T]}}$.}
\index{$\meas(\ctorus\times \ctoruspi)$ \dotfill space of measures on $\ctorus\times\ctoruspi$}
\index{$\meas^{[0,T]}$ \dotfill space of c\`adl\`ag traj. on $\meas(\ctorus\times \ctoruspi)$}
\index{$\pi^N_t$ \dotfill empirical measure at time $s$}
We endow  $\Mhat$  with Skorohod's metric defined in Appendix \ref{subsec:topo}, and the set $\mathcal{P}(\Mhat)$ of probability measures on $\Mhat$ with the weak topology. We now define $Q^N\in \mathcal{P}(\Mhat)$ the distribution of {the trajectory of the  empirical measure $\pi^N\pa{\confhat^{[0,T]}}$} of  our process $\confhat^{[0,T]}\sim \Prob^{\lambda,\beta}_{\mu^N}$.
\index{$Q^N$ \dotfill distribution of $(\pi^N_t)_{t\in[0,T]}$ for the active exclusion process}

\subsection{Hydrodynamic limit}
\label{subsec:mainthm}

\paragraph{Self-diffusion coefficient}

The hydrodynamic limit for our system involves the diffusion coefficient of a tagged particle for symmetric simple exclusion process (SSEP) in dimension $2$. Let us briefly remind here its definition. On $\Z^2$, consider an infinite equilibrium SSEP with density $\rho$ and a tagged particle placed at time $0$ at the origin. We keep track of the position  $X(t)=(X_1(t), X_2(t))\in \Z^2$ of the tracer particle at time $t$ and denote by  $Q^*_{\rho}$ the measure of the process starting with measure $\mu_{\rho}$ on $\Z^2\setminus\{0\}$ and a particle at the origin.
\begin{defi}[Self-Diffusion coefficient]
The self-diffusion coefficient $\sdc(\rho)$ is defined as the limiting variance of the tagged particle 
\[ \sdc(\rho):= \lim_{t\to \infty}\frac{\E_{Q^*_{\rho}}(X_1(t)^2)}{t}.\]
\index{$\sdc$ \dotfill self-diffusion coefficient}
\end{defi}
{The existence of this limit is a consequence of \cite{KV1986}.} A variational formula for $d_s$ has been obtained later by Spohn \cite{Spohn1990}. The regularity of the self-diffusion coefficient was first investigated in  \cite{Varadhan1994},  where  Varadhan shows that the self-diffusion matrix is Lipschitz-continuous in any dimension $d\geq 3$. Landim, Olla and Varadhan since then proved in \cite{LOV2001} that the self-diffusion coefficient is in fact of class $C^{\infty}$ in any dimension. The matter of self-diffusion being treated in full detail in Section 6, p199-240 of \cite{KLOB2012}, we do not develop it further here. We summarize in appendix \ref{subsec:sdc} some useful results on the matter. 

\paragraph{Diffusion, conductivity and alignment coefficients}
Given a density profile on the torus $\boldsymbol \dens(u, d\theta)$, recall from Definition \ref{def:densityprofile} that 
$\rho(u)=\int_{\ctoruspi}\boldsymbol \dens(u, d\theta)$ is the local density. We introduce the coefficients
\index{$\diff$ \dotfill diffusion coefficient relative to $\nabla \rho$}
\index{$\drift$ \dotfill conductivity coefficient}
\index{$\overset{\to}{\Omega}$ \dotfill local direction of the asymmetry}
\index{$\Gamma$ \dotfill local creation rate of $\theta$-particles}
\begin{equation*}
\diffh({\rho,\boldsymbol \dens})(u, d\theta)=\frac{\boldsymbol \dens(u, d\theta)}{\rho(u)}(1-\sdc(\rho(u))){{\bf{1}}_{\{\rho(u)>0\}}}, \quad\drifth({\rho,\boldsymbol \dens})(u,d\theta)=(1-\rho(u)-\sdc(\rho(u)))\frac{\boldsymbol \dens(u, d\theta)}{\rho(u)}{{\bf{1}}_{\{\rho(u)>0\}}},
\end{equation*}
where $\sdc$ is the self-diffusion coefficient described in the previous paragraph. We also define  $\overset{\rightarrow}{\Omega}(\boldsymbol \dens)$,  the vector representing the mean direction of the asymmetry under $\boldsymbol \dens$, 
\[\overset{\rightarrow}{\Omega}(\boldsymbol \dens)(u)=\int_{\ctoruspi}\boldsymbol \dens(u, d\theta')\pa{\begin{matrix}
\cos(\theta') \\ 
\sin(\theta')
\end{matrix}}.\] 
as well as $\Gamma(\boldsymbol \dens)$ the local creation {and annihilation} rate of particles with angle $\theta$
\[\Gamma(\boldsymbol \dens)(u,d\theta)=\rho(u)\E_{\boldsymbol \dens(u,\cdot)}\cro{c_{0,\beta}(\theta,\confhat)}d\theta-\boldsymbol \dens(u, d\theta),\]
where under $\E_{\boldsymbol \dens(u,\cdot)}$, each site is occupied independently w.p. $\rho(u)$, and the angle of each particle is chosen according to the probability distribution $\boldsymbol \dens(u,\cdot)/\rho(u)$. The precise definition of $\E_{\boldsymbol \dens(u,\cdot)}$ is given  just below in Definition \ref{defi:GCM}. 

\paragraph{Weak solutions of the PDE}In order to state the hydrodynamic limit of our system, we need to describe the notion of weak solutions in our case, which is quite delicate because of the angles. {For any measure $\pi\in \meas(\ctorus\times\ctoruspi)$ and any function $H:\ctorus\times\ctoruspi\to \R$ integrable w.r.t. $\pi$, we shorten $<\pi, H>=\int_{\ctorus\times\ctoruspi}H(u, \theta)d\pi(du, d\theta)$.}
\index{$G_t(u)$ \dotfill smooth function on $[0,T]\times \ctorus$}
\index{$\omega$ \dotfill smooth function on $\ctoruspi$}
\index{$H_t(u,\theta)$ \dotfill smooth function on $[0,T]\times \ctorus\times \ctoruspi$}
\begin{defi}[Weak solution of the differential equation]
\label{defi:weaksol}
Any trajectory of measures $(\pi_t)_{t\in[0,T]}\in \Mhat$ will be called a \emph{weak solution of the differential system}
\begin{equation}\label{equadiff}
\begin{cases}
 \quad \partial_t \boldsymbol\dens_t=\nabla\cdot\cro{\diffh({\rho_t,\boldsymbol\dens_t})\nabla\rho_t+\sdc(\rho_t)\nabla\boldsymbol\dens_t}-2\lambda \nabla\cdot\cro{\drifth({\rho_t,\boldsymbol\dens_t})\overset{\rightarrow}{\Omega}_t+\boldsymbol\dens_t\sdc(\rho_t)\pa{\begin{matrix}
\cos(\theta) \\ 
\sin(\theta)
\end{matrix}}}+\Gamma(\boldsymbol \dens_t)\\
\quad \dens_0(u,d\theta)=\widehat{\zeta}(u, \theta)d\theta
\end{cases}
,\end{equation}
if the following four conditions are satisfied~:
\begin{enumerate}[i)]
\item{$\pi_0(du, d\theta)=\widehat{\zeta}(u,\theta)dud\theta$}
\item{for any fixed time $t\in [0,T]$, the measure $\pi_t$ is absolutely continuous in space w.r.t. the Lebesgue measure on $\ctorus$, i.e. there exists a  density profile on the torus (in the sense of Definition \ref{def:densityprofile}) $\boldsymbol \dens_t$, such that \[\pi_t(du, d\theta)=\boldsymbol \dens_t(u, d\theta) du.\]}
\item{Letting $\rho_t(u)=\int_{\ctoruspi}\boldsymbol \dens_t(u, d\theta) $, $\rho$ is in $H^1([0,T]\times \ctorus)$, i.e. there exists a family of functions $\partial_{u_i}\rho_t$ in $L^2([0,T]\times \ctorus)$ such that for any smooth function $G\in C^{0,1}([0,T]\times \ctorus)$,
\[\int_{[0,T]\times \ctorus}\rho_t(u) \partial_{u_i}G_t(u)dt du=-\int_{[0,T]\times \ctorus}G_t(u) \partial_{u_i}\rho_t(u)dt du\]}
\item{For any function $H\in C^{1,2,1}([0,T]\times \ctorus\times\ctoruspi)$, 
\begin{multline*}<\pi_T,H_T>-<\pi_0,H_0>=\int_0^T<\pi_t,\partial_tH_t>dt\\
+\int_0^T\int_{\ctorus\times \ctoruspi}\Bigg[\sum_{i=1}^2\bigg(-\partial_{u_i}H_t(u, \theta)\big[\diffh({\rho_t,\boldsymbol \dens_t}) -\sdc'(\rho_t)\boldsymbol \dens_t\big](u, d\theta)\partial_{u_i}\rho_t(u)+\partial_{u_i}^2 H_t(u, \theta)\sdc(\rho_t) \boldsymbol \dens_t(u, d\theta)\\
+\partial_{u_i}H_t(u, \theta) \cro{2\lambda\drifth({\rho_t,\boldsymbol \dens_t}){\Omega_i}(\boldsymbol \dens_t)+2\lambda_i(\theta) \sdc(\rho_t) \boldsymbol\dens_t}(u,d\theta)\bigg)+H_t(u,\theta)\Gamma(\boldsymbol \dens_t)(u,d\theta)\Bigg]du dt
,\end{multline*}
\noindent where the various coefficients are those defined just before, and the functions $\lambda_i$ are defined in \eqref{lidef}.
}
\end{enumerate}
\end{defi}
Note that in this Definition, the only quantity required to be in $H^1$ is the total density $\rho$~: indeed, the term $\sdc(\rho_t)\nabla\boldsymbol\dens_t$ is rewritten as \[\sdc(\rho_t)\nabla\boldsymbol\dens_t=\nabla (\sdc(\rho_t) \boldsymbol\dens_t)-\sdc'(\rho_t) \boldsymbol\dens_t\nabla \rho_t,\]
and the first term in the right-hand side above allows another derivative to be applied to the test function $H$, whereas the second term only involves the derivative of $\rho$ as wanted.

We are now ready to state our main theorem~:
\begin{theo}
\label{thm:mainthm}The sequence $(Q^N)_{N\in \N}$ defined at the end of Section \ref{subsec:processdef} is weakly relatively compact, and any of its limit points $Q^*$ is concentrated on trajectories $(\pi_t)_{t\in [0,T]}$ which are solution of \eqref{equadiff} in the sense of Definition \ref{defi:weaksol}.
\end{theo}
 
 \begin{rema}[Uniqueness of the weak solutions of equation \eqref{equadiff}]One of the reasons for our weak formulation of the scaling limit of the active exclusion process is the lack of proof for the uniqueness of weak solutions of equation \eqref{equadiff}. Several features of equation \eqref{equadiff} make the uniqueness difficult to obtain~: First, our differential equation does not take the form of an autonomous differential equation~: the variation of $\dens_t(u,\theta)$ involves the total density $\rho$, therefore the differential equation is in fact a differential system operating on the vector $(\dens_t(u, \theta), \rho_t(u))$. Cross-diffusive systems can exhibit pathological behavior when the diffusion matrix has negative eigenvalues, but in our case, both eigenvalues are non-negative and this issue does not appear.
 
However, although cross-diffusive systems are quite well understood (cf. for example \cite{Amann1993}), our equation involves a drift term which factors in via the vector $\overset{\rightarrow}{\Omega}(\boldsymbol \dens_t)$ the whole profile $(\dens_t(u, \theta))_{\theta\in \ctoruspi}$. One of the consequences of this drift term, which is the main obstacle to prove uniqueness, is that even the uniqueness of the total density $\rho_t(u)$ is not well established. Indeed, contrary to \cite{Quastel1992}, in which the total density evolves according to the heat equation,  the total density in our case is driven by the Burgers-like equation
 \[\partial_t \rho_t(u)=\Delta \rho_t(u)-\lambda \nabla\cdot( m_t(u)(1-\rho_t(u)))\]
where $m$ is a quantity which depends on the whole profile $(\dens_t(u, \theta))_{\theta\in \ctoruspi}$, and for which uniqueness is hard to obtain.
 \end{rema}

\subsection{Instantaneous currents}
 
 \label{subsec:outlineandcurrents}

In order to get a grasp on the delicate points of the proof, and to introduce the particle currents on which rely the proof of Theorem \ref{thm:mainthm}, we need a few more notations.

Throughout the proof, for any function $\varphi:\statespace\to \R$ and $x\in \torus$, we will denote by $\tau_x\varphi:\statespace\to \R$ the function which associates to a configuration $\confhat$ the value $\varphi(\tau_{x}\confhat)$, where $\tau_{x}\confhat\in \statespace$ is the translation of the configuration $\confhat$ by a vector $x$~:
\[(\tau_{x}\confhat)_y=\confhat_{x+y}, \quad \forall y\in \torus.\]
\index{$\tau_x$\dotfill translation by $x$ on the discrete torus}

For any function 
\[\begin{array}{cccc}
H:&[0,T]\times\ctorus\times \ctoruspi&\to &\R\\
&(t,u, \theta)&\mapsto& H_t(u, \theta)\\
\end{array},\]
{in $C^{1,2,1}([0,T]\times\ctorus\times\ctoruspi)$,} and any measure $\pi$ on $\ctorus\times \ctoruspi$, let us denote 
\index{$<\pi,H>$\dotfill integral of $H$ w.r.t. the measure $\pi$}
\[<\pi,H_t>=\int_{\ctorus\times \ctoruspi}H_t(u, \theta)d\pi(u, \theta) \]
the integral of $H$ with respect  to the measure $\pi$.
We consider the  martingale $M_t^{H, N}$
\begin{equation}\label{martingaledef}M_t^{H, N}=<\pi_t^N,H_t>-<\pi_0^N,H_0>-\int_0^t (\partial_s+ L_N)<\pi_s^N,H_s> ds,\end{equation}
where $\pi_s^N$ is the \emph{empirical measure of the process}
\index{$\partial_s,\; \partial_t$\dotfill time derivative}
\[\pi_s^N=\frac{1}{N^2}\sum_{x\in \torus}\conf_x(s)\delta_{x/N, \theta_x(s)}.\]
{The quadratic variation of this martingale can be explicitely computed, and is equal to} (cf. Appendix 1.5 of \cite{KLB1999})
\begin{multline*}
[M^{H, N}]_t=\int_0^t L_N(<\pi_s^N,H_s>^2)-2<\pi_s^N,H_s>L_N<\pi_s^N,H_s> ds\\
=\frac{2}{N^4}\sum_{x\in \torus}\bigg[\int_0^t L_N\cro{\eta_x(s)\eta_{x+1}(s)H_s(x/N, \theta_x(s))H_s((x+1)/N, \theta_{x+1}(s))}\\
-\eta_{x+1}(s)H_s((x+1)/N, \theta_{x+1}(s))L_N\cro{\eta_x(s)H_s(x/N, \theta_x(s))}\\
-\eta_x(s)H_s(x/N, \theta_x(s))L_N\cro{\eta_{x+1}(s)H_s((x+1)/N, \theta_{x+1}(s))}ds\bigg]
.\end{multline*}
Because of the initial factor $N^{-4}$, the contributions of the asymmetric and Glauber parts of the dynamic can be crudely bounded respectively by $CN^{-1}$ and $CN^{-2}$. By computing the symmetric part, we finally obtain  
\begin{align*}
[M^{H, N}]_t=&O(1/N)+\frac{1}{N^2}\sum_{x\in \torus}\bigg[\int_0^t \eta_x(s)\Big[H^2_s(x+1/N, \theta_x(s))+H^2_s(x-1/N, \theta_x(s))-2H^2_s(x/N, \theta_x(s))\\
&+2\eta_x(s)(1-\eta_{x+1}(s))H_s(x/N,\theta_x(s))\big[H_s((x+1)/N,\theta_x(s))-H_s(x/N,\theta_x(s))\big]\\
&+2\eta_{x+1}(s)(1-\eta_{x}(s))H_s((x+1)/N,\theta_{x+1}(s))\big[H_s((x+1)/N,\theta_{x+1}(s))-H_s(x/N,\theta_{x+1}(s))\big].\end{align*}
Because we assumed that $H$ is a smooth function, the three lines above are of order at most $N^{-1}$,  and therefore $[M_t^{H, N}]_t $ vanishes as $N$ goes to infinity. The martingale thus vanishes uniformly in time, in probability under $\Prob^{\lambda,\beta}_{\mu^N}$.

\medskip

Assume now that  the function  $H$ takes the form 
\begin{equation}\label{Hdecomp}H_s( u,\theta)=G_s(u) \omega(\theta),\end{equation}
where $G$ and $\omega$ are respectively functions on $[0,T]\times\ctorus$ and $\ctoruspi$. From now on, for any function $\Phi:\ctoruspi\to \R$, any configuration $\confhat$ and any $x\in \torus$ we will shorten 
\[\conf_x^{\Phi}=\Phi(\theta_x)\conf_x.\]
\index{$\conf_x^{\Phi}$\dotfill $\Phi(\theta_x)\conf_x$}
 With these notations, recalling that
\[L_N =N^2 \left(\genas+N^{-1}\genwa\right)+\genis, \] 
we can write the generator part of the integral term of \eqref{martingaledef} as 
\begin{equation}\label{empiricalmeasure2}\int_0^T L_N<\pi_s^N,H_s> ds=\frac{1}{N^2}\int_0^T \sum_{x\in \torus}G_s(x/N)\left(N^2[\genas\com_x(s)+N^{-1}\genwa\com_x(s)\right)+\genis\com_x(s)]ds.\end{equation}
Let us introduce accordingly the three instantaneous currents in our active exclusion process. Recall that $\tau_x$ represents the translation of a function by $x$.

\begin{defi}\label{defi:currents}
Given a site $x\in \torus$, each part of the generator  $L_N$'s action over $\com_x$ can be written
\begin{equation}\label{currentssym} \genas\com_x=\sum_{i=1}^2(\tau_{x-e_i}\curom_i-\tau_x\curom_{i}) \quad \mbox{ with }\quad \curom_{i}(\confhat)=\com_0\left(1-{\conf_{e_i}}\right)-\com_{e_i}\left(1-{\conf_0}\right),\end{equation}
\index{$\curom_i$\dotfill $\omega$-weighted sym. current on $(0,e_i)$}
\begin{equation}\label{currentsasym} \genwa\com_x=\sum_{i=1}^2(\tau_{x-e_i}\curaom_i-\tau_x\curaom_i)\quad \mbox{ with }\quad\curaom_i(\confhat)=\conf^{\omega\lambda_i}_0(1-\conf_{e_1})+ \conf^{\omega\lambda_i}_{e_i}(1-\conf_{0}),\end{equation}
\index{$\curaom_i$\dotfill $\omega$-weighted asym. current on $(0,e_i)$}
and
\begin{equation}\label{currentsising} \genis\com_x=\tau_x\gamma^{\omega}\quad \mbox{ with }\quad \gamma^{\omega}(\confhat)=\eta_0\int_{\ctoruspi} c_{0,\beta}(\theta,\confhat)(\omega(\theta)-\omega(\theta_0))d\theta.\end{equation}
\index{$\gamma^{\omega}$\dotfill instant. creation rate of $\com_0$ due to $\genis$}
For $i\in\{1, 2\}$ we will at times write $\curom_{x,x+e_i}=\tau_x \curom_i$ (resp. $\curaom_{x,x+e_i}=\tau_x\curaom_i$), which is interpreted as  the \emph{instantaneous current with intensity $\omega$ in the direction $i$} along the edge $(x, x+e_i)$ of the symmetric  (resp. weakly asymmetric)  part of the process. The last quantity $\tau_x\gamma^{\omega}$ is the \emph{local  alignment rate. }

When considering the time process $(\confhat(t))_{t\in [0,T]}$ we will, for the sake of concision, write $  \curom_{i}(t)$ for $\curom_{i}(\confhat(t))$, and in the same fashion $\curaom_i(t)$  instead of  $\curaom_i(\confhat(t))$, and $\gamma^{\omega}(t)$ instead of $\gamma^{\omega}(\confhat(t))$. Finally, in the case where $\omega\equiv 1$, we will denote by \[\cur_i:=\cur_i^1=\conf_0-\conf_{e_i}.\]
\end{defi}
\index{$\cur_i$\dotfill total instant. sym. current on $(0,0+e_i)$}

Performing a first integration by parts on the exclusion part of the right-hand side of  \eqref{empiricalmeasure2}, we obtain thanks to equations \eqref{currentssym}, \eqref{currentsasym} and \eqref{currentsising} 
\begin{align}\label{empiricalmeasure1}\int_0^T L_N<\pi_s^{N},H_s> ds=\frac{1}{N^2}\int_0^T  \sum_{x\in \torus}\tau_x\cro{\sum_{i=1}^2 \Big(N\curom_i(s)+\curaom_i(s)\Big)\partial_{u_i,N} G_s(x/N)+ G_s(x/N)\gamma^{\omega}(s)} ds,
\end{align}
where $\partial_{u_i,N}$ is the discrete partial derivative \[(\partial_{u_i,N} G)(x/N)=N\left[G((x+e_i)/N)-G(x/N)\right].\]
\index{$\partial_{u_i,N}$\dotfill discrete approximation of $\partial_{u_i}$}

\index{$B_l(x)$\dotfill $\{y\in \torus, \abss{y-x}\leq l\}$}
\index{$B_l$\dotfill $B_l(0)$}
\index{$|B|$\dotfill number of sites in $B$}

The spatial averaging is of great importance throughout the proof of the hydrodynamic limit, we need some convenient notation to represent this operation. For any site $x\in \torus$ and any integer $l$, we denote by 
\[B_l(x)=\left\{\;y\in \torus \; \big|\; \norm{y-x}_{\infty}\leq l\;\right\}\] 
the box of side length $2l+1$ around $x$. In the case where $x=0$ is the origin, we will simply write $B_l:=B_l(0)$. For any \emph{finite} subset $B\subset \torus$, $|B|$ denotes the number of sites in $B$.
Given  $\varphi$ a function on $\statespace$, we denote by
\index{$\langle\varphi\rangle_x^l$\dotfill average over $B_l(x)$ of the $\tau_y\varphi$}
\begin{equation}\label{averagedef}\langle\varphi\rangle_x^l=\frac{1}{\abss{B_l(x)}}\sum_{y\in B_l(x)}\tau_y\varphi\end{equation}
the average of the function $\varphi $ over $B_l(x)$. In the case where $\varphi(\confhat)=\com_0$, (resp. $\varphi(\confhat)=\conf_0$), we will write  $\tau_x\rho^{\omega}_l=\langle\varphi\rangle_x^l$ (resp. $\tau_x\rho_l$) for the empirical average of $\com$ (resp. $\conf$) over the box centered in $x$ of side length  $2l+1$.

We will also denote for any integer $l$ by $\dens_l$ the empirical angular density defined by 
\index{$\rho_l$\dotfill empirical particle density in $B_l$}
\index{$\rho^{\omega}_l$\dotfill average of $\com_x$ over $B_l$}
\begin{equation}\label{empiricalprofile}\dens_l=\frac{1}{\abss{B_l}}\sum_{x\in B_l}\conf_x \delta_{\theta_x}\in \pset,\end{equation}
where $\pset$ is the set of non-negative measures on $\ctoruspi$ with total mass in $[0,1]$ (cf. Definition \ref{defi:angleprofile} below).
Finally, to simplify notations throughout the proof, we will write $\varepsilon N$ instead of the integer part $\lfloor \varepsilon N \rfloor$.
\index{$\dens_l$\dotfill empirical angular density over $B_l$}

\section{Canonical measures,  entropy and irreducibility}
\label{sec:3}

\subsection{Definition of the canonical measures}
\label{subsec:canonicalmeasures}

Due to the presence of angles, the canonical product measures for the active exclusion process are not parameterized by the local density $\alpha\in [0,1]$ like the SSEP, but rather by a measure $\param$ on $[0,2\pi]$ whose total mass $\int_{\ctoruspi}\param(d\theta)$ is the local density.
\begin{defi}[Grand-canonical parameters]
\label{defi:angleprofile}
Recall that $\ctorus $ is the $2$-dimensional continuous torus $(\R/\Z)^2$, and let $\meas(\ctoruspi)$ be the set of non-negative measures on $\ctoruspi$. We will call \emph{ grand-canonical parameter} any measure $\param\in \meas(\ctoruspi)$ with total mass $\alpha:=\int_{\ctoruspi} \param(d\theta)\leq 1$. We denote by 
\begin{equation}\label{pseteq}\pset=\left\{\;\param\in \meas(\ctoruspi) \; |\; \alpha\in[ 0,1]\;\right\},\end{equation}
the set of grand-canonical parameters.
\end{defi}
\index{$\param$ \dotfill  grand-canonical parameter, element of $\pset$}
\index{$\am$ \dotfill total mass of $\param$}
\index{$\pset$ \dotfill set of  grand-canonical parameters}

We now define a topological setup on $\pset$. Let us consider on $C^1(\ctoruspi)$, the set of {continuously} differentiable functions, the norm $\norm{g}^*=\max(\norm{g}_{\infty},\norm{g'}_{\infty})$, and let $B^*$ be the unit ball in $(C^1(\ctoruspi),\norm{\cdot}^*)$.
\begin{defi}\label{defi:convparam}
We endow $\mathcal{M}(\ctoruspi)$, the vector space of finite mass signed measures on $\ctoruspi$, with the norm 
\[\normm{\param}=\sup_{g\in B^*}\left\{\int g(\theta)d\param(\theta)\right\},\] 
\index{ $\normm{\cdot}$ \dotfill norm on $\pset$}
and with the corresponding distance \[d(\param, \param'):=\sup_{g\in B^*}\left\{\int_{\ctoruspi}g(\theta)d\param(\theta)-\int_{\ctoruspi}g(\theta)d\param'(\theta)\right\}.\]
We then endow $\pset$  with the topology induced by $\normm{\cdot}$. This distance is a generalization of the Wasserstein distance to measures which are not probability measures.
\end{defi}

\begin{rema}
As checked in Appendix \ref{sec:B}, this topology satisfies
\begin{itemize}
\item for any cylinder function $\psi$, the application $\param\mapsto \Egcm(\psi)$ is Lipschitz-continuous (cf. Proposition \ref{prop:Lipschitzcontinuity}).
\item any {$\param\in \mathcal{M}(\ctoruspi)$} is the limit of combinations of Dirac measures.
\item if $\theta_k\to \theta$, then $\normm{\delta_{\theta_k}-\delta_{\theta}}\to 0$.
\end{itemize}
It is therefore the natural choice for our problem.
\end{rema}

We now introduce the canonical measures of our process, which are translation-invariant particular cases of  measures associated with a density profile, introduced in Definition \ref{defi:measuresforaDPT}.
\begin{defi}[Grand canonical measures]\label{defi:GCM}
Consider a \emph{translation invariant} density profile on the torus $ \boldsymbol \dens$, i.e. such that for any $u\in \ctorus$, \[\boldsymbol \dens(u, d\theta)=\param(d\theta)\]  for some  grand-canonical parameter $\param\in \pset$ independent of $u$. We will write  $\gcm$ for the product measure $\mu^N_{\boldsymbol\dens}$, and $\E_{\param}$ will denote the corresponding expectation. This class of measures will be referred to as \emph{grand-canonical measures}.
\index{$\gcm$ \dotfill grand-canonical measure GCM($\param$)}
\index{$\Egcm$ \dotfill expectation w.r.t. $\gcm$}
Furthermore, for any $\am\in[0,1]$, the measure $\mu_{\param}$ associated with the uniform density profile on the torus \[\boldsymbol \dens(u,d\theta)\equiv\am d\theta/2\pi,\] where the angle of each particle is chosen uniformly in $\ctoruspi$, will be denoted by $\mu_{\am}^*$, and the corresponding expectation will be denoted by $\E_{\am}^*$.
\index{$\mesref$ \dotfill GCM with uniform angles}
\index{$\Eref$ \dotfill expectation w.r.t. $\mesref$}
\end{defi}
Note that these measures are dependent on $ N$, but due to their translation invariant nature, we will omit this in our notation.

\begin{rema}For any density $\am\in [0,1]$, the measure $\mu^*_{\am}$ on $\statespace$ is not invariant for our dynamic, because although it is invariant for the symmetric part of the exclusion, the weakly asymmetric part (as well as the Glauber part as soon as $\beta\neq0$) breaks this property. We will however prove in Section \ref{subsec:entropy} that due to the scaling in $N$, the stationary distribution of our dynamics is locally close to  {$\mu^*_{\am}$}.
\end{rema}

\begin{defi}[Canonical measures]
\label{defi:CM}
Fix a positive integer $l$, an integer $K\leq (2l+1)^2$ and {$\Theta_K=(\theta_1,\ldots ,\theta_K)$ a family of $K$ angles, taken up to reordering of its coordinates}, we shorten by $\K$ the pairs $(K, \Theta_K)$, which we will refer to as \emph{canonical states} on $B_l$. We will denote by $\Kset_l$ the set of canonical states $\K$ on $B_l$, 
\index{$\Theta_K$ \dotfill a family of $K$ angles}
\index{$\K$ \dotfill a pair $(K,\Theta_K)$}
\index{$\Kset_l$ \dotfill the set of possible $\K=\dens_l$}
\[\Kset_l=\{\K=(K, \Theta_K)\; | \; K\leq (2l+1)^2\}.\]
Since\index{$\Ksett_l$ \dotfill the set of $\K$ such that $K\leq \abss{B_l}-2$}
 our process loses its fast mixing properties when there is only one or less empty site (In which case mixing mainly comes from the Glauber dynamics, which is very slow w.r.t. the displacement dynamics, cf. Section \ref{subsec:irreducibility} below), we also introduce 
\begin{equation}\label{Ksettdef}\Ksett_l=\{\K\in \Kset_l \;|\; K\leq (2l+1)^2-2\},\end{equation}
the set of $\K$ for which the exclusion process on $B_l$ is irreducible.
Furthermore, for any fixed $\K\in \Kset_l$, we denote by 
\begin{equation}\label{sousespace}\subspace=\left\{{\confhat \mbox{ config. on }B_l }\quad \left| \quad \sum_{x\in B_l}\conf_x \delta_{\theta_x}=\sum_{k=1}^{K}\delta_{\theta_k}\right. \right\},\end{equation} 
the set of configurations {on $B_l$} with canonical state $\K$ in $B_l$. \index{$\subspace$ \dotfill set of confs. with $\K$ particles in $B_l$} 
\index{$\mu^*_{\am,l}$ \dotfill $\mesref$ restricted to configurations on $B_l$} 
\index{$\E_{\am,l}^*$ \dotfill expectation w.r.t. $\mu^*_{\am,l}$} 
\index{$\mu_{l, \K}$ \dotfill $\mu^*_{\am,l}$ conditioned to $\confhat\in \subspace$} 
\index{$\E_{l,\K}$ \dotfill expectation w.r.t. $\mu_{l, \K}$}

Let $\mu^*_{\am, l}$ denote the measure $\mesref$ on $B_l$, for any density $\am\in ]0,1[$, we will denote by $\mu_{l, \K}$ the conditioning of $\mu^*_{\am, l}$ to $\subspace$ {(which is therefore a measure on the set of local configurations $\confhat\in (\{0,1\}\times\ctoruspi)^{B_l}$)}, and by $\E_{l, \K}$ the corresponding expectation
\[\E_{l, \K}(\function)=\E^*_{\am, l}\pa{\;\function\;\;\left\vert\; \;\confhat\in \subspace\right.}.\]
 These measures will be referred to as \emph{canonical measures of the process.}
\end{defi}

\begin{defi}\label{defi:paramK}
Fix $l\in \N$, we associate to any $\K\in \Kset_l$ the  grand-canonical parameter 
\[\param_{\K,l}=\frac{1}{(2l+1)^2}\sum_{k=1}^{K}\delta_{\theta_k}.\]
When there is no ambiguity, we will drop the dependency in $l$ and simply write $\param_{\K}=\param_{\K,l}$.
 \end{defi}

The pertinent results regarding the metric space $(\pset, \normm{\cdot})$ are regrouped in Appendix \ref{sec:B}~: The equivalence of ensembles is proved in Section \ref{subsec:equivalenceensembles}, the Lipschitz-continuity of the expectation w.r.t. $\mesinv$ in the parameter $\param$ is proved in Section \ref{subsec:lipshitzcontinuity}, and finally, the compactness of the set $(\pset, \normm{\cdot})$ is proved in Section \ref{subsec:compaciteparam}.

\subsection{Entropy production and local equilibrium}
\label{subsec:entropy}

\intro{The proof of the replacement Lemma is based on the control of the entropy production of the process. The difficulty  here is that the invariant measures of the process are not known, and the decay of the relative entropy w.r.t. these measures cannot be computed directly. Thus we consider approximations of these measures, and for a fixed non-trivial density $\am\in ]0,1[$,  our goal is to get an estimate of the entropy of the process's time average with respect to the reference measure $\mesref$ introduced in Definition \ref{defi:GCM}.}

Let us fix $\am \in ]0,1[$, we are going to prove that regardless of the initial density profile, the entropy of the active exclusion process w.r.t the measure of a process started from $\mesref$ and following a symmetric simple exclusion process can be controlled by $CN^2$ for some constant $C$. 

The choice of $\mesref$ among the $\mu^*_{\am'}$, $\am'\in ]0,1[$ is not  important, since for any different angle density  $\am'\in]0,1[$, the relative entropy between the two product measures $\mesref$ and $\mu^*_{\am'}$ is of order $N^2$ as well. 

For some cylinder function $h\in{\mathcal C}$, and some edge $a=(a_1,a_2)$ {in $\torus$ or $\Z^2$}, we denote by $\grad$ the gradient representing the transfer of a particle from site $a_1$ to site $a_2$  under the  exclusion process
\begin{equation}\label{graddef}\grad f(\confhat)={\conf_{a_1}}\pa{1-{\conf_{a_2}}}\pa{f\pa{\confhat^{a_1, a_2}}-f(\confhat)}.\end{equation}
\index{$\grad $\dotfill gradient due to a particle jump $a_1 \to a_2$ }
We will shorten this notation in the case where $a=(0,e_j)$ by writing $\nabla_j:=\nabla_{(0,e_j)}$.
\index{$ \nabla_j$\dotfill gradient due to a particle jump  $0\to e_j$}
Before turning to the control of the entropy itself, we introduce an important quantity in the context of hydrodynamic limits.
\begin{defi}[Dirichlet form of the symmetric dynamics]
\label{defi:Dirichletform}
Let $h$ be a cylinder function, we introduce the Dirichlet form of the process 
\begin{equation}\label{dirdef}\dir_{\param}(h)=-\Egcm(h\gene h),\end{equation}
\index{$\dir$\dotfill Dirichlet form of the exclusion process}
where $\gene$ is the symmetric exclusion generator defined in equation \eqref{gensymdef}.
It can be rewritten thanks to the invariance of $\mesinv$ w.r.t the symmetric exclusion process as
\[\dir_{\param}(h)=\frac{1}{2}\Egcm \pa{\sum_{x\in \torus}\sum_{|z|=1}\pa{\nabla_{x,x+z}h}^2}.\]
If there is no ambiguity, we will omit the dependency in $\param$ of the Dirichlet form, and simply denote it by $\dir$. The Dirichlet form is  convex and  non-negative. Furthermore, any function $f$ in its kernel is such that $f(\confhat)=f(\confhat')$ for any pair $(\confhat,\confhat')$ of configurations with the same number of particles $K\leq N^2-2$ and the same family of angles.
{For any non-negative function $h$, we} also introduce the Dirichlet form 
\begin{equation}\label{reddirdef}\rdir(h)=\dir(\sqrt{h}),\end{equation} 
\index{$ \rdir(h)$\dotfill $\dir(\sqrt{h})$}
which has the same properties as $\dir$. 
\end{defi}

We now investigate the entropy production of the active exclusion process. Let  $P^{N,\lambda, \beta}_t$ be the semi-group of the active exclusion process associated with the complete generator $L_N$ introduced in equation \eqref{defgenecomplet}, and $\mu_t^N=\mu^N P_t^{N,\lambda, \beta}$ the measure of the configuration at time $t$. Because we assume the initial profile to be continuous (and therefore bounded), $\mu^N$ is absolutely continuous with respect to the product measure $\mesref$, with density
\begin{equation}
\label{eq:densitemes}
\frac{d\mu^N}{d\mesref}(\confhat)=\prod_{x\in \torus}\cro{(1-\conf_x)\frac{1-\zeta(x/N)}{1-\am}+\conf_x\frac{2\pi \widehat{\zeta}(x/N, \theta_x)}{\am}}. 
\end{equation}
This, and the fact that the alignment rates $c_{x,\beta}$ are bounded from above and below uniformly in $\theta$, guarantee that for any time $t$, $\mu_t^N$ is also absolutely continuous w.r.t. $\mesref$. We therefore denote by $f_t^N=d\mu_t^N/\mesref$ the density of the measure at time $t$ w.r.t. the reference measure $\mesref$. We now prove the following estimate on the entropy of the function $f_t^N$.

 \begin{prop}[Control on the entropy and the Dirichlet form of $f_t^N$]
 \label{prop:entropyproduction}For any density $f$ w.r.t. $ \mesref$, we denote by $H(f)=\Eref(f \log f)$ the entropy of the density $f$. Then, for any time $t>0$, there exists a constant $K_0=K_0(t,\lambda, \beta, \widehat{\zeta})$ such that 
 \begin{equation*}H\left({\frac{1}{t}\int_0^t f_s^N ds}\right)\leq K_0 N^2\eqand \rdir\left({\frac{1}{t}\int_0^t f_s^N ds}\right)\leq K_0.\end{equation*}
\end{prop}
\begin{proof}[Proof of Proposition \ref{prop:entropyproduction}] 
The density $f_t^N$  is solution  to \index{$\mu^N_t $\dotfill measure at $t$ of the active exclusion process started at $\mu^N$} \index{$f_t^N $\dotfill density of $\mu^N_t$ w.r.t. $\mesref$}
\begin{equation}\label{kolmogorov}\begin{cases}
\partial_t f_t^N= L^*_N f_t^N\\
f_0^N=d\mu^N/d\mesref,
\end{cases}
\end{equation}
where $L^*_N$ is the adjoint of $L_N$ in $L^2(\mesref)$.
To clarify the proof, we divide it in a series of steps.

\paragraph{Expression of the entropy production of the system}The relative entropy of  $\mu_t^N$ with respect to the reference measure $\mesref$ is given by
 \index{$H(\mu\mid\nu) $\dotfill entropy of $\mu$ w.r.t. $\nu$ }
\[H(\mu_t^N\mid \mesref)=H(f_t^N)=\Eref\left( f_t^N \log f^N_t \right),\]
which is non-negative due to the convexity on $[0,+\infty[$ of $x\mapsto x\log x$.
According to equation \eqref{kolmogorov}, its time derivative is 
\begin{equation}\label{entropydecomp}\partial_t H(f_t^N)=\Eref\left(\log f_t^N L_N^* f_t^N \right)+\Eref\left( L^*_N f_t^N \right).\end{equation}
The second term on the right-hand side is equal to \[\Eref\left( L^*_N f_t^N \right)=\Eref\left( f_t^N L_N\widetilde{1}  \right)=0,\] since all constant functions are in the kernel of $L_N$. Equation \eqref{entropydecomp} can be rewritten, since $ L_N^*$ is the adjoint of $L_N$ in $L^2(\mesref)$, as
\[\partial_t H(f_t^N)=\Eref\left( f_t^N L_N \log f_t^N \right).\]
Now thanks to the elementary inequality 
\[\log b -\log a\leq \frac{2}{\sqrt{a}}(\sqrt{b}-\sqrt{a}),\]
we can control $L_N \log f_t^N$ by \[\frac{2}{\sqrt{f_t^N}}L_N\sqrt{f_t^N},\]
therefore, the definition of $L_N$ yields
\[\partial_t H(f_t^N)\leq -2N^2 \rdir\left({f_t^N}\right)+2N\Eref\left(\sqrt{f_t^N} \genwa \sqrt{ f_t^N}\right)+2\Eref\left(\sqrt{f_t^N} \genis \sqrt{ f_t^N} \right),\]
where $\rdir$ is the Dirichlet form defined in Definition \ref{defi:Dirichletform}.

Integrating between the times $0$ and $t$, we get
\begin{equation}\label{entropyineq}H(\mu_t^N\mid \mesref)+2N^2\int_0^t  \rdir \left({f_s^N}\right)\leq H(\mu^N\mid \mesref)+2\int_0^t\Eref\left(\sqrt{f_s^N} (N\genwa+\genis) \sqrt{ f_s^N}\right)ds\end{equation}
Since the Dirichlet form of the symmetric exclusion process is non-negative, we now focus on showing that the part of the entropy due to the weakly asymmetric part  and Glauber part do not grow too much in $N$, in order to get an upper bound on the Dirichlet form $\rdir(f)$ and on the entropy $H(\mu_t^N\mid \mesref)$. From here, control over the initial relative entropy should suffice to ensure that the measure of the active exclusion process remains close to a  product measure.

\paragraph{Bound on the entropy production of the asymmetric part of the dynamics} by definition of the asymmetric dynamic, 
\[\Eref\left(\sqrt{f_s^N} \genwa \sqrt{ f_s^N}\right)= \Eref \pa{\sum_{x,i,\delta=\pm1}\lambda_i(\theta_x)\delta\conf_x(1-{\conf_{\delta e_i}})\sqrt{ f_s^N}(\confhat)\pa{\sqrt{ f_s^N}(\confhat^{x,x+\delta e_i})-\sqrt{ f_s^N}(\confhat)}}.\]
Despite the extra factor $N$, the jump rates of the weakly asymmetric dynamics are not very different from symmetric exclusion process jump rates, which allows us to estimate the quantity above in terms of  the Dirichlet form. More precisely, thanks to the elementary inequality 
\[\E(\varphi \psi)\leq \gamma\E(\varphi^2)/2+\E(\psi^2)/2\gamma\]
 which holds for any positive constant $\gamma$, we can write with 
\[\varphi={\conf_x(1-{\conf_{\delta e_i}})}\pa{\sqrt{ f_s^N}(\confhat^{x,x+\delta e_i})-\sqrt{ f_s^N}(\confhat)},\]
and\[\psi={\lambda_i(\theta_x)\delta\sqrt{ f_s^N}(\confhat)}\]
that
\begin{multline*}\Eref\left(\sqrt{f_s^N} \genwa \sqrt{ f_s^N}\right) \\
\leq  \sum_{x,i,\delta=\pm1}\cro{\frac{\Eref\pa{\lambda_i(\theta_x)^2f_s^N}}{2\gamma}+\frac{\gamma}{2}\Eref\pa{ \conf_x(1-{\conf_{\delta e_i}})\pa{\sqrt{ f_s^N}(\confhat^{x,x+\delta e_i})-\sqrt{ f_s^N}(\confhat)}^2}}.\end{multline*}
In right-hand side above, letting $ C_{\lambda}=4\lambda^2$ the first term can be bounded  by $ C_{\lambda}N^2/2\gamma$, since the number of terms in the sum is $4 N^2$, whereas the second sum of terms is $\gamma \rdir(f_s^N)$. We  then let $\gamma=N$ to obtain the upper bound  
\begin{equation}\label{entropwa}2N\Eref\left(\sqrt{f_s^N} \genwa \sqrt{ f_s^N}\right)\leq C_{\lambda}N^2+{ N^2\rdir(f_s^N)}.\end{equation}

\paragraph{Bound on the entropy production of the Glauber part of the dynamics}thanks to the elementary inequality $ab\leq (a^2+b^2)/2$, and since the jump rates $c_{x,\beta}$ are less than $e^{8\beta}/2\pi$, and $\conf_x$ by $1$ 
\begin{align*}\Eref\left(\sqrt{f_s^N} \genis \sqrt{ f_s^N}\right)=&\Eref\pa{ \sqrt{f_s^N}\sum_{x\in \torus}\conf_x\int_{\ctoruspi}c_{x,\beta}(\theta, \confhat)\pa{\sqrt{f_s^N}(\confhat^{x,\theta})-\sqrt{f_s^N}(\confhat)}d\theta}\\
\leq&\frac{e^{8\beta}}{2\pi}\sum_{x\in \torus}\Eref\pa{\frac{1}{2}\int_{\ctoruspi}{f_s^N}(\confhat^{x,\theta})d\theta+\frac{3}{2}{f_s^N}(\confhat)}.\end{align*}
Since $\Eref\pa{\frac{1}{2\pi}\int_{\ctoruspi}{f_s^N}(\confhat^{x,\theta})d\theta}=\Eref\pa{f_s^N}$, the expectation can be bounded from above by $2$, and we can therefore write, letting $C_{\beta}=2e^{8\beta}/\pi$
\begin{equation}\label{entropis}2\Eref\left(\sqrt{f_s^N} \genis \sqrt{ f_s^N}\right)\leq C_{\beta}N^2.\end{equation}

\paragraph{Bound on the Dirichlet form and on the entropy production}at this point, we obtain from \eqref{entropyineq}, \eqref{entropwa} and \eqref{entropis} \[H(\mu_t^N\mid \mesref)+ N^2\int_0^t \rdir\left({f_s^N}\right)ds\leq H(\mu^N\mid \mesref)+t(C_\lambda+C_{\beta})N^2\]
By \eqref{eq:densitemes}, there exists a constant  $K=K(\widehat{\zeta},\alpha)$, such that for any $N\in \N$, $\norm{\log d\mu^N/d\mesref}_{\infty}\leq KN^2$, and we can therefore estimate the relative entropy of the initial measure $\mu^N$ w.r.t. $\mesref$ by
\begin{equation}\label{entropinit}H(\mu^N\mid \mu_{\am}^*)\leq K N^2.\end{equation}
We can therefore write 
\begin{equation}\label{entropycomp3}H(\mu_t^N\mid \mesref)+\int_0^t N^2 \rdir\left({f_s^N}\right)\leq K(t)N^2. \end{equation}
where $K(t)=K+t(C_\lambda+C_{\beta})$ is a positive constant. Since $H(f)=\Eref(f\log f)$ and $D(f)$ are both non-negative and convex, we can deduce from \eqref{entropycomp3}, that for some time-dependent constant $K_0=\int_0^tK(s)ds$, we have
\begin{equation}\label{entropybounds} H\pa{\frac{1}{t}\int_0^tf_s^N}\leq K_0 N^2\eqand \rdir\left({\frac{1}{t}\int_0^t f_s^N ds}\right)\leq K_0.
\end{equation}
This upper bound proves proposition \ref{prop:entropyproduction}, and will be necessary in the next Section to apply the replacement Lemma \ref{lem:replacementlemma} to the active exclusion process. 
\end{proof}

Before taking on the problem of irreducibility, we give a result that will be needed several times throughout the proof, and comes from the entropy inequality.
Let us denote by $\geniz$ the modified Glauber generator with uniform update of the angle in $ \ctoruspi$, (i.e. $\beta=0$)
\[\geniz f(\confhat)=\sum_{x\in \torus}\conf_x\frac{1}{2\pi}\int_{\ctoruspi}(f(\confhat^{x, \theta})-f(\confhat))d\theta\]
\index{$ \geniz$\dotfill Glauber generator with $\beta=0$}
and denote in a similar fashion
\begin{equation}\label{betaz}L_N^{\beta=0}=N^2\genex+\geniz,\end{equation}
which is the complete generator of the active exclusion process with random update of the angles.
\index{$L_N^{\beta=0} $\dotfill generator of the active exclusion process for $\beta=0$}
Then, accordingly to our previous notations,  $\Prob^{\lambda, 0}_{\mesref}$ is the measure on the trajectories started from $\mesref$ and driven by the generator $L_N^{\beta=0}$. We can now state the following result.
\begin{prop}[Comparison of $\Prob_{\mu^N}^{\lambda,\beta}$ and $\Prob_{\mesref}^{\lambda,0}$]
\label{prop:comparisonbetazero}
{We endow $\statespace$  (resp. $\statespace^{[0,T]}$) with the topology induced by the mapping $\pi^N$ and the topology on $\meas(\ctorus\times\ctoruspi)$ (resp. $\Mhat$, cf. topological setup just before Section \ref{subsec:mainthm}).} There exists a  constant $K_0=K_0(T,\beta, {\widehat{\zeta}})>0$ such that for any {bounded and measurable} function  $X:\statespace^{[0,T]}\to\R$ and any  $A>0$, 
\[ \E^{\lambda,\beta}_{\mu^N}\left[X\pa{\confhat^{[0,T]}}\right]\leq \frac{1}{A}\pa{K_0N^2+\log\E^{\lambda,0}_{\mesref}\left[\exp \left(A X\pa{\confhat^{[0,T]}}\right)\right]},\]
where $\confhat^{[0,T]}$ is the notation already introduced at the end of Section \ref{subsec:processdef} for a trajectory $(\confhat(t))_{t\in [0,T]}$.
\end{prop}
\begin{proof}[Proof of Proposition \ref{prop:comparisonbetazero}]The proof of this Proposition is rather straightforward thanks to the entropy inequality. In a first step, we compare the same process starting from $\mesref$. First note that for any function $X:\statespace^{[0,T]}\to\R$, we can write 
\[\E^{\lambda, \beta}_{\mu^N}\left[X\pa{\confhat^{[0,T]}}\right]=\E^{\lambda, \beta}_{\mesref}\pa{ \frac{d\mu^N }{d\mesref}(\confhat(0)) X\pa{\confhat^{[0,T]}}}.\]
This yields that 
\begin{align}\label{entropybeta}
\E^{\lambda, \beta}_{\mu^N}\left[ X\pa{\confhat^{[0,T]}}\right]&\leq \frac{1}{A }\pa{H(\mu^N \mid{\mesref})+\log \E^{\lambda, \beta}_{\mesref}\left[\exp \left(A X\pa{\confhat^{[0,T]}}\right)\right]}.
\end{align}
In the entropy inequality above, $ \E^{\lambda, \beta}_{\mu^N}$ is the expectation under the measure of the process started from $\mu^N$, whereas $\E^{\lambda, \beta}_{\mesref}$is that of the process started from the stationary measure $\mesref$.

By \eqref{entropinit}, the first term in the right-hand side above is less than $KN^2/A$ for some fixed constant $K=K(\widehat{\zeta})$. 
Furthermore, the Radon-Nikodym derivative of the process with alignment ($\beta >0$) w.r.t the one without alignment ($\beta=0$) 
can be explicitly computed. Given a c\`adl\`ag trajectory $\confhat^{ [0,T]}\in \statespace^{[0,T]}$, 
consider $\tau_1,\ldots ,\tau_R$ the set of angle jumps between times $0$ and $T$, let us denote by $x_i$ the site at which the angle changed at time $\tau_i$, and by $\theta_i=\theta_{x_i}(\tau_i)$ the new angle at site $x_i$. Then, the density between the measures with and without alignment is given by 
\[\frac{d\Prob^{\lambda, \beta}_{\nu}}{d\Prob^{\lambda, 0}_{\nu}}(\confhat^{[0,T]})=\prod_{i=1}^R\frac{c_{x_i,\beta}(\theta_i, \confhat(\tau_i))}{c_{x_i,0}(\theta_i, \confhat(\tau_i))}\leq e^{8\beta R},\]
where $R$ is the number of angle updates between times $0$ and $T$. To establish the estimate above, we used that $c_{x,\beta}(\theta, \confhat)$ can be uniformly bounded from above by $e^{8\beta}/2\pi$, that $c_{x,0}(\theta, \confhat)=1/2\pi$, and that regardless of the configuration and the inverse temperature $\beta$, each site updates its angle at rate $1$(i.e. $\int_{\theta}c_{x,\beta}(\theta, \confhat)=1$). 
We can now estimate the second term in the right-hand side of equation \eqref{entropybeta} by 
\[\frac{1}{A}\log \E^{\lambda,0}_{\mesref}\left[e^{8\beta R}\exp \left(AX\pa{\confhat^{[0,T]}}\right)\right].\]
Applying the Cauchy-Schwarz inequality yields that the quantity above is less than 
\[\frac{1}{2A}\pa{\log\E^{\lambda,0}_{\mesref}\left[e^{16\beta R}\right]+\log\E^{\lambda,0}_{\mesref}\left[\exp \left(2A X\pa{\confhat^{[0,T]}}\right)\right]}.\]
Since the angle updates happen in each site at rate $1$ except when the site is empty, we can define on the same probability space as our process a family $P_x$ of i.i.d. Poisson variable with mean $T$, and such that $R\leq \sum_{x\in \torus}P_x$. Thanks to the elementary inequality 
\[\log \E\cro{e^{16\beta\sum_{x\in \torus}P_x}}=T(e^{16\beta}-1)N^2 ,\]
we now only have to let  \[K_0(T, \beta,\widehat{\zeta})=2K(\widehat{\zeta})+T(e^{16\beta}-1)\] and replace $A$ by $2A$ to conclude  the proof of Proposition \ref{prop:comparisonbetazero}.
\end{proof}

\subsection{Irreducibility and control on full clusters}
\label{subsec:irreducibility}

{ \textit{Unlike the exclusion process with one type of particles, the multi-type exclusion process is not irreducible when the number of particles is too large, namely when the domain has less than one empty site. When all the sites are occupied for example, the process is stuck in its current configuration, up to realignment, due to the exclusion rule. At high density, we therefore lose the mixing properties we need to reach local equilibrium. To illustrate this statement, consider a square macroscopic domain of size $\varepsilon N$, and  on it a  configuration with the bottom half filled with particles with angle $\theta$, and the top half filled with particles with angle $\theta'\neq \theta$, and letting a finite number of sites be empty, the mixing time of this setup is of order larger than $N^2$ due to the rigidity of the configuration. In order to reach equilibrium, an empty site needs to "fetch'' a particle and transport it in the other cluster, and so on, until the density is homogeneous for both types of particles. The scaling of our alignment dynamics, is, furthermore, not sufficient to ensure sufficiently frequent realignment of the particles to solve this issue. }}

{\textit{ In order to prove the scaling limit of a multi-type exclusion process, it is therefore critical to bound the particle density away from $1$. This issue was solved in \cite{Quastel1992} by using the fact that the total density of the multi-type SSEP (the angle blind model) follows the standard SSEP dynamics (with one specie). Thus the total density could be controlled by the classical argument on the hydrodynamic limit for SSEP. In our case, however, the total density does not follow the SSEP dynamics. In fact, it is not even a Markov chain due to the asymmetric parts which depend on the angles. A different argument is required to control the evolution of the total density, which is the purpose of the subsection.}}

\bigskip

\bigskip 

In the general setup where the number of types of particles in a domain $B$ can reach $|B|$ (which will often be the case when particles take their angles in $\ctoruspi$), it is known that the exclusion process with $|B|-1$ particles is no longer irreducible, as a consequence of a generalization of the $n$-puzzle (cf. Johnson \and Story, 1879, see \cite{JS1879}). We therefore need to consider only the local configurations with two empty sites, on which the exclusion process is irreducible regardless of the number of types of particles, as stated in the following Lemma. For any integers $a,b\in \Z$, $\llbracket a,b\rrbracket=\{a,\dots,b\}$ denotes the segment of integers between $a$ and $b$.
\begin{lemm}[Irreducibility of the displacement process with two empty sites]\label{lem:twoemptysites}
Consider a square domain $B=B_p(x)$, and two configurations $\confhat$, $\confhat'$ two configurations with the same types and number of particles in $B$, i.e. such that \[\sum_{x\in B}\conf_x\delta_{\theta_x}=\sum_{x\in B}\conf'_x\delta_{\theta'_x}.\]
Further assume that the number of empty sites in $\eta$ and $\eta'$ is at least $2$. Then, there exists a sequence of configurations $\confhat^{0},\ldots ,\confhat^n$, such that $\confhat^0=\confhat$, $\confhat^n=\confhat'$, and such that for any $k\in \llbracket0,n-1\rrbracket$, $\confhat^{k+1}$  is reached from $\confhat^k$ by one allowed particle jump, i.e.
\[\confhat^{k+1}=\pa{\confhat^k}^{x_k, x_k+z_k},  \eqand \conf^k_{x_k+z_k}=1-\conf^k_{x_k}=0\eqand \abss{z_k}=1.\]
Furthermore, there exists a constant $C$ such that $n\leq C p^4$.
\end{lemm}
\begin{proof}[Proof of Lemma \ref{lem:twoemptysites}]The proof of this statement is quite elementary. Fix a configuration $\confhat\in \statespace$ on a rectangular domain $B$ with two empty sites, and let $a=(a_1,a_2)$ be an edge in $\torus$. We are first going to prove that $\confhat^{a_1,a_2}$ can be reached from $\confhat$ using allowed particles jumps. Notice that according to the exclusion rule, we can consider that any empty site is allowed to move freely by exchanging their place with any site next to it. 

We first bring ourselves back to a configuration described in Fig. \ref{Irred2trous}, where the two closest empty sites are brought next to the edge $a$. More precisely, we reach a configuration where the two empty sites and the two sites $a_1$ and $a_2$ are at the vertices of a side-1 square. From here, we are able to invert the two particles in $a_1$ and $a_2$ by a circular motion of the four empty sites along the edges of the square, and then bring back the empty sites along the paths that brought them next to $a$ to their original location. Doing so, one reaches exactly the configuration $\confhat^{a_1,a_2}$  from $\confhat$ with allowed particle jumps in $B$. 

We deduce from this last statement that for any pair of configurations $\confhat$, $\confhat'$ with the same particles in $B$, $\confhat'$ can be reached from $\confhat$ with jumps in $B$ since the transition can be decomposed along switches of nearest neighbor sites. The process is thus irreducible on the sets with fixed numbers $\K$ of particles, as soon as $K$ is smaller than $|B|-2$.  Furthermore, this construction ensures that any two neighboring particles can be switched with a number of particle exchanges of order $p$ where we denoted  by $p$ the size of the box. Since one needs to invert $2p$ pairs of particles at most to move one particle to its final position in $ \confhat'$, this proves the last statement. 
\end{proof}

\begin{figure}
\centering
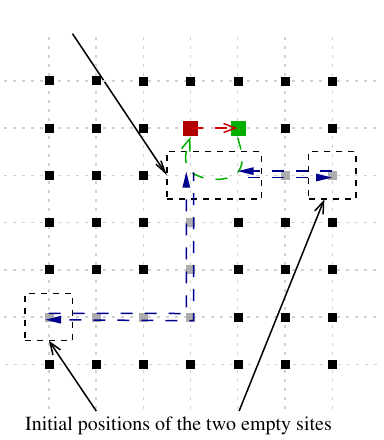
\caption{Reaching $\confhat^{a_1,a_2}$ from $\conf$ .}
\label{Irred2trous}
\end{figure} 

We now prove that large microscopic boxes are rarely fully occupied under the dynamics. Let us denote by $\epx$ the event 
\begin{equation} \label{epxdef}\epx=\left\{\sum_{y\in B_p(x)}\conf_y\leq \abss{B_p(x)}-2\right\},\end{equation}
\index{$E_{p,x} $\dotfill  $B_p(x)$ contains at least $2$ empty sites}
\index{$E_p $\dotfill $E_{p,0}$}
on which the box of size $p$ around $x$ contains at least two empty sites. When the site $x$ is the origin, we will simply write  $E_p$ instead of $E_{p,0}$. In order to ensure that full clusters very rarely appear in the dynamics, we need the following Lemma. 
\begin{prop}[Control on full clusters]
\label{prop:fullclusters}For any positive time  $T$,
\begin{equation}\label{fullceq}\lim_{p\to\infty}\lim_{N\to \infty}\E^{\lambda,\beta}_{\mu^N}\pa{\int_0^T\frac{1}{N^2}\sum_{x\in \torus}\1_{\epxc}(t)dt}=0.\end{equation}
\end{prop}
\begin{rema}[Scheme of the proof]We first sketch the proof in a continuous idealized setup to explain the general ideas before giving the rigorous proof. To prove that the box of microscopic size $p$ is not full, setting $p'=(2p+1)^2$ the cardinal of $B_p$, it is enough to prove thanks to the microscopic setting that 
\[\iint_{[0,T]\times\ctorus}\rho^{p'}_t(u)du dt\underset{p'\to \infty}{\to }0 ,\]
where $\rho_t(u)$ denotes the macroscopic density in $u$ at time $t$.

We expect the  total density $\rho$ to follow the partial differential equation 
\begin{equation}\label{aqd}\partial_t \rho=\Delta \rho-\nabla\cdot( m(1-\rho)),\end{equation}
where $m$ is an \emph{a priori} random quantity representing the local direction of the asymmetry, which can be represented as the vector field which would satisfy at any time $t$ and for any smooth function $H:\ctorus\to\R$
\[\int_{\ctorus}H(u)m_t(u)du=\lim_{N\to\infty}\frac{1}{N^2}\sum_{x\in \torus}H(x)\eta_x(t)\pa{\begin{matrix}
                                                            cos(\theta_x(t))\\
                                                            sin(\theta_x(t))
                                                           \end{matrix}}
.\]
Naturally, making sense of this quantity is not obvious, and it is not our purpose in this paragraph. For now, we carry on with our heuristic presentation, and therefore assume that \eqref{aqd} holds true. We can therefore formally write, letting $\phi(\rho)=1/(1-\rho)$
\begin{align}\label{milfullc}\partial_t \int_{\ctorus}\phi(\rho_t)du=&\int_{\ctorus}\phi'(\rho_t)\cro{\Delta \rho_t-\nabla\cdot (m_t(1-\rho_t))}du\nonumber\\
=&\int_{\ctorus}\phi''(\rho_t)\cro{-(\nabla \rho_t)^2+m_t(1-\rho_t)\nabla\rho_t}du\nonumber\\
\leq &\int_{\ctorus}\phi''(\rho_t)\cro{-(\nabla \rho_t)^2+\frac{(\nabla \rho_t)^2}{2}+{\norm{m_t}_{\infty}^2(1-\rho_t)^2}}du\\
\leq &\int_{\ctorus}\phi''(\rho_t){\norm{m_t}_{\infty}^2(1-\rho_t)^2}du\nonumber=2\norm{m_t}_{\infty}^2\int_{\ctorus}\phi(\rho_t)du\nonumber
\end{align}
One could then apply Gronwall's Lemma to obtain that for any time $t$, \[\int_{\ctorus}\phi(\rho_t)du\leq  e^{2\norm{m}_{\infty}^2 t} \int_{\ctorus}\phi(\rho_0)du .\]
Furthermore, for any time $t$, \[\int_{\ctorus}\phi(\rho_t)du\geq\frac{1}{\delta}\int_{\ctorus}\1_{\{\rho_t\geq 1-\delta\}}+\int_{\ctorus}\1_{\{\rho_t\leq 1-\delta\}}=\frac{1-\delta}{\delta}\int_{\ctorus}\1_{\{\rho_t\geq 1-\delta\}}+1,\]
therefore, for any time $t$,
\begin{equation}\label{majlimeplimNc}\int_{\ctorus}\1_{\{\rho_t\geq 1-\delta\}}\leq  \frac{\delta}{1-\delta}\cro{e^{2\norm{m}_{\infty}^2 t} \int_{\ctorus}\phi(\rho_0)du-1}\underset{\delta\to 0}{\to} 0 .\end{equation}
As a consequence, for any time $t$, we could therefore write 
\begin{equation}\label{crudemajdelta}\iint_{[0,T]\times\ctorus}\rho^{p'}_t(u)du dt\leq T(1-\delta)^{p'}+\iint_{[0,T]\times\ctorus} \1_{\{\rho_t\geq 1-\delta\}}.\end{equation}
The first term in the right-hand side vanishes for any fixed $\delta$ as $p'\to \infty$, whereas the second becomes as small as needed letting $\delta\to 0$.

Since our macroscopic density does not verify equation \eqref{aqd}, however, the operations above need to be performed in a microscopic setup. The derivation of equation \eqref{majlimeplimNc} is the purpose of Proposition \ref{prop:densityestimate}. Two intermediate Lemmas \ref{lem:Gronwalldensity} and \ref{lem:Gronwallasym} prove the microscopic equivalent of equation \eqref{milfullc}. 
\end{rema}

Before giving the proof of Proposition \ref{prop:fullclusters}, which is postponed to the end of the subsection, we give first the following estimate. 
\begin{prop}[High density estimate]\label{prop:densityestimate} Denote \[\rho_{\varepsilon N}=\frac{1}{2\varepsilon N+1}\sum_{|y|\leq \varepsilon N}\eta_y\]
the average density in a small mesoscopic box centered at $0$.
For any positive $0<\delta'<1/2$,  and any time $t>0$, we have the bound
\begin{equation}\label{deltatoz}\limep \limN\E^{\lambda,\beta}_{\mu^N}\pa{\frac{1}{N^2}\sum_{x\in \torus}\1_{\{\tau_x\rho_{\varepsilon N}(t)>1-\delta'/2\}}}\\
\leq \delta'C,\end{equation}
where $C$ is a finite constant depending continuously on $t$, and also depending on the asymmetry $\lambda$, and the initial profile $\widehat{\zeta}$.
\end{prop}
\begin{proof}[Proof of Proposition \ref{prop:densityestimate}]For any small $\delta>0$, let us denote by $\phi_{\delta}$ the application 
\[\func{\phi_{\delta}}{[0,1+\delta/2{}]}{\R_+}{\rho}{\frac{1}{1+\delta-\rho}}.\]
Note that all successive derivatives of order less than $k$ of $\phi_{\delta}$ are positive (and increasing) functions, and all are bounded by $C_k/\delta^{k+1}$ for some family of universal constants $(C_k)_{k>0}$.

We now fix a $C^1$ function $ H:\ctorus \to\R_+$, and assume that $\int_{\ctorus} H(u)du=1$. For any $u\in \ctorus$, we denote by $H_u$ the function 
\[H_u:v\mapsto H(u-v).\] 
In order to simplify the notations, for any configuration $ \confhat\in \statespace$, and given its empirical measure $\pi^N$, we shorten 
\begin{equation}\label{rnx}\rn{x}(\confhat):=<\pi^N,\hx>=\frac{1}{N^2}\sum_{y\in \torus}H\pa{\frac{x-y}{N}}\conf_y.\end{equation}
In some cases, this quantity could be larger than $1$, so that we need to take further precautions. For any fixed $\delta$ we will therefore assume that $N$ is large enough for the condition 
\[\frac{1}{N^2}\sum_{x\in \torus}H(x/N)\leq 1+\frac{\delta}{2},\]
to hold, which is possible because we assumed that $H$ is smooth and $\int_{\ctorus} H(u)du=1$. Note that this restriction to $N$ large enough is not an issue, because in all what follows, $H$ will be fixed and $N$ will go to $\infty$.

For $N$ large enough, the density $\rn{x}(\confhat)$ is now in the domain of $\phi_\delta$, we now write
\begin{equation}\label{gronfullceq}\partial_t\E^{\lambda,\beta}_{\mu^N}\pa{\frac{1}{N^2}\sum_{x\in \torus}\phi_{\delta}\pa{\rn{x}(\confhat)}}=\E^{\lambda,\beta}_{\mu^N}\pa{\frac{1}{N^2}\sum_{x\in \torus}L_N\phi_{\delta}\pa{\rn{x}(\confhat)}},\end{equation}
where $L_N$ is the generator of the complete process $L_N=N^2 \gene+ N\genwa+\genis$. Our goal is to apply Gronwall's Lemma to the expectation in the left-hand side, therefore we now need to estimate the right-hand side.

Since $\rn{x}$ does not depend on the angles of the particles, neither does $\phi_{\delta}\pa{\rn{x}}$, and the contribution of the Glauber part $\genis$ of the generator $L_N$ in the right-hand side above vanishes. The two other parts of the generator together yield the wanted bound, and are treated in separate lemmas for the sake of clarity. As mentioned earlier, these two lemmas are the microscopic equivalent of  equation \eqref{milfullc}. 
\begin{lemm}\label{lem:Gronwalldensity}[Contribution of the symmetric part] There exists a  sequence $(c_N(\delta,H))_{N\in \N}$ depending only on $\delta$ and $H$, vanishing as $N\to \infty$, and such that for any configuration $\confhat\in \statespace$
\begin{equation}\label{laplacien}\sum_{x\in \torus}\gene \phi_{\delta}\pa{\rn{x}}(\confhat)
\leq -\sum_{\substack{x\in \torus\\i=1,2}}\frac{\phi_{\delta}''\pa{\rn{x+e_i}}+\phi_{\delta}''\pa{\rn{x}}}{2}\pa{\rn{x+e_i}-\rn{x}}^2(\confhat)+ c_N(\delta,H).\end{equation}
\end{lemm}
\begin{lemm}\label{lem:Gronwallasym}[Contribution of the asymmetric part] There exists a  sequence $(\widetilde{c}_N(\delta,H))_{N\in \N}$ depending only on $\delta$ and $H$, \emph{vanishing as $N\to \infty$}, and such that  for any configuration $\confhat\in \statespace$
\begin{multline}\frac{1}{N}\sum_{x\in \torus}\genwa \phi_{\delta}\pa{\rn{x}}(\confhat)\\
\leq \sum_{x\in \torus}\cro{\sum_{i=1}^2\frac{\phi_{\delta}''\pa{\rn{x+e_i}}+\phi_{\delta}''\pa{\rn{x}}}{2}\pa{\rn{x+e_i}-\rn{x}}^2+\frac{{4}\lambda^2\phi_{\delta}\pa{\rn{x}}}{N^2}}(\confhat)+ \widetilde{c}_N(\delta,H).\end{multline}
\end{lemm}
\begin{proof}[Proof of Lemma \ref{lem:Gronwalldensity}]By definition of the symmetric part of the generator $\gene$, 
\[\sum_{x\in \torus}\gene \phi_{\delta}\pa{\rn{x}(\confhat)}=\sum_{x,y\in \torus}\sum_{i=1}^2\1_{\{\conf_y\conf_{y+e_i}=0\}}\cro{\phi_{\delta}\pa{\rn{x}(\confhat^{y,y+e_i})}-\phi_{\delta}\pa{\rn{x}(\confhat)}}.\]
We now develop the gradient of $\phi_{\delta}$ to the second order, to obtain that the right-hand side above is equal to 
\begin{multline*}\sum_{x,y\in \torus}\sum_{i=1}^2\1_{\{\conf_y\conf_{y+e_i}=0\}}\bigg[\phi_{\delta}'\pa{\rn{x}(\confhat)}\pa{\rn{x}(\confhat^{y,y+e_i})-\rn{x}(\confhat)}\\
+\frac{\phi_{\delta}''\pa{\rn{x}(\confhat)}}{2}\pa{\rn{x}(\confhat^{y,y+e_i})-\rn{x}(\confhat)}^2+o\pa{\pa{\rn{x}(\confhat^{y,y+e_i})-\rn{x}(\confhat)}^2}\bigg].\end{multline*}
Note that since the successive derivatives of order less than $k$ of $\phi_\delta$ are uniformly bounded on $[0,1]$ by $C_k/\delta^k$, the  vanishing quantity $ o\pa{\pa{\rn{x}(\confhat^{y,y+e_i})-\rn{x}(\confhat)}^2}$ can be bounded uniformly in $\confhat$, $x,y$ and $i$ (but not uniformly in $\delta$). 
Since $H$ is a smooth function,  \[\abs{\rn{x}(\confhat^{y,y+e_i})-\rn{x}(\confhat)}=\frac{1}{N^2}\abs{H_{x/N}\pa{\frac{y+e_i}{N}}-H_{x/N}\pa{\frac{y}{N}}} \] is of order $N^{-3}$, the contributions of the second line above are therefore at most of order  $N^{-2}$  and vanish in the limit $N\to \infty$. This yields
\begin{equation}\label{ippphi}\sum_{x\in \torus}\gene \phi_{\delta}\pa{\rn{x}}=\sum_{x\in \torus}\phi_{\delta}'\pa{\rn{x}(\confhat)}\sum_{y\in \torus}\sum_{i=1}^2\1_{\{\conf_y\conf_{y+e_i}=0\}}\pa{\rn{x}(\confhat^{y,y+e_i})-\rn{x}(\confhat)}+o_N(1),\end{equation}
where $o_N(1)$ is less than a vanishing sequence $(c^1_N)_{N\in \N}$ depending on $\delta$ and $H$ only.

Since for any $z\in \ctorus$, $H_u(v+z)=H_{u-z}(v)$, the definition of $\rn{x}$ yields 
\begin{align}\1_{\{\conf_y\conf_{y+e_i}=0\}}\pa{\rn{x}(\confhat^{y,y+e_i})-\rn{x}(\confhat)}=&\frac{1}{N^2}(\conf_y-\conf_{y+e_i})\pa{H_{x/N}\pa{\frac{y+e_i}{N}}-H_{x/N}\pa{\frac{y}{N}}}\nonumber\\
=&\frac{1}{N^2}\conf_{y}\pa{H_{x-e_i/N}\pa{\frac{y}{N}}-H_{x/N}\pa{\frac{y}{N}}}\nonumber\\
&-\frac{1}{N^2}\conf_{y+e_i}\pa{H_{x/N}\pa{\frac{y+e_i}{N}}-H_{x+e_i/N}\pa{\frac{y+e_i}{N}}}\nonumber.
\end{align}
Summing the quantity above over $y$, one obtains exactly $\rn{x-e_i}+\rn{x+e_i}-2\rn{x}$. This is the discrete Laplacian in the variable $x$ of $\rn{x}$, and a discrete integration by parts allows us to rewrite the first term on the right-hand side of equation \eqref{ippphi} as 
\begin{equation*}-\sum_{x\in \torus}\sum_{i=1}^2\pa{\phi_{\delta}'\pa{\rn{x+e_i}}-\phi_{\delta}'\pa{\rn{x}}}\pa{\rn{x+e_i}-\rn{x}}.\end{equation*}
We now write 
\[\pa{\phi_{\delta}'\pa{\rn{x+e_i}}-\phi_{\delta}'\pa{\rn{x}}}=\frac{(\phi_{\delta}''\pa{\rn{x+e_i}}+\phi_{\delta}''\pa{\rn{x}})}{2}\pa{\rn{x+e_i}-\rn{x}}+o\pa{\rn{x+e_i}-\rn{x}},\]
in which $\rn{x+e_i}-\rn{x}$ is of order $1/N$ because $H$ is a smooth function, to finally obtain that 
\begin{equation}\label{lapdisc}
\sum_{x\in \torus}\gene \phi_{\delta}\pa{\rn{x}}=-\sum_{x\in \torus}\sum_{i=1}^2\frac{\phi_{\delta}''\pa{\rn{x+e_i}}+\phi_{\delta}''\pa{\rn{x}}}{2}\pa{\rn{x+e_i}-\rn{x}}^2
+o_N(1),\end{equation}
where once again, the $o_N$ can be bounded by a vanishing sequence $(c_N)_N$ depending only on $\delta$, which completes the proof of Lemma \ref{lem:Gronwalldensity}
\end{proof}

\begin{proof}[Proof of Lemma \ref{lem:Gronwallasym}]
This proof follows the exact same steps as for the previous one. We first obtain by definition of $\genwa$ and developing the discrete gradient of $\phi$ that 
\begin{equation}\label{genwafullc}\frac{1}{N}\sum_{x\in \torus}\genwa \phi_{\delta}\pa{\rn{x}}=o_N(1)+\frac{1}{N}\sum_{x,y\in \torus}\sum_{i=1}^2(\tau_y\cur_i^{\lambda_i})\phi_{\delta}'\pa{\rn{x}(\confhat)}\pa{\rn{x}(\confhat^{y,y+e_i})-\rn{x}(\confhat)},\end{equation}
where $\cur_i^{\lambda_i}$ is defined according to equation \eqref{currentssym} as 
\[\cur_i^{\lambda_i}(\confhat)=\lambda_i(\theta_0)\conf_0(1-\conf_{e_i})-\lambda_i(\theta_{e_i})\conf_{e_i}(1-\conf_{0}),\]
and $o_N(1)$ is less than a vanishing sequence depending only on $\delta$ and $H$.
Once again, similar steps as in the previous case allow us to rewrite 
\begin{align*}(\tau_y\cur_i^{\lambda_i})\Big(\rn{x}&(\confhat^{y,y+e_i})-\rn{x}(\confhat)\Big)\\
=&\frac{1}{N^2}\cro{\lambda_i(\theta_y)\conf_y(1-\conf_{y+e_i})+\lambda_i(\theta_{y+e_i})\conf_{y+e_i}(1-\conf_{y})}\pa{H_{x/N}\pa{\frac{y+e_i}{N}}-H_{x/N}\pa{\frac{y}{N}}}\\
=&\frac{1}{N^2}\lambda_i(\theta_y)\conf_y(1-\conf_{y+e_i})\pa{H_{x/N}\pa{\frac{y+e_i}{N}}-H_{x/N}\pa{\frac{y}{N}}}\\
&+\frac{1}{N^2}\lambda_i(\theta_{y+e_i})\conf_{y+e_i}(1-\conf_{y})\pa{H_{x/N}\pa{\frac{y+e_i}{N}}-H_{x/N}\pa{\frac{y}{N}}}\\
=&\frac{1}{N^2}\lambda_i(\theta_y)\conf_y(1-\conf_{y+e_i})\pa{H_{x/N}\pa{\frac{y+e_i}{N}}-H_{x+e_i/N}\pa{\frac{y+e_i}{N}}}\\
&+\frac{1}{N^2}\lambda_i(\theta_{y+e_i})\conf_{y+e_i}(1-\conf_{y})\pa{H_{x-e_i/N}\pa{\frac{y}{N}}-H_{x/N}\pa{\frac{y}{N}}}
\end{align*}
Summing once again by parts in $x$, we obtain that the second term in the right-hand side of equation \eqref{genwafullc} is
\begin{align}\label{p1fullc}\frac{1}{N}\sum_{x,y\in \torus}&\sum_{i=1}^2(\tau_y\cur_i^{\lambda_i})\phi_{\delta}'\pa{\rn{x}(\confhat)}\pa{\rn{x}(\confhat^{y,y+e_i})-\rn{x}(\confhat)}\nonumber\\
&=\frac{1}{N^3}\sum_{x\in \torus}\sum_{i=1}^2\cro{\phi_{\delta}'\pa{\rn{x+e_i}(\confhat)}-\phi_{\delta}'\pa{\rn{x}(\confhat)}}\times\nonumber\\
&\quad\sum_{y\in\torus}\cro{\lambda_i(\theta_y)\conf_y(1-\conf_{y+e_i})H_{x+e_i/N}\pa{\frac{y+e_i}{N}}+\lambda_i(\theta_{y+e_i})\conf_{y+e_i}(1-\conf_{y})H_{x/N}\pa{\frac{y}{N}}}\nonumber\\
&:=S_1+S_2,\end{align}
where
\[S_1=\frac{1}{N^3}\sum_{x\in \torus}\sum_{i=1}^2\cro{\phi_{\delta}'\pa{\rn{x+e_i}(\confhat)}-\phi_{\delta}'\pa{\rn{x}(\confhat)}}\sum_{y\in\torus}\cro{\lambda_i(\theta_y)\conf_y(1-\conf_{y+e_i})H_{x+e_i/N}\pa{\frac{y+e_i}{N}}}\]
and 
\[S_2=\frac{1}{N^3}\sum_{x\in \torus}\sum_{i=1}^2\cro{\phi_{\delta}'\pa{\rn{x+e_i}(\confhat)}-\phi_{\delta}'\pa{\rn{x}(\confhat)}}\sum_{y\in\torus}\cro{\lambda_i(\theta_{y+e_i})\conf_{y+e_i}(1-\conf_{y})H_{x/N}\pa{\frac{y}{N}}}.\]
These two terms are treated in the exact same fashion, we therefore only treat in full detail the case of $S_1$, $S_2$ will follow straightforwardly. First, we develop the difference $\phi_{\delta}'\pa{\rn{x+e_i}(\confhat)}-\phi_{\delta}'\pa{\rn{x}(\confhat)}$ to the first order, 
\[\phi_{\delta}'\pa{\rn{x+e_i}}-\phi_{\delta}'\pa{\rn{x}}=\phi_{\delta}''\pa{\rn{x+e_i}}\pa{\rn{x+e_i}-\rn{x}}+o\pa{\rn{x+e_i}-\rn{x}}.\]
Once again,  $H$ being a smooth function, $\rn{x+e_i}-\rn{x}$ is of order $1/N$, therefore the $o\pa{\rn{x+e_i}-\rn{x}}$ is also a $o_N(1/N)$, and the corresponding contribution in $S_1$ vanishes in the limit $N\to\infty$. Recall that $\phi_{\delta}''$ is a positive function, we now apply in $S_1$ the elementary inequality $ab\leq a^2/2+b^2/2$ to  
\[a=\sqrt{\phi_{\delta}''}\pa{\rn{x+e_i}}\pa{\rn{x+e_i}-\rn{x}}\]
and 
\[b=\frac{1}{N^3}\sqrt{\phi_{\delta}''}\pa{\rn{x+e_i}}\sum_{y\in\torus}\cro{\lambda_i(\theta_{y+e_i})\conf_{y+e_i}(1-\conf_{y})H_{x/N}\pa{\frac{y}{N}}}.\]
This yields 
\begin{multline*}\abss{S_1}\leq o_N(1)+ \sum_{\substack{x\in \torus\\i=1,2}}\left[\frac{\phi_{\delta}''\pa{\rn{x+e_i}}}{2}\pa{\rn{x+e_i}-\rn{x}}^2\right.\\
+\left.\frac{\phi_{\delta}''\pa{\rn{x+e_i}}}{2N^6}\pa{\sum_{y\in\torus}\lambda_i(\theta_y)\conf_y(1-\conf_{y+e_i})H_{x+e_i/N}\pa{\frac{y+e_i}{N}}}^2\right].\end{multline*}
The function $H$ being non-negative, for any $y$, we can write \[\lambda_i(\theta_y)\conf_y(1-\conf_{y+e_i})H_{x+e_i/N}\pa{\frac{y+e_i}{N}}\leq \lambda(1-\conf_{y+e_i})H_{x+e_i/N}\pa{\frac{y+e_i}{N}}.\] 
Furthermore, since we assumed that $\int_{\ctorus} H=1$, and since $H$ is smooth, we  get that 
\[\frac{1}{N^2}\sum_{y\in \torus}H_{x/N}(y/N)=1+o_N(1),\]
 which yields 
\[\abs{\frac{1}{N^2}\sum_{y\in\torus}\lambda_i(\theta_y)\conf_y(1-\conf_{y+e_i})H_{x+e_i/N}\pa{\frac{y+e_i}{N}}}\leq\lambda (1-\rn{x+e_i})+o_N(1)\]
This, combined with the previous bound, yields that 
\[\abss{S_1}\leq o_N(1)+ \sum_{\substack{x\in \torus\\i=1,2}}\cro{\frac{\phi_{\delta}''\pa{\rn{x+e_i}}}{2}\pa{\rn{x+e_i}-\rn{x}}^2+\frac{\lambda^2\phi_{\delta}''\pa{\rn{x+e_i}}}{2N^2} (1-\rn{x+e_i})^2}.\]
A similar bound can be achieved for $S_2$, this time developing the difference $\phi_{\delta}'\pa{\rn{x+e_i}}-\phi_{\delta}'\pa{\rn{x}}$ in $\rn{x}$ instead of $\rn{x+e_i}$, 
\[\abss{S_2}\leq o_N(1)+ \sum_{\substack{x\in \torus\\i=1,2}}\cro{\frac{\phi_{\delta}''\pa{\rn{x}}}{2}\pa{\rn{x+e_i}-\rn{x}}^2+\frac{\lambda^2\phi_{\delta}''\pa{\rn{x}}}{2N^2} (1-\rn{x})^2}.\]
Combining these two bounds with identities \eqref{genwafullc} and \eqref{p1fullc}, we obtain that 
\begin{multline*}\frac{1}{N}\sum_{x\in \torus}\genwa \phi_{\delta}\pa{\rn{x}}\\
\leq \sum_{\substack{x\in \torus\\i=1,2}}\cro{\frac{\phi_{\delta}''\pa{\rn{x+e_i}}+\phi_{\delta}''\pa{\rn{x}}}{2}\pa{\rn{x+e_i}-\rn{x}}^2+\frac{\lambda^2\phi_{\delta}''\pa{\rn{x}}}{N^2} (1-\rn{x})^2}+o_N(1),\end{multline*}
where the $o_N(1)$ can be bounded by a vanishing sequence $(\widetilde{c}_N)_N$ depending only on $H$ and $\delta$. One easily obtains that for any non-negative $\delta$ and any $\rho\in [0,1+\delta/2]$,  \[(1-\rho)^2\phi_{\delta}''(\rho)\leq 2\phi_\delta(\rho),\] thus concluding the proof of Lemma \ref{lem:Gronwallasym}.
\end{proof}

We are now ready to apply Gronwall's Lemma and complete the proof of Proposition \ref{prop:densityestimate}. For that purpose, let us define \[\Phi(t)=\E^{\lambda,\beta}_{\mu^N}\pa{\frac{1}{N^2}\sum_{x\in \torus}\phi_{\delta}\pa{\rn{x}(t)}}.\]
according to the previous Lemmas \ref{lem:Gronwalldensity}, \ref{lem:Gronwallasym} and to equation \eqref{gronfullceq}, there exists a sequence $k_N=c_N+\widetilde{c}_N$ depending only on $\delta$ and $H$, verifying  
\[k_N\underset{N \to \infty}{\to} 0,\] 
and such that \[\partial_t \Phi(t)\leq {4}\lambda^2 \Phi(t)+k_N.\]
Since $\phi_{\delta}$ is bounded from below by $1/1+\delta$, $\Phi(t)$ also is, and therefore
\[\partial_t \Phi(t)\leq ({4}\lambda^2 +k_N(1+\delta))\Phi(t).\]
Gronwall's Lemma therefore yields that for any non-negative $t$, 
\[\E^{\lambda,\beta}_{\mu^N}\pa{\frac{1}{N^2}\sum_{x\in \torus}\phi_{\delta}\pa{\rn{x}(t)}}\leq \E^{\lambda,\beta}_{\mu^N}\pa{\frac{1}{N^2}\sum_{x\in \torus}\phi_{\delta}\pa{\rn{x}(0)}} e^{({4}\lambda^2 +k_N(1+\delta))t},\]
where this time the right-hand side depends on the trajectory only through its initial state $\confhat(0)$.

Fix a small $\delta'>0$.
$\phi_\delta$ being a non-decreasing function bounded from below by $1/1+\delta$, one can write for any $\rho\in[0,1+\delta/2]$
\begin{align*}\phi_{\delta}(\rho)\geq& \frac{1}{\delta+\delta'}\1_{\{\rho>1-\delta'\}}+\1_{\{\rho\leq 1-\delta'\}}\frac{1}{1+\delta}=\frac{1-\delta'}{(1+\delta)(\delta+\delta')}\1_{\{\rho>1-\delta'\}}+\frac{1}{1+\delta}\end{align*}
We apply this decomposition to the left-hand side of the inequality above, to obtain that 
\begin{multline}\label{gronapp}\E^{\lambda,\beta}_{\mu^N}\pa{\frac{1}{N^2}\sum_{x\in \torus}\1_{\left\{\rn{x}(t)>1-\delta'\right\}}}\\
\leq \frac{(1+\delta)(\delta+\delta')}{1-\delta'}\cro{\E^{\lambda,\beta}_{\mu^N}\pa{\frac{1}{N^2}\sum_{x\in \torus}\phi_{\delta}\pa{\rn{x}(0)}} e^{({4}\lambda^2 +k_N(1+\delta))t}-\frac{1}{1+\delta}}.\end{multline}
Coming back to the definition \eqref{rnx} of $\rn{x}$, for any smooth non-negative function $H$ with integral equal to $1$, taking the $\limsup$ $N\to \infty$, we thus obtain  from equation \eqref{gronapp}
\begin{multline}\label{fullcconc}\limN\E^{\lambda,\beta}_{\mu^N}\pa{\frac{1}{N^2}\sum_{x\in \torus}\1_{\left\{\rn{x}(t)>1-\delta'\right\}}}\\
\leq \limN\frac{(1+\delta)(\delta+\delta')}{1-\delta'}\cro{\E^{\lambda,\beta}_{\mu^N}\pa{\frac{1}{N^2}\sum_{x\in \torus}\phi_{\delta}\pa{\rn{x}(0)}} e^{{4}\lambda^2t}-\frac{1}{1+\delta}}.\end{multline}

Fix a small $\varepsilon>0$, and let us denote for any $u, v\in \ctorus$ \[H^{\varepsilon}(v)=\frac{1}{(2\varepsilon)^2}\1_{[-\varepsilon, +\varepsilon]^2}(v)\eqand H_u^{\varepsilon}(v)=\frac{1}{(2\varepsilon)^2}\1_{[-\varepsilon, +\varepsilon]^2}(v-u).\] 
Recalling that $\rho_{\varepsilon N}(t)$ is the empirical density in a box of size $\varepsilon N$ around the origin at time $t$, we can then write 
\[\tau_x\rho_{\varepsilon N}(t)=\frac{(2\varepsilon N)^2}{(2\varepsilon N+1)^2}\rho_x^{N,H^{\varepsilon}}=\rho_x^{N,H^{\varepsilon}}+o_N(1).\]

At this point, we want to apply equation \eqref{fullcconc} to $H=H^{\varepsilon}$, which is an indicator function, and thus need to be smoothed out. For that purpose, consider a sequence $(H_l^{\varepsilon})_{l\in \N}$ of functions such that 
\begin{itemize}
\item $\forall l\in \N$, $\forall u\in \ctorus$,  $ H_l^{\varepsilon}(u)\geq 0$  and $\underset{\ctorus}{\sup}\;H_l^{\varepsilon}=\underset{\ctorus}{\sup}\;H^{\varepsilon}=1/(2\varepsilon)^2$     .
\item $\forall l\in \N$, $H_l^{\varepsilon}\in C^1(\ctorus)$ and $\int_{\ctorus} H_l^{\varepsilon}(u) du=1$.
\item $H_l^{\varepsilon}(u)\neq H^{\varepsilon}(u)\Rightarrow \varepsilon -1/l <\norm{u }_{\infty}<\varepsilon+1/l$. 
\end{itemize}
The existence of such a sequence of functions is quite clear and is left to the reader.
In particular, the last condition imposes that \[I_l:=\int_{\ctorus}\1_{H_l^{\varepsilon}(u)\neq H^{\varepsilon}(u)}du\leq \frac{16\varepsilon}{l},\]
which is the area of the crown on which the two functions may differ.
The sequence $H_l^{\varepsilon}$ converges for any fixed $\varepsilon$ towards $H^{\varepsilon}$ in $L^1(\ctorus)$. Furthermore, notice that for any $x\in \torus$, since both the $H_l^{\varepsilon}$'s and $H^{\varepsilon}$ are bounded by $1/(2\varepsilon)^2$,
\begin{align*}\abs{\rho_x^{N,H_l^{\varepsilon}}-\rho_x^{N,H^{\varepsilon}}}&\leq\frac{1}{N^2}\sum_{y\in \torus}\conf_y\abs{H^{\varepsilon}_{l,x/N}\pa{\frac{y}{N}} - H_{x/N}^{\varepsilon}\pa{\frac{y}{N}}}\\
&\leq \pa{\frac{16\varepsilon}{l}+o_N(1)}\pa{\norm{H_l^{\varepsilon}}_{\infty}+\norm{H^{\varepsilon}}_{\infty}} = \frac{8}{\varepsilon l}+o_N(1),\end{align*}
where the last line represents the proportion of sites of the discrete torus in the crown around $u=x/N$ on which $H^{\varepsilon}_{l,x/N}$ and $H^{\varepsilon}_{x/N}$ can be different. The last observation yields that for any $x\in \torus$, we can write 
\[\abs{\tau_x\rho_{\varepsilon N}(t)-\rho_x^{N,H_l^{\varepsilon}}(t)}\leq  \frac{8}{\varepsilon l}+o_N(1),\]
where the $o_N(1)$ can be chosen independent of $\confhat$ and $x$. Fix $\varepsilon>0$ and consider $N_0$ and $l_0$ such that for any $N\geq N_0$ and any $l\geq l_0$, 
\[\abs{\tau_x\rho_{\varepsilon N}(t)-\rho_x^{N,H_l^{\varepsilon}}(t)}\leq \frac{\delta'}{2}.\]
For any such pair $l$, $N$, we therefore also have 
\[\1_{\{\tau_x\rho_{\varepsilon N}(t)>1-\delta'/2\}}\leq\1_{\left\{\rho_x^{N,H_l^{\varepsilon}}(t)>1-\delta'\right\}}.\]

 For any $l$, by our assumptions, equation \eqref{fullcconc} holds for $H=H_l^{\varepsilon}$ for any positive $\delta$ and $\delta'$. For any $l\geq l_0$, we can therefore write
\begin{multline}\label{eq:cor1Irred}\limN\E^{\lambda,\beta}_{\mu^N}\pa{\frac{1}{N^2}\sum_{x\in \torus}\1_{\{\tau_x\rho_{\varepsilon N}(t)>1-\delta'/2\}}}\\
\leq \limN \frac{(1+\delta)(\delta+\delta')}{1-\delta'}\cro{\E^{\lambda,\beta}_{\mu^N}\pa{\frac{1}{N^2}\sum_{x\in \torus}\phi_{\delta}\pa{ \rho^{N,H^{\varepsilon}_{l}}_x(0)}} 
e^{{4}\lambda^2t}-\frac{1}{1+\delta}}.\end{multline}
Recall that under $\Prob^{\lambda,\beta}_{\mu^N}$, the initial configuration $\confhat(0)$ is distributed according to a product measure fitting the initial profile $\zeta$ defined before \eqref{assumption0}. By law of large number, and since $\phi_\delta$ is smooth on $[0,1+\delta/2]$, we therefore obtain for any $v\in \ctorus$
\[\limsup_{N\to\infty}\E^{\lambda,\beta}_{\mu^N}\pa{\phi_{\delta}\pa{ \rho^{N,H^{\varepsilon}_{l}}_{\lfloor Nv \rfloor}(0)}}=\phi_{\delta} \pa{\zeta*H_l^{\varepsilon}(v)},\]
where $\lfloor Nv \rfloor=(\lfloor Nv_1 \rfloor,\lfloor Nv_2 \rfloor)\in \torus$ and $"*"$ denotes the convolution operator on $\ctorus$.
By dominated convergence theorem, we thus obtain 
\[\E^{\lambda,\beta}_{\mu^N}\pa{\frac{1}{N^2}\sum_{x\in \torus}\phi_{\delta}\pa{ \rho^{N,H^{\varepsilon}_{l}}_x(0)}}\xrightarrow[N\to \infty]{}\int_{\ctorus}\phi_{\delta}\pa{ \zeta*H_l^{\varepsilon}(v)}dv.\]

Since $\zeta$ and satisfies \eqref{assumption0}, it is bounded away from $1$ uniformly on $\ctorus$, $\zeta*H_l^{\varepsilon}$ is also bounded away from $1$ uniformly in $\varepsilon$, and therefore 
\[\phi_{\delta}\pa{ \zeta*H_l^{\varepsilon}(v)}\leq C^*, \]
where $C^*=C^*(\widehat{\zeta})$ is a constant which does not depend on $l$, $\varepsilon$, $v$ or $\delta$. Letting now $\delta$ go to $0$, we obtain from \eqref{eq:cor1Irred} and the limit above that for any $\varepsilon>0$ and any time $t$,
\[ \limN\E^{\lambda,\beta}_{\mu^N}\pa{\frac{1}{N^2}\sum_{x\in \torus}\1_{\{\tau_x\rho_{\varepsilon N}(t)>1-\delta'/2\}}}\\
\leq\frac{\delta'}{1- \delta'}(e^{{4}\lambda^2t}C^*-1),\]
which concludes the proof of Proposition \ref{prop:densityestimate} since we assumed $\delta'<1/2$.
\end{proof}

With the estimate stated in Proposition \ref{prop:densityestimate}, we are ready to prove Proposition \ref{prop:fullclusters}.
\proofthm{Proposition \ref{prop:fullclusters}}{First notice that in order to prove \eqref{fullceq}, it is sufficient to prove it both for $ F_{p,x}$ and $F'_{p,x}$ instead of $\epxc$, where 
\[F_{p,x}=\left\{\sum_{y\in B_p(x)}\conf_y=\abs{B_p(x)}\right\}\eqand F'_{p,x}=\left\{\sum_{y\in B_p(x)}\conf_y=\abs{B_p(x)}-1\right\}.\]
We focus on the first case, the second is derived in the exact same fashion.

  Unlike in \cite{Quastel1992}, the angle blind process's macroscopic density does not evolve according to the heat equation because of the weak drift. However, thanks to the bound  \eqref{entropybounds} on the entropy of the measure $\mu^N_t$ w.r.t. the reference measure $\mesref$ and on the Dirichlet form of the density $ f_t^N$,  local equilibrium holds for the angle-blind process. As a consequence, the replacement Lemma  \ref{lem:replacementlemma} holds for functions independent of the angles (cf. for example \cite{KLB1999}, p77). One therefore obtains that to prove 
  \begin{equation}\label{fullcF}\lim_{p\to\infty}\lim_{N\to \infty}\E^{\lambda,\beta}_{\mu^N}\pa{\int_0^T\frac{1}{N^2}\sum_{x\in \torus}\1_{F_{p,x}}(s)ds}=0, \end{equation}
one can replace $\1_{F_{p,x}(s)}$ by its expectation under the product measure with parameter $\tau_x\rho_{\varepsilon N}(s)$, namely 
\[\E_{\tau_x\rho_{\varepsilon N}(s)}(\1_{F_{p,x}})=\cro{\tau_x\rho_{\varepsilon N}(s)}^{p'},\]
where $p'=(2p+1)^2$ is the number of sites in $B_p$.

To prove  equation \eqref{fullcF}, it is therefore sufficient to prove that $\forall t\in [0,T]$, \begin{equation}\label{finfullc}\lim_{p'\to\infty}\limep\limN \E^{\lambda,\beta}_{\mu^N}\pa{\frac{1}{N^2}\sum_{x\in \torus}\cro{\tau_x\rho_{\varepsilon N}(t)}^{p'}}=0. \end{equation}
To prove the latter, since $\rho_{\varepsilon N}(t)$ is at most $1$, one only has to write, as outlined in equation \eqref{crudemajdelta},
\[ \E^{\lambda,\beta}_{\mu^N}\pa{\frac{1}{N^2}\sum_{x\in \torus}\cro{\tau_x\rho_{\varepsilon N}(t)}^{p'}}\leq (1-\delta)^{p'}+\E^{\lambda,\beta}_{\mu^N}\pa{\frac{1}{N^2}\sum_{x\in \torus}\1_{\{\tau_x\rho_{\varepsilon N}(t)>1-\delta\}}},\]
which holds for any positive $\delta$.

For any fixed $\delta>0$, the first term on the right-hand side vanishes as $p\to \infty$, whereas the second does not depend on $p$ and we can therefore let $\delta\to 0$ after $N\to \infty$, then  $\varepsilon \to 0$, then $p'\to \infty$. Since the right-hand side of equation \eqref{deltatoz} vanishes as $\delta'=2\delta$ goes to $0$, the left-hand side also does, and \eqref{finfullc} holds for any $t$ thanks to Proposition \ref{prop:densityestimate}. This proves equation \eqref{fullcF}, and the equivalent proposition with $ F'_{p,x}$ instead of $F_{p,x}$ is proved in the exact same fashion, thus concluding the proof of Proposition \ref{prop:fullclusters}.}

\section{Law of large number for the exclusion process with angles}
\label{sec:4}

\subsection{Replacement Lemma }
\label{gradientreplacement}
\intro{Our goal in this section is to close the microscopic equations and to replace in the definition of the martingale $M^{H,N}$ introduced in \eqref{martingaledef} any cylinder (in the sense of Definition \ref{def:conf}) function $\function(\confhat)$  by its spatial average $\E_{\densep}(\function)$, where $\densep$ is the empirical angular density over a small macroscopic box of size $\varepsilon N$. We use this Section to introduce new useful notations. 
The proof of the main result of this section, the Replacement Lemma \ref{lem:replacementlemma}, follows closely the usual strategy (c.f. Lemma 1.10 p.77 of \cite{KLB1999}), however it requires several technical adaptations due to the nature of our canonical and grand-canonical measure. In particular, we will need the topological setup and the various results obtained in Section \ref{sec:3}.
}

Consider a cylinder function $\function\in{ \mathcal C}$, and $l$ a positive integer. Recall from \eqref{averagedef} that $\langle \function \rangle_0^l$ is the average of the translations of $\function$ over a box of side $2l+1$ centered at the origin.  Recall from equation \eqref{empiricalprofile} and Definition \ref{defi:angleprofile} that the empirical angular density $\dens_l$ over the box $B_l$ of side $2l+1$ is the measure on $[0,2\pi[$
\[\dens_l=\frac{1}{\abss{B_l}}\sum_{x\in B_l}\conf_x \delta_{\theta_x}.\]

Define
\begin{equation}\label{Vldef}{\mathcal V}^l(\confhat)=\langle\function(\confhat)\rangle_0^{l}-\E_{\densl}(\function) \eqand {\mathcal W}^l(\confhat)=\function(\confhat)-\E_{\densl}(\function), \end{equation}
and for any smooth function  $G\in C(\ctorus)$, let 
\begin{equation}\label{xln}X^{l,N}(G,\confhat)=\frac{1}{N^2}\sum_{x\in\torus}G(x/N) \tau_x{\mathcal W}^l.\end{equation}

We first state that under the measure of active exclusion process, one can replace the average of $g$ over a small macroscopic box by its expectation w.r.t. the grand-canonical measure with  grand-canonical parameter $\densep$.
\begin{lemm}[Replacement Lemma]
\label{lem:replacementlemma}
For every $\delta>0$, we have with the notation \eqref{Vldef}
\[\limsup_{\varepsilon\to 0}\limsup_{N\to \infty}\Prob^{\lambda,\beta}_{\mu^N}\left[\int_0^T\frac{1}{N^2} \sum_{x\in \torus}\tau_x\abs{{\mathcal V}^{\varepsilon N}(\confhat(t) )}dt>\delta\right]=0.\]
\end{lemm}
The proof is postponed to Subsection \ref{subsec:replacementlemma}, and requires the control of the full clusters stated in Proposition \ref{prop:fullclusters}. For now, we can deduce from this lemma the following result, which will allow us to replace in \eqref{empiricalmeasure1} the currents by their spatial averages.
\begin{coro}
\label{LargeNumbersCor}
For every $\delta>0$, and any \emph{continuous} function \[\func{G}{[0,T]\times \ctorus}{\R}{(t,u)}{G_t(u)},\]
 we get with the notation \eqref{xln}
\[\limsup_{\varepsilon\to 0}\limsup_{N\to \infty}\Prob^{\lambda,\beta}_{\mu^N}\left[\abs{\int_0^T X^{\varepsilon N,N}(G_t,\confhat(t))dt}>\delta\right]=0.\]

\end{coro}
\begin{proof}[Proof of Corollary \ref{LargeNumbersCor}]Recall that $\varepsilon\to 0$ after $N\to \infty$, which means that the smoothness of $G$ allows us to replace in the limit  $G(x/N)$ by its spatial average on a box of size $\varepsilon$, which is denoted by
\[G^{\varepsilon N}(x/N):=\frac{1}{(2{\varepsilon N} +1)^2}\sum_{y\in B_{\varepsilon N}(x)}G(y/N).\]
More precisely, we can write, using notation \eqref{averagedef} for the local averaging, and since $g$ is a cylinder, hence bounded, function,
\begin{align}
\nonumber\limN \int_0^T\frac{1}{N^2}  \sum_{x\in\torus}G_t(x/N)\tau_x \function\; dt&=\limep\limN\int_0^T \frac{1}{N^2}\sum_{x\in\torus} G_t^{\varepsilon N}(x/N)\tau_x \function \;dt\\
\label{average1}&=\limep\limN\int_0^T\frac{1}{N^2} \sum_{y\in\torus}G_t(y/N)\langle \function\rangle_{y}^{\varepsilon N}\; dt,
\end{align}
where the average $\langle \function\rangle_{y}^{\varepsilon N}$ is defined in equation \eqref{averagedef}.

As a consequence, $\tau_yg$ can be replaced by its average $\langle g\rangle_y^{\varepsilon N}$. Note that
\[{\mathcal V}^{\varepsilon N}(\confhat)={\mathcal W}^{\varepsilon N}(\confhat)+\langle \function \rangle_{y}^{\varepsilon N}-\function,\]
 and that the replacement Lemma  \ref{lem:replacementlemma} implies in particular that for any bounded function $G\in C([0,T]\times \ctorus)$
\[\limsup_{\varepsilon\to 0}\limsup_{N\to \infty}\Prob^{\lambda,\beta}_{\mu^N}\left[\abs{\int_0^T\frac{1}{N^2} \sum_{x\in \torus}G_t(x/N)\tau_x{\mathcal V}^{\varepsilon N}(\confhat(t) )dt}>\delta\right]=0.\]
Therefore, thanks to equality \eqref{average1}, Corollary \ref{LargeNumbersCor} follows directly from Lemma \ref{lem:replacementlemma}.
\end{proof}

\subsection{Proof of the replacement Lemma }
\label{subsec:replacementlemma}

In order to prove the replacement Lemma \ref{lem:replacementlemma}, we will need the two lemmas below. The first one states that the average of any cylinder function $\langle\function(\confhat)\rangle_0^l$ over a large microscopic box (a box of size $l$ which tends to infinity after $N$) can be replaced by its expected value w.r.t. the grand-canonical measure whose parameter is the empirical density $\E_{\densl}(\function)$.

The second states that the empirical angular density does not vary much between a large microscopic box and a small macroscopic box.  We state these two results, namely the one and two-blocks estimates, in a quite general setup, because they are necessary in several steps of the proof of the hydrodynamic limit. 
\begin{lemm}[one-block estimate]
\label{lem:OBE}Consider $\am\in]0,1[$  and a density $f$ w.r.t the translation invariant measure $\mesref$ (cf. Definition \ref{defi:GCM}) satisfying
\begin{enumerate}[i)]
\item{There exists a constant $K_0$ such that for any $N$
\[H(f)\leq K_0 N^2\eqand \rdir\left(f\right)\leq K_0.\]}
\item{\begin{equation}\label{ii}\lim_{p\to \infty}\lim_{N\to \infty}\Eref\pa{f\frac{1}{N^2}\sum_{x\in \torus}\1_{\epxc}}=0.\end{equation}}
\end{enumerate}
Then, for  any cylinder function $\function$, 
\[\liml\limN \Eref\pa{f\frac{1}{N^2} \sum_{x\in \torus}\tau_x{\mathcal V}^{l}}=0,\]
where ${\mathcal V}^{l}$ was defined in \eqref{Vldef}.
\end{lemm}

\begin{lemm}[two-block estimate]
\label{lem:TBE}
For any  $\am\in]0,1[$ and any density $f$ satisfying conditions $i)$ and $ii)$ of Lemma \ref{lem:OBE}, 
\[\liml\limep\limN \sup_{y\in B_{\varepsilon N}}\Eref\pa{\frac{1}{N^2} \sum_{x\in \torus}\normm{\tau_{x+y}\densl-\tau_x\densep}f}=0,\]
where $\tau_{z}\dens_k$ is the local empirical angular density in the box of size $k$ centered in $z$ introduced in  \eqref{empiricalprofile}.
\end{lemm}
The proofs of these two lemmas will be presented resp. in Section \ref{subsec:OBE} and \ref{subsec:TBE}. For now, let us show that they are sufficient to prove the replacement Lemma \ref{lem:replacementlemma}.

\begin{proof}[Proof of Lemma \ref{lem:replacementlemma}]
Lemma \ref{lem:replacementlemma} follows from applying the two previous lemmas to the density \[\overline{f}_T^N=\frac{1}{T}\int_0^T f^N_t dt,\] where $f_t^N={d\mu_t^N}/{d\mesref}$, defined in Section \ref{subsec:entropy}, is the density of the active exclusion process at time $t$ started from $\mu^N$, and prove that Lemma \ref{lem:replacementlemma} follows. 
Proposition \eqref{prop:entropyproduction} proved that  $\overline{f}_T^N$ satisfies condition $i)$ of Lemma \ref{lem:OBE}.
Furthermore, $\overline{ f}_T^N$ also satisfies condition ii) 
\[\lim_{p\to \infty}\lim_{N\to \infty}\Eref\pa{\overline{f}_T^N\frac{1}{N^2}\sum_{x\in \torus}\1_{\epxc}}=0\]
 thanks to Proposition \ref{prop:fullclusters}, thus the one-block and two-blocks estimates apply to $f=\overline{ f}_T^N $. 

Now let us recall that we want to prove for any $\delta>0$
\[\limsup_{\varepsilon\to 0}\limsup_{N\to \infty}\Prob^{\lambda,\beta}_{\mu^N}\left[\int_0^T\frac{1}{N^2} \sum_{x\in \torus}\tau_x\abs{{\mathcal V}^{\varepsilon N}(\confhat(t) )}dt>\delta\right]=0,\]
where
\[{\mathcal V}^{\varepsilon N}(\confhat)=\langle\function(\confhat)\rangle_0^{\varepsilon N}-\E_{\densep}(\function).\]
Thanks to the Markov inequality, it is sufficient to prove that 
\[\limsup_{\varepsilon\to 0}\limsup_{N\to \infty}\E^{\lambda,\beta}_{\mu^N}\left[\int_0^T\frac{1}{N^2} \sum_{x\in \torus}\tau_x\abs{{\mathcal V}^{\varepsilon N}(\confhat(t) )}dt\right]=0.\]
We can now express the expectation above thanks to the mean density  $\overline{f}_T^N$. Since $T$ is fixed, to obtain the replacement Lemma  it is enough to show that 
\begin{equation}\label{RLexp}\limsup_{\varepsilon\to 0}\limsup_{N\to \infty}\Eref\pa{ \overline{f}_T^N\frac{1}{N^2} \sum_{x\in \torus}\tau_x\abs{{\mathcal V}^{\varepsilon N}(\confhat)}}=0.
\end{equation} 
For any function $\varphi(\cdot)$ on the torus $\torus$, recall that we denoted in \eqref{averagedef} by $\langle \varphi(\cdot) \rangle^l_x$ the average of the function $\varphi$ over a box centered in $x$ of size $l$, and that $\tau_y\densl$ is the empirical angular density in a box of size $l$ centered in $y$ defined in \eqref{empiricalprofile}. Let us add and subtract
\[\Bigl\langle \langle\function(\confhat)\rangle_0^l-\E_{\densl}(\function)\Bigr\rangle_0^{\varepsilon N}=\frac{1}{(2\varepsilon N +1)^2}\sum_{x\in B_{ \varepsilon N}}\left[\frac{1}{(2l+1)^2}\sum_{\abss{y-x}\leq l}\tau_y\function -\E_{\tau_x\densl}(\function)\right]\]
 inside $\abs{{\mathcal V}^{\varepsilon N}(\confhat)}$. We can then write thanks to the triangular inequality
\[\abs{{\mathcal V}^{\varepsilon N}(\confhat )}\leq ({\mathcal Z}^{l,\varepsilon N}_1+{\mathcal Z}^{l,\varepsilon N}_2+{\mathcal Z}^{l,\varepsilon N}_3)(\confhat), \]
where 
\[{\mathcal Z}^{l,\varepsilon N}_1=\abs{\frac{1}{(2\varepsilon N+1)^2}\sum_{x\in B_{\varepsilon N}}\left( \tau_x\function-\frac{1}{(2l+1)^2}\sum_{\abss{y-x}\leq l}\tau_y\function\right)},\] is the difference between $g$ and its local average,
\[{\mathcal Z}^{l,\varepsilon N}_2=\frac{1}{(2\varepsilon N+1)^2}\sum_{x\in B_{\varepsilon N}}\abs{ \E_{\tau_x\densl}(\function)-\frac{1}{(2l+1)^2}\sum_{\abss{y-x}\leq l}\tau_y\function},\]
is the difference between the local average of $g$ and its expectation under the product measure with parameter the local empirical angular density $\dens_l$, and
\[{\mathcal Z}^{l,\varepsilon N}_3=\frac{1}{(2\varepsilon N+1)^2}\sum_{x\in B_{\varepsilon N}}\abs{ \E_{\tau_x\densl}(\function)-\E_{{\densep}}(\function)}\]
is the difference between the expectations of $g$ under the empirical microscopic and macroscopic empirical angular density $\densl$ and $\densep$.

Let us consider the first term,  $N^{-2}\sum_x\tau_x{\mathcal Z}^{l,\varepsilon N}_1$. All the terms in  ${\mathcal Z}^{l,\varepsilon N}_1$ corresponding to the $x$'s in $B_{\varepsilon N-l}$ vanish, since they appear exactly once in both parts of the sum. The number of remaining terms can be crudely bounded by $4\varepsilon N l$, and each term takes the form $\tau_z\function/(2\varepsilon N +1 )^2$. Hence, we have the upper bound
\[\Eref\pa{\overline{f}_T^N \frac{1}{N^2} \sum_{x\in \torus}\tau_x{\mathcal Z}^{l,\varepsilon N}_1}\leq \frac{K l}{\varepsilon N}\Eref\pa{\overline{f}_T^N \frac{1}{N^2} \sum_{x\in \torus}\tau_x \abss{\function}}. \]
Since $\function$ is a bounded function, this expression can be bounded from above by \[\frac{K l\norm{g}_{\infty}}{\varepsilon N}\Eref\pa{\overline{f}_t^N}=C(l, \varepsilon, g)o_N(1),\]
which proves that 
\[\limep \limN\Eref\pa{ \frac{1}{N^2} \sum_{x\in \torus}\tau_x{\mathcal Z}^{l,\varepsilon N}_1\overline{f}_t^N}=0.\]

Now since \[\sum_{x\in \torus} \frac{1}{(2\varepsilon N+1)^2}\sum_{y\in B_{\varepsilon N}(x)}\tau_yg=\sum_{x\in \torus} \tau_x g,\] the two following terms can respectively be rewritten as
\begin{equation}
\label{secondterme}
\Eref\pa{\overline{f}_T^N\frac{1}{N^2} \sum_{x\in \torus}\tau_x{\mathcal Z}^{l,\varepsilon N}_2}=\Eref\pa{\overline{f}_T^N\frac{1}{N^2} \sum_{x\in \torus}\tau_x\abs{ \E_{\densl}(\function)-\langle\function\rangle_0^l}},
\end{equation}
and
\begin{equation}
\label{premierterme}
\Eref\pa{\overline{f}_T^N\frac{1}{N^2} \sum_{x\in \torus}\tau_x{\mathcal Z}^{l,\varepsilon N}_3}=\Eref\pa{\overline{f}_T^N\frac{1}{N^2} \sum_{x\in \torus}\abs{ \E_{{\tau_x}\densl}(\function)-\E_{\densep}(\function)}}.
\end{equation}
The quantity \eqref{secondterme} vanishes in the limit $N\to \infty$ then $l\to \infty$ thanks to the one-block estimate stated in Lemma \ref{lem:OBE}.

 Finally, according to Definition \ref{defi:convparam}, \eqref{premierterme} also vanishes thanks to the two-block estimate of Lemma \ref{lem:TBE} and the Lipschitz-continuity of the application
\[\func{\Psi_{\function}}{(\pset, \normm{\cdot})}{\R}{\param}{\Egcm\pa{\function}},\]
which was proved in Proposition \ref{prop:Lipschitzcontinuity}. 
 The Replacement Lemma \ref{lem:replacementlemma} thus follows from the one and two-blocks estimates.
 \end{proof}

In the next two Sections \ref{subsec:OBE} and \ref{subsec:TBE}, we prove the one-block and two-block estimates. The strategy for these proofs follows closely these presented in \cite{KLB1999}, albeit it requires some adjustments due to the measure-valued nature of the parameter of the product measure $\mesinv$ and the necessity to control the full clusters.

\subsection{Proof of Lemma \ref{lem:OBE}~: The one-block estimate}
\label{subsec:OBE}
\intro{The usual strategy to prove the one block estimate is to project the estimated quantity on sets with fixed number of particles, on which the density of $f$ should be constant thanks to the bound on the Dirichlet form.}

To prove the one-block estimate, thanks to the translation invariance of $\mesref$, it is sufficient to control the limit as $N$ goes to $\infty$, then $l\to \infty$ of 
\[\Eref\pa{f.\frac{1}{N^2} \sum_{x\in \torus}\tau_x{\mathcal V}^{l}}=\Eref({\mathcal V}^{l}\overline{f}),\]
where $\overline{f}=N^{-2} \sum_{\torus} \tau_xf$ is the average over the periodic domain of the translations of the density $f$.
Furthermore, define $s_g$ a fixed integer such that $g$ is measurable w.r.t. $(\confhat_x)_{x\in B_{s_g}}.$ We introduce for $l$ larger than $s_g$
\[\widetilde{{\mathcal V}}^{l}=\langle\function(\confhat)\rangle_0^{l-s_g}-\E_{\densl}(\function)={\mathcal V}^{l}+o_1(l),\]
where the $o_1(l)$ vanishes uniformly in $\confhat$ as $l\to\infty$. Proving the one block estimate for $\widetilde{{\mathcal V}}^{l}$ instead of ${\mathcal V}^{l}$ is therefore sufficient, and $\widetilde{{\mathcal V}}^{l}$ depends on the configuration only through the sites in $B_l.$

We first eliminate the configurations in which the box $B_l$ is almost full. Notice that the average $\widetilde{\mathcal V}^{l}$ is bounded because $\function$ is a cylinder function. We can therefore write \[\Eref(\widetilde{\mathcal V}^{l}\overline{f})\leq \Eref(\widetilde{\mathcal V}^{l}\1_{E_l}\overline{f})+C(g)\Eref(\1_{E_l^c}\overline{f}), \]
where $E_l$ is the event on which at least two sites are empty in $B_l$, defined after equation \eqref{epxdef}, and $E_l^c$ is its complementary event. The second term in the right-hand side vanishes by definition of $\overline{f}$, because $f$ verifies \eqref{ii}, and it is therefore sufficient to prove that \[\liml\limN\Eref(\widetilde{\mathcal V}^{l}\1_{E_l}\overline{f})=0.\]
Furthermore, the convexity of the Dirichlet form and the entropy yield that condition $i)$ of the one-block estimate is also satisfied by $\overline{f}$.
Since $\widetilde{\mathcal V}^{l}\1_{E_l}$ depends on $\confhat$ only through the $\confhat_x$'s in the cube ${B_{l}}$ we can replace the density $\overline{f}$ in the formula above by its conditional expectation $\overline{f}_{l}$, defined, for any configuration $\confhat'$ on $B_{{l}}$ by
\[\overline{f}_{l}(\confhat')=\Eref(\overline{f}\mid \confhat_x=\confhat'_x, \; x\in B_{{l}}).\]
For any function $f$ depending only on sites in $ B_l$ let $\E^*_{\am,l}$ be the expectation with respect to the product measure $\mesref$ over $B_l$.
With the previous notations, and in order to prove the one-block estimate, it is sufficient to prove that 
\[\liml\limN \E^*_{\am,l}\pa{\widetilde{\mathcal V}^{l}\1_{E_l}\overline{f}_l}\leq 0.\]

In order to proceed, we need to estimate the Dirichlet form and the entropy of $\overline{f}_l$ thanks to that of $f$, and prove the following Lemma
\begin{lemm}\label{lem:dirent}We have the following bounds
\begin{equation}\label{direntl}\rdir_{{l}}\left(\overline{f}_l\right)\leq C(l)N^{-2}\eqand H(\overline{f}_l)\leq C(l ).\end{equation}
\end{lemm}
\begin{proof}[Proof of Lemma \ref{lem:dirent}]\

\noindent {\bf Estimate on the Dirichlet form of $\overline{f}_l$ - } we denote by $\gene_{x,y}$ the symmetric part of the exclusion generator corresponding to the transfer of a particle between $x$ and $y$ 
\[\gene_{x,y}f(\confhat)=\pa{\conf_x-\conf_y}(f(\confhat^{y,x})-f(\confhat)),\]
\index{$\gene_{x,y} $\dotfill part of $\gene$ due to jumps between $x$ and $y$}
and by $\rdir^{x,y}$ the part of the Dirichlet form of the exclusion process corresponding to $\gene_{x,y}$ 
\[\rdir^{x,y}(f)=-\Eref\left( \sqrt{f} \gene_{x,y} \sqrt{f} \right).\]
\index{$ \rdir^{x,y}$\dotfill part of the Dirichlet form due to $\gene_{x,y}$}
With this notation, we have \[\rdir(f)=\sum_{\abss{x-y}=1}\rdir^{x,y}(f),\]
where $\rdir$ is the Dirichlet form introduced in equation \eqref{reddirdef}.
We denote in a similar fashion the Dirichlet form restricted to the box of size ${l}$ for any function $h$ depending only on the sites in $B_{{l}}$ by
\[D_{{l}}^{x,y}(h)=-\E^*_{\am,{l}}\left( \sqrt{h}\gene_{x,y} \sqrt{h} \right).\]
Since the conditioning  $f\mapsto f_l$ is an expectation, and since the Dirichlet elements $D_l^{x,y}$ are convex, the inequality 
\[D_{{l}}^{x,y}(\overline{f}_{l})\leq \rdir^{x,y}(\overline{f})\]
follows from Jensen's inequality. We deduce from the previous inequality, by summing over all edges $(x,y)\in B_l$, thanks to the translation invariance of $\overline{f}$, that
\[ D_{{l}}(\overline{f}_{l})\leq \sum_{(x,y)\in B_l} \rdir^{x,y}(\overline{f})= 2l(2l+1)\sum_{j=1}^2 \rdir^{0,e_j}(\overline{f})=\frac{(2l+1)^2}{N^{2}}\rdir(\overline{f}),\]
where $D_{{l}}$ is the Dirichlet form of the process restricted to the particle transfers with both the start and end site in $B_l$.
Up to this point, we have proved that for any function $f$ such that $\rdir(\overline{f})\leq \rdir(f)\leq K_0$, we have as wanted
\begin{equation}\label{boundDirOBE} D_{{l}}(\overline{f}_{l})\leq C_1(l)N^{-2}.\end{equation}

\noindent{\bf Estimate on the entropy of $\overline{f}_l$ -} recall that we defined the entropy $H(f)=\Eref(f \log f)$ and that we already established 
$H(\overline{f})\leq K_0N^2$. Let us partition $\torus$ in $q:=\lfloor N/({2l+1})\rfloor^2$  square boxes 
$B^1:=B_{{l}}(x_1),\ldots , B^q:=B_{{l}} (x_q)$, and $B^{q+1}$, which contains all the site that weren't part of any of the boxes. We can thus write \[\torus=\bigsqcup_{i=1}^{q+1} B^i.\]
We denote by $\confhat^i$ the configuration restricted to $B^i$ and by $\hat{\xi}^i$ the complementary  configuration to $\confhat^i$. In other words, for any $i\in \llbracket1,q+1\rrbracket $, we split any configuration on the torus $\confhat$ into  $\confhat^i$ and $\hat{\xi}^i$. We define for any $i\in \llbracket 1,q\rrbracket$ the densities on the $\confhat^i$'s
\[\overline{f}_l^i(\confhat^i)=\Eref\pa{\overline{f}(\confhat^i,\hat{\xi}^i)\left|\confhat^i\right.}.\]
Let us denote by $\varphi$ the product density w.r.t. $\mesref$ with the same marginals as $\overline{f}$, defined by  \[\varphi(\confhat)=\overline{f}_l^1(\confhat^1)\overline{f}_l^2(\confhat^2)\ldots \overline{f}_l^{q+1}(\confhat^{q+1}),\]
elementary entropy computations yield that \[H(\overline{f})=H_{\varphi}\pa{\overline{f}/\varphi}+\sum_{i=1}^{q+1}H\Big(\overline{f}_l^i\Big),\]
where $H_{\varphi}(f)=H( f\mesref \mid\varphi \mesref)$.
Since by construction $\overline{f}$ is translation invariant, for any $i=1,\ldots ,q$, we can write $H\Big(\overline{f}_l^i\Big)=H\Big(\overline{f}_l^1\Big)=H\pa{\overline{f}_l}$,
therefore in particular, the previous bound also yields, thanks to the non-negativity of the entropy, that \[H(\overline{f})\geq qH\pa{\overline{f}_l}.\]
Since $q$ is of order $ N^2/l^2$, this rewrites
\begin{equation}\label{boundEntOBE} H(\overline{f}_{l})\leq\frac{K_0 N^2}{q}\leq C_2(l),\end{equation}
and proves equation \eqref{direntl}.
\end{proof}

Thanks to Lemma \eqref{lem:dirent} we now reduced the proof of Lemma \ref{lem:OBE} to
\begin{equation}\label{pbcomp}\liml\limN \sup_{\substack{D_{{l}}(f)\leq C_1(l)N^{-2}\\
H(f)\leq C_2(l)}}\E^*_{\am,l}\pa{\widetilde{\mathcal V}^{l}\1_{E_l}f}=0.\end{equation}
Since the set of measures with density w.r.t. $\mesref$ such that $H(f)\leq C_2(l)$ is weakly compact, to prove the one block estimate of Lemma \ref{lem:OBE}, it is sufficient to show that
\[\liml \sup_{\substack{D_{{l}}(f)=0\\
H(f)\leq C_2(l)}}\E^*_{\am,l}\pa{\widetilde{\mathcal V}^{l}\1_{E_l}f}.\]

Before using the equivalence of ensembles, we need to  project the limit above over all sets with fixed number of particles $\subspace$ defined in equation \eqref{sousespace}. Recall from Definition \ref{defi:CM} the projection of the grand-canonical measures on the sets with fixed number of particles. For any density $f$ w.r.t. $\mesref$, such that $D_{{l}}(f)=0$, thanks to Section \ref{subsec:irreducibility} and the presence of the indicator function, $f$ is constant on $\Sigma_{l}^{\K}$ for any $\K\in \Ksett_l$. We therefore  denote, for any such $f$, by $f(\K)$ the value of $f$ on the set $\subspace$. Shortening $\int_{\K\in \Kset_l}$ for the sum $\sum_{K\leq (2l+1)^2}\int_{\theta_1\in\ctoruspi}\ldots \int_{\theta_K\in\ctoruspi}$,  we can write thanks to the indicator functions $\1_{E_l}$, for any $f$ satisfying $D_{{l}}(f)=0$,
\begin{equation}\label{dfz4}\E^*_{\am,l}\pa{\widetilde{\mathcal V}^{l}\1_{E_l}f}= \int_{\K\in \Ksett_l} f(\K)\E_{l, \K}(\widetilde{\mathcal V}^{l} ){d\mesref\pa{\subspace}},\end{equation}
where $\Ksett$ was defined in \eqref{Ksettdef}.

Since $\int_{\K\in \Kset_l}f(\K){d\mesref\pa{\subspace}}=1$ and $\E_{l,\K}\pa{\widetilde{\mathcal V}^{l}}\leq\sup_{\K\in \Ksett_l}\E_{l, \K}\pa{\widetilde{\mathcal V}^{l}} $, we obtain 
\[\liml\limN \sup_{\substack{D_{{l}}(f)\leq C_2(l)N^{-2}\\
H(f)\leq C_2(l)}}\E^*_{\am,l}\pa{\widetilde{\mathcal V}^{l}\1_{E_l}f}\leq  \liml\sup_{\K\in \Ksett_l}\E_{l, \K}\pa{\widetilde{\mathcal V}^{l}}.\]
To conclude the proof of equation \eqref{pbcomp} and the one-block estimate, it is therefore sufficient to prove that the right-hand side above vanishes.

\bigskip 

For any $\K\in \Kset_l$, recall that $\param_{\K}\in\pset$ is the grand-canonical parameter
\[\param_{\K}=\frac{1}{(2l+1)^2}\sum_{k=1}^{K}\delta_{\theta_k}\in\pset.\]
\index{$\param_{\K} $\dotfill grand-canonical parameter in $\pset$ associated with $\K$ }
Since the expectation $\E_{l, \K}$ conditions the process to having $K$ particles with angles $\Theta_K$ in $B_l$,  by definition of $\widetilde{\mathcal V}_l$, letting $l'=l-s_g$ we can write
\[\abs{\E_{l, \K}\pa{\widetilde{\mathcal V}^{l}}}\leq\E_{l, \K}\pa{\abs{\frac{1}{(2l'+1)^2}\sum_{x\in B_{l'}}\tau_x\function-\E_{\param_{\K}}(\function)}}.\]
\begin{figure}
\begin{center}
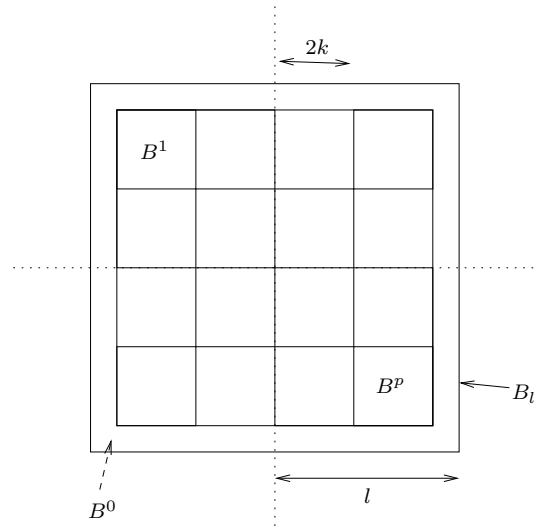
\caption{Construction of the $B^i$\label{Bis}}
\end{center}
\end{figure}
Let $k$ be an integer that will go to infinity after $l$, and let us divide $B_l$  according to Figure \ref{Bis} into $q$ boxes $B^1,\ldots ,B^q$, each of size $(2k+1)^2$, with $q=\lfloor \frac{2l+1}{2k+1}\rfloor^2$. let $k'=k-s_g$, $B'^{i}$ denotes the box of size $(2k'+1)$ centered inside $B^i$, and Let $B'^0=B_{l'}-\cup_{i=1}^qB'^{i}$, the number of sites in $B^0$ is bounded for some constant $C:=C(g)$ by $Ckl$.

With these notations, the triangular inequality yields 
\begin{align*}\E_{l, \K}\pa{\abs{\E_{\param_{\K}}(\function)-\frac{1}{(2l'+1)^2}\sum_{x\in B_{l'}}\tau_x\function}}\leq& \frac{\abss{B'^{1}}}{\abss{B_{l'}}}\sum_{i=0}^q\E_{l, \K}\pa{\abs{\E_{\param_{\K}}(\function)-\frac{1}{\abss{B'^{i}}}\sum_{x\in B'^{i}}\tau_x\function}}\\
=&\frac{(2k'+1)^2}{(2l'+1)^2}\sum_{i=1}^q\E_{l, \K}\pa{\abs{\E_{\param_{\K}}(\function)-\frac{1}{(2k'+1)^2}\sum_{x\in B'^{i}}\tau_x\function}}\\
&+O\pa{\frac{k}{l}}\end{align*}
Since the distribution of the quantity inside the expectation does not depend on $i$, the quantity above can be rewritten
\[\underset{\to 1}{\underbrace{q\frac{(2k'+1)^2}{(2l'+1)^2}}}\E_{l, \K}\pa{\abs{\E_{\param_{\K}}(\function)-\frac{1}{(2k'+1)^2}\sum_{x\in B_{k'}}\tau_x\function}}+O\pa{\frac{k}{l}}.\]
Because $\function$ is a cylinder function, and since $k$ goes to $\infty$ after $l$, the quantity inside absolute values is a local function for any fixed $k$. Letting $l$ go to $\infty$, the equivalence of ensembles stated in Proposition \ref{prop:equivalenceofensembles} allows us to replace the expectation above, uniformly in $\K$, by
\[\E_{\param_{\K}}\pa{\abs{\E_{\param_{\K}}(\function)-\frac{1}{(2k'+1)^2}\sum_{x\in B_{k'}}\tau_x\function}}.\]
Finally, since $\cup_{l\in \N}\{\param_{\K}, \K\in \Ksett_l\}\subset \pset$, where $\pset$ is the set of angle density profiles introduced in Definition \ref{defi:angleprofile}, \[\liml\ \sup_{\K\in \Kset_l}\E_{l, \K}(\widetilde{\mathcal V}^{l})\leq \sup_{\param\in \pset}\Egcm\pa{\abs{\Egcm(\function)-\frac{1}{(2k'+1)^2}\sum_{x\in B_{k'}}\tau_x\function}},\]
whose right-hand side vanishes as $k\to \infty$ by the law of large numbers, thus concluding the proof of the one-block estimate.

\subsection{Proof of Lemma \ref{lem:TBE}~: The two-block estimate}
\label{subsec:TBE}
\intro{This Sections follows the usual strategy for the two-block estimate, with small adaptations to the topological setup on the space of parameters $\pset$ introduced in Definition \ref{defi:convparam}.}

{Our goal is to show  that for any density $f$ satisfying conditions $i)$ and $ii)$ in Lemma \ref{lem:OBE},
\[\liml\limep\limN\sup_{y\in B_{\varepsilon N}}\Eref\pa{\frac{1}{N^2} \sum_{x\in \torus}\normm{\tau_{x+y}\densl-\tau_x\densep}f}=0.\]
The previous expectation can be bounded from above by triangle inequality  by
\[\Eref\pa{\frac{1}{N^2} \sum_{x\in \torus}\frac{1}{(2{\varepsilon N}+1)^2}\normm{\sum_{z\in B_{{\varepsilon N} }}\pa{\tau_{x+y}\densl-\tau_{x+z}\densl}}f}+o(l/\varepsilon N).\]
In this way, we reduce the proof to comparing average densities in two boxes of size $l$ distant of less than $2\varepsilon N$.
Let us extract in the sum inside the integral the terms in $z'$s such that $\abs{y-{z'}}\leq 2l$, the number of such terms is {at most $(4l+1)^2$}, and this quantity is bounded from above by 
\[\Eref\pa{\frac{1}{N^2} \sum_{x\in \torus}\frac{1}{(2{\varepsilon N}+1)^2}\normm{\sum_{\substack{z\in B_{{\varepsilon N} }\\
\abss{y-z}>2l}}\pa{\tau_{x+y}\densl-\tau_{x+z}\densl}}f}+o(l/\varepsilon N).\]
This separation was performed in order to obtain independent empirical measures  $\tau_{x+y}\densl$ and $\tau_{x+z}\densl$.
Regarding the expectation above, notice that we now only require to bound each term in the sum in $z$. In order to prove the two-block estimate, it is thus sufficient to show that
\[\liml{\limep}\limN \sup_{ 2l<\abss{y}<2\varepsilon N}\Eref\pa{\frac{1}{N^2} \sum_{x\in \torus}\normm{\tau_{x+y}\densl-\tau_{x}\densl}f}= 0.\]
As in the proof of  the one-block estimate, the expectation above can be rewritten
\[\Eref\pa{\normm{\tau_{y}\densl-\densl}\overline{f}},\]
where $\overline{f}=N^{-2}\sum_{x\in \torus}\tau_x f$ is the average of the density $f$. We can also introduce the cutoff functions $\1_{E_l}$ in the expectation above, thanks to $f$ satisfying \eqref{ii} and $\normm{\tau_{y}\densl-\densl}$ being a bounded quantity.

 Let $B_{y,l}$ be the set $B_l\cup \tau_y B_l$, the quantity under the expectation above is measurable with respect to the sites in $B_{y,l}$. Before going further, let us denote, for any configuration $\confhat\in \statespace$, $\confrond_1$ the configuration restricted to $B_l$ and $\confrond_2$ the configuration restricted to $y+B_l=\tau_yB_l$. We also denote by $\confrond $ the configuration $(\confrond_1, \confrond_2)$ on $B_{y,l}$. Let us finally write $\mu_{y,l}$ for the projection of the product measure $\mesref$  on $B_{y,l}$, and $\E_{y,l}$ the expectation with respect to the latter. 

\begin{figure}
\begin{center}
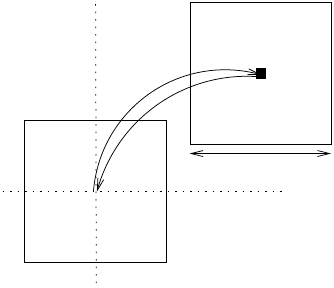
\caption{ \label{2BE}}
\end{center}
\end{figure}

With these notations, the expectation above can be replaced by
\[\Eref\pa{\normm{\tau_{y}\densl-\densl}\1_{E_l}\overline{f}_{y,l}},\]
where for any density $f$, $f_{y,l}$ is its {conditional expectation with respect to the sigma-field generated by $(\confhat_{x})_{x\in B_{y,l}},$}
\[f_{y,l}(\confrond)=\Eref\pa{f\mid \confhat_{\mid B_{y,l}}=\confrond},\]
which is well-defined because the two boxes $B_l$ and $\tau_yB_{l}$ are disjoint, thanks to the condition $\abss{y}> 2l$.

As in the proof of the one-block estimate, we now need to estimate the Dirichlet form of $\overline{f}_{y,l}$ in terms of that of $f$, on which we have some control. For that purpose, let us introduce with the notations of the previous Section
\begin{align}\label{dir12BE}D_{l,y}(h)&=-\E_{y,l}(h\gene_{0,y}h)-\sum_{\substack{x,z\in B_l\\
\abss{x-z}=1}}\E_{y,l}(h\gene_{x,z}h)-\sum_{\substack{x,z\in y+B_l\\
\abss{x-z}=1}}\E_{y,l}(h\gene_{x,z}h)\nonumber\\
&:=\quad\quad D^0_{l,y}\quad\quad\quad+\quad\quad\quad D^1_{l,y}\quad\quad\quad+\quad\quad\quad D^2_{l,y}\end{align}
the Dirichlet form corresponding to particle transfers inside the two boxes, and allowing a particle to transfer from the center of one box to the center of the other, according to Figure \ref{2BE}. The work of the previous section allows us to write that 
\[-\E_{y,l}(\overline{f}_{y,l}\gene_{x,z}\overline{f}_{y,l})\leq D^{x,z}(\overline{f}),\]
which implies, if $\rdir\pa{f}\leq C_0$ that 
\begin{equation}\label{dir22BE}D^1_{l,y}(\overline{f}_{y,l})+D^2_{l,y}(\overline{f}_{y,l})\leq 2C_0\frac{(2l+1)^2}{N^2},\end{equation}
by translation invariance of $\mu_{\param}$ and $\overline{f}$. We now only need to estimate the third term $D^0_{l,y}$. Let us consider a path $x_0=0, x_1,\ldots ,x_k=y$ of minimal length, such that $\abs{x_i-x_{i+1}}=1$ for any $i\in\{0,\ldots ,k-1\}. $ For any such path, we have $k\leq 4{\varepsilon N}$, since $\abss{y}\leq 2{\varepsilon N}$, and we can write 
\[D^0_{l,y}(\overline{f})\leq-\Eref(\overline{f}\gene_{0,y}\overline{f})=\frac{1}{2}\Eref \cro{\abs{\conf_0-\conf_y}(\overline{f}(\confhat^{0,y})-\overline{f}(\confhat))^2} \]
where $\confhat^{0,y}$ here is the state where the sites in $0$ and $y$ are inverted regardless of the occupation of either site. Since $\conf_0-\conf_y$ vanishes whenever both sites $0$ and $y$ are occupied or both are empty, we can for example assume that $\conf_0=1$ and $\conf_y=0$. For any configuration $\confhat^0=\confhat$, we let for any $i\in\{1,\ldots ,k\}$ \[{\confhat^{i}}=\pa{\confhat^{i-1}}^{x_{i-1}, x_i}\]
Thanks to the elementary inequality \[\pa{\sum_{j=1}^k a_j}^2\leq k\sum_{j=1}^k a_j^2,\]
and by definition of the sequence $({\confhat^{i}})_{i=0\ldots  k}$ (which yields in particular $\confhat^0=\confhat$ and $\confhat^k=\confhat^{0,y}$), the previous equation yields 
\begin{align*}
\Eref \cro{\conf_{0}(1-\conf_y)(\overline{f}(\confhat^{0,y})-\overline{f}(\confhat))^2}&\leq k \sum_{{i}=0}^{k-1} \Eref\cro{\conf_{0}(1-\conf_y)(\overline{f}({\confhat^{i+1}})-\overline{f}({\confhat^{i}}))^2} \\
&=k \sum_{{i}=0}^{k-1} \Eref\cro{\conf^i_{x_{i}}(1-\conf^i_{x_{i+1}})\cro{\overline{f}(\pa{\confhat^{i}}^{x_i, x_{i+1}})-\overline{f}({\confhat^{i}})}^2} \end{align*}
Since $\mesref$ is invariant through any change of variable $\confhat\to {\confhat^{i}}$, and since we can easily derive the same kind of inequalities with $\conf_y(1-\conf_0)$ instead of $\conf_{0}(1-\conf_y)$, we obtain that 
\begin{equation}\label{dir32BE}D^{0,y}_{l}(\overline{f})\leq k\sum_{i=0}^{k-1}D^{x_{i+1},x_i}(\overline{f})=k^2N^{-2}\rdir\pa{f}\leq 16\varepsilon^2 \rdir\pa{f}\end{equation}
thanks to the translation invariance of $\overline{f}$.
Finally, equations \eqref{dir12BE}, \eqref{dir22BE} and \eqref{dir32BE} yield 
\begin{equation}\label{dirly}D_{l,y}(\overline{f}_{y,l})\leq 2C_0\frac{(2l+1)^2}{N^2}+16C_0\varepsilon^2,\end{equation}
which vanishes as $N\to \infty$ then $\varepsilon\to 0$.
A bound on the entropy analogous to \eqref{direntl} is straightforward to obtain. Finally, to prove the two-block estimate, as in the proof of the one-block estimate, we can get back to proving that
\begin{equation}
\label{eq:TBEfini}
\liml\limep\limN  \sup_{2l<\abss{y}< 2\varepsilon N}\sup_{D_{l,y}(f)\leq 2C_0\frac{(2l+1)^2}{N^2}+16C_0\varepsilon^2}\E_{y,l}\pa{\normm{\tau_y\densl-\densl}\1_{E_l}f}=0. 
\end{equation}
Any {density satisfying the bound $D_{l,y}(f)\leq 2C_0\frac{(2l+1)^2}{N^2}+16C_0\varepsilon^2$} is ultimately constant on any set  with fixed number of particles and angles 
in the set $B_{y,l}$ with at least two empty sites. 
{More precisely, denote 
\[\param_{y, \ell}(\confhat)=\frac{1}{2(2l+1)^2}\sum_{x\in B_l\cup \tau_yB_l}\conf_x\delta_{\theta_x}\] 
the empirical canonical state of the configuration in 
$B_l\cup \tau_yB_l$, and denote by $\widehat{f}(\cdot)$ the conditional expectation of $f$ w.r.t. the canonical state of the configuration in $B_l\cup \tau_y B_l$, 
defined for any $\K$ on $B_l\cup \tau_y B_l$ by
\[\widehat{f}(\K)=\Eref\pa{f\left|\param_{y, \ell}(\confhat)=\param_{\K}\right.}.\]
We can now write for any $|y|> 2l$
\begin{multline*}
\E_{y,l}\pa{\normm{\tau_y\densl-\densl}\1_{E_l}f}\leq \int_{\Kset_{y,l}}\E_{\K,y,l}\pa{\normm{\tau_y\densl-\densl}}\widehat{f}(\K)d\K+\E_{y,l}\pa{\1_{E_l}\abs{f-\widehat{f}(\param_{y, \ell}(\confhat))}}\\
\leq \sup_{\K\in \Kset_{y_l,l}}\E_{\K,y_l,l}\pa{\normm{\tau_{y_l}\densl-\densl}}+\E_{y,l}\pa{\1_{E_l}\abs{f-\widehat{f}(\param_{y, \ell}(\confhat))}}, 
\end{multline*}
where we shortened $y_l=(2l+1)e_1$, $\Kset_{y,l}$ denotes the set of canonical parameters on $B_l\cup \tau_y B_l$, and 
$\E_{\K,y,l}(\cdot)=\Eref(\cdot\mid \param_{y, \ell}(\confhat)=\param_{\K})$. By compactness of the set of densities w.r.t. $\mesref$ on $B_l\cup \tau_y B_l$, 
the supremum over all densities satisfying $D_{l,y}(f)\leq 2C_0\frac{(2l+1)^2}{N^2}+16C_0\varepsilon^2$ of the second term above vanishes uniformly in $|y|> 2l$ 
as $N\to\infty$ and then $\varepsilon\to 0$, whereas the first term does not depend on $y$. To prove \eqref{eq:TBEfini}, it is therefore sufficient to prove that
\[\limsup_{l\to\infty}\sup_{\K\in \Kset_{y_l,l}}\E_{\K,y_l,l}\pa{\normm{\tau_{y_l}\densl-\densl}}=0,\]
which follows from the equivalence of ensembles.
}

}

\section{Preliminaries to the non-gradient method}
\label{sec:5}
\intro{The main focus of  Sections \ref{sec:5} and \ref{sec:6} is the symmetric part of the displacement process, whose contribution to the hydrodynamic limit requires the non-gradient method. Before engaging in the proof of the non-gradient estimates, however, we regroup several results which will be needed throughout the proof.  }

\subsection{Comparison with an equilibrium measure}
\label{subsec:FeynmanKac}
\intro{
In this section, we prove a result that will be used several times throughout the proof, and which allows to control the exponential moments of a functional $X$ by a variational formula involving the equilibrium measure $\mesref$. This control is analogous to the so called \emph{sector condition} for asymmetric processes, which ensures that the mixing due to the symmetric part of the generator is sufficient to balance out the shocks provoked by the antisymmetric part.}

\begin{rema}\label{rema:nonstationnarity}[Non-stationarity of $\mesref$ for the weakly asymmetric process]
It has already been pointed out that $\gene$ is self-adjoint w.r.t any product measure $\mesinv$, which is not in general the case of $\geniz$. However, $\geniz$ is self-adjoint w.r.t. $\mesref$ due to the uniformity in $\theta$ of that measure. Asymmetric generators are usually "almost" anti-self-adjoint, in the sense that one could expect ${\genwa}^*=-\genwa$.
This identity is for example true for the $TASEP$, for which the asymmetry is constant and does not depend on each particle.

It is not true in our case however, due to the exclusion rule and the dependency of the asymmetry in the angle of the particle. To clarify this statement, see the adjoint operator as a time-reversal, and consider a configuration with two columns of particles wanting to cross each other. This configuration would be stuck under $\genwa$, however, under the time-reversed dynamics ${\genwa}^*$, it starts to move. This illustrates that in our model, the asymmetric generator $\genwa$ is not anti-self-adjoint. 

Let us denote accordingly to the previous notation \eqref{currentssym} and recalling the definition of the $\lambda_i's$ \eqref{lidef}, for $ i=1,2$ 
\[\cur_i^{\lambda_i}=\lambda_i(\theta_0)\conf_0(1-\conf_{e_i})-\lambda_i(\theta_{e_i})\conf_{e_i}(1-\conf_0).\]
 Elementary computations yield accordingly that the adjoint in $L^2(\mesref)$ of $\genwa$ is in fact given by 
\begin{equation}\label{asymad}\genwad=-\genwa+2\sum_{x\in \torus}\sum_{i=1,2}\tau_x \cur_i^{\lambda_i}.\end{equation}
This identity will be necessary to prove the following result, which compares the measure of the process with drift to the measure $\mesref$.
\end{rema}

\begin{lemm}
\label{lem:FeynmanKac}
{Recall the topology on $\statespace$ introduced in Proposition \ref{prop:comparisonbetazero}, and fix a bounded measurable function 
\[\func{X}{\statespace\times[0,T]}{\R}{(\confhat,t)}{X_t(\confhat)}.\] 
}
For any $\gamma>0$, we have 
\[\frac{1}{\gamma N^2}\log \E^{\lambda,0}_{\mesref}\left[\exp \left(\gamma N^2\int_0^TX_t(\confhat(t))dt\right)\right]\leq\frac{2 T\lambda^2}{\gamma}+\frac{1}{\gamma }\int_0^Tdt\sup_{\varphi}\left\{\Eref\pa{\varphi \gamma X_t(\confhat)}-\frac{1}{2}\rdir(\varphi)\right\},\]
where the supremum in the right-hand side is taken on the densities w.r.t. $\mesref$.
\end{lemm}
\proofthm{Lemma \ref{lem:FeynmanKac}}{
Let us denote by $P_t^{\lambda, X}$  the modified semi-group 
\[P_t^{\lambda, X}=\exp\cro{\int_0^t L_N^{\beta=0}+\gamma N^2X_s ds}.\]
where  $L_N^{\beta=0}$ is the alignment-free generator introduced in \eqref{betaz} and let us denote in this section by $<.,.>_{\alpha}$ the inner product in $L^2(\mesref)$.   For any $i=1,\; 2$, and any $H$, and $T>0$, the Feynman-Kac formula yields 
\begin{align}\label{FC}\E^{\lambda, 0}_{\mesref}\left[\exp \left(\gamma N^2\int_0^TX_t(\confhat(t))dt\right)\right]=\;<1,P_T^{\lambda, X}1>_{\am}\;\leq\; <P_T^{\lambda, X}1,P_T^{\lambda, X}1>_{\am}^{1/2}.
\end{align}
by definition of $P_t^{\lambda, X}$, 
\begin{equation}\label{dertemps}\frac{d}{dt}<P_t^{\lambda, X}1,P_t^{\lambda, X}1>_{\am}=<P_t^{\lambda, X}1,(L_N^{\beta=0}+L_N^{\beta=0, *}+2\gamma N^2 X_t)P_t^{\lambda, X}1>_{\am},\end{equation}
where $M^*$ stands for the adjoint in $L^2(\mesref)$ of $M$. By definition of  $L_N^{\beta=0}$, we have
\[ L_N^{\beta=0,*}=N^2\gene^*+N\genwad+\genizad.\]

We now work to control the weakly asymmetric contribution in  the right-hand side of equation \eqref{dertemps}, which does not vanish in our case, as a consequence of Remark \ref{rema:nonstationnarity}. For that purpose, consider a function $\varphi\in L^{2}(\mesref)$, identity \eqref{asymad} yields 
\[<\varphi, (\genwa+\genwad)\varphi>_{\alpha}=2\sum_{x\in \torus}\sum_{i=1,2}\Eref\cro{\varphi^2\tau_x\cur_i^{\lambda_i}}.\]
Recall the definition of $\grad f$ given in equation \eqref{graddef}. A change of variable $\confhat\mapsto \confhat^{x,x+e_i}$ on the second part of $\tau_x \cur_i^{\lambda_i}$ yields that for any $x$ 
\[\Eref(\varphi^2\tau_x \cur_i^{\lambda_i})=-\Eref(\lambda_i(\theta_x)\nabla_{x,x+e_i}\varphi^2)=-\Eref\cro{\lambda_i(\theta_x)\pa{\varphi(\confhat^{x,x+e_i})+\varphi}\nabla_{x,x+e_i}\varphi}, \]
therefore applying the elementary inequality $ab\leq a^2/2+b^2/2$, to 
\[a=\sqrt{N}\nabla_{x,x+e_i}\varphi \eqand  b=-\frac{\lambda_i(\theta_0)}{\sqrt{N}}\pa{\varphi(\confhat^{x,x+e_i})+\varphi},\]
 we obtain (since $\lambda_i(\theta)$ is either $\lambda \cos(\theta)$ or $\lambda\sin(\theta)$ and is less than  $\lambda$)
\begin{equation*}<\varphi,(\genwa+\genwad)\varphi>_{\alpha}\leq \frac{N}{2}\sum_{x\in \torus}\sum_{i=1,2}\Eref\cro{\pa{\nabla_{x,x+e_i}\varphi}^2}+ \frac{\lambda^2}{2N }\sum_{x\in \torus}\sum_{i=1,2}\Eref\cro{(\varphi(\confhat^{x,x+e_i})+\varphi)^2}.\end{equation*}
Since $(\varphi(\confhat^{x,x+e_i})+\varphi)^2$ is less than $2\varphi^2(\confhat^{x,x+e_i})+2\varphi^2$, we finally obtain that, 
\[<\varphi,N(\genwa+\genwad)\varphi>_{\alpha}\leq - N^2\Eref\cro{\varphi \gene \varphi}+ 4\lambda^2 N^2\Eref\cro{\varphi^2}.\]

In particular, applying this identity to $\varphi=P_t^{\lambda, X}1$, we deduce from equation \eqref{dertemps} that
\begin{align*}\frac{d}{dt}<P_t^{\lambda, X}1,P_t^{\lambda, X}1>_{\am}\leq& <P_t^{\lambda, X}1,\cro{2\gamma N^2X_t+N^2\gene+2\geniz+4\lambda^2N^2}P_t^{\lambda, X}1>_{\am}\\
\leq&\;\pa{\nu_{\gamma} (t)+{4\lambda^2N^2}}<P_t^{\lambda, X}1,P_t^{\lambda, X}1>_{\am}+2<P_t^{\lambda, X}1,\geniz P_t^{\lambda, X}1>_{\am}
,\end{align*}
where $\nu_{\gamma }(t)$ is the largest eigenvalue of the self-adjoint operator $N^2\gene +2\gamma N^2X_t$. It is not hard to see that the second term above is non-positive. Indeed, for any function $\varphi$ on $\statespace$, by definition of $\geniz$ (cf. equation \eqref{genisdef})
\begin{align*}<\varphi,\geniz \varphi>_{\am}&=\sum_{x\in\torus}\Eref\pa{\conf_x\varphi(\confhat)\cro{\frac{1}{2\pi}\int_{\ctoruspi}\varphi(\confhat^{x,\theta})d\theta-\varphi(\confhat)}}\\
&=-\frac{1}{2}\sum_{x\in\torus}\Eref\pa{\conf_x\cro{\frac{1}{2\pi}\int_{\ctoruspi}\varphi(\confhat^{x,\theta})d\theta-\varphi(\confhat)}^2}\leq0.\end{align*}
To establish the last identity, we only used that under $\mesref$, the angles are chosen uniformly, and therefore $\Eref\pa{\conf_x\varphi(\theta_x)}=\Eref(\conf_x)(1/2\pi)\int_{\ctoruspi} \varphi(\theta')d\theta'$.
We thus obtain that  
\[\frac{d}{dt}<P_t^{\lambda, X}1,P_t^{\lambda, X}1>_{\am}\leq\pa{\nu_{\gamma} (t)+{4\lambda^2N^2}}<P_t^{\lambda, X}1,P_t^{\lambda, X}1>_{\am},\]
and Gr\"onwall's inequality therefore yields that
\[<P_T^{\lambda, X}1,P_T^{\lambda, X}1>_{\am}\leq \exp\pa{4T\lambda^2N^2+\int_0^T\nu_{\gamma}(t)dt}.\]
This, combined with \eqref{FC}, allows us to write 
\begin{equation}\label{FKForm}\frac{1}{\gamma N^2}\log \E^{\lambda, 0}_{\mesref}\left[\exp \left(\gamma N^2\int_0^T X_tdt\right)\right]\leq\frac{2 T\lambda^2}{\gamma }+ \int_0^T \frac{\nu_{\gamma}(t)}{2\gamma N^2}dt.\end{equation}
The variational formula for the largest eigenvalue of the self-adjoint operator $N^2(\gene+2\gamma X_t)$ yields that 
\begin{align*}\nu_{\gamma}(t)=&N^2\sup_{\psi, \;\Eref(\psi^2)=1}\Eref\pa{\psi(\gene+2\gamma X_t)\psi}=2N^2 \sup_{\varphi}\left\{\gamma\Eref\pa{X_t \varphi}-\frac{1}{2}\rdir(\varphi)\right\},
\end{align*}
where the second supremum is taken over all densities $\varphi$ w.r.t. $\mesref$, which together with \eqref{FKForm} concludes the proof of Lemma \ref{lem:FeynmanKac}. To prove the last identity, one only has to note that the supremum must be achieved by functions $\psi$ of constant sign, so that we can let $\varphi=\sqrt{\psi}$.

}

\subsection{Relative compactness of the sequence of measures}
\label{subsec:compactnessQN}
\intro{
We prove in this section that the sequence $(Q^N)_{N\in \N}$, defined in equation \eqref{QNdef}, is relatively compact for the weak topology.
It follows from two  properties stated in Proposition \ref{prop:relativecompactness} below. The first one ensures that the fixed-time marginals are controlled, whereas the second ensures that the time-fluctuations of the process's measure are not too wide. }

Given a  function $H:\ctorus\times \ctoruspi\to \R$, we already introduced  in the outline of Section \ref{subsec:outlineandcurrents} the notation  
\[<\pi, H>=\int_{\ctorus\times\ctoruspi}H(u, \theta )\pi(du,d\theta).\]
The following result yields sufficient conditions for the weak relative compactness of the sequence $(Q^N)_N$. Recall from equation \eqref{Mhatdef} the definition of the set of trajectories $\Mhat$.

\begin{prop}[Characterization of the relative compactness on $\mathcal{P}(\Mhat)$]
\label{prop:relativecompactness}
Let $P^N$ be a sequence of probability measures on the set of trajectories $ \Mhat$ defined in \eqref{Mhatdef}, such that
\begin{enumerate}[(1)]
\item{There exists some $A_0>0$ such that for any $A>A_0$, \[\limN P^N\pa{\sup_{s\in [ 0,T]}<\pi_s,1>\;\;\geq A}=0\]}
\item{For any $H\in C^{2,1}(\ctorus\times \ctoruspi)$, $\varepsilon>0$, \[\lim_{\delta\to 0}\limN P^N\pa{\sup_{\substack{\abss{t-t'} \leq \delta\\
0\leq t',t\leq T}}\abs{<\pi_{t'},H>-<\pi_{t},H>}>\varepsilon}=0.\]}
\end{enumerate}
Then, the sequence $(P^N)_{N\in \N}$ is relatively compact for the weak topology. 
\end{prop}
Since this proposition is, with minor adjustments, found in \cite{BillingsleyB1999} (cf. Theorem 13.2,  page 139), we do not give its proof, and refer the reader to the latter. For now, our focus is the case of the active exclusion process, for which both of these conditions are realized. The strategy of the proof follows closely that of Theorem 6.1, page 180 of \cite{KLB1999}, but requires two adjustments. First, our system is driven out of equilibrium by the drift, and we therefore need to use the Lemma \ref{lem:FeynmanKac} stated in the previous section to carry out the proof. The second adaptation comes from the presence of the angles, and since most of the proof is given for a test function $H(u,\theta)=G(u)\omega(\theta)$, we need to extend it in the general case where $H$ cannot be decomposed in this fashion.
\begin{prop}[Compactness of $(Q^N)_{N\in \N}$]
\label{prop:compactnessQN}
The sequence $(Q^N)_{N\in \N}$ defined in equation \eqref{QNdef} of probabilities on the trajectories of the active exclusion process satisfies conditions (1) and (2) above, and is therefore relatively compact.
\end{prop}
\proofthm{Proposition \ref{prop:compactnessQN}}{ The first condition does not require any work since the active exclusion process only allows one particle per site and  we can thus choose $A_0=1$. Regarding the second condition, recall that 
\begin{equation} \label{intpartcomp}<\pi^N_{t'},H>-<\pi^N_t,H>=\int_{t'}^t L_N<\pi^N_s,H> ds+M_t^H-M_{t'}^H,\end{equation}
where $M^H$ is a martingale with quadratic variation of order $N^{-2}$. For more details, we refer the reader to appendix A of \cite{KLB1999}.
First, Doob's inequality yields uniformly in $\delta$ the crude bound
 \begin{align}\label{martcompa}\E_{\mu^N}^{\lambda,\beta}\pa{\sup_{{t'},t\leq \delta}\abs{M_t^H-M_{t'}^H}}\leq 2\E^{\lambda,\beta}_{\mu^N}\pa{\sup_{0\leq t\leq T}\abs{M_t^H}}\leq C(H)N^{-1},\end{align}
where $\E_{\mu^N}^{\lambda,\beta}$ is the expectation w.r.t  the measure $\Prob_{\mu^N}^{\lambda,\beta}$ introduced just after Definition \ref{defi:GCM} of the complete process $\confhat^{[0,T]}$ started from the initial measure $\mu^N$.

Regarding the integral part of \eqref{intpartcomp}, we first assume like earlier that $H$ takes the form \[H(u, \theta)=G(u)\omega(\theta),\] where $G$ and $\omega$ are both $C^2$ functions. When this is not the case, an application of the periodic Weierstrass Theorem will yield the wanted result. Then, following the same justification as in Section \ref{subsec:outlineandcurrents} we can  write
\begin{align*}\int_{t'}^t L_N<\pi_s^{N},H> ds&=\frac{1}{N^2}\int_{t'}^t ds\sum_{x\in \torus} \tau_x\pa{\sum_{i=1}^2\cro{N\curom_i+\curaom_i}(s)\partial_{u_i,N} G(x/N)+ \tau_x\gamma^{\omega}(s) G(x/N) },
\end{align*}
where the instantaneous currents $\curom$, $\curaom$ and $\gamma^{\omega}$ were introduced in Definition \ref{defi:currents}.

The weakly asymmetric and Glauber contributions are easy to control, since both jump rates $\curaom$ and $\gamma^{\omega}$ can be bounded by a same constant $K$, and we can therefore write 
\begin{align*}\int_{t'}^t \pa{N\genwa+\genis}<\pi_s^{N},H> ds&\leq K \int_{t'}^t ds\frac{1}{N^2}\sum_{x\in \torus}\abs{ G(x/N)}+\sum_{i=1}^2 \abs{\partial_{u_i,N} G(x/N)} \\
&\to_{N\to \infty}K(t-t')\int_{\ctorus} \abs{G(u)}+\sum_{i=1}^2\abs{\partial_{u_i} G(u)}du, 
\end{align*}
which vanishes as soon as $\abs{t'-t}\leq \delta$ in the limit $\delta\to 0$. Finally, 
\begin{align*}Q^N\Bigg(\sup_{\substack{\abss{t-t'} \leq \delta\\
0\leq t',t\leq T}}|<\pi_{t'},H>-&<\pi_{t},H>|>\varepsilon \Bigg)\\
\leq&\;\Prob_{\mu^N}^{\lambda,\beta}\cro{\sup_{\substack{\abss{t-t'} \leq \delta\\
0\leq t',t\leq T}}\abs{\int_{t'}^t N^2\gene<\pi_s^{N},H> ds}>\varepsilon /3}\\
& +\Prob_{\mu^N}^{\lambda,\beta}\cro{\sup_{\substack{\abss{t-t'} \leq \delta\\
0\leq t',t\leq T}}\abs{\int_{t'}^t \pa{N\genwa+\genis}<\pi_s^{N},H> ds}>\varepsilon /3}\\
&+\Prob_{\mu^N}^{\lambda,\beta}\cro{\sup_{\substack{\abss{t-t'} \leq \delta\\
0\leq t',t\leq T}}\abs{M_t^H-M_{t'}^H}>\varepsilon /3}.\end{align*}
The second line of the right-hand side vanishes in the limit $N\to\infty$ then $\delta\to 0$ thanks to the computation above, whereas the third line also vanishes thanks to Markov's inequality and equation \eqref{martcompa}. Finally, the first term  vanishes accordingly to Lemma \ref{lem:compacteness} below and the Markov inequality, thus completing the proof in the case where $H(u,\theta)=G(u)\omega(\theta)$. The general case is derived just after the proof of Lemma \ref{lem:compacteness}.

\begin{lemm}\label{lem:compacteness}
For any function $H(u,\theta)=G(u)\omega(\theta)\in C^{2,0}(\ctorus\times \ctoruspi)$, 
\begin{equation}\label{limzero}\lim_{\delta\to 0}\limN \E^{\lambda,\beta}_{\mu^N}\pa{\sup_{\substack{\abss{{t'}-t} \leq \delta\\
0\leq {t'},t\leq T}}\abs{\int_{t'}^tN^2 \gene<\pi_s^{N},H> ds}}=0.\end{equation}
\end{lemm}
\proofthm{Lemma \ref{lem:compacteness}}{The proof of this Lemma follows, with minor adjustments to account for the drift, the proof given in \cite{KLB1999}. First, we get rid of the supremum and come back to the reference measure with fixed parameter $\alpha\in ]0,1[$ thanks to Lemma \ref{lem:FeynmanKac} of Section \ref{subsec:FeynmanKac}.  Let us denote 
\begin{equation}g(t)=\int_0^t N^2\gene<\pi_s^{N},H> ds.\end{equation} 
 We now compare the measure of the active exclusion process to that of the process started from equilibrium ($\mu^N=\mesref$), and with no alignment ($\beta=0$), 
 according to Proposition \ref{prop:comparisonbetazero} with $A=RN^2$ and 
\[X\pa{\confhat^{[0,T]}}=\sup_{\substack{\abss{{t'}-t} \leq \delta\\
0\leq {t'},t\leq T}}\abs{\int_{t'}^t N^2\gene<\pi_s^{N},H> ds}=\sup_{\substack{\abss{{t'}-t} \leq \delta\\
0\leq {t'},t\leq T}}\abss{g(t)-g({t'})}.\]
This yields that for some constant $K_0>0$, the expectation in equation \eqref{limzero} is bounded from above for any positive $R$ by 
\begin{equation}\label{entropy}\frac{1}{R N^2 }\cro{K_0N^2+\log\E_{\mesref}^{\lambda, 0}\exp\pa{R N^2\sup_{\substack{\abss{{t'}-t} \leq \delta\\
0\leq {t'},t\leq T}}\abss{g(t)-g({t'})}}}.\end{equation}
We therefore reduce the proof of Lemma \ref{lem:compacteness} to showing that
\begin{equation}\label{entropybeta2}\lim_{ \delta\to 0}\limN\frac{1}{R(\delta) N^2 }\log\E_{\mesref}^{\lambda,0}\exp\pa{R(\delta) N^2\sup_{\substack{\abss{{t'}-t} \leq \delta\\
0\leq {t'},t\leq T}}\abss{g(t)-g({t'})}}=0,\end{equation}
where $R(\delta)$ goes to $\infty$ as $\delta$ goes to $ 0$.

 Let $p$ and $\psi$ be two strictly increasing functions such that $\psi(0)=p(0)=0$ and $\psi(+\infty)=+\infty$, {with $\psi$ continuous}, we denote 
\[I=\int_{[0,T]\times [0,T]}\psi\pa{\frac{\abss{g(t)-g({t'})}}{p(\abss{{t'}-t})}}dt'dt,\]
the Garsia-Rodemich-Rumsey  inequality \cite{GRR1978} yields  that 
\begin{equation}\label{GRR1}\sup_{\substack{\abss{{t'}-t} \leq \delta\\
0\leq {t'},t\leq T}}\abss{g(t)-g({t'})}\leq 8\int_0^{\delta}\psi^{-1}\pa{\frac{4I}{u^2}}p(du).\end{equation}
Given any positive $a$, we choose $p(u)=\sqrt{u}$ and $\psi(u)=\exp(u/a)-1$, hence $\psi^{-1}(u)=a\log(1+u)$. An integration by parts yields {for any $\delta<e^{-2}$} that
\begin{align}\label{GRR2}\int_0^{\delta}\psi^{-1}\pa{\frac{4I}{u^2}}p(du)&=a\int_0^{\delta}\log\pa{1+\frac{4I}{u^2}}\frac{du}{2\sqrt{u}}\nonumber\\
&= a\sqrt{\delta}\log\pa{1+4I\delta^{-2}}+a\int_0^{\delta}\frac{8I}{u^3+4Iu}\sqrt{u}du \nonumber\\
&\leq a\sqrt{\delta}\log\pa{1+4I\delta^{-2}}+a\int_0^{\delta}\frac{2}{\sqrt{u}}du \nonumber\\
&=a\sqrt{\delta}\cro{\log\pa{\delta^{2}+4I}-2\log\delta+4}\nonumber\\
&\leq a\sqrt{\delta}\cro{-\frac{\log\delta}{2}\log\pa{\delta^{2}+4I}-4\log\delta}\nonumber\\
&\leq a\sqrt{\delta}\cro{-4\log\delta\log\pa{\delta^{2}+4I}-4\log\delta},
\end{align}
since {by assumption} $-\log(\delta)>2$. From equations \eqref{GRR1} and \eqref{GRR2} we deduce that 
\begin{align*}\log\E_{\mesref}^{\lambda, 0}\exp\pa{R N^2\sup_{\substack{\abss{{t'}-t} \leq \delta\\
0\leq {t'},t\leq T}}\abss{g(t)-g({t'})}}&\leq \log \E_{\mesref}^{\lambda, 0} \exp \pa{-32 a R N^2\sqrt{\delta}\log\delta\cro{1 +\log\pa{\delta^{2}+4I+1}}}\end{align*}
holds for any $a>0$.
For $ \delta<1$, Let us choose $ a=-(32R N^2\sqrt{\delta}\log\delta)^{-1}>0$, we can write for the second term of \eqref{entropy} the upper bound
\begin{align*}\frac{1}{R N^2}\log\E_{\mesref}^{\lambda, 0}\exp\pa{R N^2\sup_{\substack{\abss{{t'}-t} \leq \delta\\
0\leq {t'},t\leq T}}\abss{g(t)-g({t'})}}\leq \frac{1}{R N^2}\cro{1+ \log\pa{1+\delta^{2} +4\Egcm\pa{I}}}.\end{align*}
By definition, 
\[I=\int_{[0,T]\times [0,T]}\exp\pa{\frac{\abs{\int_{t'}^t N^2 \gene<\pi_u^{N},H> du}}{a\sqrt{\abss{t-t'}}}}dt'dt -T^2. \]
Let us assume, purely for convenience, that $T>1/2$, for $\delta$ sufficiently small, we have $4T^2-1-\delta^2>0$, 
and the quantity inside the limit in equation \eqref{entropybeta2} can be estimated by
\begin{multline} \label{majoration}\frac{1}{R N^2}\log\E_{\mesref}^{\lambda, 0}\exp\pa{R N^2\sup_{\substack{\abss{{t'}-t} \leq \delta\\
0\leq {t'},t\leq T}}\abss{g(t)-g({t'})}}\\
\leq \frac{1}{R N^2}\cro{1+\log4\E_{\mesref}^{\lambda, 0}\cro{\int_{[0,T]\times [0,T]}\exp\pa{\frac{\abs{\int_{t'}^t N^2 \gene<\pi_s^{N},H> ds}}{a\sqrt{\abss{{t'}-t}}}}d{t'}dt}}.
\end{multline}
If $T\leq1/2$, we simply carry out a constant term in the $\log $ above, which does not alter the proof.

Let us  take a look at the two constants $a$ and $R$. Noting the first bound on the entropy mentioned earlier, in order to keep the first term of \eqref{entropy} in check, 
$R=R(\delta)$ must simply grow to $\infty$. Furthermore, we previously obtained that $a=-(R N^232\sqrt{\delta}\log\delta)^{-1}$, we can choose $a=N^{-2}$, 
thus $R=(-1/32\sqrt{\delta}\log\delta)^{-1}$, which is non-negative, and goes to $\infty$ as $\delta\to 0^+$. Therefore, the second term above can be rewritten 
\begin{align*} 
\frac{1}{R N^2}\log\int_{[0,T]\times [0,T]}4\E_{\mesref}^{\lambda, 0}\exp\cro{\abs{\int_{t'}^t { \frac{N}{\abss{{t'}-t}^{1/2} }\sum_{\substack {x\in \torus\\{i=1,2}}} \curom_{x,x+e_i}(s)\partial_{u_i,N} G(x/N) ds}}d{t'}dt}.
\end{align*}
In order to estimate the expectation  above, we can get rid of the absolute value, since $e^{\abss{x}}\leq e^x+e^{-x}$, and since the function $G$ is taken in a symmetric class of functions. Furthermore, Lemma \ref{lem:FeynmanKac}, applied with $\gamma=1$ yields that the second term in the right-hand side of \eqref{majoration} is less than
\begin{equation}
\label{boundComp}
\frac{1}{R N^2}\log\int_{[0,T]\times [0,T]}\exp\cro{\frac{(t-t')}{2}\cro{4\lambda^2N^2+\nu_N{(G,\omega)}}}dt dt',
\end{equation}
where ${\nu_N(G,\omega)}$ is the largest eigenvalue in $L^2(\mesref)$ of the self-adjoint operator  
\[N^2\gene+  \frac{2N}{\abss{{t'}-t}^{1/2} }\sum_{\substack{x\in \torus\\ {i=1,2}}} \curom_{x,x+e_i}\partial_{u_i,N} G(x/N) ,\]
which can be rewritten as the variational formula
\begin{equation}\label{varfor}{\nu_N(G,\omega)}=\sup_{f}\left\{ \frac{2N}{\abss{{t'}-t}^{1/2} }\sum_{\substack {x\in \torus\\{i=1,2}}}\partial_{u_i,N} G(x/N)\Eref\pa{f \curom_{x,x+e_i} }-N^2\rdir(f)\right\},\end{equation}
where the supremum is taken on all densities $f$ w.r.t. $\mesref$.
 In order to prove that the eigenvalue above is of order $N^2$, we now want to transform 
 \[ \frac{N}{\abss{{t'}-t}^{1/2} }\sum_{x\in \torus}\partial_{u_i,N}G(x/N)\Eref\pa{f \curom_{x,x+e_i} }.\]
For any density $f$, and ${i=1,2}$, since $\curom_{x,x+e_i}(\confhat^{x,x+e_i})=-\tau_x{\curom_i}$, we can write 
\begin{align*}\Eref\pa{f \curom_{x,x+e_i} }\partial_{u_i,N}G(x/N)=&-\frac{1}{2}\Eref\cro{(f(\confhat^{x,x+e_i})-f)\curom_{x,x+e_i}}\partial_{u_i,N}G(x/N)\\
\leq&\frac{1}{4C}\Eref\pa{(\curom_{x,x+e_i})^2\pa{\sqrt{f}(\confhat^{x,x+e_i})-\sqrt{f}}^2}\\
&+\frac{C}{4}(\partial_{u_i,N}G(x/N))^2\Eref\pa{\pa{\sqrt{f}(\confhat^{x,x+e_i})+\sqrt{f}}^2}.\end{align*}
Since $(\curom_{x,x+e_i})^2\leq \norm{\omega}_{\infty}^2\1_{\conf_x \conf_{x+e_i}=0}$, and since $\cro{\sqrt{f}(\confhat^{x,x+e_i})+\sqrt{f}}^2\leq 2f(\confhat^{x,x+e_i})+2f$, we obtain the upper bound
\[\frac{N}{\abss{{t'}-t}^{1/2}}\sum_{x\in \torus}\partial_{u_i,N}G(x/N)\Eref\pa{f \curom_{x,x+e_i} }\leq \frac{N\norm{\omega}^2_{\infty}}{2C\abss{{t'}-t}^{1/2} }\rdir(f)+\frac{N^3C}{\abss{{t'}-t}^{1/2}}\norm{\partial_{u_i} G}^2_{\infty},\]
which holds for any positive $C$. We now set $C= \abss{{t'}-t}^{-1/2}\norm{\omega}^2_{\infty}/{N }$ so that the Dirichlet form contributions in the variational formula \eqref{varfor} cancel out. 
 We finally obtain that for some positive constant ${C_1=C_1(G, \omega)}$, independent of $N$, 
\begin{equation*}\nu_N(G,{\omega})\leq \frac{{C_1}N^2}{\abss{t-t'}},\end{equation*}
which yields that \eqref{boundComp} vanishes in the limit $N\to \infty$ and $\delta\to 0$, since $R=R(\delta)$ goes to $\infty$ as $\delta$ goes to $0$. Finally, we have proved thanks to equation \eqref{majoration} that \[\lim_{\delta\to 0}\limN\frac{1}{R N^2}\log\E_{\mesref}^{\lambda, 0}\pa{\exp\cro{R N^2\sup_{\substack{\abss{{t'}-t} \leq \delta\\
0\leq {t'},t\leq T}}\abss{g(t)-g({t'})}}}=0,\] which concludes the proof of Lemma \eqref{lem:compacteness}.}

In order to complete the proof of Proposition \ref{prop:compactnessQN}, we still have to consider the case when $H$ does not take a product form $G(u)\omega(\theta)$. In this case, since $H$ is smooth it can be approximated by a trigonometric polynomial in $u_1$, $u_2$ and $\theta$. Each term of the approximation is then of the form $G(u)\omega(\theta)$, and the previous result can therefore be applied. More precisely, consider a smooth function $H$, and for any $\alpha>0$, there exists a finite family $(p^{\alpha}_{ijk})_{0\leq i,j,k\leq M_{\am}}$ of coefficients such that 
\[\sup_{\substack{u\in \ctorus,\\ 
\theta\in \ctoruspi}}\abs{H(u,\theta)\;-\sum_{i,j,k\in \llbracket 0, M\rrbracket}p^{\alpha}_{ijk} u_1^i u_2^j\theta^k}\leq \alpha.\] 
Let us now fix an $\varepsilon>0$, and let us take $\alpha=\varepsilon/4$. Then, considering the corresponding family  $P_{ijk}(u, \theta)=p^{\alpha}_{ijk} u_1^i u_2^j\theta^k$ we have that  
\[\abs{<\pi^N_{t'},H>-<\pi^N_t,H>}\leq\abs{<\pi^N_{t'}-\pi^N_t,H-\sum_{i,j,k\leq M_{\am}}P_{ijk}>}+\sum_{i,j,j\leq M_{\am}}\abs{<\pi^N_{t'}-\pi^N_t,P_{ijk}>}.\]
Since we allow at most 1 particle per site, and since $H-\sum_{i,j,k\leq M_{\am}}P_{ijk}$ is smaller than $\varepsilon/4$, the first term of the right-hand side above is less than $\varepsilon/2$. From this, we deduce that for the left-hand side to be greater than $\varepsilon$, one of the terms $\abs{<\pi^N_{t'},P_{ijk}>-<\pi^N_t,P_{ijk}>}$ must be larger than $\varepsilon /2M_{\am}^{3}$. This yields that 
\begin{multline*} Q^N\pa{\sup_{\substack{\abs{s-t} \leq \delta\\
0\leq {t'},t\leq T}}\abs{<\pi_{t'},H>-<\pi_t,H>}>\varepsilon}\\
\leq \sum_{i,j,k\leq M_{\am}} Q^N\pa{\sup_{\substack{\abss{{t'}-t} \leq \delta\\
0\leq {t'},t\leq T}}\abs{<\pi_{t'},P_{ijk}>-<\pi_t,P_{ijk}>}>\frac{\varepsilon}{2M_{\am}^{3}}}.\end{multline*}
Since $\alpha$ is fixed, we can now take the limit $N\to \infty$ then $\delta\to 0$, in which the right-hand side vanishes since all functions are decorrelated in $u$ and $\theta$. The result thus holds for any smooth function $H$, thus completing the proof of Proposition \ref{prop:compactnessQN}.
}

We now prove that in the limit, the empirical measure of our process admits at any fixed time a density w.r.t. the Lebesgue measure on $\ctorus$.
\begin{lemm}\label{lem:Lebesguedensity}
\index{$ Q^*$ \dotfill a limit point of the sequence $(Q^N)_{N\in \N}$ }
Any limit point $Q^*$ of the sequence $Q^N$ is concentrated on measures $\pi\in \meas^{{[0,T]}}$ with time marginals absolutely continuous w.r.t the Lebesgue measure on $\ctorus$,
\[Q^*\pa{\pi, \;\pi_t(du, d\theta)=\boldsymbol \dens_t(u,  d\theta)du,  \quad\forall t\in[0,T]}=1.\]
\end{lemm}
\proofthm{Lemma \ref{lem:Lebesguedensity}}{For any smooth function $H\in C(\ctorus)$ configuration $\confhat$ in $\statespace$ and any corresponding empirical measure $\pi^N$, we have \[\abs{<\pi^N,H>}=\abs{\frac{1}{N^2}\sum_{x\in \torus}H(x/N)\conf_x}\leq\frac{1}{N^2}\sum_{x\in \torus}\abss{H(x/N)}.\]
The right-hand side above converges as $N$ goes to $\infty$ towards $\int_{\ctorus}\abss{H(u)}du$. Since for any fixed function $H$, the application \[ \pi\mapsto \sup_{0\leq t\leq T}\abs{<\pi_t,H>}\]
is continuous, any limit point $Q^*$ of $(Q^N)_N$ is concentrated on trajectories $\pi$ such that 
\[\sup_{{0\leq t\leq T}}\abs{<\pi_t,H>}\leq\int_{\ctorus}\abss{H(u)}du,\]
for any smooth function $H$ on $\ctorus$, and  therefore is absolutely continuous w.r.t. the Lebesgue measure on $\ctorus$.
}

\subsection{Regularity of the density and energy estimate}
\label{subsec:regularity}

\intro{In this section we prove {that the macroscopic particle density is regular enough for the weak hydrodynamic limit 
\eqref{equadiff} to be well defined, i.e. that criterion iii) of Definition \ref{defi:weaksol} is satisfied. 
The proof follows the same strategy as in \cite{KLB1999}, we give it for exhaustivity.}}

{

Due to the non-constant diffusion coefficients, the second derivative in equation \eqref{equadiff} cannot be applied to the test function, and we need, according to condition $iii)$ of Definition \ref{defi:weaksol}, 
to prove that the macroscopic profiles of our particle system are such that $\nabla\rho$ is well-defined. We can now state the following result.
\begin{theo}
\label{thm:regularity}Any limit point  $Q^*$  of the measure sequence $(Q^N)_N$ is concentrated on trajectories with 
$\rho_t(u)\in H_1=W^{1,2}([0,T]\times \ctorus)$. In other words, $Q^*$-a.s., there exists functions 
$\partial_{u_i}\rho_t(u)$ in $L^2([0,T]\times \ctorus)$ such that for any smooth function $H\in C^{0,1}([0,T]\times \ctorus)$
\index{$\partial_{u_i} $ \dotfill $i$-th continuous space derivative}
\begin{equation}\label{IPPsobolev}\iint_{[0,T]\times \ctorus}\rho_t(u)\partial_{u_i}H_t(u)dudt=-\iint_{[0,T]\times\ctorus} H_t(u)\partial_{u_i}\rho_t(u)dudt\end{equation}
Furthermore, there exists a constant $K=K(T, \lambda, \beta, {\widehat{\zeta}})$ such that for any limit point $Q^*$ of $ (Q^N)$, and for any $i$,
\begin{equation}\label{Enest}\E_{Q^*}\pa{\iint_{[0,T]\times\ctorus}[\partial_{u_i}\rho_t(u)]^2dudt}<K.\end{equation}
In particular, any such limit point $Q^*$ is concentrated on measures satisfying condition $iii)$ of Definition \ref{defi:weaksol}.
\end{theo}
The proof is postponed to the end of this section. The usual argument to prove this result is Riesz representation theorem, 
that yields that if
\[\iint_{[0,T]\times\ctorus}\rho_t(u)\partial_{u_i}H_t(u)dudt \leq C{\pa{\int_{[0,T]\times \ctorus}H^2}^{1/2}}\] 
for any $H$, there exists a function $\partial_{u_i}\rho\in L^2(  [0,T]\times \ctorus)$ such that \eqref{IPPsobolev} holds. 
For that purpose, we need the estimate given in Lemma \ref{lem:H1lemma1} below. Fix a direction $i\in\{1,2\}$, for any $x\in \torus$,  
shorten  $x_k=x+ke_i$, $k\in \{0,\dots,\varepsilon N\}$. Following the strategy of the energy estimate of \cite{KLB1999}, 
and recalling that $\tau_x\rho_{\delta N}$ is the empirical particle  density in  $B_{\delta N}(x)$, we let 
\begin{align*}
W_{N,i}(\varepsilon,\delta, H,\conf)&=\frac{1}{N^2}\sum_{x\in \torus}H(x/N)\pa{\frac{1}{\varepsilon }\cro{\tau_{x+\varepsilon N e_i}\rho_{\delta N}-\rho_{\delta N}}-H(x/N)}
.\end{align*}
Note that to emphasize that this quantity does not depend on the angles, we denote its third variable as $\conf$ instead of $\confhat$.

\begin{lemm}\label{lem:H1lemma1}Let $\{H^l, l\in \N\}$ be a dense sequence in the separable algebra $C^{0,1}([0,T]\times \ctorus)$ endowed with the norm 
$\norm{H}_{\infty}+\sum_{i=1}^2\norm{\partial_{u_i} H}_{\infty}$. For any ${i}=1,2,$  there exists a positive constant 
 $K=K(T,\lambda, \beta, {\widehat{\zeta}})$ such that for any $k\geq 1$ and $\varepsilon>0$, 
\[\limsup_{\delta\to 0}\limsup_{N\to\infty}\E_{\mu^N}^{\lambda, \beta}\pa{\max_{1\leq l\leq k}\int_0^TW_{N,i}(\varepsilon,\delta, H^l_t,\conf(t) )dt}\leq K_0.\]
\end{lemm}
\proofthm{Lemma \ref{lem:H1lemma1}}{
By the replacement Lemma  \ref{lem:replacementlemma}, it is sufficient to show the result above without the limit in $\delta$, 
and with $\widetilde{W}_{N,i}(\varepsilon, H,\conf)$ instead of $W_{N,i}$, where
\begin{align*}\widetilde{W}_{N,i}(\varepsilon, H,\conf)&=\frac{1}{N^2}\sum_{x\in \torus}H(x/N)\pa{\frac{1}{\varepsilon }
\cro{\conf_{x+\varepsilon N e_i}-\conf_x}-H(x/N))}\\
&=\frac{1}{N^2}\sum_{x\in \torus}H(x/N)\frac{1}{\varepsilon N }\sum_{k=0}^{\varepsilon N-1}
\Big[N(\conf_{x_{k+1}}-\conf_{x_k})- H(x/N)\Big].\end{align*}

Applying Proposition \ref{prop:comparisonbetazero} to $A=N^2$ and 
\[X\pa{\confhat^{[0,T]}}=\max_{1\leq i\leq k}\int_0^T\widetilde{W}_{N,i}(\varepsilon, {H_t^l},\conf(t))dt,\]
the contribution of the Glauber dynamics and the initial measure can be compared to the case $\beta=0$ started from $\mesref$, 
\begin{multline*}\E_{\mu^N}^{\lambda, \beta}\pa{\max_{1\leq l\leq k}\int_0^T\widetilde{W}_{N,i}(\varepsilon, H^l_t,\conf(t))dt}\\
\leq K_0(T, \beta, {\widehat{\zeta}})+ \frac{1}{N^2}\pa{\log\E_{\mesref}^{\lambda,0}\cro{\exp\pa{N^{2}\max_{1\leq l\leq k}\int_0^T\widetilde{W}_{N,i}(\varepsilon, H^l_t,\conf(t))dt}}}.\end{multline*}
The $\max$ can be taken out of the $ \log$ in the second term because for any finite family $(u_l)$, 
\[\exp\pa{\max_l u_l}\leq \sum \exp u_l \eqand  \limN N^{-2}\log\pa{\sum_l u_{l,N}}\leq \max_l \limN N^{-2}\log u_{N,l}.\]
Furthermore, we apply Lemma \ref{lem:FeynmanKac} to $\gamma=1$, and $X_t=\widetilde{W}_{N,i}(\varepsilon, H_t,\conf)$,
to obtain that
\begin{multline*}
\frac{1}{N^2}\log\E_{\mesref}^{\lambda,0}\cro{\exp\pa{N^{2}\int_0^T\widetilde{W}_{N,i}(\varepsilon, H_t,\conf(t))dt}}\\
\leq 2T\lambda^2+\frac{1}{2} \int_0^Tdt\sup_{\varphi}\left\{2\Eref\pa{\varphi\widetilde{W}_{N,i}(\varepsilon, H_t, \conf)}-\rdir\pa{\varphi}\right\}
,\end{multline*}
where the supremum is taken over all densities w.r.t. $\mesref$. Letting 
\[K(T,\lambda, \beta, {\widehat{\zeta}})=K_0(T,\beta, {\widehat{\zeta}})+2T\lambda^2,\]
to prove  Lemma \ref{lem:H1lemma1} it is therefore sufficient to show that the second term on the right-hand side of the inequality 
above is non-positive in the limit $N\to\infty$. This will be implied by  Lemma \ref{lem:H1lemma2} below, since the time integral is now only applied to $H$.
}

\begin{lemm}\label{lem:H1lemma2}For any $H\in C^{1}( \ctorus)$, and $\varepsilon>0$,  
\[\limsup_{N\to\infty} \sup_{\varphi }\left\{2\Eref\pa{\widetilde{W}_{N,i}(\varepsilon, H,\conf)\varphi}-\rdir(\varphi)\right\}\leq 0,\]
where the supremum is taken over the densities $\varphi$ w.r.t the product measure $\mesref$.
\end{lemm}
\proofthm{Lemma \ref{lem:H1lemma2}}{{The proof of this Lemma follows the exact same steps as the treatment of equation (7.3), p.106 in \cite{KLB1999}, we do not detail it:
since $\conf_{x_{k+1}}-\conf_{x_k}$ appearing in the expression of $\widetilde{W}_{N,i}(\varepsilon, H,\conf)$ can be rewritten $\conf_{x_{k+1}}(1-\conf_{x_{k}})-\conf_{x_k}(1-\conf_{x_{k+1}})$, the proof of the Lemma is just a matter of performing changes of variables $\confhat\mapsto \confhat^{x_k, x_{k+1}}$, and using the elementary inequality 
\[ab(c-d)\leq a^2(c+d)+\frac{b^2}{2}(\sqrt{c}-\sqrt{d})^2,\]
which holds for any positive $c$, $d$, to 
\[a=H(x/N),\quad b=\conf_{x_{k+1}}(1-\conf_{x_{k}}), \;\; \conf_{x_{k}}(1-\conf_{x_{k+1}}) , \quad c=\sqrt{\varphi}(\confhat^{x_k, x_{k+1}}), \eqand d=\sqrt{\varphi}\]
in the first term of $\widetilde{W}_{N,i}(\varepsilon, H,\conf)$.}
}

Lemma \ref{lem:H1lemma1} allows us to complete the proof of Theorem \ref{thm:regularity}. 
\proofthm{Theorem \ref{thm:regularity}}{Recall that we defined in Section \ref{subsec:processdef} $\Prob^{\lambda, \beta}_{\mu^N}$, 
the measure on the space $D([0,T], {\statespace})$ of the active exclusion process $\confhat(s)$ started with the measure $\mu^N$, 
and $Q^N$ is the measure on the corresponding measure space $\Mhat$. 
Let us introduce 
\[\varphi_{\delta}(u)=(2\delta)^{-2}\1_{[-\delta, \delta]^2}(u).\]
Since $\tau_x \rho_{\delta N}=\frac{(2\delta N)^2}{(2\delta N+1)^2}<\pi_t,\varphi_{\delta}(x/N-\cdot)>$ for any weak limit point  
$Q^*$ of $(Q^N)$, 
Lemma \ref{lem:H1lemma1} yields
\begin{equation*}
\limsup_{\delta\to 0}\E_{Q^*}\Bigg(\max_{1\leq l\leq k}\iint_{[0,T]\times\ctorus}\frac{H^l_t(u)}
{\varepsilon}\Big(<\pi_t,\varphi_{\delta}(u+\varepsilon e_i-\cdot)> -<\pi_t,\varphi_{\delta}(u-\cdot)>\Big)-H^l_t(u)^2 dudt\Bigg)\leq K.
\end{equation*}
Since thanks to Lemma \ref{lem:Lebesguedensity} any limit point $Q^*$ of $(Q^N)$ is concentrated on trajectories absolutely continuous w.r.t. the Lebesgue measure on $\ctorus$, letting $\delta$ then $\varepsilon$ go to $0$, by dominated convergence, we obtain that 
\begin{equation*}
\E_{Q^*}\pa{\max_{1\leq l\leq k}\iint_{[0,T]\times\ctorus}\cro{\partial_{u_i}H^l_t(u)\rho_t(u)-H^l_t(u)^2 }dudt}\leq K,
\end{equation*}
where $\rho_t(u)$ is the density of the measure $\int_{\ctoruspi}\pi_t (du, d\theta)$ w.r.t. the Lebesgue measure on $\ctorus$.
By monotone convergence, and since the sequence $(H_l)$ is dense in $C^{0,1}([0,T]\times \ctorus)$, we therefore obtain
\begin{equation}
\label{estimateH1}
\E_{Q^*}\pa{\sup_H\iint_{[0,T]\times\ctorus}\cro{\partial_{u_i}H_t(u)\rho_t(u)- H_t(u)^2 }dudt}\leq K,
\end{equation}
where the supremum is taken over all functions $H\in C^{0,1}([0,T]\times \ctorus)$. Given a limit point $Q^*$, let us denote by $\mathscr{E}$ the event on which the quantity inside parenthesis above is finite~: 
\[\mathscr{E}=\left\{\sup_H\iint_{[0,T]\times\ctorus}\cro{\partial_{u_i}H_t(u)\rho_t(u)- H_t(u)^2 }dudt<\infty\right\},\]
and denote by $\xi$ the elements of $\mathscr{E}$.
Then, thanks to the $L^1$ bound we just obtained, we have that $Q^*(\mathscr{E})=1$.

Define on $C^{0,1}([0,T]\times \ctorus)$ the linear operator
\[f_i(H)=\iint_{[0,T]\times\ctorus}\partial_{u_i}H_t(u)\rho_t(u)dudt, \]
then equation \eqref{estimateH1} yields that for any $\xi\in \mathscr{E}$, there exists a constant $K(\xi)$ 
such that for any positive constant $r$, $rf_i(H)-r^2\iint H^2\leq K(\xi)$, i.e.
\[f_i(H)\leq \frac{1}{r}K(\xi)+r \iint H^2.\]
Letting $r=\sqrt{K(\xi)/\iint H^2}$, and $C_0=2\sqrt{K(\xi)}$, we obtain that for any function $H\in C^{0,1}([0,T]\times \ctorus)$, 
\[f_i(H)\leq C_0(\xi)\pa{\iint_{[0,T]\times\ctorus} H^2}^{1/2}.\]
The functional $f_i$ can then be extended to a bounded linear functional in $L^2([0,T]\times \ctorus)$. 
The conclusion then follows from Riesz's representation theorem.
}}

\section{Non-gradient estimates}
\label{sec:6}

\subsection{Replacement of the symmetric current by a macroscopic gradient}

\intro{In this section, we focus on the complete exclusion process, and replace the current $\curom_i$ by a quantity of the form $\tau_{e_i} h-h+\gene f$, with $f$ a function of the configuration with infinite support. We then show that the perturbation $\gene f$ is of the same order as the weakly asymmetric contribution, and they both contribute to the drift term of equation \eqref{equadiff}. To obtain the non gradient estimates, we use the formalism developed in \cite{KLB1999} rather than that of \cite{Quastel1992}. This changes the proof substantially, with the upside that the orders in $N$, as well as the studied quantities, are clearly identified at any given point of the proof. }

\label{replacement1}
One of the challenges in proving the non-gradient hydrodynamic limit is to replace the local particle currents $\curom_i$ by the gradient of a function of the empirical measure. Recall that we already defined in equation \eqref{empiricalprofile}  the empirical angular density $\densl\in \pset$,
\[\densl=\frac{1}{(2l+1)^2}\sum_{x\in B_l}\conf_x\delta_{\theta_x},\]
and we denote by $\rho_l$ the empirical density 
\[\rho_l=\frac{1}{(2l+1)^2}\sum_{x\in B_l}\conf_x=\densl(\ctoruspi).\]
Let \[\densom_l=\frac{1}{(2l+1)^2}\sum_{x\in B_l}\com_x,\]
be the average of $\com$ over a box of side $2l+1$.
Finally,  for any function $\varphi$ on $\statespace$, recall that $\ddi$ is the discrete derivative  \[\ddi\varphi=\tau_{e_i}\varphi-\varphi\]
(for example, $\ddi\com_0=\com_{e_i}-\com_0$).

\bigskip

The usual strategy in the proof of the non-gradient hydrodynamic limit is to show that for some coefficients $\diffom, \; \diff:[0,1]\times \R\to \R^+$,  
\[\curom_i+\diffom\left(\rho_{\varepsilon N}, \densom_{\varepsilon N}\right)\ddi\densom_{\varepsilon N}+\diff\left(\rho_{\varepsilon N}, \densom_{\varepsilon N}\right)\ddi\rho_{\varepsilon N}\] vanishes as $N\to \infty$.
More precisely, the quantity above is in the range of the generator $\gene$, which is usually sufficient when the functions of the form $\gene f$ are negligible. In our case, however, due to the addition of a weak drift, the usual martingale estimate does not yield that $\gene f$ is negligible, but that $\genex f=(\genas+N^{-1}\genwa) f$ is negligible, therefore this perturbation can be integrated to the drift part, which is done in Section \ref{subsec:drift}. 

 For this replacement, we will need further notations similar to the ones introduced in Section \ref{gradientreplacement}. In our case, the diffusion coefficient $\diffom(\rho,\densom)$ is in fact the self-diffusion coefficient $\sdc(\rho)$, therefore we will from now on simply write $\sdc(\rho)$ for the diffusion coefficient relative to $\densom$. Note that it depends on the configuration only through the empirical density, and not on the particle angles.  For any positive integer $l$, and  any cylinder function $f$, let us thus denote 
\[{\V}_{i}^{f,\varepsilon N}(\confhat)=\curom_i+\sdc\left(\rho_{\varepsilon N}\right)\ddi\densom_{\varepsilon N}+\diff\left(\rho_{\varepsilon N}, \densom_{\varepsilon N}\right)\ddi\rho_{\varepsilon N}-\gene f,\]
where $\diff:[0,1]\times \R\to \R^+$ is the diffusion coefficient given in \eqref{diffconddef}.

We introduce for any smooth function $G\in C^{2}(\ctorus)$
\begin{equation}\label{Xdef}{X}_{i,N}^{f,\varepsilon N}(G,\confhat)=\frac{1}{N}\sum_{x\in \torus}G(x/N)\tau_x {\V}_{i}^{f,\varepsilon N}.\end{equation}
Our goal throughout this section is to prove that under the measure of our process, ${X}_{i,N}^{f,\varepsilon N}(G,\confhat)$ vanishes for any smooth function $G$, i.e. that the microscopic currents can be replaced by a macroscopic average of the gradients up to a perturbation $\gene f$ that will be dealt with later on. 

The sum contains $N^2$ terms, and the normalization is only $1/N$, therefore an order $N$ has to be gained, and this is the major difficulty of the non-gradient dynamics. To prove this statement, we decompose ${X}_{i,N}^{f,\varepsilon N}(G,\confhat)$ into distinct vanishing parts. 
 We already introduced in equation \eqref{epxdef} the set \[\epx=\left\{\sum_{\abss{y-x}\leq p}\conf_y\leq \abss{B_{p}}-2\right\},\]
such that at least two sites are empty in a vicinity of $x$ of size $p$.
 The cutoff functions $\1_{\epx}$ are crucial in order to control the local variations of the measure of the process with the Dirichlet form.
 
We set for any integer $l$ 
\begin{equation}\label{rhorhobar}\rho^{\omega,p}_l=\frac{1}{(2l+1)^2}\sum_{x\in B_l}\com_x\1_{\epx} \eqand  \rb^{\omega,p}_l=\densom_l-\rho^{\omega,p}_l=\frac{1}{(2l+1)^2}\sum_{x\in B_l}\com_x\1_{\epx^c}, \end{equation}
where $\epx^c$ is the complementary event of $\epx$.

We are now ready to split ${X}_{i,N}^{f,\varepsilon N}$ into $4$ vanishing parts. Let us denote by
\[{\W_1}={\W}_{i,1}^{f,l}(\confhat)=\curom_i-\langle\curom_i\rangle_0^{l'}-\pa{\gene f-\langle\gene f\rangle_0^{l-s_f}},\]
the difference between $\curom_i-\gene f$ and their local average, and by
\[{\W_2}={\W}_{i,2}^{\varepsilon N,p}(\confhat)=\sdc\left(\rho_{\varepsilon N}\right)\ddi\rb^{\omega,p}_{\varepsilon N}\]
the mesoscopic contributions of full clusters, where $\rb^{\omega,p}_{\varepsilon N}$ was defined in equation \eqref{rhorhobar} above. Let us also introduce 
\[{\W_3}={\W}_{i,3}^{l,\varepsilon N,p}(\confhat)=\sdc\left(\rho_{\varepsilon N}\right)\ddi\rho^{\omega,p}_{\varepsilon N}-\sdc\left(\rho_l\right)\ddi\rho^{\omega,p}_{l_p}+\diff\left(\rho_{\varepsilon N}, \densom_{\varepsilon N}\right)\ddi\rho_{\varepsilon N}-\diff\left(\rho_{l}, \densom_l\right)\ddi\rho_{l'},\]
where $l_p=l-p-1$ and $l'=l-1$, which is the difference between the cutoff microscopic and macroscopic gradients. Note that the cutoff functions are not needed for the total density $\rho$, because the gradients will vanish on full configurations. Finally, we set 
\begin{equation}\label{w4}{\W_4}={\W}_{i,4}^{f,l,p}(\confhat)=\langle\curom_i\rangle_0^{l'}+\sdc\left(\rho_l\right)\ddi\rho^{\omega,p}_{l_p}+\diff\left(\rho_{l}, \densom_l\right)\ddi\rho_{l'}-\langle\gene f\rangle_0^{l-s_f},\end{equation}
the microscopic difference between currents and gradients, taking into consideration the perturbation $\gene f$. For any smooth function $G\in C^{2}(\ctorus)$, we also introduce 
\[Y_1=Y_{i,1}^{f,l}( G,\confhat)=\frac{1}{N}\sum_{x\in \torus}G(x/N)\tau_x{\W}_{1},\qquad  Y_2=Y_{i,2}^{\varepsilon N,p}( G,\confhat)=\frac{1}{N}\sum_{x\in \torus}G(x/N)\tau_x{\W}_{2},\] 
\[Y_3=Y_{i,3}^{l,\varepsilon N,p}( G,\confhat)=\frac{1}{N}\sum_{x\in \torus}G(x/N)\tau_x{\W}_{3}\eqand   Y_4=Y_{i,4}^{f,l,p}( G,\confhat)=\frac{1}{N}\sum_{x\in \torus}G(x/N)\tau_x{\W}_{4}.\]
By construction,
\[{X}_{i,N}^{f,\varepsilon N}(G,\confhat)=\sum_{k=1}^4Y_k(G,\confhat).\]
We can now state the main result of this section.

\begin{theo}
\label{thm:NGestimates}
Let $G$ be a smooth function in $C^{1,2}([0,T]\times \ctorus)$,  $T>0$, and $i\in \{1,2\}$.
For any cylinder function $f$,
\begin{equation}\label{k1}\limsup_{l\to \infty}\limsup_{N\to\infty} \E^{\lambda,\beta}_{\mu^N} \left(\abs{\int_0^T{Y}_{i,1}^{f,l}(G_t,\confhat(t))dt}\right)=0.\end{equation}
Furthermore,
\begin{equation}\label{k2}\lim_{p\to \infty}\limsup_{\varepsilon\to 0}\limsup_{N\to\infty}\E^{\lambda,\beta}_{\mu^N} \left(\abs{\int_0^T{Y}_{i,2}^{\varepsilon N,p}(G_t,\confhat(t))dt}\right)=0.\end{equation}
For any integer $p>1$,
\begin{equation}\label{k3}\limsup_{l\to \infty}\limsup_{\varepsilon \to 0}\limsup_{N\to\infty} \E^{\lambda,\beta}_{\mu^N}\left(\abs{\int_0^T{Y}_{i,3}^{l,\varepsilon N,p}(G_t,\confhat(t))dt}\right)=0.\end{equation}
Finally,
\begin{equation}\label{k4}\inf_{f}\lim_{p\to \infty}\limsup_{l\to \infty}\limsup_{N\to\infty}\E^{\lambda,\beta}_{\mu^N}\left(\abs{\int_0^T{Y}_{i,4}^{f,l,p}(G_t,\confhat(t))dt}\right)=0,\end{equation}
where the infimum in $f$ is taken over the set ${\mathcal C}$ of cylinder functions.
\end{theo}

The core of this section is dedicated to proving these four estimates. The proof of equation \eqref{k1} is immediate and is sketched in Section \ref{subsec:k1}. 

Equation \eqref{k2} is quite delicate, and requires both the control on full clusters derived in equation \eqref{fullceq} and the energy estimate \eqref{Enest}. It is proved in Section \ref{subsec:k2}, in which the main challenge, as in the control of full clusters, is to carry out the macroscopic estimate  \eqref{Enest} in a microscopic setup. 

The proof of equation \eqref{k3} is given in Section \ref{subsec:k3}. This limit is the non-gradient counterpart of the two-block estimate stated in Lemma \ref{lem:TBE}. It follows closely the replacement of local gradients by their macroscopic counterparts performed in Lemma 3.1, p.156 of \cite{KLB1999}, but needs some technical adaptation due to the presence of the cutoff functions. 

The last limit \eqref{k4} requires the tools developed by Varadhan and Quastel \cite{Varadhan1994b} \cite{Quastel1992} for the hydrodynamic limit for non-gradient systems, and therefore requires more work. It is the non-gradient counterpart of the one-block estimate of Lemma \ref{lem:OBE}. However, if the latter was essentially a consequence of the law of large numbers, \eqref{k4} is analogous to the central limit theorem, where the gradient term plays the role of $-\E(\curom_i)$.  The limit \eqref{k4} is the focus of Sections \ref{subsec:localvariance}-\ref{subsec:halpha}.

Finally, Section \ref{subsec:drift}, and in particular Lemma  \ref{lem:driftexpectation}, is dedicated to the integration of the contribution $\gene f$ to the drift part of the scaling limit.

\bigskip

These four estimates are sufficient to allow the replacement of currents by macroscopic averages of gradients, up to a perturbation $\gene f$.
\begin{coro}
\label{currents3}
Let $G$ be a smooth function in $C^{1,2}([0,T]\times \ctorus)$, and ${T>0}$, and consider ${X}_{i,N}^{f,\varepsilon N}$ introduced in \eqref{Xdef}. Then for $i\in \{1,2\}$  
\begin{equation}\label{curreplacementthm}\inf_{f}\limsup_{\varepsilon\to 0}\limsup_{N\to\infty} \E^{\lambda,\beta}_{\mu^N}\left[\abs{\int_0^T{X}_{i,N}^{f,\varepsilon N}(G_t,\confhat(t))dt}\right]=0.\end{equation}
\end{coro} 
\proofthm{Corollary \ref{currents3}}{Since \[{X}_{i,N}^{f,\varepsilon N}(G,\confhat)=\sum_{k=1}^4Y_k(G,\confhat),\] this Corollary follows immediately from the triangular inequality, and Theorem \ref{thm:NGestimates} above, taking the limits $N\to\infty$, then $\varepsilon \to 0$ then $l\to\infty$, then $p\to \infty$, and finally the infimums over the local functions $f$.
}

\subsection{Replacement of the currents and $\gene f$ by their local average }
\label{subsec:k1}
In this paragraph, we prove equation \eqref{k1}, i.e. that  for any $i=1,2$, any function $G\in C^{1,2}([0,T]\times \ctorus)$, and any cylinder function $f$,
\begin{equation*}\limsup_{l\to \infty}\limsup_{N\to\infty} \E^{\lambda,\beta}_{\mu^N}\left( \abs{\int_0^TY_{1}(G_t,\confhat(t))dt}\right)=0.\end{equation*}
We set \[G^{l,N}(x/N)=\frac{1}{(2l+1)^2}\sum_{y\in \torus,  \;\abss{y-x}\leq l}G(y/N),\] an integration by parts yields that, shortening $l'=l-1$ 
\begin{multline*}\frac{1}{N}\sum_{x\in \torus}G(x/N)\pa{\curom_{x,x+e_i}-\frac{1}{(2l'+1)^2}\sum_{\abss{y-x}\leq l'}\curom_{y,y+e_i}}\\
=\frac{1}{N}\sum_{x\in \torus}\pa{G(x/N)-G^{l',N}(x/N)}\curom_{x,x+e_i}\leq\frac{ C(G) l^2}{N}.\end{multline*}
since the difference $G(x/N)-G^{l,N}(x/N)$ is a discrete Laplacian, and is therefore of order $l^2/N^2$, and the currents $\curom_{x,x+e_i}$ are bounded. By the same reasoning, letting $l_f=l-s_f$, we obtain a similar bound on the difference 
\begin{equation*}\frac{1}{N}\sum_{x\in \torus}G(x/N)\pa{\tau_x\gene f-\frac{1}{(2l_f+1)^2}\sum_{\abss{y-x}\leq l_f}\tau_y\gene f}\leq \frac{C'(G,f)l^2}{N},\end{equation*}
since $\gene f$ is a bounded function (this last statement comes from the fact that $f$ is,  and depends only on a finite number of sites).
These two bounds finally yield that for some constant $K=C(G)+C'(G,f)$, \[\abs{Y_1(G,\confhat)}\leq \frac{Kl^2}{N},\]
which immediately yields equation \eqref{k1} for any cylinder function $f$.

\subsection{Estimation of the gradients on full clusters}
\label{subsec:k2}
We now prove that equation \eqref{k2} holds. Our goal is to bound $Y_{i,2}^{\varepsilon N,p}(G,\confhat(s))$ thanks to the control of full clusters functions obtained in \eqref{fullceq}, and to the energy estimate \eqref{Enest}. For the sake of clarity, we drop the various dependencies, and simply write 
\[Y_2=Y_{i,2}^{\varepsilon N,p}.\]
By definition of $Y_2$ and $\rb^{\omega,p}_{\varepsilon N}$ \eqref{rhorhobar},
\begin{align*}Y_2(G,\confhat)&=\frac{1}{N}\sum_{x\in\torus}G(x/N)\tau_x\pa{\sdc\left(\rho_{\varepsilon N}\right)\ddi\rb^{\omega,p}_{\varepsilon N}}\\
&=\frac{1}{N}\sum_{x\in\torus}G(x/N)\tau_x\pa{\sdc\left(\rho_{\varepsilon N}\right)\cro{\frac{1}{(2\varepsilon N+1)^2}\sum_{y\in B_{\varepsilon N}(e_i)}\com_y\1_{E_{p,y}^c}-\frac{1}{(2\varepsilon N+1)^2}\sum_{y\in B_{\varepsilon N}}\com_y\1_{E_{p,y}^c}}},\end{align*}  
and we can rewrite it by summation by parts as 
\begin{equation}\label{Y2S}Y_2(G, \confhat)=\frac{1}{N}\sum_{x\in \torus}\com_x\1_{\epx^c}\frac{1}{(2\varepsilon N+1)^2}\pa{\sum_{y\in B_{\varepsilon N}(x-e_i)}G(y/N)\tau_y\sdc(\rho_{\varepsilon N})-\sum_{y\in B_{\varepsilon N}(x)}G(y/N)\tau_y\sdc(\rho_{\varepsilon N})}.\end{equation}
Most of the terms in the parenthesis above cancel out, since the boxes $B_{\varepsilon N}(x-e_i)$ and $B_{\varepsilon N}(x)$ overlap except on the two sides (cf. Figure \ref{yz}).

\begin{figure}
\centering
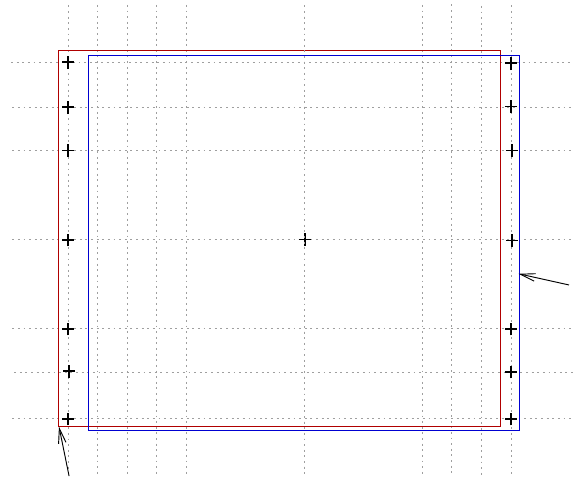
\caption{Definition of the $y_k$'s and $z_k$'s.}      
\label{yz}
\end{figure}
For  any $k\in\llbracket-\varepsilon N,\varepsilon N\rrbracket$, we let according to Figure \ref{yz} 
\[y_k=-(\varepsilon N+1)e_i+ke_{i'} \eqand z_k=\varepsilon Ne_i+ke_{i'},\]
where $i'\neq i$ is the second direction on the torus, which are defined so that $B_{\varepsilon N}(-e_i)\backslash B_{\varepsilon N}=\{y_{-\varepsilon N},\ldots ,y_{\varepsilon N}\}$ and $B_{\varepsilon N}\backslash B_{\varepsilon N}(-e_i)=\{z_{-\varepsilon N},\ldots ,z_{\varepsilon N}\}$.

We thus obtain  from \eqref{Y2S}
\begin{multline}\label{Y2decomp}Y_2(G,\confhat(s))\\
=\frac{1}{N}\sum_{x\in \torus}\com_x\1_{\epx^c}\frac{1}{(2\varepsilon N+1)^2}\pa{\sum_{k=-\varepsilon N}^{\varepsilon N}G\pa{\frac{x+y_k}{N}}\sdc(\tau_{x+y_k}\rho_{\varepsilon N})-G\pa{\frac{x+z_k}{N}}\sdc(\tau_{x+z_k}\rho_{\varepsilon N})}.\end{multline}
We can now rewrite the quantity inside the parenthesis as  the sum over $k$ of
\[\cro{G\pa{\frac{x+y_k}{N}}-G\pa{\frac{x+z_k}{N}}}\sdc(\tau_{x+y_k}\rho_{\varepsilon N})-G\pa{\frac{x+z_k}{N}}\cro{\sdc(\tau_{x+z_k}\rho_{\varepsilon N})-\sdc(\tau_{x+y_k}\rho_{\varepsilon N})}.\]
Since $y_k$ and $z_k$ are distant of $2\varepsilon N+1$, the first term in the decomposition above can be bounded in absolute value uniformly in $x$ and $k$ by $(2\varepsilon N+1)\norm{\partial_{u_i}G}_{\infty}/N$. Let $C(G, \omega)=\norm{\partial_{u_i} G}_{\infty}\norm{\omega}_{\infty}\norm{\sdc}_{\infty}$, the corresponding contribution in \eqref{Y2decomp}  is 
\[\frac{1}{N}\sum_{x\in \torus}\underset{\leq \norm{\omega}_{\infty}}{\underbrace{\com_x}}\1_{\epx^c}\frac{1}{(2\varepsilon N+1)^2}\pa{\sum_{k=-\varepsilon N}^{\varepsilon N}\underset{\leq (2\varepsilon N+1)\norm{\partial_{u_i}G}_{\infty}/N}{\underbrace{\cro{G\pa{\frac{x+y_k}{N}}-G\pa{\frac{x+z_k}{N}}}}}\underset{\leq \norm{\sdc}_{\infty}}{\underbrace{\sdc(\tau_{x+y_k}\rho_{\varepsilon N})}}},\]
and can therefore be bounded by
\[\frac{C(G,\omega)}{N^2}\sum_{x\in \torus}\1_{\epx^c}.\]

Furthermore, since $\sdc$ is $C^{\infty}$ on $[0,1]$, it is Lipschitz-continuous on $[0,1]$ with Lipschitz constant $c$, we let $C'(G,\omega)=c\norm{G}_{\infty}\norm{\omega}_{\infty}/2$. We can now write thanks to the previous considerations that 
\[\abss{Y_2}\leq \frac{C(G,\omega)}{N^2}\sum_{x\in \torus}\1_{\epx^c}+\frac{C'(G,\omega)}{N^2}\sum_{x\in \torus}\frac{1}{(2\varepsilon N+1)}\sum_{k=-\varepsilon N}^{\varepsilon N}\1_{\epx^c}\frac{\abs{\tau_{x+y_k}\rho_{\varepsilon N}-\tau_{x+z_k}\rho_{\varepsilon N}}}{\varepsilon}.\]
For any positive $\gamma$, we have the elementary bound
\[\1_{\epx^c}\frac{\abs{\tau_{x+y_k}\rho_{\varepsilon N}-\tau_{x+z_k}\rho_{\varepsilon N}}}{\varepsilon}\leq \gamma \1_{\epx^c}+\frac{1}{\gamma}\frac{\pa{\tau_{x+y_k}\rho_{\varepsilon N}-\tau_{x+z_k}\rho_{\varepsilon N}}^2}{\varepsilon^2},\]
and finally, for any positive $\gamma$, 
\begin{align}\abss{Y_2}&\leq \frac{C+\gamma C'}{N^2}\sum_{x\in \torus}\1_{\epx^c}+\frac{C'}{\gamma N^2}\sum_{x\in \torus}\frac{1}{(2\varepsilon N+1)}\sum_{k=-\varepsilon N}^{\varepsilon N}\frac{\pa{\tau_{x-(\varepsilon N+1)e_i}\rho_{\varepsilon N}-\tau_{x+\varepsilon Ne_i}\rho_{\varepsilon N}}^2}{\varepsilon^2}\nonumber\\
\label{2terms}&= \frac{C+\gamma C'}{N^2}\sum_{x\in \torus}\1_{\epx^c}+\frac{C'}{\gamma N^2}\sum_{x\in \torus}\frac{\pa{\tau_{x-(\varepsilon N+1)e_i}\rho_{\varepsilon N}-\tau_{x+\varepsilon Ne_i}\rho_{\varepsilon N}}^2}{\varepsilon^2}.\end{align}
Recall that we want to prove \eqref{k2}, i.e. 
\begin{equation*}\lim_{p\to \infty}\limsup_{\varepsilon\to 0}\limsup_{N\to\infty}\E^{\lambda,\beta}_{\mu^N} \left(\int_0^T\abs{Y_2(G_t,\confhat(t))}dt\right)=0.\end{equation*}
The contribution of the first term in the bound for $\abss{Y_2}$ in equation \eqref{2terms} vanishes for any $\gamma$ as $N$ then $p$ goes to $\infty$, thanks to Proposition \ref{prop:fullclusters}. 

Furthermore, we can replace  $\tau_{x-(\varepsilon N+1)e_i}\rho_{\varepsilon N}$ by $\tau_{x-\varepsilon N e_i}\rho_{\varepsilon N}$ in \eqref{2terms} since the difference between these two quantities is of order $1/N$ and vanishes in the limit $N\to\infty$. This replacement allows us to work only with quantities that can be expressed in terms of the empirical measure of the process. Equation \eqref{k2} therefore holds according to Lemma \ref{lem:H1estimate} below, letting $\gamma$ go to $\infty$ after $N\to\infty$ then $\varepsilon \to 0$ then $p\to \infty$. $\hfill \blacksquare$

\begin{lemm}
\label{lem:H1estimate}There exists a positive constant $K$ such that 
\begin{equation*}\limsup_{\varepsilon\to 0}\limsup_{N\to\infty}\E^{\lambda,\beta}_{\mu^N} \left({\int_0^T\frac{1}{ N^2}\sum_{x\in \torus}\frac{\pa{\tau_{x-\varepsilon Ne_i}\rho_{\varepsilon N}(t)-\tau_{x+\varepsilon Ne_i}\rho_{\varepsilon N}(t)}^2}{\varepsilon^2}dt}\right)\leq K.\end{equation*}
\end{lemm}
\proofthm{Lemma \ref{lem:H1estimate}}{This Lemma states that the difference of macroscopic densities between two points distant from $2\varepsilon$ is also of order $\varepsilon$, and is a consequence of the energy estimate \eqref{Enest}. We are going to prove this macroscopic estimate in the topological setup of the space of c\`adl\`ag trajectories of measures on $\ctorus\times\ctoruspi$ .
Recall from Section \ref{subsec:compactnessQN} that $\meas(\ctorus\times\ctoruspi)$ is the space of non-negative measures on the continuous configuration space, 
\[\Mhat=D\pa{[0,T],\meas(\ctorus\times\ctoruspi)}\] 
is the space of right-continuous, left-limit trajectories on the set of  measures on $\ctorus\times \ctoruspi$, and that  $Q^N$ is the distribution on $\Mhat$ of the process's empirical measure $\pi^N$. We have proved in Proposition \ref{prop:compactnessQN} that the sequence $(Q^N)_{N\in \N}$ is relatively compact for the weak topology.
Let $\Lambda_{\varepsilon}=[\varepsilon,\varepsilon]^2\subset\ctorus$ be the cube of size $\varepsilon$, and $(\varphi_{\varepsilon})_{\varepsilon>0}$ be a family of localizing functions on $\ctorus$ 
\[\varphi_{\varepsilon}(\cdot)=\frac{1}{(2\varepsilon)^2}\1_{\Lambda_{\varepsilon}}(\cdot),\]
we then have
\[\tau_x\rho_{\varepsilon N}(
t)=\frac{(2\varepsilon N)^2}{(2\varepsilon N+1)^2}<\pi^N_t,\varphi_{\varepsilon}(.+x/N)>.\]

\begin{figure}
        \centering
\hspace{-13em}
	\begin{subfigure}{0.3\textwidth}
		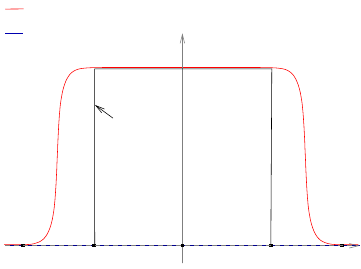
         \caption{}
                \label{phiep}
        \end{subfigure}
\hspace{5em}
\begin{subfigure}{0.3\textwidth}
		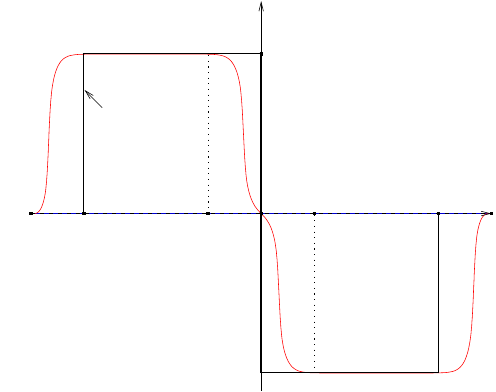
         \caption{}
                \label{gradphiep}
        \end{subfigure}
\caption{(\subref{phiep}) Representations of $ \widetilde{\varphi}_{\varepsilon}(\cdot ,v)$ depending on the value of $v$.\\
(\subref{gradphiep}) Representation of $h_{\varepsilon}(\cdot ,v)=\nabla^{\varepsilon} \widetilde{\varphi}_{\varepsilon}(\cdot ,v)$ depending on the value of $v$.}

\label{phitot}\end{figure}

We define the \emph{mesoscopic gradient} \[\boldsymbol\nabla_i^{\varepsilon}\varphi(\cdot)=\varepsilon^{-1}(\varphi(\cdot-\varepsilon e_i)-\varphi(\cdot+\varepsilon e_i)),\] represented in Figure \ref{gradphiep}. Note that $\boldsymbol\nabla_i^{\varepsilon}\varphi_{\varepsilon}$ is at most of order $\varepsilon^{-3}$ since $\varphi_{\varepsilon}$ is of order $\varepsilon ^{-2}$. We can rewrite the left-hand side in Lemma \ref{lem:H1estimate} as 
\begin{align}\label{phinonc}\E_{Q^N}\pa{\int_0^T\frac{1}{ N^2}\sum_{x\in \torus}<\pi_t,\boldsymbol\nabla_i^{\varepsilon}\varphi_{\epsilon}(.+x/N)>^2dt}+o_N(1).
\end{align}
Furthermore, since for any two sites $x,x'\in \ctorus$ distant from less than $1/N$, 
\begin{equation*}\abs{<\pi_t,\boldsymbol\nabla_i^{\varepsilon}\varphi_{\epsilon}(.+x/N)>-<\pi_t,\boldsymbol\nabla_i^{\varepsilon}\varphi_{\epsilon}(.+x'/N)>}\leq C(\varepsilon)\frac{1}{N},\end{equation*}
we can replace the sum above by the integral over the continuous torus. 

However, regarding the weak topology on $\meas(\ctorus\times\ctoruspi)$, it will be convenient later on to consider smooth functions instead of $\varphi_{\varepsilon}$. We therefore introduce for any $\varepsilon$ a function $\widetilde{\varphi}_{\varepsilon}$, represented in Figure  \ref{phiep} verifying 
\begin{itemize}
\item $\widetilde{\varphi}_{\varepsilon}={\varphi}_{\varepsilon}$ on $\Lambda_{\varepsilon}$ and on $\ctorus\backslash \Lambda_{\varepsilon+\varepsilon^3}$.
\item $\norm{\widetilde{\varphi}_{\varepsilon}}_{\infty}=\norm{{\varphi}_{\varepsilon}}_{\infty}$.
\item $\widetilde{\varphi}_{\varepsilon}$ is in $C^1(\ctorus)$.
\end{itemize}
 Since $\widetilde{\varphi}_{\varepsilon}$ and ${\varphi}_{\varepsilon}$ coincide everywhere except on $\Lambda_{\varepsilon +\varepsilon^3}\backslash \Lambda_{\varepsilon}$, and since $\norm{\widetilde{\varphi}_{\varepsilon}}_{\infty}=(2\varepsilon )^{-2}$ we can write for any $x\in \torus$
\begin{align*}\abs{<\pi^N_t,{\varphi}_{\varepsilon}(.+x/N)>-<\pi^N_t,\widetilde{\varphi}_{\varepsilon}(.+x/N)>}&\leq\frac{1}{(2\varepsilon )^2}\underset{\leq 4\varepsilon\times \varepsilon^3}{\underbrace{<\pi^N_t,\1_{\Lambda_{\varepsilon+\varepsilon^3}\backslash\Lambda_{\varepsilon}}(.+x/N)>}}.\nonumber\\
&\leq C \varepsilon^2,\end{align*}
for some positive constant $C$.
This bound immediately yields \[\abs{<\pi^N_t,\boldsymbol\nabla_i^{\varepsilon}{\varphi}_{\varepsilon}(.+x/N)>-<\pi^N_t,\boldsymbol\nabla_i^{\varepsilon}\widetilde{\varphi}_{\varepsilon}(.+x/N)>}\leq C \varepsilon,\]
which allows us to replace in equation \eqref{phinonc}, in the limit $N\to \infty$ then $\varepsilon \to0$,  $\varphi_{\epsilon}$ by $\widetilde{\varphi}_{\epsilon}$.

To prove Lemma \ref{lem:H1estimate} it is therefore sufficient to prove that 
\begin{equation}\label{limsupQN}\limsup_{\varepsilon\to 0}\limsup_{N\to\infty}\E_{Q^N}\pa{\iint_{[0,T]\times\ctorus}<\pi_t,h_{\varepsilon}(.+u)>^2 du dt}\leq K,\end{equation}
where $h_{\varepsilon}=\boldsymbol\nabla_i^{\varepsilon}\widetilde{\varphi}_{\varepsilon}$, is a continuous bounded function, represented in Figure \ref{gradphiep}.
 Let us denote by $\Pi$ the subset of $\Mhat$
 \[\Pi=\left\{\pi\in \Mhat\;\Big|\;\sup_{t\in [0,T]}<\pi_t,1>\leq 1\right\}\]
 of trajectories with mass less than one at all times, which is compact w.r.t Skorohod's topology introduced in Section \ref{subsec:compactnessQN}.
 
 Consider a weakly convergent subsequence $Q_{N_k}\to Q^*$, in order to substitute $Q^*$ to $Q^N$ in the limit above, we want to prove that for any fixed $\varepsilon>0$, the application 
 \[I_{\varepsilon}:\pi\mapsto\iint_{[0,T]\times\ctorus}<\pi_t,h_{\varepsilon}(.+u)>^2 dudt\]
 is bounded, and continuous on $\Pi$ w.r.t. Skorohod's topology.
 
 Note that this application is bounded on $\Pi$ by construction, we now prove the following Lemma.
 \begin{lemm}\label{lem:Skorohodcont}
 Fix $\varepsilon>0$, the application $I_{\varepsilon}$ is continuous on $(\Pi,d)$, where $d$ is the Skorohod metric defined in equation \eqref{metricMhat}.
\end{lemm}
\proofthm{Lemma \ref{lem:Skorohodcont}}{For any two trajectories $\pi$ and $\pi'$ in $\Pi$, and some continuous strictly increasing function $\kappa$ from $[0,T]$ into itself, such that $\kappa_0=0$ and $\kappa_T=T$, we can write 
\[I_{\varepsilon}(\pi)-I_{\varepsilon}(\pi')=\iint_{[0,T]\times\ctorus}du<\pi'_t+\pi_t,h_{\varepsilon}(.+u)><\pi'_t-\pi_{\kappa_t}+\pi_{\kappa_t}-\pi_t,h_{\varepsilon}(.+u)> dt.\]
The first factor $<\pi'_t+\pi_t,h_{\varepsilon}(.+u)>$ can be crudely controlled by $2\norm{h_{\varepsilon}}_{\infty}$, which yields 
\begin{align}\label{cvdie}\abs{I_{\varepsilon}(\pi)-I_{\varepsilon}(\pi')}\leq &2\norm{h_{\varepsilon}}_{\infty}{\iint_{[0,T]\times\ctorus}\abs{<\pi'_t-\pi_{\kappa_t},h_{\varepsilon}(.+u)>+<\pi_{\kappa_t}-\pi_t,h_{\varepsilon}(.+u)>}dudt}.
\end{align}
Note that by definition of $\norm{\kappa}$, one easily gets that for any $t\in [0,T]$, $\abs{t- \kappa_t}\leq T(e^{\norm{\kappa}}-1)$, therefore, $\kappa_t\to t$ uniformly on $[0,T] $ as $\norm{\kappa}\to 0$. Let us fix $\pi\in \Pi$, and assume that $d(\pi,\pi^n)\to0$ for some sequence of trajectories $(\pi^n)_n\in \Pi^{\N}$, there exists a sequence $(\kappa^n)_{n\in\N}$ such that $\norm{\kappa^n}\to 0$ and 
$\lim_{n\to\infty}\sup_{t\in [0,T]}\delta(\pi^n_t,\pi_{\kappa^n_t})=0$. This last statement yields in particular that for any $t\in [0,T]$, $\delta(\pi^n_t,\pi_{\kappa^n_t})\to 0$, therefore for any $t\in [0,T]$, and for any $u\in \ctorus$, 
\[\lim_{n\to\infty}<\pi^n_t-\pi_{\kappa^n_t}, h_{\varepsilon}(.+u)>=0,\]
since $h_{\varepsilon}(.+u)$ is a continuous bounded function, and $\delta$ is a metric of the weak convergence.
Furthermore, since $\kappa^n_t$ converges uniformly towards $t$ on $[0,T]$ and since $t\to\pi_t$ is weakly continuous almost everywhere on $[0,T]$ by definition of $\Mhat$, we also have that {for almost every} $(t,u)\in [0,T]\times \ctorus$, 
\[\lim_{n\to \infty}<\pi_{\kappa^n_t}-\pi_t,h_{\varepsilon}(.+u)>=0.\]
Since $\pi$ and the $\pi^n$'s are in $\Pi$, both of these quantities are crudely bounded in absolute value by ${2\norm{h_{\varepsilon}}_{\infty}}$, which is naturally integrable on $[0,T]\times \ctorus$. One finally obtains  by dominated convergence, from \eqref{cvdie} applied to $\pi'=\pi^n$ and $\kappa=\kappa^n$, that 
\begin{align*}\abs{I_{\varepsilon}(\pi)-I_{\varepsilon}(\pi^n)}\underset{n\to \infty}{\to} 0.
\end{align*}
Lemma \ref{lem:Skorohodcont} is complete.}

We have now proved that the application $I_{\varepsilon}$ is continuous for any fixed $\varepsilon$, therefore the left-hand side of \eqref{limsupQN}
 is less than 
 \begin{equation*}
 \limsup_{\varepsilon\to 0}\sup_{Q^*}\E_{Q^*}\pa{\iint_{[0,T]\times\ctorus}du<\pi_t,h_{\varepsilon}(.+u)>^2 dt},
 \end{equation*}
where the supremum is taken over all limit points $Q^*$ of the sequence $Q^N$.
Since by definition $h_{\varepsilon}=\boldsymbol\nabla_i^{\varepsilon} \widetilde{\varphi}_{\varepsilon}$ does not depend on $\theta$, 
we drop the dependency of $\pi$ on $\theta$ and consider simply for any $u\in \ctorus$, $\rho(t,u)=\int_{\ctoruspi}\boldsymbol \dens_t(u,d\theta)$, 
where $\boldsymbol \dens_t(u,d\theta)$ is the density of $\pi_t(\cdot ,d\theta)$ w.r.t. the Lebesgue measure $\ctorus$, which exists $Q^*$-a.s. 
according to Lemma \ref{lem:Lebesguedensity}. We can write 
\begin{equation}
\label{densQstar}
\E_{Q^*}\pa{\iint_{[0,T]\times\ctorus}du<\pi_t,h_{\varepsilon}(.+u)>^2 dt}=\E_{Q^*}\pa{\iint_{[0,T]\times\ctorus} \pa{\int_{v\in \ctorus}\rho(t,v)\boldsymbol\nabla_i^{\varepsilon} \widetilde{\varphi}_{\varepsilon}(v+u)dv}^2 dudt}.
\end{equation}
We can now express $\boldsymbol\nabla_i^{\varepsilon} \widetilde{\varphi}_{\varepsilon}$ as a gradient, by writing 
\[\boldsymbol\nabla_i^{\varepsilon} \widetilde{\varphi}_{\varepsilon}(u)=\partial_{u_i}\int_{-1/2}^{u_i}\boldsymbol\nabla_i^{\varepsilon} \widetilde{\varphi}_{\varepsilon}(\upsilon e_i+u_{i'}e_{i'})d\upsilon=\partial_{u_i} \Phi_{\varepsilon,i},\]
where $i'\neq i$ still denotes the second direction on the torus.

\begin{figure}
\centering
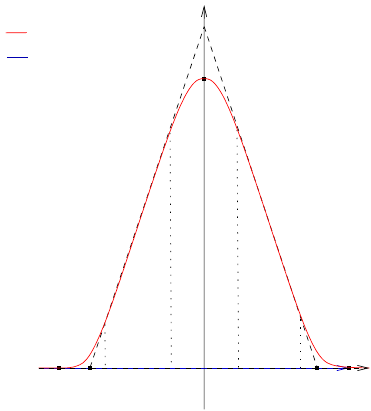
\caption{Representation of $\Phi_{\varepsilon,i}(\cdot ,v)$ depending on $v$.}      
\label{phifig}
\end{figure}

Furthermore, $\Phi_{\varepsilon,i}$, represented in Figure \ref{phifig}, is in $C^2({\ctorus})$ because $\widetilde{\varphi}_{\varepsilon}$ is $C^1$, and the various integrals can be freely swapped since all quantities are bounded at any fixed $\varepsilon$. Since $Q^*$-a.s. $\rho\in W^{1,2}([0,T]\times\ctorus )$ according to Theorem \ref{thm:regularity}, the right-hand side in equation \eqref{densQstar} is therefore equal to 
 \begin{equation}\label{densQstar2}\E_{Q^*}\pa{\iint_{[0,T]\times\ctorus} \pa{\int_{v\in \ctorus} \Phi_{\varepsilon,i}(v+u) \partial_{u_i}\rho(t,v) dv}^2 dudt}.\end{equation}
In order to conclude, we adapt the proof of Young's Inequality, and apply Cauchy-Schwarz inequality to $f=\pa{\Phi_{\varepsilon,i}(v+u)}^{1/2}$ and $g=\pa{\Phi_{\varepsilon,i}(v+u)}^{1/2}\partial_{u_i}\rho(t,v)$, to finally obtain that   
 \begin{align*}\E_{Q^*}\bigg(\iint_{[0,T]\times\ctorus}du<\pi_t,h_{\varepsilon}(.+u)&>^2 dt\bigg) \\
 \leq \;&\E_{Q^*}\pa{\iint_{[0,T]\times\ctorus} \norm{\Phi_{\varepsilon,i}}_1 \cro{\int_{v\in \ctorus}\Phi_{\varepsilon,i}(v+u)(\partial_{u_i}\rho(t,v))^2dv } dudt}\\
 = \;&\norm{\Phi_{\varepsilon,i}}^2_1\E_{Q^*}\pa{\iint_{[0,T]\times\ctorus}(\partial_{u_i}\rho(t,u))^2du  dt},\end{align*}
 where the last identity was obtained by integrating first w.r.t. $u$, then w.r.t. $v$.
Since $  \norm{\Phi_{\varepsilon,i}}_1=1+o_{\varepsilon}(1)$, Lemma \ref{lem:H1estimate} follows from equation \eqref{Enest}.
}

\subsection{Replacement of the macroscopic gradients by their local counterparts}
\label{subsec:k3}
We now prove equation \eqref{k3}, i.e. that the macroscopic average of the gradients can be replaced by a local average.  To simplify the notations, throughout this section, we drop the various dependencies of ${Y}_{i,3}^{l,\varepsilon N,p}$ and simply denote it by $ Y_{3}$.

Recall that $\geniz$ stands for the modified Glauber generator without alignment of the angles, where each angle is updated uniformly in $\ctoruspi$,
\[\geniz f(\confhat)=\sum_{x\in \torus}\conf_x\int_{\ctoruspi}\frac{(f(\confhat^{x, \theta})-f(\confhat))}{2\pi}d\theta,\]
and 
\begin{equation*}L_N^{\beta=0}=N^2\genex+\geniz.\end{equation*}
Recall that  $\Prob^{\lambda, 0}_{\mesref}$ is the measure on the trajectories starting from the equilibrium measure $\mesref$ and driven by the generator $L_N^{\beta=0}$, and that the expectation w.r.t the latter is denoted by $\E^{\lambda,0}_{\mesref}$. 
We first apply Proposition \ref{prop:comparisonbetazero} to the positive functional 
\[X\pa{\confhat^{[0,T]}}=\abs{\int_0^T{Y}_3(G_t,\confhat(t))dt},\] letting  $A=\gamma N^2$, and obtain that for some constant $K_0=K_0(T,\beta, {\widehat{\zeta}})$,
\[ \E^{\lambda,\beta}_{\mu^N}\left(\abs{\int_0^T{Y}_3(G_t,\confhat(t))dt}\right)\leq\frac{K_0}{\gamma}+\frac{1}{\gamma N^2}\log \E^{\lambda,0}_{\mesref}\left[\exp \left(\gamma N^2\abs{\int_0^T{Y}_3(G_t,\confhat(t))dt}\right)\right].\]
Letting $\gamma$ go to $\infty$ after $N$, to prove \eqref{k3} it is therefore  enough to show that for any integer $p>1$
\begin{equation}\label{k3old}\lim_{\gamma\to\infty}\limsup_{l\to \infty}\limsup_{\varepsilon \to 0}\limsup_{N\to\infty}\frac{1}{\gamma N^2}\log \E^{\lambda,0}_{\mesref}\left[\exp \left(\gamma N^2\abs{\int_0^T{Y}_3(G_t,\confhat(t))dt}\right)\right]=0.\end{equation}
We now get rid of the absolute value by using both of the elementary inequalities  
\[e^{\abss{x}}\leq e^x+e^{-x}\] and \[ \limN\frac{1}{N^2}\log (a_N+b_N)\leq \max\pa{\limN\frac{1}{N^2}\log a_N,\limN \frac{1}{N^2}\log b_N}.\]
Both of these imply that the limit in equation \eqref{k3} is bounded up by the maximum of the limits of  
\[\frac{1}{\gamma N^2}\log \E^{\lambda, 0}_{\mesref}\left[\exp \left(\gamma N^2 \int_0^TY_3(G_t,\confhat(t))dt\right)\right]\]
and 
\[\frac{1}{\gamma N^2}\log \E^{\lambda, 0}_{\mesref}\left[\exp \left(-\gamma N^2 \int_0^TY_3(G_t,\confhat(t))dt\right)\right].\]
Since $-Y_3(G,\confhat)=Y_3(-G,\confhat)$, and since the identity above must be true for any function $G$, to obtain the wanted result it is sufficient to show that for any $\gamma$ and any $G\in C^{1,2}([0,T]\times \ctorus)$ 
\begin{equation}
\label{limitcurrents}
\lim_{\gamma\to \infty}
{\liml}\limsup_{\varepsilon\to 0}\limsup_{N\to\infty}\frac{1}{\gamma N^2}\log \E^{\lambda, 0}_{\mesref}\left[\exp \left(\gamma N^2\int_0^T Y_3(G_t,\confhat(t))dt\right)\right]\leq 0.
\end{equation} 

We now get back to a variational  problem, since Lemma \ref{lem:FeynmanKac} yields 
\[\frac{1}{\gamma N^2}\log \E^{\lambda,0}_{\mesref}\left[\exp \left(\gamma N^2\int_0^TY_3(G_t,\confhat(t))dt\right)\right]\leq\frac{2 T\lambda^2}{\gamma}+\frac{1}{\gamma }\int_0^T\sup_{\varphi}\left\{\Eref\pa{\varphi \gamma Y_3(G_t,\confhat)}-\frac{1}{2}\rdir(\varphi)\right\}.\]
The first term in the right-hand side above vanishes as $\gamma $ goes to $\infty$. Furthermore, the time integral is now only applied to the function $G_t$, therefore to obtain equation \eqref{k3}, it is sufficient to prove that for any $\gamma$ and any function $G\in C^{2}(\ctorus)$,
\begin{equation}\label{k3var}\limsup_{l\to \infty}\limsup_{\varepsilon \to 0}\limsup_{N\to\infty}\sup_{\varphi}\left\{2\gamma\Eref\pa{\varphi  Y_3(G,\confhat)}-\rdir(\varphi)\right\}\leq 0.\end{equation}
Since this must be true for any $G$ and any $\gamma$, we can safely assume that $\gamma=1/2$, and equation \eqref{k3var} follows from Lemma \ref{lem:NGTBE} below. Thus this completes the proof of \eqref{k3}.

In order to avoid repeating a similar proof twice, we forget for the moment that $\diffom\left(\rho, \rho^{\omega}\right)=d_s(\rho) $ only depends on the total particle density, and present the proof of the following Lemma in the most difficult case where the gradient is on $\rho^{\omega,p}$ and where the diffusion coefficient depends on both $\rho$ and  $\rho^{\omega}$. 
We simply assume throughout this proof that the diffusion coefficient $\diffom$ is a uniformly continuous function of $\rho$ and  $\rho^{\omega}$ on the set 
\[\Big\{(\am, \aom)\in[0,1]\times[-\norm{\omega}_{\infty},\norm{\omega}_{\infty}]\mbox{ such that }  |\aom|\leq \norm{\omega}_{\infty} \alpha\Big\}.\] 
\begin{lemm}
\label{lem:NGTBE}
Let us fix $1\leq i,\;j \leq 2$, we shorten \[\D_k=\diffom\left(\rho_k, \rho_k^{\omega}\right) \mbox{ and } v_k=\ddi\rho^{\omega,p}_k.\]
For any $G\in C^{2}(\ctorus)$ 
  \begin{equation}\label{micromacroeq}\liml\limep\limN\sup_{\varphi}\left\{\sum_{x\in\torus}\cro{\frac{1}{N}G(x/N)\Eref\Big(\varphi\tau_x(\D_{\varepsilon N}v_{\varepsilon N}-\D_lv_{l_p})\Big)}-\rdir(\varphi)\right\} \leq 0,\end{equation} 
where as before $l_p=l-p-1$, and the supremum is taken over all probability densities with respect to $\mesref$. 
The same result is true for the  gradients $v_k=\ddi\rho_k$ instead of $\ddi\rho^{\omega,p}_k$, $\diff$ instead of $\diffom$, and $l'=l-1$ instead of $l_p$.
\end{lemm}
\begin{proof}[Proof of Lemma \ref{lem:NGTBE}]The difficulty of this Lemma comes from the extra factor $N$, which prevents us from using directly the replacement Lemma \ref{lem:replacementlemma}. We hence need to get some precise control over each term to ensure that they are small enough. 
We start by splitting in two parts the quantity in Lemma \ref{lem:NGTBE} by noticing that 
\begin{equation}\label{splitmicromacro}\D_{\varepsilon N}v_{\varepsilon N}-\D_lv_{l_p}=\D_{\varepsilon N}(v_{\varepsilon N}-v_{l_p})+(\D_{\varepsilon N}-\D_{l})v_{l_p}.\end{equation}
Both terms are treated in the same fashion due to the continuity of the diffusion coefficients (which follows directly from their explicit expression). More precisely, we intend to show that the difference between the average over a microscopic and macroscopic box is of order $1/N$, and hence yields the extra factor $N$  needed to use the replacement Lemma. Let us thus consider the first term appearing in the Lemma, namely 
\[\frac{1}{N}\Eref\pa{\varphi\sum_{x\in \torus} G(x/N)\tau_x\D_{\varepsilon N}(v_{\varepsilon N}-v_{l_p})}.\]
Recall that we denoted $B_l=\{x\in \torus, \abss{x}\leq l\}$, and $\abss{B_l}=(2l+1)^2$. 
Since both $v_{\varepsilon N}$ and $v_{l_p}$ are merely spatial averages of the gradients $\ddi(\com_0\1_{E_{p}})$, a first summation by parts yields that the quantity above is equal to 
\begin{multline*}\frac{1}{N}\Eref\Bigg(\varphi\sum_{x\in \torus}(\com_{x+e_i}\1_{E_{p,x+e_i}}-\com_x\1_{\epx}  )\Bigg[\frac{1}{\abss{B_{\varepsilon N}}}\sum_{\abss{y-x}\leq \varepsilon N}G(y/N)\tau_y\D_{\varepsilon N}\\
-\frac{1}{\abss{B_{l_p}}}\sum_{\abss{y-x}\leq l_p}G(y/N)\tau_y\D_{\varepsilon N}\Bigg]\Bigg).\end{multline*}
Now let $S_{x}(\confhat)$ denote the quantity inside braces, i.e 
\[S_{x}(\confhat)=\frac{1}{\abss{B_{\varepsilon N}}}\sum_{\abss{y-x}\leq \varepsilon N}G(y/N)\tau_y\D_{\varepsilon N}-\frac{1}{\abss{B_{l_p}}}\sum_{\abss{y-x}\leq l_p}G(y/N)\tau_y\D_{\varepsilon N}.\]
We are now going to prove that
\begin{equation}
\label{vanishingsubqtty1}\liml\limep\limN\sup_{\varphi}\left\{\frac{1}{N}\Eref\pa{\varphi\sum_{x\in \torus}S_{x}(\com_{x+e_i}\1_{E_{p,x+e_i}}-\com_x\1_{E_{p,x}} )}-\frac{1}{2}\rdir(\varphi)\right\} \leq 0.
\end{equation}
In order to transfer the gradient appearing in the expression above on $\varphi$ and $S_{x}$, 
we need {the specific} change of variable represented in Figure \ref{Tip}.
For any direction $i\in \{1,2\}$, let $i'\neq i$ be the second direction on the torus. Given $x$ in the torus, we denote for any $k\in \llbracket -p,p\rrbracket$ (See Figure \ref{yz})
\[y_k=x-pe_i+k e_{i'}\in B_p(x)\eqand z_k=x+(p+1)e_i+k e_{i'}\in B_p(x+e_i).\]
Given these, {we  denote}, for any configuration $\confhat$, by 
\[T_{i,p}^x(\confhat)=\pa{\pa{(\confhat^{x,x+e_i})^{y_{-p},z_{-p}}}^{\ldots }}^{y_p,z_p}\]
the configuration where the sites $x$ and $x+e_i$ have been swapped, as well as the boundary sites $y_k$ and $z_k$.
\index{$T_{i,p}^x $ \dotfill exchanges $\com_x 1_{\epx}$ and $\tau_{e_i}\com_{x}1_{E_{p,x}}$}
\begin{figure}
     \centering
		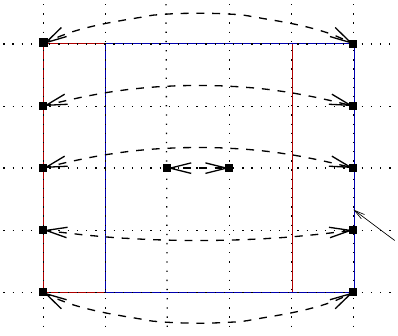
         \caption{Change of variable $\confhat\to T^x_{i,p}\confhat$.}
                \label{Tip}
\end{figure}

By {construction}, we have
\[\com_x 1_{\epx}(T_{i,p}^x\confhat)=\com_{x+e_i}\1_{E_{p,x+e_i}}(\confhat)\]
The first term in the left-hand side of \eqref{vanishingsubqtty1} can be rewritten as 
\begin{align}
\label{vanishingqtty2}\frac{1}{N}\Eref\pa{\varphi\sum_{x\in \torus}S_{x}(\com_{x+e_i}\1_{E_{p,x+e_i}}-\com_x\1_{E_{p,x}})}
=&\frac{1}{N}\Eref\pa{\sum_{x\in \torus}\com_x\1_{E_{p,x}}\pa{(\varphi S_{x})(T_{i,p}^x\confhat)-\varphi S_{x}}}\nonumber\\
=&\frac{1}{N}\sum_{x\in \torus}\Eref\left(\com_x\1_{E_{p,x}}\left[\varphi(T_{i,p}^x\confhat)\pa{S_{x}(T_{i,p}^x\confhat)-S_{x}}\right.\right.\nonumber\\
&\hfill\left.\left.+\pa{\varphi(T_{i,p}^x\confhat)-\varphi }S_{x}\right]\right)
.\end{align}

We are going to show that the contribution of the first term of the right-hand side in \eqref{vanishingqtty2} vanishes in the limit $N\to \infty$, whereas the second term can be controlled with the Dirichlet form $\rdir(\varphi)$.
Recall that $S_{x}$ is defined as 
\[S_{x}(\confhat)=\frac{1}{\abss{B_{\varepsilon N}}}\sum_{\abss{y-x}\leq \varepsilon N}G(y/N)\tau_y\D_{\varepsilon N}-\frac{1}{\abss{B_{l_p}}}\sum_{\abss{y-x}\leq l_p}G(y/N)\tau_y\D_{\varepsilon N}.\]
Since the only dependency of $S_{x}$ in $\confhat$ lies in $\D_{\varepsilon N}$, which is the diffusion coefficient evaluated in the macroscopic empirical density $\densep$, in order to control the first term in the right-hand side of \eqref{vanishingqtty2}, we can write
\begin{multline}\label{decompSx}S_{x}(T_{i,p}^x\confhat)-S_{x}=\\
\frac{1}{\abss{B_{\varepsilon N}}}\sum_{\abss{y-x}\leq \varepsilon N}G(y/N)\tau_y\cro{\D_{\varepsilon N}(T_{i,p}^x\confhat)-\D_{\varepsilon N}(\confhat)}-\frac{1}{\abss{B_{l_p}}}\sum_{\abss{y-x}\leq l_p}G(y/N)\tau_y\cro{\D_{\varepsilon N}(T_{i,p}^x\confhat)-\D_{\varepsilon N}(\confhat)}.\end{multline}
Recall that $\tau_y\D_{\varepsilon N}(\confhat)=\diffom(\tau_y\rhoep, \tau_y\rhoep^{\omega})$.  Since it depends on the configuration through an average over $B_{\varepsilon N}(y)$, $\tau_y\D_{\varepsilon N}(\confhat)$ is invariant under any exchange of a pair of sites with both ends in $B_{\varepsilon N}(y)$. We deduce from this remark that for any $\abss{y-x}\leq l_p$, the quantity  \[\tau_y\cro{\D_{\varepsilon N}(T_{i,p}^x\confhat)-\D_{\varepsilon N}(\confhat)}\] vanishes, since all the exchanges happen between sites at a distance at most $p$ of $x$, and therefore at a distance at most $p+l_p$ of $y$. This yields that the second term in the right-hand side of \eqref{decompSx} vanishes.

We now consider the first term in the right-hand side of \eqref{decompSx}. For the same reason as before, for any $y$ in $B_{\varepsilon N-p-1}(x)$, all the exchanges in $T_{i,p}^x$ have both ends in $B_{\varepsilon N}(y)$, and  $\tau_y\cro{\D_{\varepsilon N}(T_{i,p}^x\confhat)-\D_{\varepsilon N}(\confhat)}$ vanishes.  We can finally rewrite \eqref{decompSx} as 
\begin{equation}\label{decompSx2}S_{x}(T_{i,p}^x\confhat)-S_{x}=\frac{1}{\abss{B_{\varepsilon N}}}\sum_{y\in B_{\varepsilon N}(x)\setminus B_{\varepsilon N-p-1}(x)}G(y/N)\tau_y\cro{\D_{\varepsilon N}(T_{i,p}^x\confhat)-\D_{\varepsilon N}(\confhat)}.\end{equation}
We now take a closer look at each of the remaining term. By definition, the configuration $T_{i,p}^x\confhat$ can be obtained from $\confhat$ by inverting $2p+2$ pair of sites in $\confhat$. Furthermore, fix a $y$ in the sum above, and consider any inversion $\confhat^{z_1,z_2}$ with $z_1\in B_{\varepsilon N}(y)$ and $z_2\notin B_{\varepsilon N(y)}$, we wan write by definition of $\rhoep$ and $\rhoep^{\omega}$
\[\abs{\tau_y\rhoep(\confhat^{z_1,z_2})-\tau_y\rhoep(\confhat)}\leq\frac{1}{\abss{B_{\varepsilon N}}}\eqand \abs{\tau_y\rhoep^{\omega}(\confhat^{z_1,z_2})-\tau_y\rhoep^{\omega}(\confhat)}\leq\frac{2\norm{\omega}_{\infty}}{\abss{B_{\varepsilon N}}}.\]
By assumption, $\diffom(\am,\aom)$ is uniformly continuous on the set
\[\Big\{(\am, \aom)\in[0,1]\times[-\norm{\omega}_{\infty},\norm{\omega}_{\infty}] \mbox{ such that } \abss{\aom}\leq \norm{\omega}_{\infty} \alpha\Big\}.\]
We deduce from this that 
\[\tau_y\pa{\D_{\varepsilon N}(\confhat^{z_1,z_2})-\D_{\varepsilon N}(\confhat)}=o_N(1),\]
therefore 
\[\abs{\tau_y\pa{\D_{\varepsilon N}(T_{i,p}^x\confhat)-\D_{\varepsilon N}(\confhat)}}\leq o_N(1),\]
where this time $o_N(1)$ stands for a constant depending on $p$ which vanishes as $N\to \infty$.
We inject the latter identity in equation \eqref{decompSx2}, to obtain that 
\[S_{x}(T_{i,p}^x\confhat)-S_{x}=\frac{\abs{B_{\varepsilon N}(x)\setminus B_{\varepsilon N-p-1}(x)}}{\abss{B_{\varepsilon N}}}o_N(1)=\frac{1}{N}o_N(1),\]
where the last $o_N(1)$ depends on $p$ and $\varepsilon$, but vanishes as $N\to \infty$.
This allows us to get back to equation \eqref{vanishingqtty2}, in which the first term in the right-hand side can be rewritten 
\[\abs{\frac{1}{N}\sum_{x\in \torus}\Eref\pa{\com_x\1_{E_{p,x}}\varphi(T_{i,p}^x\confhat)\pa{S_{x}(T_{i,p}^x\confhat)-S_{x}}}}\leq \frac{\norm{\omega}_{\infty}}{N^2}\sum_{x\in \torus}\Eref\pa{\varphi(T_{i,p}^x\confhat)}o_N(1)=o_N(1),\]
since $\mesref$ is invariant under the change of variable $T_{i,p}^x\confhat$, and therefore $\Eref\pa{\varphi(T_{i,p}^x\confhat)}=\Eref(\varphi)=1$.
\bigskip

We now work on the contribution of the second part of  \eqref{vanishingqtty2}, namely 
\begin{equation}\label{contrib2}\Eref\pa{N^{-1}\sum_{x\in \torus} \com_x\1_{E_{p,x}}S_{x}(\confhat)\cro{\varphi\pa{T_{i,p}^x\confhat}-\varphi}},\end{equation} 
that we wish to estimate  by the Dirichlet form $\rdir(\varphi)$.
The elementary bound 
\begin{equation*}cd\pa{a-b}\leq\frac{A c^2}{2}\pa{\sqrt{a}+\sqrt{b}}^2+\frac{d^2}{2A}\pa{\sqrt{a}-\sqrt{b}}^2,\end{equation*}
which holds for any positive constant $A$, applied to \[a=\varphi \pa{T_{i,p}^x\confhat}, \quad b=\varphi,\quad  c=\com_x S_{x}\mbox{ and } d=\1_{E_{p,x}} \] yields that the quantity above \eqref{contrib2} can be bounded from above for any positive $A$ by 
\begin{equation}\label{decomp4}\frac{1}{N}\sum_{x\in \torus} \Eref\pa{\frac{A}{2}\pa{\com_xS_{x}}^2\pa{\sqrt{\varphi}\pa{T_{i,p}^x\confhat}+\sqrt{\varphi}}^2+\frac{1}{2A}\1_{E_{p,x}}\pa{\sqrt{\varphi}\pa{T_{i,p}^x\confhat}-\sqrt{\varphi}}^2}.\end{equation}
Since we already established that $S_{x}\pa{T_{i,p}^x\confhat}=S_{x} +(\varepsilon N )^{-1}o_N(1)$, since $\com_x$ can be  bounded by $C(\omega)>0$, and since $\1_{ \epx}\leq \1_{E_{p+1,x}}$ the sum above is less than
\begin{equation}\label{on1}\frac{AC^2}{N}\sum_{x\in \torus}\Eref(\varphi S_{x}^2)+\frac{1}{2AN}\sum_{x\in \torus}\Eref\pa{\1_{E_{p+1,x}}\pa{\sqrt{\varphi}\pa{T_{i,p}^x\confhat}-\sqrt{\varphi}}^2}+o_N(1).\end{equation}

According to Section \ref{subsec:irreducibility}, on the event $E_{p+1,x}$ on which there are two empty sites in $B_{p+1}$, there exists a sequence of allowed jumps permitting to reach $T_{i,p}^x\confhat$ from $\confhat$. However, this sequence is random, which we avoid by crudely bounding 
\[\1_{E_{p+1,x}}\leq \sum_{z_1, z_2\in B_{p+1}}(1-\conf_{z_1})(1-\conf_{z_2}),\]
since the right-hand side only vanishes when there are less than one empty site in $B_{p+1}$. 
Given two fixed empty sites  $z_1$ and $z_2$  there exists an integer $n_p(z_1,z_2)$ bounded by a constant $C_p$, and a sequence of edges $((a_m, b_m))_{m\in \llbracket 0,n_p\rrbracket}$  such that   \[\confhat=\confhat(0), \qquad T_{i,p}^x\confhat=\confhat(n_p), \qquad \mbox{ and } \confhat(m+1)=\confhat(m)^{a_m,b_m}\;\forall m\in \llbracket 0,n_p-1\rrbracket,\]
where $a_m$ and $b_m$ are neighboring sites in $B_{p+1}(x)$ and $\conf_{a_m}(\confhat(m))=1-\conf_{b_m}(\confhat(m))=1$.
We can therefore write
\begin{align*}\Eref\pa{\1_{E_{p,x}}\pa{\sqrt{\varphi}\pa{T_{i,p}^x\confhat}-\sqrt{\varphi}}^2}\leq&\sum_{z_1,z_2\in B_{p+1}}\Eref\pa{ n_p\sum_{m=0}^{n_p-1}\1_{E_{p,x}}\pa{\sqrt{\varphi}\pa{\confhat(m+1)}-\sqrt{\varphi}(\confhat(m))}^2}\\
\leq & \;K_p D_{N,p+1}(\varphi), \end{align*}
since $\confhat(m+1)$ is reached from $\confhat(m)$ by an allowed particle jump, where $D_{N,p+1}(\varphi)$ is the contribution of edges in $B_{p+1}$ in $\rdir(\varphi)$.

The sum in the second term of \eqref{on1} can therefore be bounded by  $C^*_p\rdir\pa{\varphi}$, where $C^*_p=(2p+1)^2K_p$. Finally, \eqref{contrib2} can be bounded, for any positive $A$ by 
\[\frac{AC^2}{N}\sum_{x\in \torus}\Eref(\varphi S_{x}^2)+\frac{C^*_p}{2AN}\rdir\pa{\varphi}+o_N(1).\]
We can now set $A=C_p^* / N$, to obtain that
\begin{equation*}\Eref\pa{N^{-1}\sum_{x\in \torus} \com_x\1_{E_{p,x}}S_{x}(\confhat)\cro{\varphi\pa{T_{i,p}^x\confhat}-\varphi}}\leq \frac{C(p,\omega)}{N^2}\sum_{x\in \torus}\Eref(\varphi S_{x}^2)+\frac{1}{2}\rdir\pa{\varphi}+o_N(1).\end{equation*} 
The first term in the right-hand side above vanishes as a consequence of the two-block estimate stated in Lemma \ref{lem:OBE}, since the diffusion coefficients are continuous according to their explicit expression. This concludes the proof of equation \eqref{vanishingsubqtty1}. 

\bigskip

The contribution of the second part of equation \eqref{splitmicromacro} is treated in a similar fashion. Denoting by 
\[S'_{x}(\confhat)=\frac{1}{\abss{B_{l_p}}}\sum_{\abss{y-x}\leq l_p}G(y/N)(\tau_y\D_{\varepsilon N}-\tau_y\D_{l}).\]
As before, the corresponding contribution in the left-hand side of \eqref{micromacroeq} can be written as 
\begin{align*}
-\frac{1}{N}\sum_{x\in \torus}\Eref\left(\com_x\1_{E_{p,x}}\pa{\varphi(T_{i,p}^x\confhat)-\varphi }S'_{x}\right)
,\end{align*}
since this time, $S'_{x}$  is invariant under the action of $T_{i,p}^x$ by definition of $l_p$, whereas the second term can be controlled in the limit $N\to \infty$ as well by $\rdir(\varphi)/2$.
This completes the proof of Lemma \ref{lem:NGTBE} in the case where $\D_k=\diffom\left(\rho_k, \rho_k^{\omega}\right) \mbox{ and } v_k=\ddi\rho^{\omega,p}_k$.

\bigskip

In the case where $\D_k=\diff\left(\rho_k, \rho_k^{\omega}\right) \mbox{ and } v_k=\ddi\rho_k$, the proof is easier and no longer requires indicator functions, since unlike $\ddi \com_x$, $\ddi\conf_x$ vanishes when there is no empty site. We do not give a detailed proof, which would be an easier version of the previous case. We will instead just give a brief outline and the equivalent quantities to the previous ones. The same summation by parts allows us to rewrite 
\[\frac{1}{N}G(x/N)\Eref\Big(\varphi\tau_x(\D_{\varepsilon N}v_{\varepsilon N}-\D_lv_{l_p})\Big)=\frac{1}{N}\Eref\pa{\varphi\sum_{x\in \torus}(S_{x}+S'_{x})(\conf_{x+e_i}-\conf_x)},\]
where 
\[S_{x}=\frac{1}{\abss{B_{\varepsilon N}}}\sum_{\abss{y-x}\leq \varepsilon N}G(y/N)\tau_y\D_{\varepsilon N}-\frac{1}{\abss{B_{l'}}}\sum_{\abss{y-x}\leq l'}G(y/N)\tau_y\D_{\varepsilon N},\]
and
\[S'_{x}(\confhat)=\frac{1}{\abss{B_{l'}}}\sum_{\abss{y-x}\leq l'}G(y/N)(\tau_y\D_{\varepsilon N}-\tau_y\D_{l}).\]
We can now rewrite $\conf_{x+e_i}-\conf_x=\conf_{x+e_i}(1-\conf_x)-\conf_x(1-\conf_{x+e_i})$, to obtain that the quantity above is 
\[\frac{1}{N}\sum_{x\in \torus}\Eref\pa{\conf_x(1-\conf_{x+e_i})\pa{(S_{x}+S'_{x})\varphi}(\confhat^{x,x+e_i})-(S_{x}+S'_{x})\varphi}.\]
The gradients of $S_{x}$ and $S'_{x}$ still vanish, whereas the average of the gradients $\varphi(\confhat^{x,x+e_i})-\varphi$ can be controlled by the sum of  a vanishing term and the Dirichlet form of $\varphi$, since this time the jump rates $\conf_x(1-\conf_{x+e_i})$ are already present.
This concludes the proof of Lemma \ref{lem:NGTBE}.
\end{proof}

\subsection{Projection on non-full sets and reduction to a variance problem}
\label{subsec:localvariance}

We now prove the limit \eqref{k4}, which states that in a local average,  the current can be replaced by gradients, up to a perturbation $\gene f$. Following the exact same steps as in Section \ref{subsec:k3}, up until the statement of Lemma \ref{lem:NGTBE}, where we reduced the proof of equation \eqref{k3} to \eqref{k3var}, we reduce the proof of equation \eqref{k4} to the variational formula
\begin{equation}\label{eqloc}\inf_f\lim_{p\to \infty}\liml\limN\sup_{\varphi}\left\{\Eref\pa{\varphi  Y_4(G,\confhat)}-\rdir(\varphi)\right\}\leq 0,\end{equation}
where we shortened \[Y_4(G,\confhat)=Y_{i,4}^{f,l,p}(G,\confhat)=\frac{1}{N}\sum_{x\in \torus}G(x/N)\tau_x{\W}_{i,4}^{f,l,p},\]
and ${\W}_{i,4}^{f,l,p}$ was introduced in equation \eqref{w4}.
Since this step is performed in the exact same way as in the beginning of Section \ref{subsec:k3}, we do not detail them here and refer the reader to the latter.
To simplify notations, we shorten
\[{\W}_{i}^{l}={\W}_{i,4}^{f,l,p}\]
for the local average of the difference between gradients and currents in the direction $i$.

We will now work to get an estimate of the largest eigenvalue of the small perturbation $\gene + Y_{4}$ of $\gene$. The strategy is close to the one used in the one-block estimate of Section \ref{subsec:OBE}. To do so, we break down the process on finite boxes with a fixed number of particles, where the generator $\gene$ has a positive spectral gap.
In order to introduce this restriction, we adopt once again the notations introduced in Section \ref{subsec:OBE}, which we briefly recall here. Let $B_l=\llbracket-l, l\rrbracket^2$ be the box of size $l$,  $\K=(K,  \{\theta_1,\ldots ,\theta_K\})$ be some particle number and angles. Recall that $\Kset_l$ is the set of $\K$'s such that $K\leq (2l+1)^2$, and denote by $\param_{\K}$ the grand-canonical parameter 
\[\param_{\K}=\frac{1}{(2l+1)^2}\sum_{k=1}^K\delta_{\theta_k}\in \pset.\]
Recall that we already defined in \eqref{sousespace}
\begin{equation*}\subspace=\left\{\confhat\in\statespace \left| \quad \dens_l=\param_{\K}\right. \right\}\end{equation*}
the set of configurations with $K$ particles in $B_l$ with angles $\theta_k$'s. Also recall that  $\cmlk$ is the canonical measure $\mesref(\;.\;\mid\subspace)$ conditioned to particle configurations of the form $\K$ in $B_l$. 

We denote for any site $x$ $\varphi^x=\tau_{-x}\varphi$, and by $\varphi^x_{l,\K}$ the density induced by $\varphi^x$ on $\subspace$. It can be defined for any configuration $\zetahat$ on $B_l$ by
\[\varphi^x_{l,\K}(\zetahat)=\frac{\Eref(\varphi^x\mid\confhat_{\mid B_l}=\zetahat)}{\Eref(\varphi^x\mid \subspace)}.\]
Let us now get back to the quantity of interest, 
\begin{equation}
\label{decompnbrepart}
\Eref\pa{\varphi {Y}_4(G,\confhat)}=\frac{1}{N}\sum_{x\in \torus}G(x/N)\Eref\pa{\varphi\tau_x{\W}_i^l}=\frac{1}{N}\sum_{x\in \torus}G(x/N)\Eref\pa{{\W}_i^l\varphi^x}.
\end{equation}
Because ${\W}_i^l$ only depends on the vertices in $B_l$, we can replace the expectation under $\mesref$ by the integral over  $\Kset_l$ of the expectation under $\cmlk$. More precisely, let us denote 
\[m_x(d\K)=\Eref\pa{\varphi^x\1_{\Sigma_l^{d\K}}},\]
the infinitesimal probability of being on the set $\subspace$ under the measure with density $\varphi^x$ w.r.t $\mesref$.
Thanks to \eqref{decompnbrepart}, letting $\E^*_{l,\alpha}$ be the conditional expectation of $\Eref$ w.r.t the sites inside of $B_l$,  we can  write
\begin{align}\label{local1}\Eref\pa{\varphi Y_4(G,\confhat)}&=\frac{1}{N}\sum_{x\in \torus}G(x/N)\E^*_{l,\alpha}\pa{{\W}_i^l\varphi^x}\nonumber\\
&=\frac{1}{N}\sum_{x\in \torus}G(x/N)\int_{\K\in\Kset_l}\Ecmlk\pa{{\W}_i^l\varphi^x_{l,\K}}m_x(d\K).\end{align}

Let us now decompose in a similar fashion the Dirichlet form. For $\varphi$ some density with respect to $\gcm$, let $D_{l,\K}$ be the Dirichlet form on $\subspace$ 
\[D_{l,\K}(\varphi)=\frac{1}{2}\sum_{\substack{x,y\in B_l\\ \abss{x-y}=1}}\Ecmlk\cro{\conf_x(1-\conf_y)\pa{\sqrt{\varphi\pa{\confhat^{x, y}}}-\sqrt{\varphi}}^2}.\]
We have with the same tools as in the proof of Lemma \ref{lem:OBE} \begin{equation}\label{local2}\sum_{x\in \torus}\int_{\K\in \Kset_l}D_{l,\K}\pa{\varphi^x_{l,\K}}m_x(d\K)\leq  (2l+1)^2 \rdir(\varphi).\end{equation}

From the previous considerations, we can localize the quantity inside braces in equation \eqref{eqloc}, which is bounded above thanks to \eqref{local1} and \eqref{local2} by 
\begin{align}\label{varformula1}\Eref\pa{\varphi  Y_4(G,\confhat)}-\rdir(\varphi)=&\sum_{x\in \torus}\int_{\K\in \Kset_l}m_x(d\K) \Bigg(\frac{1}{N}G(x/N)\Ecmlk\pa{{\W}_i^l \varphi^x_{l,\K}}-(2l+1)^{-2}D_{l,\K}\pa{\varphi^x_{l,\K}}\Bigg)\nonumber \\
\leq&\kappa_1\sum_{x\in \torus}\sup_{\K\in \Kset_l} \cro{\frac{\kappa_2}{N}\Ecmlk\pa{{\W}_i^l\varphi^x_{l,\K}}-D_{l,\K}\pa{\varphi^x_{l,\K}}} \nonumber\\ 
\leq&\kappa_1\sum_{x\in \torus}\sup_{\K\in \Kset_l} \sup_{\psi}\cro{\frac{\kappa_2}{N}\Ecmlk\pa{{\W}_i^l\psi}-D_{l,\K}\pa{\psi}},
\end{align}
since $\int_{\K\in\Kset_l}m_x(d\K)=1$, where \[\kappa_1=(2l+1)^{-2} \quad \mbox{ and }\quad \kappa_2=G(x/N)(2l+1)^2,\]
and the supremum is taken over all densities $\psi$ with respect to $\cmlk$.

\bigskip

We now wish to exclude in the supremum over $\K$ above the configurations with one or less empty sites since on the corresponding sets, the exclusion process is not irreducible as investigated in Section \ref{subsec:irreducibility}. First note that for any $\K$ such that $K= \abss{B_l}$, ${\W}_i^l$ vanishes. Indeed, thanks to our cutoff functions $\1_{E_p}$, and since $l$ goes to $\infty$ before $p$, in that case, the currents, the gradients as well as the $\gene f$'s in ${\W}_i^l$ all vanish as well as $D_{l,\K}\pa{\psi}$.

We now consider the case where $K=\abss{B_l}-1$, i.e. when there is one empty site in $B_l$. We state the corresponding estimate as a separate lemma for the sake of clarity. 

\begin{lemm}\label{lem:controle1trou}
There exists a constant $C=C(G,\omega, f)$ such that for any $\K$ such that $K=\abss{B_l}-1$, 
\[\frac{\kappa_2}{N}\Ecmlk\pa{{\W}_i^l\psi}\leq D_{l,\K}\pa{\psi}+\frac{C}{N^2}.\]
\end{lemm}
\proofthm{Lemma \ref{lem:controle1trou}}{
First note that all the gradients $\ddi\comp$ vanish in the expression of ${\W}_i^l$ due to the cutoff functions. We can therefore write, for any configuration with one or less empty site, that \[{\W}_i^l=\frac{1}{(2l'+1)^2}\sum_{x\in B_{l'}}\pa{\curom_{x,x+e_i}+\diff_{\K}\cur_{x,x+e_i}}-\frac{1}{(2l_f+1)^2}\gene_l \overline{f},\]
where we denoted by $\diff_{\K}$ the value on $\Sigma_l^{\K}$ of $\diff\left(\rho_{l}, \densom_l\right)$, which does not depend on the configuration, and $\overline{f}=\sum_{x\in B_{l_f}}\tau_x f$. The quantity we want to estimate can therefore be rewritten
 \[\frac{\kappa_2}{N}\Ecmlk\pa{{\W}_i^l\psi}=\frac{\kappa_2}{N(2l'+1)^2}\Ecmlk\pa{\psi \sum_{x\in B_{l'}}\pa{\curom_{x,x+e_i}+\diff_{\K}\cur_{x,x+e_i}}}-\frac{\kappa_2}{N(2l_f+1)^2}\Ecmlk\pa{\psi \gene_l \overline{f}},\]
 {where $\gene_l$ is the generator of the symmetric exclusion process restricted to jumps with both ends in $B_l$}. 
Since $\kappa_2$, $(2l'+1)^2$, and $(2l_f+1)^2$ are of order $(2l+1)^2$, and since the sign of $f$ is arbitrary, to prove Lemma \ref{lem:controle1trou} it is sufficient to prove that for any $A>0$, we have both
\begin{multline}\label{duo1trou}\frac{1}{N}\Ecmlk\pa{\psi \sum_{x\in B_{l'}}\pa{\curom_{x,x+e_i}+\diff_{\K}\cur_{x,x+e_i}}}\leq\frac{D_{l,\K}\pa{\psi}}{2A}+\frac{AC(\omega)}{N^2} \\
\eqand  \frac{1}{N}\Ecmlk\pa{\psi \gene_l \overline{f}}\leq \frac{D_{l,\K}\pa{\psi}}{2A}+\frac{AC(f)}{N^2}.
\end{multline}

\bigskip

The two inequalities above are proved in the same way. We treat in detail the second, which is the most delicate, and simply sketch the adaptations to obtain the first. Using the elementary inequality \begin{equation}\label{elemprod}ab\leq\frac{ \gamma a^2}{2}+\frac{b^2}{2 \gamma },\end{equation} which holds for any positive $\gamma$, we first write
\begin{align*}\Ecmlk\pa{\psi \gene_l \overline{f}}&=\sum_{x, x+z \in B_l}\Ecmlk\pa{\psi\nabla_{x,x+z}\overline{f}}\\
&=-\frac{1}{2}\sum_{x, x+z\in B_l}\Ecmlk\pa{\nabla_{x,x+z}\psi\nabla_{x,x+z}\overline{f}}\\
&\leq \sum_{x, x+z\in B_l}\frac{\gamma}{4}\Ecmlk\pa{(\nabla_{x,x+z}\sqrt{\psi})^2}+\frac{1}{4\gamma}\Ecmlk\pa{(\nabla_{x,x+z}\overline{f})^2(\sqrt{\psi}+\sqrt{\psi}(\confhat^{x,x+z}))^2}\\
&=\frac{\gamma}{2}\rdir_{l,\K}\pa{\psi}+\frac{1}{4\gamma}\Ecmlk\pa{\sum_{x, x+z\in B_l}\conf_x(1-\conf_{x+z})(\overline{f}-\overline{f}(\confhat^{x,x+z}))^2(\sqrt{\psi}+\sqrt{\psi}(\confhat^{x,x+z}))^2}.\end{align*}
One only has now to carefully account for the order of the different quantities in the second term. Since $f$ is a bounded local function, by definition of  $\overline{f}$, it is invariant under particle jumps with both ends outside of its domain. There hence exists a constant $C(f)$ such that for any $x$ and $x+z$, $\overline{f}-\overline{f}(\confhat^{x,x+z})\leq C(f)$. In particular, the constant $C(f)$ does not depend on $l$. We can also crudely bound $\conf_x$ by $1$ and $(\sqrt{\psi}+\sqrt{\psi}(\confhat^{x,x+z}))^2$ by $2\psi+\psi(\confhat^{x,x+z})$. These bounds and a change of variable $\confhat\to \confhat^{x,x+z}$  finally yield that for any positive $\gamma$,
\[\Ecmlk\pa{\psi \gene_l \overline{f}}\leq\frac{\gamma}{2}\rdir_{l,\K}\pa{\psi}+\frac{C(f)}{2\gamma}\Ecmlk\pa{\sum_{x, x+z\in B_l}(2-\conf_x-\conf_{x+z})\psi}.\]
Furthermore, since there is only one empty site in $B_l$, 
\[\sum_{\abss{y}\leq l-1}(2-\conf_y-\conf_{y+e_i})=\underset{\leq 1}{\underbrace{\abss{B_{l-1}}-\sum_{y\in B_{l-1}}\conf_y}}+\underset{\leq 1}{\underbrace{\abss{\tau_{e_i} B_{l-1}}-\sum_{y\in \tau_{e_i}B_{l-1}}}}\conf_y\leq 2,\] 
therefore,  since $\psi$ is a probability density, and setting $\gamma=N/A$ proves the second identity of \eqref{duo1trou}.

\bigskip

The second identity is obtained in the same way, since 
\[\frac{1}{N}\Ecmlk\pa{\psi \sum_{x\in B_{l'}}\pa{\curom_{x,x+e_i}+\diff_{\K}\cur_{x,x+e_i}}}=\frac{1}{N}\sum_{\abss{y}\leq l-1}\Ecmlk\pa{(\omega(\theta_y)+\diff_{\K})\nabla_{y,y+e_i}\psi},\]
we also obtain \begin{multline*}\frac{1}{N}\Ecmlk\pa{\psi \sum_{x\in B_{l'}}\pa{\curom_{x,x+e_i}+\diff_{\K}\cur_{x,x+e_i}}}\\\leq \frac{\gamma}{2}\rdir_{l,\K}\pa{\psi}+\frac{\pa{\norm{\omega}_{\infty}+\norm{\diff}_{\infty}}^2}{2\gamma}\Ecmlk\pa{\sum_{x, x+e_i\in B_l}(2-\conf_x-\conf_{x+e_i})\psi}.\end{multline*}
The last estimate, in turn, yields the first inequality in \eqref{duo1trou}, which concludes the proof of Lemma \ref{lem:controle1trou}.
}

In the limit $N\to \infty$ then $l\to \infty$, Lemma \ref{lem:controle1trou} yields, since $\kappa_1$ vanishes as $l\to \infty$, and since all quantities vanish when $K=\abss{B_l}$, that 
\[\kappa_1\sum_{x\in \torus}\sup_{\substack{\K\in \Kset_l\\
K\geq \abss{B_l}-1}} \sup_{\psi}\cro{\frac{\kappa_2}{N}\Ecmlk\pa{{\W}_i^l\psi}-D_{l,\K}\pa{\psi}}\to 0. \]

We can therefore restrict the supremum over $\K$ to those satisfying $K\leq\abss{B_l}-2$. Recall that we denoted in equation \eqref{Ksettdef} by $\Ksett_l$ the set of such $\K$, the left-hand side of \eqref{eqloc} is bounded by 
\begin{equation}\label{varformulakmoins1}\inf_{f}\lim_{p\to \infty}\liml\limsup_{N\to\infty}\kappa_1\sum_{x\in \torus}\sup_{\K\in \Ksett_l} \sup_{\psi}\cro{\frac{\kappa_2}{N}\Ecmlk\pa{{\W}_i^l\psi}-D_{l,\K}\pa{\psi}}, \end{equation}
where the supremum is taken over all densities $\psi$ w.r.t. $\mu_{l,\K}$. On all the sets $\subspace$ considered, $\gene_l$ is invertible and the supremum over $\psi$ is a variational formula for the largest eigenvalue of the operator $\gene_l+\kappa_2{\W}_i^l/N$.  Proposition \ref{prop:markoveigenvalue} then allows us to bound  the quantity whose limit  is taken in \eqref{varformulakmoins1} by
\begin{equation*}\limsup_{N\to\infty}\sup_{\K\in \Ksett_l} \frac{\kappa_1\kappa_2^2}{1-2\gamma_l\norm{{\W}_i^l}_{\infty}\kappa_2 N^{-1}}\Ecmlk\pa{{\W}_i^l(-\gene_{l})^{-1}{\W}_i^l}\leq  \norm{G}_{\infty}^2 (2l+1)^2\sup_{\K\in \Ksett_l}\Ecmlk\pa{{\W}_i^l(-\gene_{l})^{-1}{\W}_i^l}.\end{equation*}
To obtain the last inequality, we denoted by $\gamma_l$ the spectral gap of the local generator $\gene_l$, which is positive, and used that $\norm{{\W}_i^l}_{\infty}$ is finite, and $\kappa_1 \kappa_2^2=  \norm{G}_{\infty}^2(2l+1)^2$. In order to obtain inequality \eqref{eqloc}, and conclude the proof of equation \eqref{k4}, it is therefore sufficient to prove the following result.

\begin{prop}[Estimate of the local covariance]
\label{prop:TCL}
Recall that ${\W}_i^l$ is the local average of the difference between currents and gradients up to $\gene f$, namely  
\begin{equation*}{\W}_i^l=\langle\curom_{i}\rangle_0^{l'}+d_s\left(\rho_l\right)\ddi\rho^{\omega,p}_{l_p}+\diff\left(\rho_l,\rho_l^{\omega}\right)\ddi\rho_{l'}-\langle\gene f\rangle_0^{l_f},\end{equation*}
where $\diff$ is given by equation \eqref{diffconddef}. Recall that $\Ksett_l$ only takes into account configurations with two empty sites in $B_l$. Then, 
\begin{equation}\label{STcov}\inf_{f}\lim_{p\to \infty}\liml\sup_{\K\in \Ksett_l}{(2l+1)^2}\Ecmlk\pa{{\W}_i^l(-\gene_{l})^{-1}{\W}_i^l}=0.\end{equation}
\end{prop}

\subsection{Limiting variance and diffusion coefficients}
\label{subsec:halpha}
\intro{In Section \ref{subsec:localvariance}, we reduced the proof of \eqref{k4}, and that of Theorem \ref{thm:NGestimates}, to estimating a local variance. In this section, we introduce the limiting variance $\scal{\cdot}$ and investigate its properties and the structure of a set of functions with mean-$0$ w.r.t. any canonical measures, equipped with $\scal{\cdot}$. The presence of indicator functions in $\ddi\comp_{0}$ and the necessity for a uniform estimate in the canonical state $\K\in \Ksett_l$ makes this section fairly technical, however, most of the results come from elementary linear algebra.
The main results of this section is Proposition \ref{prop:currentsdecomposition}, which is the main ingredient to prove Proposition \ref{prop:TCL}, and therefore concludes the proof of Theorem \ref{thm:NGestimates}. 
}

%

To prove Proposition \ref{prop:TCL}, we are now going to investigate the limit as $l\to \infty$ and $\param_{\K_l}\to \param$ (cf Definition \ref{defi:convparam}) of
\begin{equation}\label{covlim}\frac{1}{(2l+1)^2}\Ecml\pa{(-\gene_{l})^{-1} \sum_{x\in B_{l_{\psi}}}\tau_x \psi\; .\sum_{x\in B_{l_{\psi}}}\tau_x \psi}:=\;\;\scal{\psi},\end{equation}
where $\psi$ is supported by $B_{s_\psi}$ and $ l_{\psi}=l-s_\psi-1$ is chosen such that $\sum_{x\in B_{l_{\psi}}}\tau_x \psi$ is measurable w.r.t. sites in $B_l$. There are therefore two important steps to prove \eqref{STcov}~:
\begin{itemize}
\item prove that the limit \eqref{covlim} is well-defined for any function $\psi$ in a convenient class of functions containing at least the currents, the gradients and $\gene \mathcal{C}$. This is done in Definitions \ref{defi:limitingcovariance1}, \ref{defi:limitingcovariance2}, and Theorem \ref{thm:limcovariance} below.
\item Prove that, {shortening $\Egcm(\omega)=\Egcm(\omega(\theta_0)|\eta_0=1)$ and letting}
\begin{equation}
\label{diffomdef}\diff(\param)=\Egcm(\omega)(1-d_s(\am)),\end{equation} we have
\begin{equation}\label{decompheu}\inf_{f\in \mathcal{C}}\lim_{p\to\infty} \sup_{\param} \scal{\curom_{i}+d_s(\am)\ddi(\com_0\1_{E_p})+\diff(\param)\ddi\conf_0-\gene f}=0.\end{equation}
which is done below in Proposition \ref{prop:NGEuniforme}.
\end{itemize}

\bigskip

We introduce a class of local functions with mean $0$ w.r.t. any $\cm$. When there are less than one empty site in the domain $B$, we  require these functions to vanish in order to avoid classifying the irreducible subsets of $\statespace$ when there is only one empty site.
Recall that we already introduced in Definition \ref{defi:CM} the  sets $\Kset_l$ and $\Ksett_l$. 
\index{$ s_\psi$\dotfill  smallest $l$,  $\psi$ depends only on sites in $B_{l}$}
We now define  
\begin{equation}
\label{CzeroDef}
\czero=\left\{\psi\in {\mathcal C}\; \;\Big| \;\; \E_{s_\psi,\K}(\psi)=0 \;\; \forall\K\in\Ksett_{s_\psi}\quad \mbox{ and  } \quad \psi_{\vert\Sigma^{\K}_{s_{\psi}}}\equiv 0\; \;\forall  \K\in\Kset_{s_\psi}\smallsetminus \Ksett_{s_\psi}\right\}.
\end{equation}
In particular, any function  $\psi\in \czero$ has mean zero w.r.t any  canonical measure.  Note that $\psi\in \czero$, and any $\param\in\pset$, conditioning w.r.t. the canonical state of the configuration in $B_{s_\psi}$, we obtain in particular that  $\Egcm(\psi)=0$.
\index{$ \czero$\dotfill space of mean $0$ func. w.r.t. any $\mu_{l,\K}$ }
Further define 
\begin{equation}
\label{eq:DefTomega}
{T^\omega} =\left\{f{\in \mathcal{C}}\; \;\Big| \;\; f(\confhat)=\varphi(\conf)+\sum_{x\in \Z^2}\com_x\psi_x({\conf}), \quad \varphi,\psi_x\in \Sp,\; \forall x \in \Z^2\right\},
\end{equation}
of functions whose only dependency in the $ \theta_x$'s is a linear combination of the $\omega(\theta_x)$. {Note that since we only consider local functions, this set is well-defined.}

Denote
\begin{equation}
\label{eq:DefT0}
\tzero=\czero\cap T^\omega.
\end{equation}
Note that $\tzero$ and $\czero$ are stable by the symmetric exclusion generator $\gene$. Further note that by construction, $\ddi(\com_0\1_{E_p})\in \tzero$.

Recall that for any function $\Phi$ on $\ctoruspi$, $\cur^\Phi_i=\Phi(\theta_0)\conf_0(1-\conf_{e_i})-\Phi(\theta_{e_i})\conf_{e_i}(1-\conf_0)$ denotes the symmetric current associated with $\Phi$ (we also shortened $\cur_i=\cur_i^1=\conf_0-\conf_{e_i}$).
We define $J^*$ the set of {linear} combinations of currents spanning any smooth angular functions, 
\begin{equation}
\label{eq:DefJstar}
J^*=\Big\{ \cur_1^{\Phi_1}+\cur_2^{\Phi_2},\quad \mbox{for }  \Phi_1, \;\Phi_2\in C^1(\ctoruspi)\Big\},
\end{equation}
and let 
\begin{equation}
\label{eq:DefJomega}
J^\omega=J^*\cap T^\omega=\Big\{\cur^{a,b}:=\sum_{i=1,2}a_i\curom_i+b_i\cur_i,\quad a, b\in \R^2\Big\}. 
\end{equation}

We now have all the notations needed to introduce the limiting variance $\scal{\cdot}$. In order to be able to estimate concisely the drift term later on, and to solve a technical issue, we need a rather general result. In particular, we give two distinct constructions for $\scal{f}$ depending on the nature of the function $f$. Fix $\param\in \pset$. Although it is not clear at this point that those two definitions actually coincide, this difficutly is adressed by Theorem \ref{thm:limcovariance} below, which states that the object $\scal{\cdot}^{1/2}$ is a semi norm, and  that for any function to which both Definitions \ref{defi:limitingcovariance1} and \ref{defi:limitingcovariance1} apply, the two definitions actually coincide.
\index{$\scal{\cdot}$ \dotfill limit of the space time covariance}
\begin{defi}[Definition of  $\scal{\cdot}$ on $J^*+\gene\mathcal{C}$]
\label{defi:limitingcovariance1}
For any  $\Phi_1, \;\Phi_2\in C^1(\ctoruspi)$ and for any local function $g\in\mathcal{C}$, we define
\begin{equation}
\label{scal0}
\scal{ \cur_1^{\Phi_1}+\cur_2^{\Phi_2}+\gene g}=\sum_{i=1,2}\Egcm\pa{\conf_0(1-\conf_{e_i})\Big[\Phi_i(\theta_0)+\Sigma_g(\confhat^{0,e_i})-\Sigma_g)\Big]^2},
\end{equation}
where $\Sigma_g=\sum_{x\in \Z^2}\tau_x g$, which is not a priori well-defined, but whose gradient $\Sigma_g(\confhat^{0,e_i})-\Sigma_g$ is, because $g$ is a local function. For any function  $\psi\in \tzero+ J^*+\gene\mathcal{C}$, define
\begin{equation}
\label{scal}
\scal{\psi\;,\;\gene g+\cur_1^{\Phi_1}+\cur_2^{\Phi_2}}=-\Egcm\Bigg(\psi\bigg[ \Sigma_g+\sum_{x\in\Z^2} \pa{x_1\conf_x^{\Phi_1}+x_2\conf_x^{\Phi_2}} \bigg]\Bigg) 
\end{equation}
which once again is well-defined because any $\psi\in \tzero+ J^*+\gene\mathcal{C}$ is a local function with mean-$0$ w.r.t. any $\mesinv$, 
therefore the expectation above only involves a finite number of non-$0$ contributions. 
{In particular, an elementary computation yields that for any $g\in \mathcal{C}$, and $j\in J^*$
\[\scal{\gene g+j\;,\;\gene g+j}=\scal{\gene g+j}\]
where the left hand-side is given by \eqref{scal} and the right-hand side by \eqref{scal0}.}
\end{defi}
\begin{defi}[Definition of  $\scal{\cdot}$ on $\tzero$]
\label{defi:limitingcovariance2}
For any $\psi\in \tzero$, define
{\begin{equation}
\label{varformulascal}
\scal{\psi}=\sup_{\substack{g\in \tcal\\
\cur\in J^\omega}}\bigg\{2\scal{\psi\;,\;\gene g+\cur}-\scal{\gene g+\cur}\bigg\},\end{equation}
where $\tcal, \tzero $ and $J^\omega$ were defined in \eqref{eq:DefT0} and \eqref{eq:DefJomega}, and the two terms inside braces are respectively given by 
\eqref{scal0} and \eqref{scal}.}

For $\psi\in \tzero$ and  $\cur_1^{\Phi_1}+\cur_2^{\Phi_2}+\gene g\in J^*+\gene\mathcal{C}$, we also define
\[\scal{\psi+\gene g+ \cur_1^{\Phi_1}+\cur_2^{\Phi_2}}\egal \scal{\gene g+\cur_1^{\Phi_1}+\cur_2^{\Phi_2}}+\scal{\psi}+2\scal{\psi,\gene g+\cur_1^{\Phi_1}+\cur_2^{\Phi_2}},\]
where the three terms in the right-hand side are respectively given by \eqref{scal0}, \eqref{varformulascal} and \eqref{scal}.

These definitions allow us to finally define on $\tzero+J^*+\gene\mathcal{C}$ a bilinear form  $\scal{\cdot, \cdot}$ by letting $\scal{\psi, \psi}=\scal{\psi}$ for any $\psi\in\tzero+J^*+\gene\mathcal{C}$, by polarization identity on ${\tzero}^2$ and $(J^*+\gene\mathcal{C})^2$, and by \eqref{scal} on $\tzero\times (J^*+\gene\mathcal{C})$.
\end{defi}

\begin{rema}
We will see in the proof of Theorem \ref{thm:limcovariance} below that this definition coincides with Definition \ref{defi:limitingcovariance1} for any $\psi\in\tzero\cap \{J^*+\gene\mathcal{C}\}\subset J^\omega+\gene\tcal$, since in this case the supremum in \eqref{varformulascal} is reached for $f=\gene g+j^{a,b}$ itself. 
\end{rema}

For any cylinder function $\psi$,  recall that  $s_\psi$ is the smallest fixed integer such that $\psi$ is measurable with respect to $\mathcal{F}_{s_\psi}$, and let $l_{\psi}=l-s_\psi-1$ for any integer $l$ large enough. 
The following result justifies the definitions above, and states that $\scal{\psi}$ 
defined for any $\psi\in\tzero+J^*+\gene\mathcal{C}$ is the limit of \eqref{covlim}.

\begin{theo}\label{thm:limcovariance}
Fix $\param \in \pset$, and a sequence  $(\K_l)_{l\in \N}$ such that $\K_l\in \Ksett_l$ and  $\normm{\param_{\K_l}-\param}\to 0$, where $\param_{\K_l}\in \pset $ is the grand-canonical parameter defined in \eqref{defi:paramK}. 

The bilinear form $\scal{\cdot,\cdot}$ introduced in  Definition \ref{defi:limitingcovariance2} is a semi-inner product on $\tzero+J^*+\gene\mathcal{C},$ and, for any functions $\psi, \varphi \in  \tzero+J^*+\gene\mathcal{C}$,
\begin{equation}
\label{limitcovQtty}
\lim_{l\to \infty}\frac{1}{(2l+1)^2}\Ecml\pa{(-\gene_{l}^{-1}) \sum_{x\in B_{l_{\psi}}}\tau_x \psi\; .\sum_{x\in B_{l_{\varphi}}}\tau_x \varphi}\egal \scal{\psi, \varphi}.
\end{equation}
Furthermore, for any $\psi, \varphi \in  \tzero+J^*+\gene\mathcal{C}$, the application $\param\to \scal{\psi, \varphi}$ is continuous in $\param$, and the convergence above is uniform in $\param$. In particular, for any $\psi\in \tzero+J^*+\gene\mathcal{C}$, 
\begin{equation}
\label{covunif}
\lim_{l\to \infty}\sup_{\K\in \Ksett_l}\frac{1}{(2l+1)^2}\Ecmlk\pa{(-\gene_{l}^{-1}) \sum_{x\in B_{l_{\psi}}}\tau_x \psi\; .\sum_{x\in B_{l_{\psi}}}\tau_x \psi}\egal\sup_{\param\in \pset}\scal{\psi}.
\end{equation}
\end{theo}

The proof of Theorem \ref{thm:limcovariance} is the purpose of Section \ref{sec8}, and is postponed for now. It requires many adaptations because of the angles, but follows the global strategy presented in \cite{KLB1999}.
Let us explicitly write the dependency in $p$ and $f$ of ${\W}_i^l={\W}_{i,p}^{f,l}$ appearing in Proposition \ref{prop:TCL}, and define  for any $\param\in \pset$
\begin{equation}\label{vipdiff}\V_{i,p}^f(\param)=\curom_i+d_s(\am)\ddi\comp_{0}+\diff(\param)\ddi\conf_{0}+\gene f\in \tzero+J^*+\gene\mathcal{C}.\end{equation}
Recall that $l_f=l-s_{f}-1$, where  $s_f$ is also the size of the support of $\V_{i,p}^f$ (since we can safely increase $s_f$, in order to have $s_f=s_{\V_{i,p}^f}$) and define 
\[Q_1=(2l+1)^2{\W}_{i,p}^{f,l}-\sum_{x\in B_{l_f}}(\tau_x\V_{i,p}^f)(\densl)\eqand Q_2=\sum_{x\in B_{l_f}}\cro{(\tau_x\V_{i,p}^f)(\densl)-(\tau_x\V_{i,p}^f)(\param)}.\] 
For $h$ a cylinder function measurable w.r.t. sites in $B_l$, define $\dir_{l, \K}(h)=\Ecmlk(h(-\gene_l) h)$. For $\param_{\K_l}\to \param$, the variational formula for the variance yields

\begin{align*}\Ecml\Big({\W}_{i,p}^{f,l}&(-\gene_{l}^{-1}){\W}_{i,p}^{f,l}\Big)= \sup_{h}\left\{\Ecml\pa{h{\W}_{i,p}^{f,l}}-\dir_{l, \K_l}(h)\right\}\\
\leq& \sup_{h}\left\{\frac{1}{(2l+1)^2}\Ecml\pa{h \sum_{x\in B_{l_f}}(\tau_x\V_{i,p}^f)(\param)}-\frac{1}{3}\dir_{l, \K_l}(h)\right\}\\
&+ \sup_{h}\left\{ \frac{1}{(2l+1)^2}\Ecml\pa{hQ_1}-\frac{1}{3}\dir_{l, \K_l}(h)\right\}+ \sup_{h}\left\{ \frac{1}{(2l+1)^2}\Ecml\pa{hQ_2}-\frac{1}{3}\dir_{l, \K_l}(h)\right\}\\
\leq& \frac{3}{(2l+1)^4}\Ecml\pa{(-\gene_{l}^{-1}) \sum_{x\in B_{l_{f}}}(\tau_x \V_{i,p}^{f})(\param)\; .\sum_{x\in B_{l_{f}}}(\tau_x\V_{i,p}^{f})(\param)}\\
&+\sup_{h}\left\{ \frac{1}{(2l+1)^2}\Ecml\pa{hQ_1}-\frac{1}{3}\dir_{l, \K_l}(h)\right\}+\sup_{h}\left\{ \frac{1}{(2l+1)^2}\Ecml\pa{hQ_2}-\frac{1}{3}\dir_{l, \K_l}(h)\right\}.
\end{align*}
Since the discrepancies in ${Q_1}=(2l+1)^2{\W}_{i,p}^{f,l}-\sum_{x\in B_{l_f}}\V_{i,p}^f(\densl)$ occur only in $B_{l-1}\setminus B_{l_f}$, letting $\gamma=1/(2l+1)^2$, Lemma \ref{lem:techps} below yields that the second term above is less than 
\[C_f\abs{B_{l-1}\setminus B_{l_f}}(2l+1)^{-4}=O(l^{-3}).\] 
The last term multiplied by $(2l+1)^2$ vanishes as well thanks to Lemma \ref{lem:techps} and because the diffusion coefficients $d_s$ and $\diff$ are continuous in $\param$. Furthermore, as in Lemma \ref{lem:techps}, both of these convergences are uniform in $\K_l$ and $\param$. We can therefore apply Theorem \ref{thm:limcovariance} to the first term to obtain that for any $f\in \mathcal{C}$,
\begin{equation*}\lim_{l\to \infty}\sup_{\K}{(2l+1)^2}\Ecmlk\pa{{\W}_{i,p}^{f,l}(-\gene_{l})^{-1}{\W}_{i,p}^{f,l}}\leq  3\sup_{\param\in \pset} \scal{\V_{i,p}^{f}(\param)},\end{equation*}
therefore to prove Proposition \ref{prop:TCL}, and thus Equation \eqref{k4}, it is sufficient to prove
\begin{equation}\label{decompunif}\inf_{f\in\mathcal{C}}\lim_{p\to\infty}\sup_{\param\in \pset} \scal{\V_{i,p}^{f}(\param)}=0.
\end{equation} 
This estimate is proved later on in Proposition \ref{prop:NGEuniforme}, and requires to understand the structure of the space $\tzero+J^*+\gene\mathcal{C}$ equipped with $\scal{\cdot}$. It is the main result of this section.
\bigskip

For any $\Phi\in C^1(\ctoruspi)$ and any $\param\in \pset$, {we shorten}
\[\Egcm(\Phi):=\Egcm(\Phi(\theta_0)\mid \conf_0=1)  \eqand V_{\param}(\Phi):=Var_{\param}(\Phi(\theta_0)\mid \conf_0=1),\] 
\[Cov_{\param}(\omega, \Phi)=\Egcm(\omega \Phi)-\Egcm(\omega)\Egcm(\Phi),\eqand \widehat{\Phi}( \theta)=\Phi(\theta)-\Egcm(\Phi).\]
In particular, we denote by $\cur^{\widehat{\Phi}}_i=\cur^{\Phi}_i-\Egcm(\Phi)\cur_i=\cur^{\Phi}_i+\Egcm(\Phi)\ddi \conf$ the associated current. Note that any element $\cur_1^{\Phi_1}+\cur_2^{\Phi_1}$ of $J^*$ can be written as a linear combination of the $\cur_i^{\widehat{\Phi}_i}$ and $\cur_i$'s, $i=1,2$.
For any fixed $\param$, we finally define the function $h_i^p$ by
\begin{align*}h_i^p(\confhat)=d_s(\alpha)(\ddi \comp_0+\Egcm(\omega)\cur_i)&=\ddi \cro{ d_s(\alpha)(\com_0\1_{E_p}-\Egcm(\omega)\conf_0)} \\
&=d_s(\am)(\comz_{e_i}-\comz_0)-d_s(\am)\cro{\com_{e_i}\1_{\tau_{e_i}E_p^c}-\com_0\1_{E_p^c}},
\end{align*}
where as before $E_p=\Big\{\sum_{x\in B_p}\conf_x\leq\abss{B_{p}}-2\Big\}$. 
\medskip

We can now rewrite \eqref{vipdiff} as
\begin{equation}\label{vipdef2}\V_{i,p}^f(\param)=\curohat_i+h_i^p+\gene f.\end{equation}
Note that both $\curohat_i$ and $h_i^p$ depend on $\param$ as well as $\omega$, but to simplify notations, we do not write it explicitly. Throughout this section, we will not indicate the dependencies in $\omega$ which is a fixed smooth function.
We now compute the inner product $\scal{\cdot, \cdot}$ of $h_i^p$ with elements of $J^*+\gene \mathcal{C}$.
\begin{coro} \label{actioncurp}
For any $\param\in \pset$, $g\in \mathcal{C}$, $\Phi\in C^1(\ctoruspi)$ and $i,k=1,2$,
\begin{equation}
\label{eq:scalhi}
\scal{ h_i^p,\gene g} \egal 0, \quad  \scal{h_i^p,\cur^{\widehat{\Phi}}_k}\;=\1_{\{i=k\}}q^\Phi_p(\param) \eqand \scal{h_i^p,\cur_k}\;=\1_{\{i=k\}}r_p(\param),
\end{equation}
where we shortened 
\[q^\Phi_p(\param)=-\am d_s(\alpha)Cov_{\param}(\omega, \Phi)\gcm(E_p|\conf_0=1)\]
and
\[r_p(\param)=d_s(\alpha)\Egcm(\omega)\Egcm\Big(\conf_0\1_{E_p^c}\big[1-\conf_{e_1}-(2p+1)^2(\am-\conf_{e_1})\big]\Big).\]
Furthermore, shortening $q_p(\param):=q_p^\omega(\param)$,
\begin{equation}
\label{eq:limqprp}
\lim_{p\to\infty}\sup_{\param\in\pset}|q_p(\param)\gcm(E_p|\conf_0=1)+\am d_s(\alpha)V_{\param}(\omega)|=0\eqand \lim_{p\to\infty}\sup_{\param\in\pset}\frac{r^2_p(\param)}{\alpha(1-\alpha)}=0.
\end{equation}
In particular, $q_p(\param)\to -\am d_s(\alpha)V_{\param}(\omega)$ and $r_p(\param)\to 0$ as $p\to\infty$ uniformly in $\param\in \pset$.
\end{coro}
\proofthm{Corollary \ref{actioncurp}}{The three identities in \eqref{eq:scalhi} are consequences of \eqref{scal}. Regarding the first one, 
\[\scal{ h_i^p,\gene g}=-\Egcm(h_i^p\Sigma_g)=-d_s(\alpha)\Egcm\pa{\Sigma_g\cro{\com_{e_i}\1_{\tau_{e_i}E_p}-\com_0\1_{E_p}-\Egcm(\omega)\conf_{e_i} +\Egcm(\omega)\conf_0)}}=0\]
by translation invariance of $\gcm$.

\medskip

For the second, we write
\begin{align*}
\scal{ h_i^p,\cur_k^{\widehat{\Phi}}}&=-\sum_{x\in \Z^2}x_k\Egcm(h_i^p \conf_x^{\widehat{\Phi}})\\
&=-d_s(\am)\sum_{x\in \Z^2}x_k\Egcm\big((\conf_{e_i}^{\widehat{\omega}}-\conf_0^{\widehat{\omega}}) \conf_x^{\widehat{\Phi}}\big)+d_s(\am)\sum_{x\in \Z^2}x_k\Egcm\big((\conf_{e_i}^{\omega}\1_{\tau_{e_i}E^c_p}-\conf_0^{\omega}\1_{E^c_p}) \conf_x^{\widehat{\Phi}}\big). 
\end{align*}
Since by construction $\widehat{\Phi}$ has mean $0$ w.r.t. the product measure $\mesinv$, for any function $\psi$ which does not depend on $\theta_x$, $\Egcm(\psi\conf_x^{\widehat{\Phi}})=0$. In particular, in both sums, any term $x\neq 0,e_i$ vanishes. The terms for $x=0$ also vanishes because of the factor $x_k$, and so does the term for $x=e_i$ if $i\neq k$. This yields
\begin{multline*}
\scal{ h_i^p,\cur_k^{\widehat{\Phi}}}=-\1_{\{i=k\}}d_s(\am)\left\{\Egcm\big(\conf_{e_i}^{\widehat{\omega}}\conf_{e_i}^{\widehat{\Phi}}\big)-\Egcm\big(\conf_{e_i}^{\omega}\conf_{e_i}^{\widehat{\Phi}}\1_{\tau_{e_i}E^c_p}\big)\right\}\\
=-\1_{\{i=k\}}\am d_s(\am) Cov_{\param}(\omega, \Phi)\gcm(E_p|\conf_0=1)
\end{multline*}
as wanted. 

\medskip

We now turn to the third identity, for which we can write, applying the same steps as before
\begin{align*}
\scal{ h_i^p,\cur_k}=-d_s(\am)\sum_{x\in \Z^2}x_k\Egcm\big((\conf^{\widehat{\omega}}_{e_i}-\conf^{\widehat{\omega}}_0) \conf_x\big)+d_s(\am)\sum_{x\in \Z^2}x_k\Egcm\Big((\com_{e_i}\1_{\tau_{e_i}E^c_p}-\com_0\1_{E^c_p}) \conf_x\Big). 
\end{align*}
By definition of $\widehat{\omega}$, each term in the first sum vanishes. Regarding the second term, recall that $B_p(x)=x+B_p$, for any $x\in (B_p\cup B_p(e_i))^c$ and any $x\in B_p\cap B_p(e_i)\setminus\{0,e_i\}$, the corresponding contribution vanishes, because $\com_{e_i}\conf_x\1_{\tau_{e_i}E^c_p}$ and $ \com_0\conf_x\1_{E^c_p}$ have the same distribution. The term for $x=0$ vanishes once again because of the factor $x_k$. We can therefore write
\begin{multline*}
\scal{ h_i^p,\cur_k}=\1_{\{i=k\}}d_s(\am)\Egcm\Big((\com_{e_i}\1_{\tau_{e_i}E^c_p}-\com_0\1_{E^c_p}) \conf_{e_i}\Big)\\
+  d_s(\am)\sum_{\substack{x\in B_p,\; x_i=-p\\\mbox{ \small{or} }x\in B_p(e_i),\;x_i=p+1}}x_k\Egcm\Big((\com_{e_i}\1_{\tau_{e_i}E^c_p}-\com_0\1_{E^c_p}) \conf_x\Big). 
\end{multline*}
If $i\neq k$, the sum in the second line vanishes because the contributions for $x_k=q$ cancel out the contributions for $ x_k=-q$. 
If $i=k$, all the contributions for $x_i=-p$ (i.e. $x\in B_p\setminus B_p(e_i)$) are identical and equal to 
$-pd_s(\am)\Egcm(\omega)\Egcm\Big(\am \conf_{e_i}\1_{\tau_{e_i}E^c_p}-\conf_x\conf_0\1_{E^c_p}\Big)=-pd_s(\am)\Egcm(\omega)\Egcm\Big((\am-\conf_{e_1})\conf_0\1_{E^c_p}\Big)$
and the contributions for $x_i=p+1$ (i.e. $x\in B_p(e_i)\setminus B_p$) are each equal to $-(p+1)d_s(\am)\Egcm(\omega)\Egcm\Big((\am-\conf_{e_1})\conf_0\1_{E^c_p}\Big)$.
Since each of those contributions appear $2p+1$ times, we finally obtain as wanted that
\begin{equation*}
\scal{ h_i^p,\cur_k}=\1_{\{i=k\}}d_s(\am)\Egcm(\omega)\cro{\Egcm\Big((1-\conf_{e_1})\conf_0\1_{E^c_p}\Big)-(2p+1)^2\Egcm\Big((\am-\conf_{e_1})\conf_0\1_{E^c_p}\Big)}. 
\end{equation*}

\medskip

According to Proposition \ref{prop:regularitySDC}, $c(1-\rho)\leq d_s(\rho)\leq C(1-\rho)$ for some positive constants $c, $ $C$. Using this fact, the uniform estimates \eqref{eq:limqprp} follow from {elementary computations~: for high densities,  the factor} $\gcm(E^c_p|\conf_0)$ fail to converge uniformly in $\param$, but then $d_s(\am)$ provides the needed control. Regarding $r_p$ the principle is the same, and the extra factor $(2p+1)^2$ is balanced out as $\alpha\to 1$ by the factor $\am-\conf_1$.  
{We start with the first estimate. To prove that $q_p(\param)\gcm(E_p|\conf_0=1)+\am d_s(\alpha)V_{\param}(\omega)$ vanishes uniformly in $\param$, by definition of $q_p$ and since $\omega$ is bounded, it is enough to prove that $|1-\gcm(E_p|\conf_0=1)^2|\am d_s(\alpha)$ also does.
The probability $\gcm(E_p|\conf_0=1)$ is explicit, and given by  
\[\gcm(E_p|\conf_0=1)=1-\am^{P}-P(1-\am)\am^{P-1}\]
where we shortened $ P=(2p+1)^2-1=|B_p\setminus\{0\}|$. In particular, since $d_s(\am)\leq C(1-\am)$, 
\[|1-\gcm(E_p|\conf_0=1)^2|\am d_s(\alpha)\leq C\am(1-\am)\cro{2\am^{P}+2P(1-\am)\am^{P-1}-[\am^{P}+P(1-\am)\am^{P-1}]^2}.\]
Thanks to the prefactor $1-\am$, Each of the terms above is bounded by $P^a(1-\am)^{a+1}\am^{C_1P}$ for some different constants $a\in \{0,1,2\}$ and $C_1>0$ independent of $P$. The previous expression  is maximal in $\am_P=C_1P/(a+1+C_1P)$, and is therefore, uniformly in $\param\in \pset$, less than
\[P^a\pa{\frac{a+1}{a+1+C_1P}}^{a+1},\]
which vanishes as wanted as $P\to\infty$.

\medskip

We now turn to the second estimate. Once again, since  $d_s(\am)\leq C(1-\am)$, we obtain immediately
\[\frac{r_p(\param)^2}{\am(1-\am)}\leq C' \frac{1-\alpha}{\am}\Egcm\Big(\conf_0\1_{E_p^c}\big[1-\conf_{e_1}-(2p+1)^2(\am-\conf_{e_1})\big]\Big)^2.\]
The expectation above can be split in two terms, resp. $\pa{1-(2p+1)^2}\Egcm\Big(\conf_0(1-\conf_{e_1})\1_{E_p^c}\Big)$ and $(1-\am)(2p+1)^2\Egcm\Big(\conf_0\1_{E_p^c}\Big)$. We still shorten $ P=(2p+1)^2-1=|B_p\setminus\{0\}|$, to obtain the bound
\[\abs{\Egcm\Big(\conf_0\1_{E_p^c}\big[1-\conf_{e_1}-(2p+1)^2(\am-\conf_{e_1})\big]\Big)}\leq P(1-\am)\am^P+\am(1-\am)(P+1)\gcm(E_p^c|\conf_0=1).\]
the last probability $\gcm(E_p^c|\conf_0=1)$ has already been computed for the previous estimate, and one obtains straightforwardly that $r_p(\param)^2/\am(1-\am)$ is also bounded from above by a (finite) sum of terms of the form $C_1 P^a(1-\am)^{a+1}\am^{C_2p}$ for $a\in \{2,3,4\}$ and $C_1$, $C_2$ positive constants. As before, each of those vanishes uniformly in $\param\in \pset$, which concludes the proof.
}
}

We are ready to investigate the structure of $\tzero$ with respect to the semi-norm $\scal{\cdot}$ on $\tzero+J^*+\gene\mathcal{C}$. 
Denote by $\mathcal{N}_{\param}=Ker\scal{\cdot}$ and define $\halpha$ the completion of ($\tzero+J^*+\gene\mathcal{C})/\mathcal{N}_{\param}$ 
with respect to $\scal{\cdot}^{1/2}$.  {We need to define $\scal{\cdot}$ on a rather general space, including in particular $J^*+\gene\mathcal{C}$, 
in order to be able later on to estimate the drift contribution to the hydrodynamic limit. However for now, we focus on the symmetic current, and further define  $H^\omega$ the closure in $\halpha$ of $(\tzero+J^{\omega}+\gene\tcal)/\mathcal{N}_{\param}$}.
\index{$ \halpha$\dotfill quotient of $\tzero$ by $\mathcal{N}_{\param}$}
{\begin{prop}[Structure of $H^\omega$]\label{prop:structureHalpha}
For any $\param\in \pset$,  $(\halpha, \scal{\cdot}^{1/2})$ is a Hilbert space, and 
\[H^\omega=\frac{\overline{\gene \tcal}}{\mathcal{N}_{\param}} \oplus J^{\omega},\]
where $\overline{\gene \tcal}/\mathcal{N}_{\param} $ is the closure of $\gene \tcal/\mathcal{N}_{\param}$ w.r.t. $\scal{\cdot}$ in $\halphaz$.
\end{prop}
\proofthm{Proposition \ref{prop:structureHalpha}}{First note that if $\am=0$ or $1$, $\scal{\cdot}\equiv0$ and therefore $\halpha=\{0\}$ is trivial. We now assume that $\param$ is such that $\am\in]0,1[$. 
Since we took the quotient by $\mathcal{N}_{\param}$, the fact that $(\halpha, \scal{\cdot}^{1/2})$ is a Hilbert space is immediate. 
By construction $H^{\omega}$  is a closed linear subspace of $\halpha$,  and the inclusion 
\[\frac{\overline{\gene \tcal}}{\mathcal{N}_{\param}}+ J^{\omega}\subset H^\omega\] 
is immediate, because $J^{\omega}=J^{\omega}/\mathcal{N}_{\param}$. Since both sets are closed subspaces of $\halpha$, we have 
\[H^\omega=\pa{\frac{\overline{\gene \tcal}}{\mathcal{N}_{\param}}+ J^{\omega}}\oplus\pa{\frac{\overline{\gene \tcal}}{\mathcal{N}_{\param}} + J^{\omega}}^{\perp, H^\omega},\]
where the second set on the right-hand side denotes the orthogonal complement of $\frac{\overline{\gene \tcal}}{\mathcal{N}_{\param}} + J^{\omega} $ in $H^\omega$. 
To prove the converse inclusion, it is therefore sufficient to prove that this orthogonal complement is reduced to $\{0\}$. 
This is rather straightforward, although a bit technical because of the different definitions for $\scal{\cdot}$. 
For that purpose, and to give a proof as clear as possible, let us shorten $M=\gene\tcal/\mathcal{N}_{\param}+J^{\omega}$, 
and denote by $m=j^{a,b}+ \gene h$ its elements. Since $\overline{M}^{\perp, H^\omega}\subset H^\omega$, and since $ H^\omega$ is by definition the closure of $(\tzero +M)/\mathcal{N}_{\param}$  any of its element can be written either as $g+m$, 
where $g\in \tzero$ and $m\in M$, or as the limit of elements of this type.
In order to avoid taking convergent sequences, fix 
\[g_0+m_0\in \overline{M}^{\perp, H^\omega},\] 
where $g_0\in \tzero$ and $m_0\in M$, we want to prove that $g_0+m_0=0$. By construction, for any $m\in M$
\[\scal{g_0+m_0,m}=0.\]
and since $g_0\in \tzero$, we can rewrite by the definition  of $\scal{\cdot}$ on $\tzero$ (cf. \eqref{varformulascal})
\begin{align}
\scal{g_0}&=\sup_{m\in M}\{2\scal{g_0,m}-\scal{m}\},
\end{align}
therefore there exists a sequence $(m^k)_{k\to\infty}$ of elements of $M$ such that $\scal{g_0+m^k}\to0$ as $k\to\infty$. We can thus write 
\[\scal{g_0+m_0}=\scal{g_0+m^k, g_0+m_0}+\scal{m_0-m^k, g_0+m_0}.\]
The second term vanishes because $m_0-m^k\in M$, whereas the first term in the right-hand side vanishes as $k\to\infty$, therefore $\scal{g_0+m_0}=0$ as wanted. The same proof holds if $g_0+m_0$ is replaced by a convergent sequence of elements of $\tzero +M$, which proves the reverse inclusion. 

Only remains to prove that the sum $\frac{\overline{\gene \tcal}}{\mathcal{N}_{\param}} + J^{\omega}$ is direct. Assume that for some coefficients $a_i$, $b_i$, and for some cylinder function $g\in \tzero$ 
\[\scal{\sum_{i=1,2}a_i\curohat_i+b_i\cur_i-\gene g}=0.\]
(We should really write this identity for a sequence $g_n$ instead of $g$, with the identity above holding only as $n\to\infty$, but this is purely cosmetic and the proof below holds in this case as well).
Thanks to equation \eqref{eq:scalhi}, we can take the inner product of the identity above w.r.t. $h_i^p$ and since we assumed that $0<\am<1$ let $p\to\infty$ to obtain that for $i=1,2$, $a_id_s(\alpha)V_{\param}(\omega)\am(1-\am)=0$,
therefore $ a_1V_{\param}(\omega)=a_2V_{\param}(\omega)=0$. In both cases, we therefore have $\scal{a_1\curohat_1}=\scal{a_2\curohat_2}=0$. This yields 
\[\scal{b_1\cur_1+b_2\cur_2-\gene g}=0,\]
so that we can now take the inner product with $\ddi\conf_0=-j_i$ (which is orthogonal to $\gene g$), to obtain that $b_1\am(1-\am)=b_2\am(1-\am)=0$, therefore $b_1=b_2=0$ as wanted. This proves that the sum is direct, and concludes the proof of Proposition \ref{prop:structureHalpha}.
}}

The next Proposition states that in $\halpha$, $\curom_i$ can be written as a combination of $h_i^p$ and $j_i$, up to a function which takes the form $\gene g$, and that the coefficients converge as $p\to \infty$ to those given in \eqref{vipdef2}.
\begin{prop}[Decomposition of the currents]
\label{prop:currentsdecomposition}
For any positive integer $p$, define 
\[c_p(\am)=\begin{cases}\gcm(E_p|\conf_0=1)^{-1}&\mbox{ if }\am<1\\
1&\mbox{else} \end{cases},
\eqand d_p(\param)=
\begin{cases}-r_p(\param)c_p(\am)/\am(1-\am)& \mbox{ if }0<\am<1\\
0&\mbox{else}                    
\end{cases},\] 
where $r_p$ was defined in Corollary \ref{actioncurp}. Then, for any $i\in 1,2$ and $\param \in \pset$. 
\begin{equation}
\label{scal3p}
\inf_{g\in \tcal}\scal{\curohat_i+c_p(\am)h_i^p+d_p(\param)\cur_i+\gene g}=0.
\end{equation}
Furthermore, any sequence $(g_m)_m$ ultimately realizing \eqref{scal3p} can be chosen independently of $p$, and also ultimately realizes 
\begin{equation}\label{arginf}\inf_{g\in \tcal}\scal{\curohat_i+\gene g}.\end{equation}
\end{prop}
\proofthm{Proposition \ref{prop:currentsdecomposition}}{
We start by clearing out the trivial cases when $\am=0$ and $\am=1$. In those, all quantities vanish and \eqref{scal3p} is trivially true for any coefficients. 
Another trivial case is when $V_{\param}(\omega)=0.$ In this case, $\curohat_i=0$ in $\halpha$, therefore, the $h_i^p$ and $\cur_i$ being orthogonal (as local gradients) to $\gene \tcal$, and $h_i^p$ being orthogonal to $\cur_k$ for $k\neq i$, as a consequence of Proposition \ref{prop:structureHalpha} we can then write 
$\scal{h_i^p+a_p\cur_i}=0$ for some constant $a_p$. 
This constant can be determined using Lemma \ref{actioncurp} and taking the inner product of the previous quantity with $\cur_i$, which yields  
$a_p=-r_p(\param)/\scal{\cur_i}=-r_p(\param)/\alpha(1-\am)$. In this case, $\scal{c_p(\am)h_i^p+d_p(\param)\cur_i}=0$ for any $p$, as wanted.

We now fix $\param\in \pset$ satisfying $\am\in ]0,1[$ and $V_{\param}(\omega)>0$.
Fix $p\in\N$, and define $c_p$, $d_p$ as in Proposition \ref{prop:currentsdecomposition}, we now prove that \eqref{scal3p} holds. According to Proposition \ref{prop:structureHalpha},  there exists coefficients $a_{i,k}^p$ and $b_{i,k}^p$ such that, 
\begin{equation}\label{projgrad}\inf_{ g \in \tcal}\scal{h_i^p+\sum_{k=1,2}a^p_{i,k}\curohat_k+b^p_{i,k}\cur_k+\gene g}=0.\end{equation}
In order not to burden the proof, we will assume that the infimum in $g$ is reached, i.e. that there exists a function $g_i^p\in \tcal$ such that
\begin{equation}\label{projass}\scal{h_i^p+\Big[\sum_{k=1,2}a^p_{i,k}\curohat_k+b^p_{i,k}\cur_k\Big]+\gene g_i^p}=0.\end{equation}
This assumption is purely for convenience, and we can substitute at any point to $g_i^p$ a sequence of functions $(g^p_{i,m})_{m\in \N}$ such that the previous identity holds in the limit $m\to \infty$. 

\bigskip

Using \eqref{scal}, one obtains immediately that $\scal{\curohat_i, \curohat_k}=\1_{\{i=k\}}V_{\param}(\omega)\am(1-\am)$, $\scal{\curohat_i, \cur_k}=0$ and $\scal{\cur_i,\cur_k}=\1_{\{i=k\}}\am(1-\am)$.
Using these formulas and Corollary \ref{actioncurp}, we take the inner product of the function in \eqref{projass} with $\curohat_l$, $\cur_l$, $\gene g_l^p$, and $h_l^p$, to obtain the four identities
\[\1_{\{i=l\}}q_p(\param)  +a_{i,l}^p V_{\param}(\omega)\alpha(1-\alpha)+\scal{\gene g_i^p,\curohat_l}=0\, ,\quad  \1_{\{i=l\}} r_p(\param)+ b_{i,l}^p \am(1-\am)=0,\]
\begin{equation}
\label{eq:2ndmatrix}
\sum_{k=1,2}a^p_{i,k}\scal{\curohat_k,\gene g_l^p}+\scal{\gene g_i^p,\gene g_l^p}=0\eqand \scal{h_i^p, h_l^p}+a^p_{i,l}q_p(\param) +b_{i,l}^pr_p(\param)=0.
\end{equation}
Note that since we assumed $\alpha\in]0,1[$, $V_{\param}(\omega)>0$ and $p>0$, we have $q_p(\param)<0$. Define  $A_p$, $B_p$, $H_p$, $G_p$ and $J_p$ the matrices whose respective elements are given for $ i,k=1,2$ by $a_{i,k}^p$, $b_{i,k}^p$, $\scal{h_i^p, h_k^p}$, $\scal{\gene g_i^p,\gene g_k^p}$ and $\scal{\gene g_i^p,\curohat_k}$. Note in particular that $H_p$ and $G_p$ are symmetric with non-negative eigenvalues. Further denote by $I$ the two-dimensional identity matrix. The four identities above then rewrite in matrix form as 
\[J_p=-q_p(\param)I  - V_{\param}(\omega)\alpha(1-\alpha)A_p,\quad B_p=-\frac{r_p(\param)}{\alpha(1-\alpha)}I\]
\[-A_pJ_p^{\dagger}=G_p\eqand -q_p(\param)A_p-r_p(\param)B_p=H_p,\]
where $J_p^{\dagger}$ is the transposed matrix of $J_p$.
The second and last identities show that $B_p$ and $A_p$ are symmetric, therefore so is $ J^p$, and that
\[A_p=-\frac{1}{q_p(\param)}\cro{H_p-\frac{r_p(\param)^2}{\alpha(1-\alpha)}I}.\] 
In particular, since $H_p$ is positive in the matrix sense, it is diagonalizable, and thus so is $A_p$. Finally, the first and third identities then yields
\[A_p[q_p(\param)I  + V_{\param}(\omega)\alpha(1-\alpha)A_p]=G_p.\]
therefore, since $G_p$ is positive in the matrix sense, any eigenvalue $\lambda$ of $A_p$ must satisfy
\[\lambda[q_p(\param) + V_{\param}(\omega)\alpha(1-\alpha)\lambda]\geq 0,\]
and therefore $\lambda> -q_p(\param)/V_{\param}(\omega)\alpha(1-\alpha)>0$. Let $C_p$ denote the inverse of $A_p$, which is a positive matrix with eigenvalues bounded from above by $-V_{\param}(\omega)\alpha(1-\alpha)/q_p(\param)$. Since $A_p$ is invertible, we can therefore rewrite \eqref{projass} as 
\begin{equation}\label{projass2}\scal{\curohat_i+\Big[\sum_{k=1,2}c^p_{i,k} h_k^p+ d^p_{i,k}\cur_k\Big]+\gene \widetilde{g}_i^k}=0.\end{equation}
which holds for $i=1,2$, where $\widetilde{g}_i^k=\sum_{k=1,2} c^p_{i,k} g_k^p$, and the $c_{i,k}^p$  (resp. $d^p_{i,k}$) are the matrix elements of $C_p$ (resp. $D_p:=C_pB_p$). 
For $x,y\in \R^2$, shorten $x\cdot y=x_1y_1+x_2y_2$ their usual inner product. Let $\curohat=(\curohat_1, \curohat_2)$, and define the quadratic form $Q$ as 
\begin{equation*}
x^\dagger Q x=\inf_{g\in \tcal}\scal{x\cdot\curohat+\gene g }.
\end{equation*}
Then,  \eqref{projass2} yields for any $x\in \R^2$
\begin{equation}
\label{projass3}
\inf_{g\in \tcal}\scal{x\cdot\curohat +\Big[\sum_{i,k=1,2}x_i c^p_{i,k} h_k^p+ x_i d^p_{i,k}\cur_k\Big]+\gene g}=0.
\end{equation}
Taking the inner product of the expression above with $x\cdot\curohat +\gene g$, and since the terms in the sum are orthogonal to any $\gene g$, we obtain
\begin{align*}
 x^\dagger Q x=\inf_{g\in \tcal}\scal{x\cdot\curohat +\gene g}=&-\scal{x\cdot\curohat, \sum_{i,k=1,2}x_i c^p_{i,k} h_k^p+ x_i d^p_{i,k}\cur_k}\\
=&-\sum_{i,k=1,2}x_ix_k c^p_{i,k} \scal{h_k^p,\curohat_k}+ x_ix_k d^p_{i,k}\scal{\cur_k, \curohat_k}\\
=&-q_p(\param) x^\dagger C_p x,
\end{align*}
thanks to Corollary \ref{actioncurp} and because $\cur_k$ and  $\curohat_k$ are orthogonal. We prove in Appendix \ref{subsec:sdc}, equation \eqref{eq:Qid}, that $Q=\alpha V(\param) d_s(\alpha)I$, therefore
\[C_p=-\frac{\alpha V(\param) d_s(\alpha)}{q_p}I=\gcm(E_p\mid \conf_0=1)^{-1}I=c_p(\am)I,\]
and $D_p=[-c_p(\am)r_p(\param)/\am(1-\am)]I=d_p(\param)I$, where $c_p$, $d_p$ were defined in Proposition \ref{prop:currentsdecomposition}.
We can now rewrite \eqref{projass3} as wanted as
\begin{equation}
\label{projass4}
\inf_{g\in \tcal}\scal{\curohat_i +c_p(\alpha) h_i^p+d_p(\param)\cur_i+\gene g}=0.
\end{equation}
Since $h_i^p$ and $\cur_i$ are both orthogonal to any $\gene g$, taking the inner product of the identity above with $\curohat_i+ \gene g $, one obtains that any sequence of functions realizing the infimum above also realizes $\inf_{g\in \tcal}\scal{\curohat_i +\gene g}$, which proves the last statement and concludes the proof of  Proposition \eqref{prop:currentsdecomposition}.
} 

\begin{rema}[Bound on $\scal{h_i^p}$]
We already obtained in \eqref{eq:2ndmatrix}  $\scal{h_i^p, h_l^p}+a^p_{i,l}q_p(\param) +b_{i,l}^pr_p(\param)=0$. Since we now have an explicit expression for the matrix $A_p=C_p^{-1}=c_p^{-1}(\alpha)I$, and $B_p=-r_p(\param)/\am(1-\am)I$, we obtain
$\scal{h_i^p}=-q_p(\param)c_p^{-1}(\am)+\frac{r_p(\param)^2}{\am(1-\am)}.$
Equation \eqref{eq:limqprp} then yields the uniform bound
\begin{equation}
\label{eq:normhi}
\lim_{p\to\infty}\sup_{\param\in \pset}|\scal{h_i^p}-\alpha d_s(\alpha)V_{\param}(\omega)|=0.
\end{equation}
\end{rema}

We now prove equation \eqref{decompunif},  and thus concludes the proof of Theorem \ref{thm:NGestimates}. Up until now, we have only used $\scal{\cdot}$ for functions in  ${T^\omega}$, but in \eqref{decompunif} the function $f$ is a priori no longer in $\tcal$ bur rather in $\mathcal{C}$, we therefore need the extension of $\scal{\cdot}$ to $\gene \mathcal{C}$ introduced in Definitions \ref{defi:limitingcovariance1} and \ref{defi:limitingcovariance2}. Thanks to \eqref{vipdef2}, the result can be stated as follows.
\begin{prop}
[Uniform bound on $\scal{\mathcal{V}_{i,p}^f}$]
\label{prop:NGEuniforme}
Identity \eqref{decompunif} holds, in the sense that there exists a sequence of local functions $f_n\in \mathcal{C}$ such that
\begin{equation}
\label{eq:unifscal}
\limsup_{n\to\infty}\limsup_{p\to\infty}\;\sup_{\param\in \pset}\;\scalstar{\curohat_i+h_i^p+\gene f_n}\;=0.
\end{equation}
Furthermore, for any $\param\in \pset$, $\lim_{n\to\infty}\scalstar{\curohat_i+\gene f_n}=\inf_{g\in \tcal}\scal{\curohat_i+\gene g}$
\end{prop}
\proofthm{Proposition \ref{prop:NGEuniforme}}{ In order not to burden with technical estimates, we start by cutting off the extreme densities for which the convergences as $p\to\infty$ can be problematic. For any $\param$, we can write by triangular inequality and using \eqref{eq:normhi},
\begin{align*}
\scalstar{\curohat_i+h_i^p+\gene f}\leq& \scal{\curohat_i}+\scal{h_i^p}+\scalstar{\gene f}\\
\leq &V_{\param}(\omega)\alpha(1-\alpha)+\am d_s(\am)V_{\param}(\omega)(1+o_p(1))+\sum_{i=1,2}\Egcm(\conf_0(1-\conf_{e_i})[\Sigma_f(\confhat^{0,e_i})-\Sigma_f]^2) 
,\end{align*}
where the $o_p(1)$ does not depend on $\param$. As stated in Proposition \ref{prop:regularitySDC}, $d_s(\am)\leq C(1-\am)$, $\omega$ is bounded, and $f$ is a cylinder function and therefore $\Sigma_f(\confhat^{0,e_i})-\Sigma_f$ is bounded as well. Fix $\epsilon>0$, in particular, the estimate above yields, for some constant $C_{\omega,f }$, and for any $\param$ such that $\am\notin[\epsilon,1-\epsilon]$
\begin{equation*}
\scalstar{\curohat_i+h_i^p+\gene f}\; \leq \;C_{\omega,f}  (1+o_p(1))\epsilon.\end{equation*}

We now fix $\param$ such that $\epsilon\leq \am\leq 1-\epsilon$, by triangular inequality, 
\[\scalstar{\curohat_i+h_i^p+\gene f}\;\leq \; \scalstar{\curohat_i+c_p(\am)h_i^p+d_p(\param)\cur_i+\gene f}+\scal{(c_p(\am)-1)h_i^p+d_p(\param)\cur_i}.\]
Since $\param$ is bounded away from the extreme densities, the second term in the right-hand side is $C_\epsilon o_p(1)$, and we can therefore write
\[\sup_{\param}\scalstar{\curohat_i+h_i^p+\gene f}\;\leq \;\sup_{\param}\scalstar{\curohat_i+c_p(\am)h_i^p+d_p(\param)\cur_i+\gene f}+ C_{\omega, f}\epsilon +C_{\epsilon, \omega, f} o_p(1).\]
We then let $p\to\infty$ and then $\epsilon\to0$ to obtain that 
\[\limsup_{p\to\infty}\sup_{\param}\scalstar{\curohat_i+h_i^p+\gene f}\;\leq \;\limsup_{p\to\infty}\sup_{\param}\scalstar{\curohat_i+c_p(\am)h_i^p+d_p(\param)\cur_i+\gene f}.\]
Proposition \eqref{prop:NGEuniforme} is therefore a consequence of Lemma \eqref{lem:1stborneunif} below.}
\begin{lemm}
\label{lem:1stborneunif}
There exists a sequence of local functions $f_n\in \mathcal{C}$ such that 
\begin{equation*}\limsup_{p\to\infty}\sup_{\param} \scalstar{\curohat_i+c_p(\am)h_i^p+d_p(\param)\cur_i+\gene f_n}\;\leq \; \frac{3}{n},\end{equation*}
and for any $\param\in  \pset$,  $\lim_{n\to\infty}\scalstar{\curohat_i+\gene f_n}=\inf_{g\in \tcal}\scal{\curohat_i+\gene g}$
\end{lemm}

\proofthm{Lemma \ref{lem:1stborneunif}}{The proof of this Lemma is analogous to that of Theorem 5.6, p.176 of \cite{KLB1999}. We now write explicitly the dependency of $h_i^p$ in $\param$. According to Theorem \ref{thm:limcovariance} the application $\param \mapsto\scalstar{\psi}$ is continuous on $\pset$, and thanks to equation \eqref{decompunif}, for any $\param_0\in \pset$, there exists a function $g_{\param_0}\in \tcal$ and  a neighborhood ${\mathcal N}_{\param_0}$ of $\param_0$ such that  for any $\param\in {\mathcal N}_{\param_0}$, 
\[\scal{\curohat_i+c_p(\am_0)h_i^p(\param_0)+d_p(\param_0)\cur_i+\gene g_{\param_0}}\;\leq\; n^{-1}.\]
Furthermore, thanks to the last statement in Proposition  \ref{prop:currentsdecomposition}, this function is an approximation of the one realizing  $\inf_{g\in \tcal}\ll \curohat_i+\gene g\gg_{\widehat{\alpha_0}}$, and can be chosen  independently of $p$.
\medskip

We prove in Proposition \ref{prop:psetcompact} that $\pset$ is compact,  it therefore admits a finite covering $\pset\subset\cup_{j=1}^m {\mathcal N}_{\param_j}$. We can build a $C^2$ interpolation of the $g_{\param_j}$'s, and therefore obtain a function $(\param,\conf)\mapsto \psi(\param, \conf)$ which coincides in $\param=\param_j$ with $g_{\param_j}$, with the two following properties~:
\begin{itemize}
  \item let $B$ be a finite set of edges in $\Z^2$ containing the support of all the $g_{\param_j}$'s, $\psi(\param,\;.\;)$ is a cylinder function in $\tcal$ with support included in $B$ for any $\param\in \pset$. 
  \item For any fixed configuration $\confhat$, $\psi(\;.\;,\confhat)$ is in $C^2(\pset)$.
\item  for any $\param\in \pset$
\begin{equation}\label{Phiscal}\scalstar{\curohat_i+c_p(\am)h_i^p(\param)+d_p(\param)\cur_i+\gene \psi(\param,\cdot)}\;\leq\; 2n^{-1}.\end{equation}
  \end{itemize}
Recall that we  introduced in \eqref{empiricalprofile} $\dens_r= \abss{B_r}^{-1}\sum_{x\in B_r}\conf_x \delta_{\theta_x}$ the empirical angular density in the box of side $(2r+1)$ around the origin. Define 
\[f_r(\confhat)=\psi(\dens_r,\confhat),\] 
for any $r$ large enough for the support $B$ of the $\psi(\param,\conf)$'s to be contained in $B_r$. Note that $f_r$ is not necessarily in  ${T^\omega}$, but it is a local function for $r$ fixed. 

By triangle inequality, 
\begin{equation}\label{supphia}\sup_{\param} \scalstar{\curohat_i+c_p(\am)h_i^p(\param)+d_p(\param)\cur_i+\gene f_r}\;\leq\; 2n^{-1}+\sup_{\param}\scalstar{\gene(f_r-\psi(\param,\cdot))}.\end{equation}
The second term in the right-hand side is
\[\sum_i\Egcm\pa{\pa{\nabla_{0,e_i}\sum_{x\in \Z^2}\tau_x\cro{f_r-\psi(\param,\cdot)}}^2}=\sum_i\Egcm\pa{\pa{\sum_{x\in \Z^2}\tau_{-x}\nabla_{x,x+e_i}\cro{f_r-\psi(\param,\cdot)}}^2}.\]
Note once again that $\sum_{x\in \Z^2}\tau_xf$ is merely a notation, and is not a well-defined function as such, but instead, is meant to either be integrated against a mean-$0$ local function, or taken a gradient of, as is the case here. We extend B by $1$ in such a way that for any edge $a$ outside of $B$, $\grad\psi(\param, .)$ vanishes. Therefore, the only contributions outside of $B$ in the sums above are at the boundary of $B_r$, where $f_r$ has a variation in its first argument of order $(2r+1)^{-2}$. Thanks to the regularity of $\psi$ in $\param$, and since the number of corresponding edges is roughly $4(2r+1)$, the contribution of all these jumps is of order $r^{-1}$ in the whole sum.

Then, since the number of edges in $B$ depends only on $\psi$, and since $\Egcm\pa{(\grad f)^2} \leq 4\Egcm(f^2)$, we obtain by definition of $f_r$ that
\begin{equation}\label{majorunif}\sup_{\param}\scalstar{\gene(f_r-\psi(\param,\cdot))}\;\leq \sup_{\param}C(\psi)\Egcm\cro{\pa{\psi(\dens_r,\;.)-\psi(\param,\cdot)}^2}+O(r^{-2}),\end{equation}
whose right-hand side vanishes as $r$ goes to infinity  by the law of large numbers. 

\bigskip

Let us fix $r_{n}$ such that the right-hand side of \eqref{majorunif} is less than $1/n$, and let $f_n=f_{r_n}$, \eqref{supphia} finally yields 
\begin{equation}
\label{steppunif}
\sup_{\param} \scalstar{\curohat_i+c_p(\am)h_i^p(\param)+d_p(\param)\cur_i+\gene f_n}\;\leq\; 3n^{-1},
\end{equation}
as wanted. The last statement of the Lemma is a direct consequence of the construction of $f_n$ and of Proposition \ref{prop:currentsdecomposition}. This concludes the proof of Lemma \ref{lem:1stborneunif}.
}

\subsection{Drift part of the hydrodynamic limit}
\label{subsec:drift}
\intro{Recall that $L_N=N^2\gene+N\genwa+\genis$ is the complete generator of our process introduced in \eqref{defgenecomplet}. In the previous section, we proved that the symmetric currents can be replaced by a gradient, up to a perturbation $\gene f$. In our case, this perturbation is not negligible, and must be added to the asymmetric currents induced by the asymmetric generator $\genwa$ to complete the drift term in equation \eqref{equadiff}. This is the purpose of this section.}

To achieve that goal, we   need notations similar to the ones introduced in Section \ref{gradientreplacement}. For any positive integer $l$, and any smooth function $G\in C([0,T]\times \ctorus)$, let us introduce 
\[{\mathcal R}_i^{f,l}(\confhat)=\curaom_i+\genwa f-\E_{\densl}(\curaom_i+\genwa f),\]
and
\[Y_{i,N}^{f,l}(G,\confhat)=\frac{1}{N^2}\sum_{x\in \torus}G(x/N)\tau_x{\mathcal R}_i^{f,l},\]
where $\curaom_i$ is the asymmetric current introduced in \eqref{currentsasym}.
According to Theorem \ref{thm:NGestimates}, for any $i$, there exists a family of cylinder functions  $(f_{i,n}^\omega)_{n\in \N}$ introduced in Proposition \ref{prop:NGEuniforme} such that
\begin{equation*}\lim_{\gamma\to \infty}\lim_{n\to\infty}\limsup_{\varepsilon\to 0}\limsup_{N\to\infty}\frac{1}{\gamma N^2}\log \E^{\lambda,\beta}_{\mesref}\left[\exp \left(\gamma N^2\abs{\int_0^T{X}_{i,N}^{f_{i,n}^\omega,\varepsilon N}(G_t,\confhat(t))dt}\right)\right]=0,\end{equation*}
where ${X}_{i,N}^{f,\varepsilon N}$ was defined in equation \eqref{Xdef}. Furthermore, we also established in Proposition \ref{prop:NGEuniforme} that this sequence satisfies for any $\param\in \pset$
\begin{equation} \label{infr}\lim_{n\to \infty}\scal{\curom_i+\gene f_{i,n}^\omega}=\inf_{f\in \tcal}\scal{\curom_i+\gene f}.\end{equation}
The replacement Lemma \ref{lem:replacementlemma} applied to $g(\confhat)=\curaom_i+\genwa f$ yields the following result.
\begin{lemm}
\label{lem:k4total}
Let $G$ be some smooth function in $C^{1,2}([0,T]\times \ctorus)$, and $T\in \R^*_+$, then for $i\in \{1,2\}$ we have 
\[\lim_{n\to \infty}\limsup_{\varepsilon\to 0}\limsup_{N\to\infty} \E^{\lambda,\beta}_{\mu^N}\left[\abs{\int_0^T Y_{i,N}^{f_{i,n}^\omega,\varepsilon N}(G,\confhat)ds}\right]=0.\]
\end{lemm}
Furthermore, we now prove the following result, which states that any function of the form $N\genex f$ vanishes in the hydrodynamic limit, where $\genex=\gene+N^{-1}\genwa$ is the generator of whole exclusion  process.

\begin{lemm}\label{lem:Lfnegligeable}
For any function $G:[0,T]\times \ctorus\to \R$ in $C^{1,2}$, and any cylinder function $f$,
\begin{equation*}\limsup_{N\to\infty}\E_{\mu^N}\left[\abs{\int_0^T \frac{1}{N} \sum_{x\in \torus}G\left(s,x/N\right)\tau_x\genex f(\confhat(s)) ds}\right]=0.\end{equation*}
\end{lemm}
\proofthm{Lemma \ref{lem:Lfnegligeable}}{For any such smooth function $H$ and cylinder function $f$, let us denote 
\[F_{G}(s,\confhat(s))=N^{-2}\sum_{x\in \torus}G(s,x/N)\tau_xf(\confhat(s)).\] 
The process 
\[M_{G}(t)= F_{G}(t, \confhat(t))-F_{G}(0,\confhat(0))-\int_0^T \partial_sF_{ G}(s,\confhat(s))ds-\int_0^T L_NF_{ G }(s,\confhat(s))ds\]
is a martingale, where $L_N$ is the complete generator of our process, introduced in \eqref{defgenecomplet}. Since $f$ is bounded, the first three terms are of order $1$, it remains to control $\int_0^T L_NF_{ G }ds$. The quadratic variation of this martingale is given (cf. Appendix 1.5, Lemma 5.1 in \cite{KLB1999}) by 
\begin{align*}[M_G(\cdot ,\confhat(\cdot))]_t=&\int_0^T L_NF_{ G }(s,\confhat(s))^2-2F_{ G}(s,\confhat(s))L_NF_{ G}(s,\confhat(s)) ds\\
=&\int_0^T ds N^2\sum_{\substack{x\in \torus\\
\delta=\pm 1, i\in \{1,2\}}}\hspace{-1em}\tau^{\lambda}_{x,z,i, \delta}\cro{F_{ G }(s,\confhat^{x,x+\delta e_i}(s))-F_{ G }(s,\confhat(s))}^2\\
&+\int_0^T ds\sum_{x\in \torus}\conf_x\int_{\ctoruspi}c_{x,\beta}(\theta, \confhat)\cro{F_{ G }(s,\confhat^{x, \theta}(s))-F_{ G}(s,\confhat(s))}^2d\theta\\
=&\frac{1}{N^2}\int_0^T ds \sum_{\substack{x\in \torus\\
\delta=\pm 1, i\in \{1,2\}}}\tau^{\lambda}_{x,z,i, \delta}(\confhat(s))\cro{\sum_{y\in \torus}G(s,y/N)\pa{\tau_yf(\confhat^{x,x+z}(s))-\tau_yf(\confhat(s))}}^2\\
&+\frac{1}{N^4}\int_0^T ds\sum_{x\in \torus}\conf_x\int_{\ctoruspi}c_{x,\beta}(\theta, \confhat)\cro{\sum_{y\in \torus}G(s,y/N)\pa{\tau_yf(\confhat^{x,x+z}(s))-\tau_yf(\confhat(s))}}^2d\theta,
\end{align*}
where 
\[\tau^{\lambda}_{x,z,i, \delta}(\confhat)=\pa{1+\frac{\delta\lambda_i(\theta_x)}{N}}\conf_x(1-\conf_{x+z})\]
is the total displacement jump rate.

Since $f$ is a local function, all but a finite number of terms in the $y$ sums vanish, and the quadratic variation is hence of order $N^{-2}$.
We deduce from the estimate of the quadratic variation of $M_G$ and the order of the three first terms in the expression of $M_G$ that
\begin{align*}\E_{\mu^N}\pa{\abs{\int_0^T N^{-1}L_N F_{ G }(s,\confhat(s))ds}}\leq N^{-1}\cro{\underset{O(N^{-1})}{\underbrace{\E_{\mu^N}\pa{[M_G(t,\confhat(t))]}^{1/2}}}+O_N(1)}\underset{N\to \infty}{\to} 0.\end{align*}
The previous martingale estimate  shows that $\E_{\mu^N}\pa{\abs{\int_0^T N^{-1}L_N F_{ G }(s,\confhat(s))ds}}$ vanishes in the limit $N\to \infty$. 
Furthermore, elementary computations yield a crude bound on the contribution of the Glauber generator of order $N^{-1}$.  Finally,  since $L_N=N^2\genex+\genis$, we obtain
\begin{align*}\E_{\mu^N}\pa{\abs{\int_0^T N\genex F_{ G }(s,\confhat(s))ds}}\underset{N\to \infty}{\to} 0,\end{align*}
which completes the proof of Lemma \ref{lem:Lfnegligeable}.}

We now use these two Lemmas to prove that the total displacement current can be replaced by the wanted averages. More precisely, let 
\[{\mathcal U}_i^{f,l}(\confhat)=\curom_i+\frac{1}{N}\curaom_i+\sdc\left(\rho_l\right)\ddi\densom_l+\diff\left(\rho_l, \densom_l\right)\ddi\rho_l-\frac{1}{N}\E_{\densl}(\curaom_i+\genwa f),\]
we can state the following result.
\begin{coro}
\label{currentstot}
For any $G\in C^{1,2}([0,T]\times \ctorus)$, $T\in \R^*_+$, and $i\in \{1,2\}$,
\[\lim_{n\to\infty}\limsup_{\varepsilon\to 0}\limsup_{N\to\infty} \E^{\lambda,\beta}_{\mu^N}\left[\abs{\int_0^T \frac{1}{N}\sum_{x\in \torus}G(x/N){\mathcal U}_i^{f_{i,n}^\omega,\varepsilon N}(G,\confhat)ds}\right]=0.\]
\end{coro}
\proofthm{Corollary \ref{currentstot}}{Adding and subtracting $\genex f_{i,n}^\omega$ to ${\mathcal U}_i^{f_{i,n}^\omega,\varepsilon N}$, we can split it into three parts, 
\[\curom_i+d_s\left(\rho_{\varepsilon N}\right)\ddi\densom_{\varepsilon N}+\diff\left(\rho_{\varepsilon N}, \densom_{\varepsilon N}\right)\ddi\rho_{\varepsilon N}+\gene f_{i,n}^\omega ,\]
\[\frac{1}{N}(\curaom_i+\genwa f_{i,n}^\omega)-\frac{1}{N}\E_{\densep}(\curaom_i+\genwa f_{i,n}^\omega),\quad \mbox{ and }\quad -\genex f_{i,n}^\omega.\]

The contribution of the first quantity vanishes in the limit of Corollary \ref{currentstot}, according to Corollary \ref{currents3}. The second contribution also does thanks to Lemma \ref{lem:k4total}, as well as the third due to Lemma \ref{lem:Lfnegligeable}, thus completing the proof of the Corollary.}

We now derive an explicit expression for the limit of $\E_{\densep}(\curaom_i+\genwa f_{i,n}^\omega)$, appearing in ${\mathcal U}_i^{f_n,l}$, as $n$ goes to $\infty$.
\begin{lemm}\label{lem:driftexpectation}For any $\param\in \pset$, 
\begin{equation}\label{idcoeff}\lim_{n\to\infty}\E_{\param}\pa{\curaom_i+\genwa f^\omega_{i,n}}=2 \sdc(\am)\am_{\omega  \lambda_i}+2\frac{\aom\am_{\lambda_i}}{\am}(1-\am-\sdc(\am)),\end{equation}
where for any function $\Phi\in C^1(\ctoruspi)$, we defined $\am_{\Phi}=\Egcm(\Phi(\theta_0)\conf_0)$.
\end{lemm}

\proofthm{Lemma \ref{lem:driftexpectation}}{
By definition of $\curaom_i=\lambda_i(\theta_0)\omega(\theta_0)\conf_0(1-\conf_{e_1})+\lambda_i(\theta_{e_i})\omega(\theta_{e_i})\conf_{e_i}(1-\conf_{0})$, 
we can write, shortening as before $\Egcm(\Phi)=\Egcm(\Phi(\theta_0)|\conf_0=1)$,
\begin{equation}\label{Ecuraom}\Egcm(\curaom_i)=2\Egcm(\lambda_i\omega)\alpha(1-\alpha)=2
\scalstar{\cur_i^{\lambda_i},\cur_i^{\omega}}.\end{equation}
For any cylinder function $f$, by translation invariance of $\mesinv$ and Definition \ref{defi:limitingcovariance1}, one also obtains by elementary computations that
\begin{equation}\label{Ecuraom2}\Egcm(\genwa f)=2\scalstar{\cur_1^{\lambda_1}+\cur_2^{\lambda_2}, \gene f}.\end{equation}

Recalling Corollary \ref{actioncurp}, we can then write  
\begin{align*}
\scalstar{\cur_k^{\lambda_k},h_i^{p,\omega}}&=\scalstar{\cur_k^{\widehat{\lambda}_k},h_i^{p,\omega}}+\Egcm(\lambda_k)\scalstar{\cur_k,h_i^{p,\omega}}\\
&=-\1_{\{i=k\}}[\am d_s(\am)Cov_{\param}(\omega, \lambda_i)(1-o_p(1))-\Egcm(\lambda_i)o_p(1)]
\end{align*}
where as before $\widehat\lambda_k=\lambda_k-\Egcm(\lambda_k)$.  
We can also write by Definition \ref{defi:limitingcovariance1}
\[\scalstar{\cur_k^{\lambda_k},\cur_i^{\omega}}=\1_{\{i=k\}}\Egcm(\lambda_k\omega)\am(1-\am).\]

Once again, in order to avoid taking everywhere limits $n\to\infty$, we assume for the convenience of notations, that there exists a local function 
$f_i^{\omega}$ realizing the infimum \eqref{infr}. Recall then from equation \eqref{scal3p} that in $\halpha$, we have the identity $\curohat_i+\gene f_i^{\omega} =-c_p(\am)h_i^p-d_p(\param)\cur_i$.
Then, using \eqref{Ecuraom}, \eqref{Ecuraom2}, and the explicit formulas for the inner products which prove orthogonality of directions $ i\neq k$, 
\begin{align}
\Egcm(\curaom_i+\genwa f_i^{\omega})&=2\scalstar{\cur_1^{\lambda_1}+\cur_2^{\lambda_2}, \gene f_i^{\omega}}+2\scalstar{\cur_i^{\lambda_i},\cur_i^{\omega}}\nonumber\\
&=2\scalstar{\cur_1^{\lambda_1}+\cur_2^{\lambda_2}, \curohat_i+\gene f_i^{\omega}}-2\scalstar{\cur_1^{\lambda_1}+\cur_2^{\lambda_2}, \curohat_i}+2\scalstar{\cur_i^{\lambda_i},\cur_i^{\omega}}\nonumber\\
&=-2\scalstar{\cur_1^{\lambda_1}+\cur_2^{\lambda_2}, c_p(\am)h_i^{p,\omega}+d_p(\param)\cur_i}-2\scalstar{\cur_i^{\lambda_i}, \curohat_i}+2\scalstar{\cur_i^{\lambda_i},\cur_i^{\omega}}\nonumber\\
&=-2c_p(\am)\scalstar{\cur_i^{\lambda_i},h_i^{p,\omega}}-2d_p(\param)\scalstar{\cur_i^{\lambda_i},\cur_i}+2\Egcm(\omega)\scalstar{\cur_i^{\lambda_i}, \cur_i}.
\end{align}
We now let $p\to\infty$, so that $d_p$ vanishes, $c_p$ goes to $1$, to obtain as wanted, by Definition \ref{defi:limitingcovariance1} and Corollary \ref{actioncurp}, 
\[\Egcm(\curaom_i+\genwa f_i^{\omega})=2\am d_s(\am)Cov_{\param}(\omega, \lambda_i)+2\Egcm(\omega)\Egcm(\lambda_i)\am(1-\am).\]
Reorganizing the terms yield Lemma \ref{lem:driftexpectation}.
}

\section{Proof of the hydrodynamic limit}
\label{sec:7}

 \intro{We now have all the pieces to prove Theorem \ref{thm:mainthm}. The last remaining difficulty is to perform the second integration by parts, since even the gradients obtained in Section \ref{sec:6} are not exactly microscopic gradients due to the non-constant diffusion coefficient. This is not a problem when the variations only depend on one quantity, the density for example, since we can then simply consider a primitive of the diffusion coefficient and obtain at the highest order in $N$ a discrete gradient. This is not the case here, and we need some more work to obtain the wanted gradient.
 }

Let us recall from Section \ref{subsec:outlineandcurrents} that for any smooth function $H\in C^{1,2,1}([0,T]\times\ctorus\times \ctoruspi)$, that we denoted by $M_t^{H,N}$ the martingale 
\begin{equation}M_t^{H,N}=<\pi_t^{N},H_t>-<\pi_0^{N},H_0>-\int_0^t  \cro{<\pi_s^{N},\partial_sH_s>+L_N<\pi_s^{N},H_s> }ds,\end{equation}
where 
\[\pi_s^{N}=\frac{1}{N^2}\sum_{x\in \torus}\conf_x(t)\delta_{x/N,\theta_x(s)}\]
is the empirical measure of the process on $\ctorus\times\ctoruspi$.

\proofthm{Theorem \ref{thm:mainthm}}{ The quadratic variation $ [M^{H,N}]_t$ of $M_t^{H,N}$ (cf.  A1.5. Lemma 5.1 in \cite{KLB1999})  is 
\begin{align*}[M^{H,N}]_t=&\int_0^t  L_N <\pi_s^{N},H_s>^2-2<\pi_s^{N},H_s>L_N <\pi_s^{N},H_s>ds \\
=&\int_0^t\frac{1}{N^4}\sum_{x\in \torus}\left[\sum_{|z|=1}A_1(\confhat,x,z)H_s(x/N)H_s((x+z)/N)+A_2(\confhat,x)H_s(x/N)^2\right] ds\\
\leq&\int_0^t \frac{1}{N^4}\sum_{x\in \torus} C\norm{H}_{\infty}^2 ds\leq \frac{1}{N^2}tC\norm{H}_{\infty}^2,
\end{align*}
where $C$, $A_1(\confhat,x,z)$ and $A_2(\confhat,x)$ are bounded uniformly in $N$. The quadratic variation $[M^{H,N}]_t$  is therefore  of order $N^{-2}$, and  vanishes as $N$ goes to infinity. 
Doob's inequality hence gives us for any $T>0$, $\delta>0$ 
\[\lim_{N\to \infty}\Prob_{\mu^N}^{\lambda,\beta}\left(\sup_{0\leq t\leq T}\abs{M_t^{H,N}}\geq \delta\right)=0,\]
and in particular 
\begin{equation}\label{martingale}\lim_{N\to \infty}\Prob_{\mu^N}^{\lambda,\beta}\left(\abs{M_T^{H,N}}\geq \delta\right)=0.\end{equation}

We first consider the case of a function $H$ such that 
\[H_t(u,\theta)=G_t(u)\omega(\theta),\]
the general case will be a simple consequence of a periodic version of the Weierstrass approximation Theorem. For any such $H$, we can write 
\begin{align}\label{intLN}\int_0^T L_N<\pi_t^{N},H_t> dt=\frac{1}{N^2}\int_0^Tdt  \sum_{x\in \torus} \tau_x \cro{\sum_{i=1}^2[N\curom_i+\curaom_i](t)\partial_{u_i,N} G_t(x/N) + G_t(x/N)\gamma^{\omega}(t)} ,
\end{align}
where $\curom_i$, $\curaom_i$ and $\gamma^{\omega}$ were introduced in Definition \ref{defi:currents}, and 
\[\partial_{u_i,N} G(x/N)=N( G(x+e_i/N)- G(x/N))\] is a microscopic approximation of the spatial  derivative $\partial_{u_i}G$.

Thanks to Sections \ref{sec:4} and \ref{sec:6}, we can perform the following replacements, in the expectation of the  expression above, and in the limit $N\to \infty$ then $\varepsilon \to 0$: 
\begin{itemize}
\item Thanks to Corollary \ref{currentstot}, we can replace $\curom_i$ by \begin{equation}\label{curomremp}-\cro{\sdc(\rhoep)\ddi\densom_{\varepsilon N} +\diff(\rho_{\varepsilon N},\densom_{\varepsilon N})\ddi\rho_{\varepsilon N}}, \end{equation}
where $\diff$ is given by equation \eqref{diffomdef}, \[\diff(\rho, \densom)=\densom(1-\sdc(\rho))/\rho,\]
\item Thanks to Corollary \ref{currentstot} and Lemma \ref{lem:driftexpectation}, $\curaom_i$ can be replaced   by
\[R_i^{\omega}(\densep):=2\cro{  \sdc(\rhoep)\E_{\densep}(\conf^{\omega\lambda_i}_0)+\frac{\E_{\densep}(\com_{0}) \E_{\densep}(\conf_0^{\lambda_i})}{\rhoep}\pa{1-\rhoep-\sdc(\rhoep)}}.\]
\item Finally, the Replacement Lemma \ref{lem:replacementlemma} yields that $\gamma^{\omega}$ can be replaced by $\E_{\densep}(\gamma^{\omega})$.
\end{itemize}

\bigskip

In other words, thanks to equation \eqref{martingale}, for any $H_s(u,\theta)=G_s(u)\omega(\theta)$, we can write  
\begin{equation}\label{martingalet}\limep\lim_{N\to \infty}\Prob_{\mu^N}^{\lambda,\beta}\left(\abs{\widetilde{M}_T^{H,N, \varepsilon}}\geq \delta\right)=0,\end{equation}
where 
\begin{multline}\label{martdeft}\widetilde{M}_T^{H,N, \varepsilon}=<\pi_T^{N},H_T>-<\pi_0^{N},H_0>-\int_0^T <\pi_t^{N},\partial_tH_t>dt\\
+\int_0^Tdt  \Bigg[\frac{1}{N^2}\sum_{x\in \torus} \tau_x \sum_{i=1}^2\cro{N \pa{\sdc(\rhoep)\ddi\densom_{\varepsilon N} +\diff(\rho_{\varepsilon N},\densom_{\varepsilon N})\ddi\rho_{\varepsilon N}}-R_i^{\omega}(\densep)}\partial_{u_i,N} G_t(x/N) \\
- G_t(x/N)\E_{\densep}(\gamma^{\omega})\Bigg](t)  ,\end{multline}

In order to give a clear scheme, we divide the end of the proof in a series of steps.

\paragraph{Performing the second integration by parts}Due to the presence of the diffusion coefficients, one cannot switch directly the last discrete derivatives $\ddi \rho_{\varepsilon N}$ and $\ddi \densom_{\varepsilon N}$ onto the smooth function $G$. In one dimension, one would consider a primitive $d(\rho)$ of the diffusion coefficient $D(\rho)$, and write that \[D(\rho_{\varepsilon N})\ddi \rho_{\varepsilon N}=\ddi d(\rho_{\varepsilon N})+o_N(\ddi \rho_{\varepsilon N}).\] However, our case cannot be solved that way because the differential form 
\[(\rho, \rho^{\omega})\mapsto \sdc(\rho)d\densom +\diff(\rho,\densom)d\rho,\]
 is not closed, and therefore not exact either, which means that we cannot express \eqref{curomremp} as 
\[\ddi F(\rho_{\varepsilon N},\densom_{\varepsilon N})+o_N(1/N).\]
We thus need another argument to obtain the differential equation \eqref{equadiff}.

First, we get rid of the part with $\ddi\densom$. To do so, notice that 
\begin{align*}\ddi \cro{\sdc(\rhoep)\densom_{\varepsilon N}}&=\sdc(\rhoep)\ddi\densom_{\varepsilon N}+\densom_{\varepsilon N}\ddi \sdc(\rhoep)+o_N(1/N)\nonumber\\
&=\sdc(\rhoep)\ddi\densom_{\varepsilon N}+\densom_{\varepsilon N}\sdc'(\rhoep)\ddi \rhoep+o_N(1/N).\end{align*}
We can therefore write 
\begin{equation}\label{IPPend}\sdc(\rhoep)\ddi\densom_{\varepsilon N}=\ddi \cro{\sdc(\rhoep)\densom_{\varepsilon N}}-\densom_{\varepsilon N}\sdc'(\rhoep)\ddi \rhoep+o_N(1/N).\end{equation}
Let us denote for any $x\in \torus$  
\[D^{\varepsilon N}_x=\tau_x\pa{\diff(\rhoep, \densom_{\varepsilon N})-\densom_{\varepsilon N} \sdc'(\rhoep)}.\]
We perform a second integration by parts in the contribution of the first term in the right-hand side of \eqref{IPPend}, whereas the left-hand side is added to the existing contribution of $\ddi\rhoep$, with the modified diffusion coefficient $D^{\varepsilon N}_x$ defined above. We can now rewrite $\widetilde{M}_T^{H,N, \varepsilon}$ as 
\begin{equation}\label{LNI1}<\pi_T^{N},H_T>-<\pi_0^{N},H_0>-\int_0^T <\pi_t^{N},\partial_tH_t>dt-\int_0^T I_1(t,\confhat_t)-I_2(t,\confhat_t) dt +o_N(1),\end{equation}
 where \[I_1(t,\confhat)=\frac{1}{N^2}  \sum_{x\in \torus} \tau_x\cro{ \sum_{i=1}^2\sdc(\rhoep)\densom_{\varepsilon N}\partial_{u_i,N}^2 G_t(x/N)+R_i^{\omega}(\densep)\partial_{u_i,N} G_t(x/N) + G_t(x/N)\E_{\densep}(\gamma^{\omega}))}\]
and 
\begin{align*}I_2(t,\confhat)&=\frac{1}{N^2}\sum_{x\in \torus} \tau_x \sum_{i=1}^2ND_0^{\varepsilon N}\ddi\rhoep \partial_{u_i,N} G_t(x/N)\\
&=\frac{1}{N^2}\sum_{x\in \torus}  \sum_{i=1}^2ND_x^{\varepsilon N}(\tau_{x+e_i}\rhoep-\tau_x\rhoep) \partial_{u_i,N} G_t(x/N).\end{align*}
In $I_1$, we regrouped all the terms for which taking the limit $N\to \infty$ is not a problem, whereas $I_2$ is the term where the extra factor $N$ still has to be absorbed in a spatial derivative.

\paragraph{Replacement of the microscopic gradient by a mesoscopic gradient}Since we cannot switch the derivative on the smooth function $G$ due to the diffusion coefficient, we need to obtain the gradient of $\rho$ in another way. For this purpose, we need to replace the microscopic gradient $\tau_{x+e_i}\rhoep-\tau_x\rhoep$ by a mesoscopic gradient, and make the derivative (in a weak sense) of $\rho$ appear directly. More precisely, let us define 
\[\widetilde{I}_2(t,\confhat)=\frac{1}{N^2}\sum_{x\in \torus}  \sum_{i=1}^2 D_x^{\varepsilon N}\frac{\tau_{x+\varepsilon^3Ne_i}\rhoep-\tau_{x-\varepsilon^3Ne_i}\rhoep}{2\varepsilon ^3} \partial_{u_i,N} G_t(x/N).\]
We are going to prove that for any configuration $\confhat$, \begin{equation}\label{I2tilde}\abs{I_2(t,\confhat)-\widetilde{I}_2(t,\confhat)}\leq o_N(1)+o_{\varepsilon}(1),\end{equation}
uniformly in $\confhat$.
To prove the latter, for any $k\in \llbracket-\varepsilon^3N,\varepsilon^3N\rrbracket$, let us denote by $x_k=x+ke_i$,
\[\tau_{x+\varepsilon^3Ne_i}\rhoep-\tau_{x-\varepsilon^3Ne_i}\rhoep=\sum_{k=-\varepsilon^3N}^{k=\varepsilon^3N-1}\tau_{x_{k+1}}\rhoep-\tau_{x_k}\rhoep.\]
A summation by parts therefore allows us to rewrite $\widetilde{I}_2$ as 
\[\widetilde{I}_2(t,\confhat)=\frac{1}{N^2}\sum_{x\in \torus}  \sum_{i=1}^2\cro{\frac{1}{2{\varepsilon ^3 N}}\sum_{k=-\varepsilon^3N}^{k=\varepsilon^3N-1}D_{x_k}^{\varepsilon N} \partial_{u_i,N} G_t(x_k/N)}N(\tau_{x+e_i}\rhoep-\tau_x\rhoep).\]
Furthermore, we can write for any $x\in \torus$
\begin{multline*}\abs{D_{x}^{\varepsilon N} \partial_{u_i,N} G_t(x/N)-\frac{1}{2\varepsilon^3N}\sum_{k=-\varepsilon^3N}^{k=\varepsilon^3N-1}D_{x_k}^{\varepsilon N} \partial_{u_i,N} G_t(x_k/N)}\\
\leq\frac{1}{2\varepsilon^3N}\sum_{k=-\varepsilon^3N}^{k=\varepsilon^3N-1}\abs{D_{x}^{\varepsilon N}( \partial_{u_i,N} G_t(x/N)- \partial_{u_i,N} G_t(x_k/N))}+\abs{\partial_{u_i,N} G_t(x_k/N)(D_{x}^{\varepsilon N}-D_{x_k}^{\varepsilon N})}.\end{multline*}
Since the diffusion coefficients are bounded and $G_s$ is $C^2$, and since $x$ and the $x_k$'s are distant of $\varepsilon^3N$, we can write 
\[\abs{D_{x}^{\varepsilon N}( \partial_{u_i,N} G_t(x/N)- \partial_{u_i,N} G_t(x_k/N))}\leq C(G_t)\varepsilon^3.\]
Since $D_{x_k}^{\varepsilon N}$ depends on the macroscopic density $\dens_{\varepsilon N}$, and since the diffusion coefficients can be extended as $C^1$ functions due to their explicit expression, we also have 
\begin{align*}\abs{\partial_{u_i,N} G_t(x_k/N)(D_{x}^{\varepsilon N}-D_{x_k}^{\varepsilon N})}&\leq {C'(G_t)}\pa{\abs{\tau_x \rhoep-\tau_{x_k} \rhoep}+\abs{\tau_x \densom_{\varepsilon N}-\tau_{x_k} \densom_{\varepsilon N}}}\\
&\leq C''(G_t, \omega)\frac{\varepsilon^3 N}{\varepsilon N}.\end{align*}
These two bounds finally yield that 
\begin{equation}\label{boundabsD}\abs{D_{x}^{\varepsilon N} \partial_{u_i,N} G_t(x/N)-\frac{1}{2\varepsilon^3N}\sum_{k=-\varepsilon^3N}^{k=\varepsilon^3N-1}D_{x_k}^{\varepsilon N} \partial_{u_i,N} G_t(x_k/N)}\leq C(G_t)\varepsilon^3 +C''(G_t, \omega)\varepsilon^2=o_{\varepsilon }(\varepsilon).\end{equation}
By definition of  $I_2$ and $\widetilde{I}_2$, the triangular inequality yields 
 \begin{multline*}\abss{I_2-\widetilde{I}_2}\leq\\
 \frac{1}{N^2}\sum_{x\in \torus}  \sum_{i=1}^2\abs{D_{x}^{\varepsilon N} \partial_{u_i,N} G_t(x/N)-\frac{1}{2\varepsilon^3N}\sum_{k=-\varepsilon^3N}^{k=\varepsilon^3N-1}D_{x_k}^{\varepsilon N} \partial_{u_i,N} G_t(x_k/N)}N(\tau_{x+e_i}\rhoep-\tau_x\rhoep) .\end{multline*}
The quantity inside the absolute values in the right-hand side above is $o_N(1)+o_{\varepsilon}({\varepsilon})$, thanks to \eqref{boundabsD}, whereas $N(\tau_{x+e_i}\rhoep-\tau_x\rhoep) $ is of order at most $1/\varepsilon$, whereas the quantity inside absolute values is $o_{\varepsilon}(\varepsilon)$, therefore their product vanishes as $\varepsilon\to 0$, which proves equation \eqref{I2tilde}.
We therefore have obtained as wanted that 
\begin{equation}\label{itilde}\limep\limN I_2(t,\confhat)-\widetilde{I}_2(t,\confhat)=0,\end{equation}
uniformly in $\confhat$. We can now replace in equation \eqref{LNI1} $I_2$ by $\widetilde{I_2}$.

\paragraph{Embedding in the space of trajectories of measures $\Mhat$}Recall that $Q^N$ is the distribution of the empirical measure of our process. We now wish to express the martingale $\widetilde{M}^{H,N,\varepsilon}_t$ introduced after equation \eqref{martingalet} as an explicit function of the empirical measure $\pi^N$ in order to characterize the limit points $Q^*$ of the compact sequence $Q^N$. For that purpose, let $(\varphi_{\varepsilon})_{\varepsilon \to 0}$ be a family of localizing functions on $\ctorus$,  \[\varphi_{\varepsilon}(\cdot)=(2\varepsilon)^{-2}\1_{[-\varepsilon,\varepsilon]^2}(\cdot),\]
and recall that we defined the empirical measure as 
\[\pi^N_t=\frac{1}{N^2}\sum_{x\in\torus}\conf_x(t)\delta_{x/N,\theta_x(t)}.\]
Then, for any function $\Phi:\ctoruspi\to \R$, and any $u\in \ctorus$ we denote by $\varphi_{\varepsilon,u}^{\Phi}$ the function 
\[\func{\varphi_{\varepsilon,u}^{\Phi}}{\ctorus\times\ctoruspi}{\R}{(v, \theta)}{\varphi_{\varepsilon}(v-u)\Phi(\theta)}.\] 
With this notation, we can therefore write
 \[\E_{\tau_x\densep}(\conf_0^{\Phi})=\frac{1}{(2\varepsilon N+1)^2}\sum_{\norm{y-x}_{\infty}\leq \varepsilon N}\conf^{\Phi}_y=\frac{(2\varepsilon N)^2}{(2\varepsilon N+1)^2}<\pi^{N},\varphi_{\varepsilon,x/N}^{\Phi}>.\]
In the particular case where $\Phi\equiv1$, (resp. $\Phi=\omega$), this rewrites
\[\tau_x\rho_{\varepsilon N}=\frac{(2\varepsilon N)^2}{(2\varepsilon N+1)^2}<\pi^{N},\varphi_{\varepsilon,x/N}^{1}>\quad \pa{\mbox{resp.}\tau_x\densom_{\varepsilon N}=\frac{(2\varepsilon N)^2}{(2\varepsilon N+1)^2}<\pi^{N},\varphi_{\varepsilon,x/N}^{\omega}>}.\]
Since  $(2\varepsilon N)^2/(2\varepsilon N+1)^2=1+o_N(1)$, we can replace in the limit $N\to \infty$ the quantity $\E_{\tau_x\densep}(\conf_0^{\Phi})$ (resp. $\tau_x{\rhoep}$, $ \tau_x\densom$) by the function of the empirical measure $<\pi^{N},\varphi_{\varepsilon,x/N}^{\Phi}>$ (resp. $<\pi^{N},\varphi_{\varepsilon,x/N}^{1}>$, $<\pi^{N},\varphi_{\varepsilon,x/N}^{\omega}>$).

We  deduce from equations  \eqref{martingalet}, \eqref{LNI1} and \eqref{itilde} and what precedes that  for any positive $\delta$,
\begin{equation}\label{delta}\limep\limN Q^N\left(\abs{N_T^{H,N}\pa{\pi^{[0,T]}}}\geq \delta\right)=0.\end{equation}
where $ N_T^{H,N}$ is defined as 
\begin{align}\label{mart2}
N&^{H,N}_T\pa{\pi^{[0,T]}}=<\pi_T,H_T>-<\pi_0,H_0>-\int_0^T <\pi_t,\partial_tH_t>dt\\
&-\int_0^T \cro{\frac{1}{N^2}\sum_{x\in \torus}\sum_{i=1}^2\widetilde{d}_{x/N, \varepsilon}(\pi_t)\partial^2_{u_i,N}G_t(x/N)+\widetilde{R}_{x/N, \varepsilon, i}(\pi_t)\partial_{u_i,N}G_t(x/N)+ \Gamma^{\omega}_{x/N, \varepsilon}\pa{\pi_t}G_t(x/N)}dt\nonumber\\
&+\int_0^T\cro{\frac{1}{N^2} \sum_{x\in \torus}\sum_{i=1}^2\widetilde{D}_{x/N, \varepsilon}(\pi_t)<\pi_t,\frac{\varphi_{\varepsilon,x/N+\varepsilon^3 e_i}^{1}-\varphi_{\varepsilon,x/N-\varepsilon^3 e_i}^{1}}{2\varepsilon^3}>\partial_{u_i,N}G_t(x/N)}dt.\nonumber
\end{align}
In the identity above, we denoted
\[ \widetilde{d}_{x/N, \varepsilon}(\pi)=\sdc(<\pi,\varphi_{\varepsilon,x/N}^{1}>)<\pi,\varphi_{\varepsilon,x/N}^{\omega}>\]
\[ \widetilde{D}_{x/N, \varepsilon}(\pi)=\diff(<\pi,\varphi_{\varepsilon,x/N}^{1}>,<\pi,\varphi_{\varepsilon,x/N}^{\omega}>)-<\pi,\varphi_{\varepsilon,x/N}^{\omega}> \sdc'(<\pi,\varphi_{\varepsilon,x/N}^{1}>)\]
\begin{multline*}\widetilde{R}_{x/N, \varepsilon, i}(\pi)= \sdc\pa{<\pi,\varphi_{\varepsilon,x/N}^{1}>}<\pi,\varphi_{\varepsilon,x/N}^{\omega \lambda_i}>\\
+\frac{<\pi,\varphi_{\varepsilon,x/N}^{\omega}><\pi,\varphi_{\varepsilon,x/N}^{\lambda_i}>}{<\pi,\varphi_{\varepsilon,x/N}^{1}>}\cro{1-<\pi,\varphi_{\varepsilon,x/N}^{1}>-\sdc\pa{<\pi,\varphi_{\varepsilon,x/N}^{1}>}},\end{multline*}
and $\Gamma^{\omega}_{u, \varepsilon}\pa{\pi}=\E_{\param_{x/N,\varepsilon}(\pi)}(\gamma^{\omega})$, where $\param_{x/N,\varepsilon}(\pi)\in\pset$ is the measure on $\ctoruspi$  
\[\param_{x/N,\varepsilon}(\pi)(d\theta)=\int_{\ctorus}\varphi_{\varepsilon}(.-x/N)\pi(du, d\theta).\]

 \paragraph{Limit $N\to\infty$} We have now successfully balanced out all the factors $N$, and can thus let $N$ go to $\infty$ in \eqref{delta}. Since $G$ is a smooth function, one can replace in \eqref{mart2} the discrete space derivatives $\partial_{u_i,N}$ by the continuous derivative $\partial_{u_i}$, the sums $N^{-2} \sum_{x\in \torus} $ by the integral $\int_{\ctorus}du$, and the variables $x/N$ by $u$. We proved in Proposition \ref{prop:compactnessQN} that the sequence of distributions $(Q^N)_N$ is relatively compact. Since the quantity inside the absolute values is a continuous function (for Skorohod's topology defined in Appendix \ref{subsec:topo}) of $\pi^{[0,T]}$, the whole event is an open set,  we obtain that for any weak limit point  $Q^*$ of $(Q^N)$, and any positive $\delta$,
\begin{align}\label{mart2*}
\limep Q^*\Bigg(\Bigg\vert&<\pi_T,H_T>-<\pi_0,H_0>-\int_0^T <\pi_t,\partial_tH_t>dt\nonumber\\
&-\int_0^T \int_{\ctorus}\sum_{i=1}^2\cro{\widetilde{d}_{u, \varepsilon}(\pi_t)\partial^2_{u_i}G_t(u)+\widetilde{R}_{u, \varepsilon, i}(\pi_t)\partial_{u_i}G_t(u)+ \Gamma^{\omega}_{u, \varepsilon}\pa{\pi_t}G_t(u)}dudt\nonumber\\
&+\int_0^T\int_{\ctorus}\sum_{i=1}^2\cro{\widetilde{D}_{u, \varepsilon}(\pi_t)<\pi_t,\frac{\varphi_{\varepsilon,u+\varepsilon^3 e_i}^{1}-\varphi_{\varepsilon,u-\varepsilon^3 e_i}^{1}}{2\varepsilon^3}>\partial_{u_i}G_t(u)}dudt.\Bigg\vert>\delta\Bigg)=0
\end{align}

 \paragraph{Limit $\varepsilon\to0$}In order to consider the limit $\varepsilon \to 0$, we need to express \[<\pi_t,\frac{\varphi_{\varepsilon,u+\varepsilon^3 e_i}^{1}-\varphi_{\varepsilon,u-\varepsilon^3 e_i}^{1}}{2\varepsilon^3}>\] in the third line above as an approximation of the gradient of the density $\partial_{u_i}\rho_t(u)$. As in the proof of Lemma \ref{lem:H1estimate}, consider a smooth function $h_{\varepsilon,i,u}$ such that 
\begin{equation}\label{approxphi}\int_{\ctorus}\abs{\frac{\varphi_{\varepsilon,u+\varepsilon^3 e_i}^{1}-\varphi_{\varepsilon,u-\varepsilon^3 e_i}^{1}}{2\varepsilon^3}(v)-h_{\varepsilon,i,u}}dv=o_{\varepsilon}(1).\end{equation}
Since such a function is very similar to the one already presented in  Lemma \ref{lem:H1estimate}, we do not give a detailed construction here. Then, we can build a smooth anti-derivative  $H_{\varepsilon,u}$ of $h_{\varepsilon, i, u}$, and we can write for any  $u\in \ctorus$, and any density $\rho$ in $H^1$, 
\[\int_{\ctorus}\rho(v)h_{\varepsilon,i,u}(v) dv=\int_{\ctorus}\partial_{u_i}\rho(v)H_{\varepsilon,u}(v) dv.\]
Regarding the third line of \eqref{mart2*}, this yields
\[<\pi_t,\frac{\varphi_{\varepsilon,u+\varepsilon^3 e_i}^{1}-\varphi_{\varepsilon,u-\varepsilon^3 e_i}^{1}}{2\varepsilon^3}>=\int_{\ctorus}\partial_{u_i}\rho(v)H_{\varepsilon,u}(v) dv+o_{\varepsilon}(1),\]
 where $H_{\varepsilon,u}$ is a smooth approximation of a Dirac in $u$ and $o_{\varepsilon}(1)$ is uniform in $u$.  According to \eqref{Enest}, $\partial_{u_i}\rho$ is in $L^2([0,T]\times \ctorus)$ $Q^*$-a.s, therefore
\begin{equation}\label{convL2}\int_{\ctorus}\partial_{u_i}\rho_t(v)H_{\varepsilon,u}(v) dv\underset{\varepsilon\to 0}{\xrightarrow{\makebox[3cm]{$L^2([0,T]\times \ctorus)$}}}\partial_{u_i}\rho_t(u),\end{equation}
$Q^*$-a.s. (see, for example, Theorem 4.22, p.109 in \cite{BrezisB2010}).

By Lemma \ref{lem:Lebesguedensity} any limit point $Q^*$ of $(Q^N)$ is concentrated on measures absolutely continuous w.r.t. the Lebesgue measure on $\ctorus$. For any such measure $\pi^{[0,T]}$, we denote by $\boldsymbol \dens_t(u,d\theta)$ its corresponding density profile on the torus at time $t$, and let
\[\rho_t^{\omega}(u)=\int_{\ctoruspi}\omega(\theta)\boldsymbol\dens_t(u,d\theta).\] 
We also shorten $\rho(u)=\rho^1(u)$. Thanks to this last remark and using both \eqref{convL2} and the dominated convergence theorem for the second line of \eqref{mart2*}, we can now let $\varepsilon $ go to $0$ in equation \eqref{mart2*}, to obtain that for any limit point $Q^*$ of $(Q^N)$ and any $\delta>0$,
\begin{multline}\label{idabo}Q^*\Bigg(\Bigg\vert<\pi_T,H_T>-<\pi_0,H_0>-\int_0^T<\pi_t,\partial_tH_t>dt\\
-\int_0^T\int_{ \ctorus} \sum_{i=1}^2\sdc(\rho_t)\rho_t^{\omega}\partial^2_{u_i}G_t(u)+2\cro{\sdc(\rho_t)\rho_t^{\lambda_i\omega}+\frac{\rho_t^{\omega}}{\rho_t}(1-\rho_t-\sdc(\rho_t))\rho_t^{\lambda_i}}\partial_{u_i}G_t(u)+ \E_{\dens_t}(\gamma^{ \omega})G_t(u)\bigg)dudt\\
+\int_0^T\int_{ \ctorus} \sum_{i=1}^2\Big[\diff(\rho_t, \rho_t^{\omega})-\sdc'(\rho_t) \rho^{\omega}_t\Big](\partial_{u_i}\rho_t) \partial_{u_i}G_t(u)dudt \Bigg\vert>\delta\Bigg)=0.
\end{multline}

\paragraph{Conclusion}As expected, all the quantities above are linear in $\omega$, and elementary computations yield that 
\[\E_{\boldsymbol\dens_t(u,\cdot)}(\gamma^{ \omega})=\int_{\ctoruspi}\omega(\theta)\Big[\rho_t(u)\E_{\boldsymbol \dens_t(u,\cdot)}(c_{u,\beta}(\theta,\confhat))d\theta-\boldsymbol\dens_t(u,d\theta)\Big].\]
Furthermore, since $H_t(u,\theta)=G_t(u)\omega(\theta)$, we can write for $k=1$, $2$
\[\rho_t^{\omega}\partial^k_{u_i}G_t(u)=\int_{\ctoruspi}\omega(\theta)\partial^k_{u_i}G_t(u)\boldsymbol\dens_t(u,d\theta)=\int_{\ctoruspi}\partial^k_{u_i}H_t(u, \theta)\boldsymbol\dens_t(u,d\theta).\]
analogous identities can be obtained when $\omega$ is replaced by another function $\Phi\in C^1(\ctoruspi)$. Using in equation \eqref{idabo}  the identities above finally yield, as wanted, that for any $\delta>0$
\begin{multline*}Q^*\Bigg(\Bigg\vert<\pi_T,H_T>-<\pi_0,H_0>-\int_0^T<\pi_t,\partial_tH_t>dt\\
-\int_0^T\int_{\ctorus\times \ctoruspi}\Bigg[\sum_{i=1}^2\bigg(-\partial_{u_i}H_t(u, \theta)\big[\diffh( \rho_t,\boldsymbol \dens_t) -\sdc'(\rho_t)\boldsymbol \dens_t\big](u, d\theta)\partial_{u_i}\rho_t(u)+\partial_{u_i}^2 H_t(u, \theta)\sdc(\rho_t) \boldsymbol \dens_t(u, d\theta)\\
+ \partial_{u_i}H_t(u, \theta) \cro{2\lambda\drifth(\rho_t,\boldsymbol \dens_t)\overset{\rightarrow}{\Omega}(\boldsymbol \dens_t)+2\lambda_i(\theta) \sdc(\rho_t) \boldsymbol\dens_t}(u,d\theta)\bigg)+H_t(u,\theta)\Gamma_t(\boldsymbol \dens)(u,d\theta)\Bigg]du dt\Bigg\vert>\delta\Bigg)=0.
\end{multline*}

As in the proof of Proposition \ref{prop:compactnessQN}, this last identity can be extended in the case where $H_t(u,\theta)$ does not take the form $G_t(u)\omega(\theta)$ by using a periodic version of the Weierstrass Theorem, thus letting $\delta \to 0$ completes the proof of Theorem \ref{thm:mainthm}.
}

\section{Limiting  space-time covariance}
\label{sec8}

\intro{
This section is entirely dedicated to the proof Theorem \ref{thm:limcovariance}, that was postponed. 
The strategy of the proof, follows the same scheme as in Section 7.4 of  \cite{KLB1999}. 
One of its core ingredients is a decomposition theorem (cf. Proposition \eqref{prop:GCF}) for translation-invariant closed differential forms. To prove this decomposition, one requires a sharp estimate on the spectral gap of the symmetric exclusion 
generator, which is not uniform w.r.t. the density in our case, and some adaptations w.r.t. the classical scheme are necessary to account for the angles. The non-uniformity of the spectral gap comes from the slow mixing occurring at high densities, and requires some minor adaptation w.r.t. \cite{Quastel1992} where this issue was not dealt with.  It is solved by cutting off large densities (cf. equation \eqref{def:phin} and Lemma \ref{lem:bulkconvergence}). }

\subsection{Spectral gap for the symmetric exclusion process with angles}
\label{subsec:spectralgap}
\intro{As investigated in Section \ref{subsec:irreducibility}, the mixing time for the exclusion dynamics on configurations of size $n$ with angles is not of order $n^2$. 
We therefore cannot consider a general class of functions as dependent on the $\theta_x$'s as wanted, and need to restrict to a subclass of functions with low levels of correlations between particle angles, but large enough for the non-gradient method to apply. In this section, we prove that the spectral gap of the symmetric exclusion process on this class of functions is of order $C(\rho)n^{-2}$ if the density in the box is less than $\rho<1$.
The core estimate was first derived  by Quastel in \cite{Quastel1992}. We present here a modified version to take into account the continuous angles.}

Throughout this section, we consider the square domain 
\[B_n=\llbracket-n,n\rrbracket^2\] 
with \emph{closed boundaries}. Recall that $\Sp$ was introduced in Definition \ref{def:conf} as the set of angle-blind functions, and that $\omega$ is the angular dependency of our test function $H$ (cf. equation \eqref{Hdecomp}). 
We already defined 
\begin{equation*}
T^\omega=\left\{{f\in \mathcal{C}} \; \;\Big| \;\;  f(\confhat)=\varphi(\conf)+\sum_{x\in \Z^2}\com_x\psi_x({\conf}), \quad \varphi,\psi_x\in \Sp,\; \forall x \in \Z^2\right\},
\end{equation*}
and now denote by $\mathcal{C}_n$ (resp.  $\Sp_n$) the set of cylinder functions (resp. angle-blind functions) depending only on sites in $B_n.$  {Finally, we define $T_n^\omega=\mathcal{C}_n\cap T^\omega$.}

\begin{rema}
The purpose of the non-gradient method is to replace the instantaneous current $\curom_i$ introduced in equation \eqref{currentssym} by a gradient quantity $D(\conf_0-\conf_{e_i})+d(\com_0-\com_{e_i})$, and the class $T^\omega$ above is the simplest set of functions, stable by $\gene_n$ and containing both the currents and the gradients.

We expect that it is not the biggest class of functions on which a spectral gap estimate of order $n^{-2}$ holds. Indeed, we believe that introducing some finite numbered correlations between angles might not alter too much the order of the spectral gap. It is not, however, the purpose of this section, and this remark is therefore left as a conjecture at this point.  
\end{rema}

Recall from Definition \ref{defi:CM} that  we encoded in the canonical state $\K\in \Kset_n$ the number and angles of the particles in $B_n$, and that we denote by $\cmnk=\gcm\pa{\;\cdot\;\mid \confhat\in\Sigma_n^{\K}}$ the canonical measure with $\K$ particles inside $B_n$. Finally, define 
\[\Dnk(f)=\Ecmnk(f\gene_n f),\]
where $\gene_n$ is the symmetric exclusion generator restricted to jumps with both extremities in $B_n$. We are now ready to state the main result of this section.

\index{$ \mathcal{T}_0^{\omega,n}$ \dotfill set of mean $0$ functions linear in the angles}

\begin{prop}[Estimate on the spectral gap for the SSEP with angles]
\label{prop:spectralgap}
For any $0\leq \alpha<1$, there exists a constant $C(\alpha)$ such that for any $ \K\in \Kset_n$ such that $K\leq \alpha|B_n|$, and any $f\in  T_n^\omega$ such that $\E_{n,\K}(f)=0$, 
\[\Ecmnk(f^2)\leq C(\alpha)n^2\Dnk(f).\]
\end{prop}
\begin{rema}[Non-uniformity of the spectral gap]
Note that this estimate is not uniform in the density. Actually, the constant $C(\alpha)$ behaves as $1/(1-\am)$, and therefore even on the set $T^\omega$, the spectral gap of the exclusion process when there are only a finite number of empty sites in $B_n$ is or order $n^{-4}$. This high density estimate is sharp~: define $\K_n$ by $K_n=(2n+1)^2-1$, and for $k=1,\dots,K_n$, $\theta_k=2k\pi/K_n$, then for  
\[f_n(\confhat)=\sum_{x\in B_n}(\theta_x-\pi)\conf_x\cos\pa{\frac{2\pi x_1}{2n+1}},\]
one easily checks that there exists a positive constant $C$ such that
\[n^4\frac{\dir_{n, \K_n}(f_n)}{Var_{n, \K_n}(f_n)}\xrightarrow[n\to\infty]{} C.\]
This non-uniformity is not an issue here, however, because when we later on classify the germs of closed forms for our model, we are able to cutoff the large densities (cf. equation \eqref{def:phin}). 
\end{rema}
In order to prove Proposition \ref{prop:spectralgap}, we need the following lemma, which states that the angle-blind process has a uniform  spectral gap of order $n^{-2}$. 
For any angle-blind function $\psi\in \Sp_n$, we will write $\psi(\conf)$ instead of $\psi(\confhat)$ to emphasize that it does not depend on the angles.
\begin{lemm}[Spectral gap for the angle-blind exclusion process]
\label{lem:blindspectralgap} Denote by $\Ecmnkt$ the expectation w.r.t. the angle-blind canonical measure with $K$ particles inside $B_n$, defined for any angle-blind function $\psi\in \Sp_n$ by
\[\Ecmnkt(\psi)=\Egcm\pa{\psi\;\;\bigg|\; \sum_{x\in B_n}\eta_x=K},\]
which holds for any $\param$ with density $\am\in (0,1)$.
There exists a universal constant $C_1>0$ such that for any $n\geq 1$, any $0\leq K\leq (2n+1)^2$ and any  $\psi\in \Sp_n$ satisfying $\Ecmnkt(\psi)=0$,
\[\Ecmnkt(\psi^2)\leq C_1n^2\Dnkt(\psi),\]
where $\Dnkt(\psi)=\Ecmnkt(\psi(-\gene_n)\psi)$.
\end{lemm}
\noindent This result is fairly classical, its proof can be found for instance in \cite{KLB1999}, we do not repeat it here. 
Note in particular that for the angle-blind process, the constant can be chosen independently of the cap on the density $\alpha$.
Before proving Proposition \ref{prop:spectralgap}, we need one more definition. 
Fix $\alpha\in [0,1)$, and a canonical state $\K\in \Kset_n$ such that $K\leq \alpha|B_n|$. We then define for any site $x\in \Z^2$, 
\index{$ \conftilde_x$ \dotfill modification of $\com_x$ with mean $0$ w.r.t. $\theta_x$}
\begin{equation}
\label{eq:Defetat}
\widehat{\omega}=\omega-\E_{n,\K}(\omega)\eqand \conftilde_x=\cro{\omega(\theta_x)-\E_{n,\K}(\omega)}\conf_x,
\end{equation}
where $\Ecmnk(\omega)$ stands for $\Ecmnk(\omega(\theta_0)\mid \eta_0=1)$. In particular, for any configuration $\confhat$, $\sum_{x\in B_n}\conftilde_x=0$ under $\mu_{n,\K}$.
This centered occupation variable plays a particular role in the proof of the spectral gap, and we state in the following Lemma  two identities regarding $\conftilde$, which will be used later on.  

\begin{lemm}[Properties of $\conftilde$]
\label{lem:conftilde}Define $V_{n,\K}(\omega)= Var_{n, \K}(\omega(\theta_0)\mid \eta_0=1)$. For any $x\neq y\in B_n$, $\K\in \Kset_n$, and any angle-blind function $\psi\in \Sp_n$, we have $\Ecmnk\pa{\conftilde_x \psi}=0$,
\[\Ecmnk\pa{(\conftilde_x)^2\psi}=V_{n,\K}(\omega)\Ecmnkt(\eta_x\psi) \quad\mbox{ and } \quad \Ecmnk\pa{\conftilde_x\conftilde_y \psi}=\begin{cases}
                                                                                                                                      -\frac{V_{n,\K}(\omega)}{K-1}\Ecmnkt(\conf_x\conf_y\psi)& \mbox{if }K> 1\\
                                                                                                                                      0 &\mbox{else}
                                                                                                                                     \end{cases}.\]
\end{lemm}
\proofthm{Proof of Lemma \ref{lem:conftilde}}
{
This Lemma follows from elementary computations. Under $\cmnk$,  for any angle-blind function $\psi\in \Sp_n$ and any function $\Phi$ on $\ctoruspi$, we have \[\Ecmnk(\conf^{\Phi}_x \psi)=\Ecmnk(\Phi(\theta_0)\mid \eta_0=1)\Ecmnkt(\conf_x\psi).\]
For the first (resp. second) identity, we set $\Phi=\omega-\E_{n,\K}(\omega)$ (resp. $\Phi=(\omega-\E_{n,\K}(\omega))^2$), which by construction has mean $0$ (resp. $V_{n,\K}(\omega)$) w.r.t. $\cmnk(\cdot\mid \eta_0=1 )$.
Regarding the last identity, we obtain similarly
\[\Ecmnk\pa{\conftilde_x\conftilde_y \psi}=\Big[\Ecmnk(\omega(\theta_x)\omega(\theta_y)\mid \conf_x=\conf_y=1)-\Ecmnk(\omega)^2\Big]\Ecmnkt(\conf_x\conf_y\psi)=-\frac{V_{n,\K}(\omega)}{K-1}\Ecmnkt(\conf_x\conf_y\psi)\]
if $K>1$, and trivially vanishes if $K=0,1$.}
\bigskip

We now estimate the spectral gap of the angle process on $ { T_n^\omega}$. 

\proofthm{Proposition \ref{prop:spectralgap}}{Fix $\alpha\in [0,1)$, $\K\in \Kset_n$ such that $K\leq \alpha|B_n|$, and consider a function $f=\varphi(\eta)+\sum_{x\in B_n}\com_x\psi_x(\eta)$ in $ {T_n^\omega}$, where $\varphi, \psi_x\in \Sp_n$, such that $\E_{n,\K}(f)=0$. Recall the notation introduced in \eqref{eq:Defetat}, and denote
\[f_1=\sum_{x\in B_n}\conftilde_x\psi_x, \quad \quad f_{b}=\varphi+ \Ecmnk(\omega)\sum_{x\in B_n}\eta_x\psi_x\in \Sp_n .\]
By construction, $f=f_1+f_{b}$. Since for any $\psi\in \Sp_n$, $\Ecmnk\pa{\conftilde_x \psi}=0$, it is straightforward to obtain that
\[\Ecmnk\pa{f^2}= \Ecmnk\pa{f_1^2}+ \Ecmnkt\pa{f_{b}^2}\eqand \Ecmnk\pa{f\gene_n f}= \Ecmnk\pa{f_1\gene_n f_1}+ \Ecmnkt\pa{f_{b}\gene_n f_{b}},\]
(i.e $\Dnk(f)=\Dnk(f_1)+\Dnkt(f_{b})$). By assumption $\Ecmnk(f)=0$, therefore, since by construction $\Ecmnk(f_1)=0$, we also have $\Ecmnk(f^{b})=0$. 
Lemma \ref{lem:blindspectralgap} can therefore be applied to $f_{b}$. 
To prove Proposition \ref{prop:spectralgap}, it is thus sufficient to prove it for any function of the form $f=\sum_{x\in B_n}\conftilde_x\psi_x(\conf)$. 
We can further assume, without loss of generality, that $ \sum\psi_x=0$ and that each $\psi_x$ vanishes if $\eta_x=0$ since we can rewrite
\[f(\confhat)=\sum_{x\in B_n}\conftilde_x \ptilde_x(\eta)\]
where
\[\ptilde_x=\eta_x(\psi_x-\pbar)\quad \mbox{  and } \quad\pbar=\frac{\sum_{x\in B_n}\eta_x\psi_x}{\sum_{x\in B_n}\eta_x}=\frac{\sum_{x\in B_n}\eta_x\psi_x}{K(\confhat)}.\]
Note that we only consider $ K>0$, since if $K=0$, Proposition \ref{prop:spectralgap} is immediate.

\bigskip

To prove Proposition \ref{prop:spectralgap}, it is therefore sufficient to prove it for  any function
\[f=\sum_{x\in B_n}\conftilde_x\psi_x,\]
where $\psi_x=\conf_x\psi_x$, and satisfy $\sum_{x\in B_n}\psi_x=0$. For any such $f$, if $K= 1$, there is  only one particle in $B_n$ and $\conftilde_x=0$ for any $x$, therefore $f=0$.
We now assume that $1<K\leq \alpha |B_n|$. By Lemma \ref{lem:conftilde}, since by assumption $\sum_x \psi_x=0$,
\begin{equation}\label{normf}\Ecmnk\pa{f^2}=\sum_{x,y\in B_n}\Ecmnk\pa{\conftilde_x\conftilde_y\psi_x\psi_y}= \frac{K}{K-1}V_{n,\K}(\omega)\sum_{x\in B_n}\Ecmnk\pa{\psi_x^2}.\end{equation}

We now turn our attention to $\Ecmnk(f\gene_n f)$. For any site $x$ and any angle-blind function $\psi\in \Sp_n$, we can write 
\[\gene_n (\conftilde_x\psi_x)=\conftilde_x \gene_n \psi_x +\sum_{|z|=1}\1_{\{\conf_x\conf_{x+z}=0\}}\psi_x(\conf^{x,x+z})((\conf^{x,x+z})^{\widehat{\omega}}_x-\conftilde_x).\]
Since we assumed that $\psi_x$ vanishes when the site $x$ is empty, the quantity above can be rewritten
\[\gene_n (\conftilde_x\psi_x)=\conftilde_x \gene_n \psi_x +\sum_{|z|=1}\conftilde_{x+z}(1-\conf_x)\psi_x(\conf^{x,x+z}).\]
It follows that
\[\Dnk(f)=\sum_{x,y\in B_n}\cro{\Ecmnk(\conftilde_x\conftilde_y\psi_x(-\gene_n)\psi_y)-\Ecmnk\pa{\conftilde_x\psi_x\sum_{|z|=1}\conftilde_{y+z}(1-\conf_y)\psi_y(\conf^{y,y+z})}}.\]
Using once again that $\sum_{x\in B_n}\psi_x =0$, and Lemma \ref{lem:conftilde} the identity above rewrites 
\begin{equation}\label{dirf}\Dnk(f)=\frac{K}{K-1}V_{n,\K}(\omega)\sum_{x\in B_n}\cro{\Dnkt(\psi_x)-\sum_{|z|=1}\Ecmnkt\pa{(1-\conf_{x+z})\psi_x\psi_{x+z}\pa{\conf^{x,x+z}}}}.\end{equation}
Let us introduce the Dirichlet form locally cropped  in $x$ 
\begin{equation}
\label{dircrop}
\Dnkxt(\psi)=\frac{1}{2}\Ecmnkt\pa{\sum_{\substack{ y, y+z\in B_n\setminus\{x\}\\
|z|=1}}\conf_y(1-\conf_{y+z})(\psi(\conf^{y,y+z})-\psi(\conf))^2},\end{equation}
which forbids jumps to and from the site $x$. 
Since $\psi_x$ vanishes whenever the site $x$ is empty, the quantity 
$\conf_x(1-\conf_{x+z})(\psi_x(\conf^{x,x+z})-\psi_x(\conf))^2$ is also equal to $(1-\conf_{x+z})\psi_x(\conf)^2$, and a similar argument with $\psi_{x+z}$ allows us to rewrite equation \eqref{dirf}
\[{\Dnk(f)=\frac{K}{K-1}V_{n,\K}(\omega)\sum_{x\in B_n}\cro{\Dnkxt(\psi_x)+\frac{1}{2}\sum_{|z|=1}\Ecmnkt\pa{(1-\conf_{x+z})\cro{\psi_{x+z}\pa{\conf^{x,x+z}}-\psi_x(\conf)}^2}}}.\]
To obtain Proposition \ref{prop:spectralgap}, thanks to the identity above together with \eqref{normf} it is enough to prove that for some constant $C(\alpha)$,
\begin{equation}\label{newgap}\sum_{x\in B_n}\Ecmnkt\pa{\psi_x^2}\leq C(\alpha) n^2 \sum_{x\in B_n}\cro{\Dnkxt(\psi_x)+\frac{1}{2}\sum_{|z|=1}\Ecmnkt\pa{(1-\conf_{x+z})\cro{\psi_{x+z}\pa{\conf^{x,x+z}}-\psi_x}^2}}.\end{equation}

We now state a technical Lemma, which gives a spectral gap estimate when one site remains frozen.
\begin{lemm}[Spectral gap for the exclusion process with a frozen site]
\label{lem:frozensitespectralgap}
Fix $x\in B_n$. There exists a universal constant $C_2$ such that for any angle-blind function $\psi\in \Sp_n$ satisfying $\Ecmnkt(\psi\mid\conf_x=1)=0 $,
\[\Ecmnkt(\psi^2\mid\conf_x=1)\leq C_2n^2\Dnkxt(\psi\mid\conf_x=1),\]
where the conditioned Dirichlet form is defined by the conditional expectation $\Ecmnkt(.\mid\conf_x=1)$ instead of $\Ecmnkt$,\[\Dnkxt(\psi\mid\conf_x=1)=-\Ecmnkt(\psi\gene_n\psi\mid\conf_x=1).\] 
\index{$ \Dnkxt$\dotfill Dirichlet form with $x$ frozen}
\end{lemm}
\proofthm{Lemma \ref{lem:frozensitespectralgap}}{We do not give the detail of this proof. It is quite similar to the proof without the frozen site for an angle-blind function, the only difference being that whenever a path should go through the site $x$, the path is bypassed around it, which results in a larger constant $C$ but does not affect the order $n^2$. }

We now take a look at the left-hand side of equation \eqref{newgap}. 
Since $\psi_x$ vanishes whenever $\conf_x=0$ we have $\Ecmnkt(\psi_x\mid\conf_x=1)=\frac{|B_n|}{K}\Ecmnkt(\psi_x)$, the previous lemma applied to 
$\psi_x-\Ecmnk(\psi_x\mid\conf_x=1)$ yields
\begin{equation}\label{variancebound} \sum_{x\in B_n}\Ecmnkt\pa{\psi_x^2}-\frac{|B_n|}{K}\Ecmnkt\pa{\psi_x}^2\leq  C_2n^2\sum_{x\in B_n}\Dnkxt(\psi_x).\end{equation}
Furthermore, 
\begin{align*}\sum_{x,y\in B_n}[\Ecmnkt\pa{\psi_x}-\Ecmnkt(\psi_y)]^2&=\sum_{x,y\in B_n}[\Ecmnkt(\psi_x)^2+\Ecmnkt(\psi_y)^2]-2\sum_{x,y\in B_n}\Ecmnkt(\psi_x)\Ecmnkt(\psi_y)\\
&=2n^2\sum_{x\in B_n}\Ecmnkt(\psi_x)^2,\end{align*}
because the last term of the first line vanishes by the assumption $\sum_{x\in B_n}\psi_x=0$. Furthermore, consider the family of paths 
$(\gamma_{x,y})_{x,y\in B_n}$ going from $x$ to $y$, defined as follows~: starting from $x$, the path  $\gamma_{x,y}$ starts straight in the first direction, until reaching the first coordinate of $y$. then, it goes in the second direction until reaching $y$. With this construction, each edge $a$ is used at most  a number of times $p_a\leq Cn^3$ in the $\gamma_{x,y}$'s, for some universal constant $C$. 
Furthermore, each path $\gamma_{x,y}$ has length at most $4n$. With this construction, we therefore write, since 
\[\psi_x-\psi_y=\sum_{a=(a_1,a_2) \in \gamma_{x,y}}(\psi_{a_1}-\psi_{a_2}),\]
and $(\sum_{k=1}^p x_k)^2\leq p\sum_{k=1}^p x_k^2$ that 
\begin{align*}\sum_{x,y\in B_n}[\Ecmnkt\pa{\psi_x}-\Ecmnkt(\psi_y)]^2\leq& \sum_{x,y\in B_n} 4n\sum_{(a_1,a_2)\in\gamma_{x,y}}[\Ecmnkt\pa{\psi_{a_1}}-\Ecmnkt(\psi_{a_2})]^2\\
=&4n\sum_{(a_1,a_2)\subset B_n}p_{a}[\Ecmnkt\pa{\psi_{a_1}}-\Ecmnkt(\psi_{a_2})]^2\\
\leq&4Cn^4 \sum_{(a_1,a_2)\subset B_n}[\Ecmnkt\pa{\psi_{a_1}}-\Ecmnkt(\psi_{a_2})]^2\\
=&4Cn^4 \sum_{\substack{x, x+z\in B_n,\\ |z|=1}}[\Ecmnkt\pa{\psi_{x+z}}-\Ecmnkt(\psi_{x})]^2.\end{align*}
Using the two previous identities, we obtain that 
\begin{equation}\label{expectationbound}\sum_{x\in B_n}\Ecmnkt(\psi_x)^2\leq Cn^2 \sum_{x\in B_n, |z|=1}[\Ecmnkt\pa{\psi_{x+z}}-\Ecmnkt(\psi_{x})]^2,\end{equation}
so that using equations \eqref{newgap}, \eqref{variancebound}, and \eqref{expectationbound}, to prove Proposition \ref{prop:spectralgap} it is enough to show that for some constant $C(\alpha)$,
\begin{multline}
\label{newgap2} 
\sum_{x\in B_n, |z|=1}[\Ecmnkt\pa{\psi_{x+z}}-\Ecmnkt(\psi_{x})]^2\\
\leq \frac{K}{|B_n|} C(\alpha) \sum_{x\in B_n}\cro{\Dnkxt(\psi_x)+\sum_{|z|=1}\Ecmnkt\pa{(1-\conf_{x+z})\cro{\psi_{x+z}\pa{\conf^{x,x+z}}-\psi_x}^2}}.
\end{multline}

\bigskip

\newcommand{\nt}{\widetilde{\nabla}}

Let us denote by $e_{x+z}$ the empty site nearest to $x+z$ other than $x$, chosen arbitrarily if there are multiple candidates. We want to reach from $\conf$ a configuration with an empty site in $x+z$, where the successive jumps will be controlled by the Dirichlet form of the $\psi_x's$, and the resulting difference will be controlled by the second term above.  
To do so, we merely have to "move" the empty site from $e_{x+z}$ to $x+z$, using a path of minimal length.
We denote by $a_1,\ldots ,a_p$ the sequence of edges along which the empty site travels. For any integer $r\leq p$ let $\conf^{(r-1)}=\conf^{a_1\ldots a_r}$ be the configuration where the empty site has traveled along $r$ edges. 
In particular,  $\conf^{(0)}=\conf$, and  $\conf^{(p)}_{x+z}=0$.  Furthermore, for each edge $a_r$ in this sequence, we denote by $a_{r,1}$ the position throughout this construction of the displaced particle at the $r-th$ stage, 
and $a_{r,2}$ the position of the empty site, therefore, $a_r=(a_{r,1},a_{r,2})$. 
One easily sees that if $e_{x+z}\neq x$, we can perform this construction with the following conditions satisfied.
\begin{enumerate}[1)] 
 \item The path followed by the empty site contains at most $p(e_{x+z})\leq 2\abs{e_{x+z}-x}$ jumps. 
\item None of the edges $a_r$ connects $x$ and one of its neighbors.
\item The only edge linking $x+z$ to one of its neighbor is the last edge $a_p$, and it is of the form $a_p=(x+z, x+z+z')$, with $z $ and $z'$ orthogonal.  In other words, we assume that the empty site comes from the direction orthogonal to the direction of the edge $(x,x+z)$.
\end{enumerate}

\medskip

With this construction, for any function $h$, since every successive jump is allowed (each initial site is occupied, each end site is empty) we have 
\begin{multline*}
\pa{1-{\conf^{(p)}_{x+z}}}h\pa{\conf^{(p)}}=h\pa{\conf^{(p)}}=h(\conf)+\sum_{r=1}^p\pa{h\pa{\conf^{(r-1)}}-h\pa{\conf^{(r-1)}}}\\
=h(\conf)+\sum_{r=1}^p\conf^{(r-1)}_{a_{r,1}}(1-\conf^{(r-1)}_{a_{r,2}})\nt_{a_r}h\pa{\conf^{(r-1)}},
\end{multline*}
where $\nt_{a}f=f(\eta^{a_1,a_2})-f(\eta)$. We can rewrite this identity 
\[h(\conf)=\pa{1-{\conf^{(p)}_{x+z}}}h\pa{\conf^{(p)}}-\sum_{r=1}^p\conf^{(r-1)}_{a_{r,1}}(1-\conf^{(r-1)}_{a_{r,2}})\nt_{a_r}h\pa{\conf^{(r-1)}}.\]
Note that in the formula above, both $p$ and the path $\conf^{(r-1)}$ depends on the position of $e_{x+z}$. 

\medskip

We not let $h(\eta)=\psi_{x+z}(\conf^{x,x+z})-\psi_x$. This function vanishes if there is an empty site in $x$, which is the only case for which the construction above 
does not hold (because in particular the empty site cannot avoid the edges surrounding $x$). Using the construction above, we obtain
\begin{align*}\Ecmnkt\pa{\psi_{x+z}}-\Ecmnkt(\psi_{x})=&\Ecmnkt\pa{\psi_{x+z}(\conf^{x,x+z})}-\Ecmnkt(\psi_{x})\\
=&-\Ecmnkt\pa{\sum_{r=1}^p\conf^{(r-1)}_{a_{r,1}}(1-\conf^{(r-1)}_{a_{r,2}})\nt_{a_r}\cro{\psi_{x+z}((\conf^{(r-1)})^{x,x+z})-\psi_x(\conf^{(r-1)})}}\\
&+\Ecmnkt\pa{\pa{1-\conf^{(p)}_{x+z}}\cro{\psi_{x+z}((\conf^{(p)})^{x,x+z})-\psi_x(\conf^{(p)})}}.
\end{align*}
We now project on the possible positions for $e_{x+z}$, by Cauchy-Schwarz inequality, and since $(\sum_{i=1}^p a_i)^2\leq p\sum_{i=1}^p a_i^2$, we obtain
\begin{multline}\label{eq:quadts}
\Big|\Ecmnkt\pa{\psi_{x+z}}-\Ecmnkt(\psi_{x})\Big|\leq \sum_{e\in B_n\setminus\{x\}}\sqrt{(2p(e)+1)\widetilde{\mu}_{n,K}\big(e_{x+z}=e, \eta_x=1\big)}\\
\times \Bigg[\Ecmnkt\pa{\1_{\{e_{x+z}=e, \eta_x=1\}}\pa{1-\conf^{(p(e))}_{x+z}}\cro{\psi_{x+z}((\conf^{(p(e))})^{x,x+z})-\psi_x(\conf^{(p(e))})}^2}\\
+\sum_{r=1}^{p(e)}\Ecmnkt\pa{\1_{\{e_{x+z}=e, \eta_x=1\}}\conf^{(r-1)}_{a_{r,1}}(1-\conf^{(r-1)}_{a_{r,2}})\cro{\nt_{a_r}\psi_{x+z}((\conf^{(r-1)})^{x,x+z})}^2}\\
+\sum_{r=1}^{p(e)}\Ecmnkt\pa{\1_{\{e_{x+z}=e, \eta_x=1\}}\conf^{(r-1)}_{a_{r,1}}(1-\conf^{(r-1)}_{a_{r,2}})\cro{\nt_{a_r}\psi_x(\conf^{(r-1)})}^2}\Bigg]^{1/2}.
\end{multline}
We now estimate each of the three terms in the bracket.

\medskip

The empty site $e$ being fixed, the sequence of edges $(a_r)$ and its length $p$ are also fixed. The first term in the bracket can therefore be rewritten, thanks the one-to-one change of variables $\conf^{(p-1)}\mapsfrom\conf$
\begin{multline*}
\Ecmnkt\pa{\1_{\{e_{x+z}=e, \eta_x=1\}}(\eta')\pa{1-\conf_{x+z}}\cro{\psi_{x+z}(\conf^{x,x+z})-\psi_x(\conf)}^2}\\
\leq \Ecmnkt\pa{\pa{1-\conf_{x+z}}\cro{\psi_{x+z}(\conf^{x,x+z})-\psi_x(\conf)}^2}, 
\end{multline*}
where $\eta'$ denotes the invert change of variable $\conf\mapsfrom\conf^{(p-1)}$.
Since none of the edges $a_r$ connects $x$ to one of its neighbors, and since each edge is used at most once, one-to-one changes of variable $\conf^{(r-1)}\mapsfrom\conf$ also allow us to crudely estimate
\begin{multline*}
\sum_{r=1}^{p}\Ecmnkt\pa{\1_{\{e_{x+z}=e, \eta_x=1\}}\conf^{(r-1)}_{a_{r,1}}(1-\conf^{(r-1)}_{a_{r,2}})\cro{\nt_{a_r}\psi_x(\conf^{(r-1)})}^2}\\
=\sum_{r=1}^{p}\Ecmnkt\pa{\1_{\{e_{x+z}=e, \eta_x=1\}}(\eta'^{(r)})\conf_{a_{r,1}}(1-\conf_{a_{r,2}})\cro{\nt_{a_r}\psi_x(\conf)}^2}\leq \Dnkxt(\psi_x). 
\end{multline*}

\medskip

Finally, for the third contribution, we can write the same estimate, except for the last gradient which is over an edge $(a_{p,1},a_{p,2})=(x+z,x+z+z')$, with $|z'|=|z|=1$. 
We therefore write
\begin{multline*}
\sum_{r=1}^{p}\Ecmnkt\pa{\1_{\{e_{x+z}=e, \eta_x=1\}}\conf^{(r-1)}_{a_{r,1}}(1-\conf^{(r-1)}_{a_{r,2}})\cro{\nt_{a_r}\psi_{x+z}((\conf^{(r-1)})^{x,x+z})}^2}\\
\leq\Dnkxzt(\psi_{x+z})+\Ecmnkt\pa{\1_{\{e_{x+z}=e, \eta_x=1\}}\conf^{(p-1)}_{a_{p,1}}(1-\conf^{(p-1)}_{a_{p,2}})\cro{\nt_{a_p}\psi_{x+z}((\conf^{(p-1)})^{x,x+z})}^2}\\
\leq\Dnkxzt(\psi_{x+z})+\Ecmnkt\pa{\conf_{x+z}(1-\conf_{x+z+z'})\cro{\psi_{x+z}\pa{\pa{\conf^{x+z, x+z+z'}}^{x,x+z}}-\psi_{x+z}\pa{\conf^{x,x+z}}}^2}.
\end{multline*}
One easily obtains that $\conf^{x,x+z+z'}=\pa{\pa{\conf^{x,x+z}}^{x+z, x+z+z'}}^{x,x+z}$, therefore performing the change of variable $ \conf^{x,x+z}\mapsfrom \conf$ in the bound above yields
\begin{multline*}
\sum_{r=1}^{p}\Ecmnkt\pa{\1_{\{e_{x+z}=e, \eta_x=1\}}\conf^{(r-1)}_{a_{r,1}}(1-\conf^{(r-1)}_{a_{r,2}})\cro{\nt_{a_r}\psi_{x+z}((\conf^{(r-1)})^{x,x+z})}^2}\\
\leq \Dnkxzt(\psi_{x+z})+\underset{\leq 2\Ecmnkt((\nabla_{x,x+z'}\psi_{x+z})^2)+2\Ecmnkt((\nabla_{x+z',x+z+z'}\psi_{x+z})^2)}{\underbrace{\Ecmnkt\pa{\conf_{x}(1-\conf_{x+z+z'})\cro{\psi_{x+z}\pa{\conf^{x, x+z+z'}}-\psi_{x+z}\pa{\conf}}^2}}}\leq 3\Dnkxzt(\psi_{x+z}),
\end{multline*}
where we used that $z'$ and $z$ are orthogonal by assumption, which means that the gradients in the last term are not of the form $(x+z, x+z+z'')$.
We now use these three bounds in \eqref{eq:quadts}, to obtain that for some universal constant $C_3$
\begin{multline*}
\Big(\Ecmnkt\pa{\psi_{x+z}}-\Ecmnkt(\psi_{x})\Big)^2\leq C_3 \pa{\sum_{e\in B_n\setminus\{x\}}\sqrt{(1+2p(e))\widetilde{\mu}_{n,K}(e_{x+z}=e, \eta_x=1)}}^2\\
\times\Bigg[\Ecmnkt\pa{\pa{1-\conf_{x+z}}\cro{\psi_{x+z}(\conf^{x,x+z})-\psi_x(\conf)}^2}+\Dnkxt(\psi_{x})+\Dnkxzt(\psi_{x+z})\Bigg].
\end{multline*}
Since we assumed $K\leq \alpha|B_n|$, for $\alpha<1$ one straightforwardly obtains by elementary computations that 
\begin{align*}
\sum_{e\in B_n\setminus\{x\}}\sqrt{(1+2p(e))\widetilde{\mu}_{n,K}(e_{x+z}=e, \eta_x=1)}&\leq  \sqrt{\frac{K}{|B_n|} C(\alpha)},
\end{align*}
therefore \eqref{newgap2} holds as desired. This concludes the proof of Proposition \ref{prop:spectralgap}.
}

\subsection{Discrete differential forms in the context of particles systems}
\label{subsec:differentialforms}
\intro{We  introduce in this section the concept of discrete differential forms in the context of particle systems. 
The key point of the non-gradient method is that any translation-invariant closed form can be decomposed as the sum of a 
gradient of a translation-invariant function and the currents. This result is stated in Proposition \ref{prop:GCF}, and  
directly rewrites as an approximation (in the sense of equation \eqref{STcov}) of any function in $\tzero$ by a linear combination 
of the currents up to an element of $\gene \mathcal{C}$. 
}

\bigskip

\index{$\statespaceinf $\dotfill set of configurations on $\Z^2$}
\index{$ {\mathcal G}$\dotfill the graph $(\statespaceinf,E)$}
\index{$ E$\dotfill set of edges $(\confhat,\confhat^{x,x+z})$, $\conf_x=1$, $\conf_{x+z}=0$}
Let us denote by $\statespaceinf$ the set of configurations on $\Z^2$ 
 \[\statespaceinf=\left\{(\conf_x,\theta_x)_{x\in \Z^2}\in (\{0,1\}\times \ctoruspi)^{\Z^2} \; \;\big| \;\; \theta_x=0\mbox{ if }\conf_x=0\right\}.\]
 We consider here the graph ${\mathcal G}=(\statespaceinf,E)$ with oriented edge set 
\begin{equation}\label{defedge}E=\left\{(\confhat,\confhat')\in\statespaceinf^2\;  \mid\; \;\confhat'=\confhat^{x, x+z}\mbox{ for some }x\in \Z^2, |z|=1\mbox{ and }\conf_x(1-\conf_{x+z})=1\right\}.\end{equation}
In other words, there is an edge from $\confhat$ to $\confhat'$ if and only if the latter can be reached from the former with 
exactly one licit particle jump (i.e. the jump of a particle to an \emph{empty} site).
We endow ${\mathcal G}$ with the usual distance $d$ on graphs, i.e. $d(\confhat, \confhat')$ is the minimal number of particle jumps 
necessary to go from one configuration to the other. Note that this graph is not connected, since for example the configuration 
$\confhat$ with no particles is not accessible from any configuration $\confhat'$ with any number of particles. 
This is also the case for two configurations with different angle distributions. In such a case where there is no path between 
$\confhat'$ and  $\confhat$, we will adopt the usual convention $d(\confhat, \confhat')=\infty$.
By abuse of notation, we also denote by $\mu_{\param}$ (cf. Definition \ref{defi:GCM})  the grand-canonical measure measure on 
$\Z^2$ with parameter $\param$, and write $\Egcm(\cdot)$ for the expectation w.r.t $\mu_{\param}$.

\medskip

We call \emph{differential form} on $({\mathcal G},d)$ a collection of $L^2(\mesinv)$ variables associated with each edge in $E$. More precisely, it is a collection 
$\ufbar=(\ufbar_{x,x+z})_{x\in \Z^2, |z|=1}$, satisfying
\[\ufbar_{x,x+z}(\confhat)=\conf_x(1-\conf_{x+z})\ufbar_{x,x+z}(\confhat)\in L^2(\mesinv).\]
This definition arbitrarily attributes to $\ufbar_{x,x+z}(\confhat)$ the value $0$ if $\conf_x(1-\conf_{x+z})$ vanishes (i.e. if the jump from $x$ to $x+z$ cannot be performed in $\confhat$), which is just a notation shortcut to define $\ufbar$ on all configurations rather than only on those such that $\conf_x(1-\conf_{x+z})=1$. 
Another way to look at these objects is that with each possible particle jump in a configuration $\confhat$ is associated a weight. 
In this section, we will only consider \emph{closed forms}, 
i.e. differential forms for which the added weight of any finite-length path (composed only of licit jumps, i.e. jumps from $x$ to $x+z$ with $x$ occupied and $x+z$ empty) 
between two configuration does not depend on the path chosen but only on the two endpoints. {Equivalently, closed forms are those for which the integral over a closed loop of licit jumps vanishes.} 

{We call \emph{path} a finite sequence of jumps coordinates $\gamma=(x_i, x_i+z_i)_{0\leq i\leq q_\gamma}$, where the $x_i$'s are in $\Z^2$, and $|z_i|=1$. Given a configuration $\confhat$, we denote $\Gamma(\confhat)$ (resp. $\Gamma_c(\confhat)$)
the set of \emph{licit paths} (resp.  \emph{licit loops}, i.e. licit closed paths) such that all successive jumps in the path are licit starting from $\confhat$, (resp. and such that the configuration reached at the end of the sequence of jumps is $\confhat$)
\[\Gamma(\confhat)=\{\gamma=(x_i, x_i+z_i)_{0\leq i\leq q_\gamma} \;\big|\;\; \confhat^{(i, \gamma)}_{x_i}(1-\confhat^{(i,\gamma)}_{x_i+z_i})=1,\;0\leq i\leq q_\gamma\},\]
(resp. $\Gamma_c(\confhat)=\{\gamma=(x_i, x_i+z_i)_{0\leq i\leq q_\gamma}\in \Gamma(\confhat) \; \mid\;\; \confhat^{(q_\gamma+1, \gamma)}=\confhat\;\}$,)
where for any path $\gamma$, and any configuration $\confhat$, we denote $\confhat^{(0,\gamma)}=\confhat$, and $\confhat^{(i+1, \gamma)}=\pa{\confhat^{(i, \gamma)}}^{x_i, x_i+z_i}$ for $0\leq i\leq q_\gamma $.
For any differential form $\ufbar=(\ufbar_{x,x+z})_{x\in \Z^2, |z|=1}$, and any finite path $\gamma$, we denote by
\[I_{\gamma, \ufbar}(\confhat)=\1_{\{\gamma\in \Gamma(\confhat)\}}\sum_{0\leq i\leq q_\gamma }\ufbar_{x_i, x_i+z_i}(\confhat^{(i,\gamma)}),\]
the random variable representing the integral of $\ufbar$ along the path $\gamma$. We assign for convenience the value $0$ to the integral if one of the jumps in the path was not licit.
}


\begin{defi}[Closed and exact forms on $({\mathcal G},d)$]
{A differential form $\ufbar=(\ufbar_{x,x+z})_{x\in \Z^2, |z|=1}$ is \emph{closed} if for any finite path $\gamma$, 
\[\1_{\{\gamma\in \Gamma_c(\confhat)\}}I_{\gamma, \ufbar}(\confhat)=0 \; \;\mesinv-a.s.,\]
i.e. if its integral along any \emph{closed} loop vanishes a.s.. Note that we require the above to hold for any finite path, but for non-closed path the indicator function vanishes. The reason for defining closed forms this way is that closedness of a finite path is a random property that also depends on the configuration, not only on the jump succession.
}

For any cylinder function $f\in \mathcal{C}$, we {say that $\ufbar^f$ is \emph{an exact differential form associated with $f$} if}
\[\ufbar^f_{x,x+z}(\confhat)={\conf_{x}}(1-{\conf_{x+z}})(f(\confhat^{x,x+z})-f(\confhat))\]
{a.s.}. It is easily checked that for any $f\in \mathcal{C}$, $\ufbar^f$ is a \emph{closed form}, {since then
\begin{equation}
\label{eq:Igamma}
I_{\gamma, \ufbar^f}(\confhat)=\1_{\{\gamma\in \Gamma(\confhat)\}}\cro{f(\confhat^{(q_\gamma+1, \gamma)})-f(\confhat)},
\end{equation}
which vanishes a.s. if the loop is closed.}
\end{defi}


We now consider the case of translation invariant closed forms.
\begin{defi}[Germs of a closed form]
\label{defi:GCF}
A pair $\uf=(\uf_{1},\uf_{2}):\statespaceinf\to\R^2$  in $L^2(\mesinv)$ is a \emph{germ of a closed form} if $\ufbar$ defined by 
\begin{equation}
\label{eq:Defubar}
\ufbar_{x,x+ e_i}(\confhat)=\tau_x \uf_i(\confhat) \quad \mbox{ and } \quad \ufbar_{x+e_i,x}(\confhat)=-\tau_x \uf_i(\confhat^{x,x+e_i})=-\ufbar_{x,x+ e_i}(\confhat^{x,x+e_i}) 
\end{equation}
is a closed form.
We endow the set of germs of closed forms with its $L^2(\mesinv)$ norm 
\begin{equation}
\label{L2FF}
\norm{\uf}_{\param,2}=\cro{\Egcm(\uf_{1}^2+\uf_{2}^2)}^{1/2}.
\end{equation}
{Denote by $\tzn$ the closure in $L^2(\mesinv)$ of $T^\omega$ (the set of cylinder functions, defined in \eqref{eq:DefTomega}, depending on the angles through a linear combination of the $\omega(\theta_x)$), and let $\bff{T}^{\omega}=\overline{\bff{T}_0^{\omega}}$ denote the closure in $L^2(\mesinv)$ of the set $\bff{T}_0^{\omega}$ of germs of closed forms with components in $\tzn$, namely
\begin{equation}
\label{frakt}
\bff{T}_0^{\omega}=\Big\{\uf=(\uf_1,\uf_2) \; \;\big| \;\; \uf\mbox{ is a $ L^2(\mesinv) $ germ of a closed form},\quad \uf_i\in \tzn,\quad\forall i\in\{1,2\}\Big\}.
\end{equation}
}

\end{defi}

\begin{defi}[Germs of an exact form]
\label{defi:GEF}
A pair $\uf=(\uf_{1},\uf_{2})$ will be called \emph{germ of an exact form {associated with a cylinder function $h\in \mathcal{C}$}} if we can write 
\[(\uf_{1},\uf_{2})=\boldsymbol{\nabla}\Sigma_h:=(\nabla_{0,e_1}\Sigma_h,\nabla_{0,e_2}\Sigma_h)\]
{pointwise}, where $\Sigma_h$ is the formal sum $\Sigma_h=\sum_{x\in\Z^2}\tau_x h$. Note that although the formal sum 
$\Sigma_h$ is ill-defined a priori, its gradient $\boldsymbol{\nabla}\Sigma_h$ is not, because $h$ is assumed to be a cylinder function, and therefore only depends on a finite number of sites.

One easily verifies that any \emph{germ of an exact form} is also the \emph{germ of a closed form}.
In particular, for any function $h\in  { T^\omega}$, (cf. \eqref{eq:DefTomega}), we have $\boldsymbol{\nabla}\Sigma_h\in \bff{T}^{\omega}$.
We denote by $\mathfrak E^{\omega}=\overline{\mathfrak E^{\omega}_0}$ the closure in $L^2(\mesinv)$ of the set $\mathfrak E^{\omega}_0$ of germs of exact forms {associated with functions in} ${T^\omega}$,
\[\mathfrak E^{\omega}_0=\{\boldsymbol{\nabla}\Sigma_h,\;{\color{white}\big(}\;  h\in {T^\omega}\}\subset \bff{T}_0^{\omega}.\]
\end{defi}

\begin{defi}[Germs of a closed form associated with the currents] 
\label{defi:GCFc} 
Define $ \curg^1$, $\curg^2$, $\curg^{1,\omega}$, and $\curg^{2,\omega}$ as
\begin{equation}
\label{eq:Defcurg}
\curg^k_{i}(\confhat)=\ind{i=k}\conf_0(1-\conf_{e_i})\quad \mbox{ and } \quad \curg^{k,\omega}_{i}(\confhat)=\ind{i=k}\com_0(1-\conf_{e_i})\quad \mbox{ for }k,i=1,2. 
\end{equation}
These four functions are germs of closed forms, and {can be seen as germs of  "almost" exact forms associated with the formal functions
\[f^k=\sum_{x\in\Z^2}x_k\conf_x\quad \mbox{ and }\quad f^{k, \omega}=\sum_{x\in\Z^2}x_k\com_x,\]
which are not well defined, but for which the gradient along any licit jumps is. Of course, since the functions $f^k$, $f^{k,\omega}$ above are merely formal sums, the $ \curg^k$, $\curg^{k,\omega}$'s are \emph{not} germs of exact forms}.  In other words, the closed form $\bar\curg^k$ associated with the germ $\curg^k$ is equal to $\pm1$ on any edge representing a particle jump in the direction $\pm e_k$, and  the closed form $\bar\curg^{k,\omega}$  associated with $\curg^{k,\omega}$ is equal to $\pm \omega(\theta)$ 
on any edge representing a jump in the direction $\pm e_k$ of a particle with angle $\theta$. We denote by $\boldsymbol{\mathfrak{J}}^{\omega}$ the linear span of the $\curg^k$, $\curg^{k,\omega}$
\[\boldsymbol{\mathfrak{J}}^{\omega}=\left\{\curg^{a,b}:=a_1\curg^1+a_2\curg^2+b_1\curg^{1,\omega}+b_2\curg^{2,\omega}, \quad a\in \R^2, b\in \R^2\right\}\subset \bff{T}_0^{\omega}.\]
\end{defi}

We are now ready to state the main result of this section.
\begin{prop}[Structure of $\bff{T}^{\omega}$]
\label{prop:GCF}
We have the decomposition
\[\bff{T}^{\omega}=\boldsymbol{\mathfrak{J}}^{\omega}\oplus\mathfrak E^{\omega}.\]
\end{prop}
\begin{rema}Note that we can safely assume that the total density $\am$ is in $]0,1[$. If not, the graph $\G$ is trivial since its edge set is empty. This assumption will be made throughout the rest of this section.
\end{rema}

Before turning to the proof of the last proposition, we investigate the case of a finite domain.  {We start by a technical Lemma. Recall that $\mathcal{C}_n$ is the set of functions depending only on sites in $B_n$, and $C^1$ with respect to each $\theta_x$ for $x$ in $B_n$, we denote $T_n^\omega=T^\omega\cap\mathcal{C}_n$, the set of functions depending only on sites in $B_n$, and depending on the angles through a linear combination of the $\omega(\theta_x)$. In order to be as clear as possible, recall that $\param$ is fixed, we denote by $\widehat{T}_n^\omega$ the set of functions a.s. equal to a function in $T_n^\omega$. Note that we need to be cautious because the various forms considered in this section are not explicit and are merely $L^2(\mesinv)$ functions of the infinite configuration. However, once their conditional expectation w.r.t. the sigma-algebra generated by sites in $B_n$, all those forms are, up to modification on a negligible set, in $T^\omega_n$. Since $\omega$ is a smooth function, and was fixed once and for all at the very begining of the proof (cf. \eqref{Hdecomp}), $T_n^\omega$ is actually a finite dimensional vector space, and all the results below are therefore analogous to the ones one would obtain with a finite number of particle types.
\begin{lemm}
\label{lem:CnTom}
For any $n\geq 0, $ $ \widehat{T}_n^\omega$ is closed in $ L^2(\mesinv)$, where $\mesinv$, here, stands for the product measure on $B_n$.
\end{lemm}
\proofthm{Lemma \ref{lem:CnTom}}{Since $\widehat{T}_n^\omega$ is roughly a finite-dimensional subspace of $L^2(\mesinv)$, this result is quite natural, but we detail the proof for the sake of exhaustivity. We need to show that if a sequence of functions $\big(\varphi_k(\eta)+\sum_{x\in B_n}\com_x\psi_{k,x}(\eta)\big)_{k\in \N}$ converges as $k\to\infty$ in $L^2(\mesinv)$ to $f$, then there exists angle-blind functions $\varphi^*$, $\psi_x^*$ such that $f=\varphi^*(\eta)+\sum_{x\in B_n}\com_x\psi^*_{x}(\eta)$ a.s.. Here, the $\varphi_k$, $\psi_{k,x}$, $\varphi^*$ and $\psi^*_{x}$ are angle-blind functions depending only on sites in $B_n$. Denote $\sigma_x\confhat$ the configuration equal to $\confhat$ everywhere in $B_n$ except in $x$ where it is distributed as an independent copy $\confhat_x'=(\eta'_x, \theta'_x)$ with distribution $\param$. Then, we abuse our notation, and also denote $\Egcm$ the expectation taken w.r.t. both $\confhat$ and $\confhat_x'$. 

We can now write 
\[\Egcm\cro{\pa{f(\confhat)-f(\sigma_x \confhat)}^2\1_{\{\eta_x=\eta'_x=1\}}}=
\lim_{k\to\infty}\Egcm\cro{\pa{\omega(\theta_x)-\omega(\theta'_x)^2\psi^2_{k,x}(\eta)\1_{\{\eta_x=\eta'_x=1\}}}}.\]
Now assume that the variance of $\omega(\theta_x)$ w.r.t. $\mesinv$ does not vanish (else, the result obviously holds, because in $L^2(\mesinv)$, $ T_n^\omega$ is the set of angle blind functions), we can write for some constant $C:=C(\omega, \param)$
\[\lim_{k\to\infty}\Egcm\cro{\psi^2_{k,x}(\eta)\mid \eta_x=1}\leq C\Egcm(f^2).\]
In particular, since the set of angle blind configurations in $B_n$ is finite, and since we can assume without loss of generality that $\psi_{k,x}(\eta)$ vanishes if $\eta_x=0$, all the $\psi_{k,x}$ must be bounded, uniformly in $x$, $k$, and $\eta$ by some constant $M$, and therefore remain in a compact set. Up to successive extractions, we can as a consequence assume that each sequences $(\psi_{k,x })_k$ converges uniformly in $\eta$ as $k\to\infty$ to $\psi^*_x$. In particular, the sequence $\varphi_k$ also converges to a function $\varphi^*$, and we can thus write as wanted 
$f=\varphi^*(\eta)+\sum_{x\in B_n}\com_x\psi^*_{x}(\eta)$ a.s..
}
}
We now consider closed differential forms in a finite box. Considering the graph ${\mathcal G}_n$ with vertices the {configurations} $\confhat$ on the box $B_n$, and connected, as on the infinite graph, if one configuration can be reached from another with one licit jump {along an edge of $B_n$}.
\begin{prop}
\label{prop:GCFn}
Fix a  parameter $\param$, $n\geq 0$, and a closed form $\ufbar=(\ufbar_{x,x+z})_{x, x+z\in B_n}$ on  ${\mathcal G}_n$  {satisfying for any $x, x+z\in B_n$
\begin{enumerate}[i)]
\item $\ufbar_{x,x+z}$ identically vanishes when there are $1$ or less empty sites in $B_n$,
\item $\ufbar_{x,x+z}\in T_n^\omega$ and is therefore smooth. 
\end{enumerate}
Then,  there exists a cylinder function $h\in T_n^\omega$} such that  
\[\ufbar_{x,x+z}=\nabla_{x,x+z} h \quad \forall x, x+z\in B_n, \mbox{ {pointwise}},\]
i.e. on a finite set, all closed forms are exact forms. Furthermore, one can assume without loss of generality that for any $\K\in\Kset_n$, $\E_{n, \K}(h)=0$. 
\end{prop}
\proofthm{Proposition \ref{prop:GCFn}}{{
Since $\ufbar$ is a closed form with each element in $T_n^\omega$ (therefore in particular smooth in the angle variables), we have that $1_{\{\gamma\in \Gamma_c(\confhat)\}}I_{\gamma, \ufbar'}$ vanishes  pointwise for any finite path $\gamma$.
Recall that $\ufbar_{x,x+z}$ vanishes if there is one or less empty site in $B_n$, we split the set of configurations on $B_n$ into components $(\Sigma^{\K}_n)_{\K\in \Kset_n}$ each connected on the graph ${\mathcal G}_n$. In particular, for any two configurations $\confhat$, $\confhat'$ in the same $\Sigma^{\K}_n$, we must have by construction $d(\confhat, \confhat')<\infty$. 

For any $\K$ with at least two empty sites, let us denote $\confhat^{\K}$ the configuration where the particles are inserted from the bottom left, row by row, and in the order of increasing angles from $0$ to $2\pi$. In other words, we insert the particle with the angle closest to $0$ at site $(-n, -n)$, the second closest at $(-n, -n+1)$, and so on until all  particles have been placed. The choice of  this reference configuration is arbitrary, but depends continuously in the angles in $\K\in \Ksett_n$. We then set $h(\confhat^{\K})=0$  for each $ \K\in \Ksett_n$, and for any other configuration $\confhat\in \Sigma^{\K}_n$, we fix a path $\gamma_{\confhat}$ of licit jumps from $\confhat^{\K}$ to $\confhat$, and  let
\[h(\confhat)= I_{\gamma_{\confhat},\ufbar'}(\confhat^{\K}).\]
Since $\ufbar$ is a \emph{pointwise} closed form, this expression does not depend on the choice of $\gamma_{\confhat}$  and pointwise, we have for any $x, x+z$, $\ufbar_{x,x+z}=\nabla_{x,x+z}h$. Furthermore, by construction, because both $\ufbar$ and $\confhat^{\K}$ depend smoothly on the particle's angles, so does $h$, and therefore $h\in \mathcal{C}_n$. We now show that $h\in T_n^\omega$.

To do so, we now consider the space $L^2(\mesinv)$, recall that  $\widehat{T}_n^\omega$ is the trace of $T_n^\omega$ in $L^2(\mesinv)$. Since, according to Lemma \ref{lem:CnTom}, $\widehat{T}_n^\omega$ is a closed linear subspace of $L^2(\mesinv)$, we can write on $B_n$ that $L^2(\mesinv)=\widehat{T}_n^\omega\oplus \pa{\widehat{T}_n^\omega}^{\perp}$. Straightforwardly, one can show that both $\widehat{T}_n^\omega$ and $\pa{\widehat{T}_n^\omega}^{\perp}$  are stable under any symmetric gradient $\widetilde{\nabla}_{x,x+z}f:=\1_{\{\conf_x \conf_{x+z}=0\}}(f(\confhat^{x,x+z})-f(\confhat))$, for $x, x+z\in B_n$.  In particular, since $\ufbar_{x,x+z}\in \widehat{T}_n^\omega$, we also have 
$\widetilde\nabla_{x,x+z}h= \ufbar_{x,x+z}(\confhat)+ \ufbar_{x,x+z}(\confhat^{x,x+z})\in \widehat{T}_n^\omega$ for any $x,x+z\in B_n$. Let now write $h$ as $h_1+h_2$, where $h_1\in \widehat{T}_n^\omega$ and $h_2\in (\widehat{T}_n^\omega)^{\perp}$, we must have $\widetilde\nabla_{x,x+z}h=\widetilde\nabla_{x,x+z}h_1$. All gradients of $h_2$ therefore vanish a.s., we conclude that $h_2$ is a.s. constant on each connected component, therefore we can choose it to be $0$ without changing $\widetilde\nabla_{x,x+z}h$. We thus have as wanted $\ufbar'_{x,x+z}=\nabla_{x,x+z}h_1$, we can therefore choose $h=h_1\in \widehat{T}_n^\omega$. Since $h$ is smooth in the angle coordinates, it implies as wanted $h\in T_n^\omega$ in a pointwise sense.

\medskip

Regarding the second claim of the Proposition, given} a configuration $\confhat$ on $B_n$, let us denote by $\K_n(\confhat):=(K(\confhat), \Theta_{K(\confhat)}(\confhat))$ the parameter giving the number and angles of particles 
in $\confhat$, i.e. 
\[K(\confhat)=\sum_{x\in B_n}\eta_x \eqand \Theta_{K(\confhat)}(\confhat)=\left\{\theta_{x_1},\dots,\theta_{x_{K(\confhat)}}\right\},\]
where $x_1,\dots,x_K$ are the positions of the $K$ particles in $\confhat$. Since the function $K(\confhat)$ is unchanged under any gradient inside $B_n$, 
we can replace $h$ by $h_0=h-\E_{n,\K_n(\confhat)}(h)$  (where $\E_{n,\K}$ is the expectation w.r.t. the canonical measure corresponding to 
having $\K$ particles in $B_n$) and still satisfy $ \ufbar_{x,x+z}(\confhat)=\nabla_{x,x+z} h_0(\confhat)$.

}

We now turn to the proof of the decomposition of germs of closed forms on the infinite graph.
\proofthm{Proposition \ref{prop:GCF}}{We first prove that the sum is direct~: assume that for $a,b\in \R^2$, there exists a cylinder function $h$ such that 
$\curg^{a,b}=a_1\curg^1+a_2\curg^2+b_1\curg^{1,\omega}+b_2\curg^{2,\omega}=\boldsymbol{\nabla} \Sigma_h$. In particular fix $i=1,2$, one easily obtains that
\[a_ij_i+b_ij_i^{\omega}=\nabla_{0,e_i}\Sigma_h-\nabla_{0,e_i}\Sigma_h(\confhat^{0,e_i})=\1_{\{\conf_0\conf_{e_i}=0\}}(\Sigma_h(\conf)-\Sigma_h(\conf^{0,e_i})),\]
where the $j_i$'s are the currents defined in \eqref{defi:currents}.
Multiplying by $\conf_{e_i}-\conf_0$ (resp. $\com_{e_i}-\com_0 $) and taking the expectation w.r.t. $\mesinv$, the identity above rewrites 
\[2(a_i+b_i\Egcm(\omega))\alpha(1-\alpha)=0 \quad (\;\mbox{resp.} \quad 2(a_i\Egcm(\omega)+b_i\Egcm(\omega^2))\alpha(1-\alpha)=0),\]
where, as in Section \ref{subsec:spectralgap}, $\Egcm(\omega^k)$ stands for $\Egcm(\omega^k(\theta_0)|\conf_0=1)$.
In particular, since $\alpha\in (0,1)$ this yields that $a_i+b_i\Egcm(\omega)=0$ and that $\Egcm(\omega^2)=\Egcm(\omega)^2$, therefore $\omega(\theta_0)$ is constant under $\mesinv$. In particular, $a_ij_i+b_ij_i^{\omega}$ vanishes in $L^2(\mesinv)$ as wanted.
The inclusion $\bff{T}^{\omega}\supset\boldsymbol{\mathfrak{J}}^{\omega}+\mathfrak E^{\omega}$ is immediate. 

\bigskip

We now prove the reverse inclusion. The set of germs of an exact form being a linear (therefore convex) subset of $L^2(\mesinv)$, its weak and strong closure in $L^2(\mesinv)$ coincide. In order to prove Proposition \ref{prop:GCF}, it is therefore sufficient to prove that for any $\uf\in \bff{T}^\omega$, there exists a sequence of cylinder functions ${h_n\in T^\omega}$ such that the sequence $(\boldsymbol{\nabla}\Sigma_{h_n})_{n\in \N}$ is weakly relatively compact in $L^2(\mesinv)$, and for any of its weak limit points $\bff{h}$, there exists $a$ and $b$ in $\R^2$ such that 
\begin{equation*}\bff{h}=\uf+\curg^{a,b}.\end{equation*} 

Fix $\uf\in\bff{T}^{\omega}$, and $(\ufbar_{x,x+z})_{x, x+z}$ the associated closed form defined by \eqref{eq:Defubar}. For any fixed integer $n$, let ${\mathcal F}_{n}$
be the $\sigma$-algebra generated by the sites inside $B_n$ 
\[{\mathcal F}_n=\sigma\pa{\;\confhat_x, \;x\in B_n\;},\]
and let $\ufbar^n_{x,x+z}$ denote the conditional expectation
\[\ufbar^n_{x,x+z}=\Egcm(\ufbar_{x,x+z}\mid{ \mathcal F}_n).\]
Note in particular that  since  $\uf$ is in $\bff{T}^{\omega}$, $\ufbar^n$ is a closed form on ${\mathcal G}_n$, and each of its coordinate is in $ {\widehat{T}_n^\omega}${, according to Lemma \ref{lem:CnTom}, and because each of the $\ufbar_{x,x+z}$ is the limit in $L^2(\mesinv)$ of a sequence of functions in $T^\omega$. In particular, up to simultaneous modification of all the $\ufbar^n_{x,x+z}$, $x,x+z\in B_n$ on a $\mesinv$-negligible set of configurations, we can safely assume that $\ufbar^n_{x,x+z}\in T^\omega_n$ in a pointwise sense.}

{Fix once and for all a density $\alpha<\alpha'<1$, and define $\rho_n=\frac{1}{|B_n|}\sum_{x\in B_n}\eta_x$ the density in $B_n$, according to Proposition \ref{prop:GCFn},  there exists a  family of ${\mathcal F}_n$-measurable functions $\varphi_n\in T_n^\omega$ with mean $0$ w.r.t. any canonical measure on $B_n$ such that 
\begin{equation*}
\label{def:phin} 
\ind{\rho_n\leq \alpha'}\ufbar^n_{x,x+z}=\nabla_{x, x+z}\varphi_n \quad \forall x,x+z\in B_n\;\; \mesinv-a.s.,
\end{equation*}
where the identity holds only a.s. and not pointwise because we may have modified the $\ufbar^n_{x,x+z}$ on a negligible set. Note that we would need a weaker indicator function to respect the conditions of Proposition \ref{prop:GCFn} (namely, that there are two empty sites in $B_n$) however in order to estimate the $L^2(\mesinv)$-norm of the $\varphi_n$, we will need the stronger indicator function above.}

Let us fix $n\in \N$, and consider the germ of an exact form $\frac{1}{(2n)^2}\boldsymbol{\nabla}\Sigma_{\varphi_n}$, whose coordinates can be rewritten for $i=1,2$
\[\frac{1}{(2n)^2}\nabla_{0,e_i}\Sigma_{\varphi_n}=\frac{1}{(2n)^2}\sum_{x\in \Z^2}\tau_{-x}\nabla_{x,x+e_i}\varphi_n.\]
Since $\varphi_n$ is ${\mathcal F}_n$-measurable, $\nabla_{x,x+e_i}\varphi_n$ vanishes as soon as neither $x$ nor $x+e_i$ is in $B_n$. Hence, the previous quantity is equal to 
\begin{equation}\label{gradientdecomp}\frac{1}{(2n)^2}\nabla_{0,e_i}\Sigma_{\varphi_n}=\frac{1}{(2n)^2}\sum_{\substack{-n-1\leq x_i\leq n\\
x\in B_n}}\tau_{-x}\nabla_{x,x+e_i}\varphi_n=R_{n,i}+\frac{1}{(2n)^2}\sum_{\substack{-n\leq x_i\leq n-1\\
x\in B_n}}\tau_{-x}\nabla_{x,x+e_i}\varphi_n, \end{equation}
where the boundary term $R_{n,i}$ is 
\[R_{n,i}=\frac{1}{(2n)^2}\cro{\sum_{\substack{x_i=-n-1\\
x\in B_n(-e_i)}}\tau_{-x}\nabla_{x,x+e_i}\varphi_n+\sum_{\substack{x_i=n\\
x\in B_n}}\tau_{-x}\nabla_{x,x+e_i}\varphi_n}.\]

For any $n$, the left-hand side in \eqref{gradientdecomp} the germ of an exact form as introduced in Definition \ref{defi:GEF}. We will see that the second term of the right-hand side converges in 
$L^2(\mu_{\param})$ as $n$ goes to infinity towards $\uf_{i}$. Hence to prove Proposition \ref{prop:GCF} it will be sufficient to show that the boundary 
 term $R_{n,i}$ is weakly relatively compact in $L^2(\mesinv)$, and that any of its weak limit points is in $\bff{J}^\omega$. 
 Since $\varphi_n$ is supported in $B_n$, the exchanges at the boundary act as reservoirs with creation (first term in $R_{n,i}$) at the sites $x+e_i$ with $x_i=-n-1$, and annihilation of particles  (second term in $R_{n,i}$) at the sites $x$ such that $x_i=n$, and cannot be expressed as such as particle transfers. 
To prove that the sequence of boundary terms is weakly relatively compact, we therefore need to smooth out the $\varphi_n$'s, by letting 
\begin{equation}
\label{def:phint}
\modphi_n=\Egcm(\varphi_{3n}\mid {\mathcal F}_n).
\end{equation}
{Not in particular that we still have $\modphi_n\in T_n^{\omega}$.}

Rewrite \eqref{gradientdecomp} with $\modphi_n$ instead of $\varphi_n$
\begin{equation}
\label{decompGFFlemme}
\frac{1}{(2n)^2}\nabla_{0,e_i}\Sigma_{\modphi_n}=\frac{1}{(2n)^2}\sum_{\substack{-n\leq x_i\leq n-1\\
x\in B_n}}\tau_{-x}\nabla_{x,x+e_i}\modphi_n+\widetilde{R}_{n,i}, \end{equation} 
where this time 
\begin{equation}
\label{rnitilde}
\widetilde{R}_{n,i}=\frac{1}{(2n)^2}\cro{\sum_{\substack{x_i=-n-1\\
x\in B_n(-e_i)}}\tau_{-x}\nabla_{x,x+e_i}\widetilde{\varphi}_n+\sum_{\substack{x_i=n\\
x\in B_n}}\tau_{-x}\nabla_{x,x+e_i}\widetilde{\varphi}_n}.\end{equation}
We are going to show that 
\begin{itemize}
\item the bulk term converges in $L^2(\mesinv)$ to $\uf_{i}$,
\item the sequence of boundary term is bounded in $L^2(\mesinv)$, and any of its weak limit points is an element of $\bff{J}^\omega$.
\end{itemize}
For the sake of clarity, we state both of these results as separate lemmas, and we will prove them afterwards.

%

\begin{lemm}[Convergence of the bulk term towards $\uf_{i}$]
\label{lem:bulkconvergence}
For any $i\in \{1,2\}$,
\begin{equation}
\label{eq:blukL2}
\limsup_{n\to\infty}\Egcm\cro{\pa{\frac{1}{(2n)^2}\sum_{\substack{-n\leq x_i\leq n-1\\
x\in B_n}}\Big[\tau_{-x}\nabla_{x,x+e_i}\modphi_n-\uf_{i}\Big]}^2}=0.
\end{equation}
\end{lemm}
\begin{lemm}[Limit of the boundary term]
\label{lem:boundaryconvergence}
For any $i\in \{1,2\}$, we split the boundary term as
\[\widetilde{R}_{n,i}=\widetilde{R}^-_{n,i}+\widetilde{R}^+_{n,i},\]
where
\begin{equation}\label{BoundaryDef}\widetilde{R}^-_{n,i}=\frac{1}{(2n)^2}\sum_{\substack{x_i=-n-1\\
x\in B_n(-e_i)}}\tau_{-x}\nabla_{x,x+e_i}\modphi_n,\eqand  \widetilde{R}^+_{n,i}=\frac{1}{(2n)^2}\sum_{\substack{x_i=n\\
x\in B_n}}\tau_{-x}\nabla_{x,x+e_i}\modphi_n,\end{equation}
which will be referred to respectively as negative and positive boundary terms. 
With the previous notations, both sequences $(\widetilde{R}^-_{n,i})_{n\in \N}$ and $(\widetilde{R}^+_{n,i})_{n\in \N}$ are bounded in $L^2(\mesinv)$. 
Furthermore, for any weakly convergent subsequence $\widetilde{R}^-_{k_n,i}\to \mathfrak{R}_i^-$, there exists $a_i,b_i\in \R$ such that 
\[ \mathfrak{R}_i^-=a_i\com_0(1-{\conf_{e_i}})+b_i\conf_0(1-{\conf_{e_i}}).\]
The same is true for the positive boundary term.
\end{lemm}

Thanks to \eqref{decompGFFlemme}, these two lemmas prove Proposition \ref{prop:GCF}.
}

The proof of Lemma \ref{lem:bulkconvergence} is simple, we treat it right now before turning to the proof of Lemma \ref{lem:boundaryconvergence}, which is the main difficulty of this section.
\proofthm{Lemma \ref{lem:bulkconvergence}}{By construction, for any $x, x+e_i\in B_n$,  
\begin{align}
\label{eq:estBn}
\nabla_{x,x+e_i}\modphi_n&=\nabla_{x,x+e_i}\Egcm(\varphi_{3n}\mid{ \mathcal F}_n)\nonumber\\
&=\Egcm(\nabla_{x,x+e_i}\varphi_{3n}\mid{ \mathcal F}_n)\nonumber\\
&=\Egcm(\ind{\rho_{3n}\leq \alpha'}\ufbar^{3n}_{x,x+e_i}\mid{ \mathcal F}_n)\nonumber\\
&=\Egcm(\ind{\rho_{3n}\leq \alpha'}\Egcm(\ufbar_{x,x+e_i}\mid{ \mathcal F}_{3n})\mid{ \mathcal F}_n)\nonumber\\
&=\Egcm(\ind{\rho_{3n}\leq \alpha'}\ufbar_{x,x+e_i}\mid{ \mathcal F}_n).
\end{align} 
By triangular inequality, translation invariance of $\mesinv$, and using $(\sum_{i=1}^k a_i)^2\leq k\sum_{i=1}^ka_i^2$, we can bound the expectation in \eqref{eq:blukL2} by
\begin{equation}
\label{eq:splitbulk}
\frac{1}{2n^2}\sum_{\substack{-n\leq x_i\leq n-1\\
x\in B_n}}\bigg(\Egcm\cro{\big(\Egcm(\ufbar_{x,x+e_i}\mid{ \mathcal F}_n)- \ufbar_{x,x+e_i}\big)^2}+ \Egcm\cro{\ind{\rho_{3n}> \alpha'}\ufbar_{x,x+e_i}^2}\bigg).
\end{equation}
We start by estimating the contribution of the first expectation in the sum. To do so, split it for any positive $ \varepsilon$ as
\begin{equation*}\frac{1}{2n^2}\sum_{x\in B_{n(1-\varepsilon)}}\Egcm\cro{\big(\Egcm(\ufbar_{x,x+e_i}\mid{ \mathcal F}_n)- \ufbar_{x,x+e_i}\big)^2}
+\frac{1}{2n^2}\sum_{\substack{-n\leq x_i\leq n-1\\
x\in B_n\setminus B_{n(1-\varepsilon)}}}\Egcm\cro{\big(\Egcm(\ufbar_{x,x+e_i}\mid{ \mathcal F}_n)- \ufbar_{x,x+e_i}\big)^2}
\end{equation*}
By definition of $\ufbar$, $\tau_x\uf_i=\ufbar_{x,x+e_i}$, thus for any $\varepsilon >0$, the expectations in the first term vanish uniformly in $ x\in B_{n(1-\varepsilon)}$ as $n\to\infty$ by martingale convergence theorem, whereas the second sum can be crudely estimated by Jensen inequality and is less than
\[C\varepsilon\max_{\substack{-n\leq x_i\leq n-1\\
x\in B_n\setminus B_{n(1-\varepsilon)}}}\Egcm\cro{\pa{\Egcm(\ufbar_{x,x+e_i}\mid{ \mathcal F}_n)- \ufbar_{x,x+e_i}}^2}\leq 4C\varepsilon\Egcm(\uf_{i}^2)
\]
which vanishes as $\varepsilon\to 0$ regardless of $n$. 

We now consider the contributions of the second part in \eqref{eq:splitbulk}. That each term vanishes is a direct consequence of the dominated convergence theorem, however since we need a convergence that is uniform in $x$, we give a more detailed and quantitative argument. We can rewrite by translation invariance of $\mesinv$, for any $x, x+e_i\in B_n$, and for any $p<2$
\begin{align*}
\Egcm\cro{\ind{\rho_{3n}> \alpha'}\ufbar_{x,x+e_i}^2}&=\Egcm\cro{\uf_i^2(\tau_{-x}\ind{\rho_{3n}> \alpha'})}\\ 
&\leq \Egcm\cro{\abs{\uf_i^2-|\uf_i|^p}}+\Egcm\cro{|\uf_i|^p(\tau_{-x}\ind{\rho_{3n}> \alpha'})}\\
&\leq \Egcm\cro{\abs{\uf_i^2-|\uf_i|^p}}+\Egcm\pa{\uf_i^2}^{p/2}\mesinv\pa{\rho_{3n}> \alpha'}^{1-p/2}
\end{align*}
by Holder inequality. By a standard large deviation estimate, $\mesinv\pa{\rho_{3n}> \alpha'}=O(e^{-Cn^2})$. We then choose $p=p(n)=2-1/n$, to obtain that second term in the right-hand side above is less than $C(\uf_i)e^{-Cn}$. The function inside the expectation in the first term is pointwise less than $\max(2\uf^2_i,1)$ which is integrable and the first term therefore vanishes by dominated convergence as $p(n)\to 2$. Since the bound above does not depend on $x$, we finally obtain 
\begin{equation}
\lim_{n\to\infty}\frac{1}{2n^2}\sum_{\substack{-n\leq x_i\leq n-1\\
x\in B_n}}\Egcm\cro{\ind{\rho_{3n}> \alpha'}\ufbar_{x,x+e_i}^2}=0
\end{equation}
as wanted, which proves Lemma \ref{lem:bulkconvergence}.
}

\proofthm{Lemma \ref{lem:boundaryconvergence}}{The proof of this Lemma being long, we split it into three steps. 
\begin{itemize}
\item We first control the $L^2(\mesinv)$ norm of the $\widetilde{\varphi}_n$'s.
\item Thanks to this control, we prove that the sequence of boundary terms $\widetilde{R}^{\pm}_{n,i}$ is bounded in $L^2(\mesinv)$.
\item Finally, we prove that any weak limit point $\mathfrak{R}_i^{\pm}$ of the boundary term can only depend on the configuration through $\confhat_0$ and $\confhat_{e_i}$, and that they can be written as a combination of the $\curg^i$ and $\curg^{i,\omega}$.
\end{itemize}
The scheme follows closely that of Theorem 4.14 in Appendix 3 of \cite{KLB1999} however adjustments are needed in the second and third step to take into account the presence of the angles.

\step{First}{Control on the $L^2$ norm of the $\varphi_n$'s. }

We proved in Section  \ref{subsec:spectralgap} that, even though we do not have a general spectral gap of order $n^{-2}$, 
we could circumvent this difficulty by staying in a  convenient class of functions linear in the angles and by cutting off 
the large densities. This spectral gap estimate is needed to control the norm of the $\varphi_n'$s. 
This is the reason for limiting the result to closed forms in $\bff{T}^{\omega}$ defined in \eqref{frakt}, 
and for introducing the indicator functions $\ind{\rho_n\leq \alpha'}$. We state this step as a separate lemma for the sake of clarity.
\begin{lemm}\label{lem:normephiu}
There exists a constant $K:=K(\param,  \alpha',\uf)$ such that for any $n\in \N$, 
\[\Egcm(\varphi_n^2)\leq Kn^4,\]
where $\varphi_n$ was introduced in \eqref{def:phin}.
 \end{lemm}
\proofthm{Lemma \ref{lem:normephiu}}{For any $\K\in \Kset_n$, we proved in Proposition \ref{prop:GCFn} that we could assume 
$\E_{n,\K}(\varphi_n)=0$, and thanks to the indicator function $\ind{\rho_n\leq \alpha'}$, $\varphi_n$ vanishes when the density in 
$B_n$ is larger than $\alpha'$, therefore the spectral gap estimate given in Proposition  \ref{prop:spectralgap}, {since $\varphi_n\in T_n^\omega$,} yields
\[\Egcm(\varphi_n^2)=\Egcm(\varphi_n^2\ind{\rho_n\leq \alpha'})\leq C(\param, \alpha')n^2{\mathscr D}_n(\varphi_n),\]
where ${\mathscr D}_n(f)=-\Egcm(f\gene_n f)$ is the Dirichlet  form relative to the symmetric exclusion process restricted to $B_n$,
\[{\mathscr D}_n(\varphi_n)=\frac{1}{2}\sum_{i=1}^2\sum_{\delta\in \{-1,1\}}\sum_{x, x+\delta e_i\in B_n}\Egcm\cro{(\nabla_{x,x+\delta e_i}\varphi_n)^2}.\]

By construction (cf. \eqref{def:phin}), $\nabla_{x,x+e_i}\varphi_n=\ind{\rho_n\leq \alpha'}\ufbar^n_{x,x+e_i}$ and $\nabla_{x+e_i,x}\varphi_n=-\ind{\rho_n\leq \alpha'}\ufbar^n_{x,x+e_i}(\confhat^{x,x+e_i})$. Thus, since $\uf$ is in $L^2(\mesinv)$, and since $\mesinv$ is invariant under the change of variables $\confhat\mapsto\confhat^{x,x+e_i}$, Jensen's inequality yields
\begin{equation}\label{Dirichletbound}{\mathscr D}_n(\varphi_n)\leq \sum_{i=1}^2\sum_{x, x+e_i\in B_n}\Egcm\cro{(\ufbar^n_{x,x+e_i})^2}\leq \sum_{i=1}^2\sum_{x, x+e_i\in B_n}\Egcm\cro{(\uf_{i})^2}\leq  C'(\uf) n^2.\end{equation}
We obtain as wanted, thanks to the spectral gap estimate above,  
\begin{equation}\label{l2bound}\Egcm(\varphi_n^2)\leq K n^4,\end{equation}
where $K=CC'$ depends only on $ \param$, $\alpha'$, and $\uf$.
}

\step{Second}{Control on the $L^2$ norm of the boundary terms. }

We now prove thanks to Lemma \ref{lem:normephiu} that the boundary terms are bounded in $L^2(\mesinv)$.

\begin{lemm}
\label{lem:L2boundaryterm}
There exists a constant $C=C(\param, \alpha', \uf)$ such that for any $n$,
\begin{equation}
\label{eq:boundboundary}
\Egcm\cro{\big(\widetilde{R}^-_{n,i}\big)^2}\leq C,
\end{equation}
The statement remains true if $\widetilde{R}^-_{n,i}$ is replaced by $\widetilde{R}^+_{n,i}$.
\end{lemm}
\proofthm{Lemma \ref{lem:L2boundaryterm}}{We will treat in full detail only the case of the negative boundary term 
\[\widetilde{R}^-_{n,i}=\frac{1}{(2n)^2}\sum_{\substack{x_i=-n-1\\
x\in B_n(-e_i)}}\tau_{-x}\nabla_{x,x+e_i}\modphi_n,\]
analogous arguments yield the bound for $\widetilde{R}^+_{n,i}$.
Using $(\sum_{i=1}^k a_i)^2\leq k\sum_{i=1}^k a_i^2$, we obtain
\[\Egcm\cro{\big(\widetilde{R}^-_{n,i}\big)^2}\leq\frac{(2n+1)}{(2n)^{4}}\sum_{\substack{x_i=-n-1\\
x\in B_n(-e_i)}}\Egcm\cro{(\tau_{-x}\nabla_{x,x+e_i}\modphi_n)^2}\leq Cn^{-3}\sum_{\substack{x_i=-n-1\\
x\in B_n(-e_i)}}\Egcm\cro{(\nabla_{x,x+e_i}\modphi_n)^2},\]
for some universal constant $C$, by translation invariance of $\mesinv$. For $x$ in the negative boundary, under $\mesinv$, we can rewrite
\begin{equation}\label{gradphicf}\nabla_{x,x+e_i}\modphi_n(\confhat)=\conf_x(1-{\conf_{x+e_i}})\pa{\modphi_n(\confhat+\delta^{\theta}_{x+e_i})-\modphi_n(\confhat)},\end{equation}
where $\confhat+\delta^{\theta}_{x+e_i}$ is the configuration equal to $\confhat$ everywhere except in $x+e_i$,  where the site contains a particle with angle $\theta$ 
distributed as $\param/\alpha$ independently of $\confhat$. Note that in the expectation $\Egcm$, we will also take the expectation w.r.t. $\theta$, but still denote it  
$\Egcm$ not to burden the notations. Since $\varphi_n$ is independent of $\confhat_{x}$ for any $x$ in the negative boundary term, 
\begin{equation}\label{eqbound}\Egcm\cro{\big(\widetilde{R}^-_{n,i}\big)^2}\leq \alpha Cn^{-3}\sum_{\substack{x_i=-n-1\\
x\in B_n(-e_i)}}\Egcm\Big[(1-{\conf_{x+e_i}})\big(\modphi_n(\confhat+\delta^{\theta}_{x+e_i})-\modphi_n(\confhat)\big)^2\Big],\end{equation}
where the expectation w.r.t $\theta$ is also taken, under the distribution $\param/\alpha$. Recall that 
$\modphi_n=\Egcm(\left.\varphi_{3n}\;\right\vert\;{\mathcal F}_n)$, since the number of terms in the sum is $O(n)$, Lemma \ref{lem:L2boundaryterm} is a consequence of Lemma \ref{lem:L2boundphi} below.
}
\begin{lemm}
\label{lem:L2boundphi}
There exists a constant  $C=C(\param, \alpha', \uf)$ such that for any $x\in B_n(-e_i)$ such that $x_i=-n-1$, 
\[\Egcm\cro{(1-{\conf_{x+e_i}})\pa{\Egcm(\left.\varphi_{3n}\right\vert{\mathcal F}_n)(\confhat+\delta^{\theta}_{x+e_i})-\Egcm(\left.\varphi_{3n}\right\vert{\mathcal F}_n)(\confhat)}^2}\leq C n^2,\]
where the expectation above is taken w.r.t. $\mesinv$ on $B_{3n}$ and w.r.t. $\theta$ distributed under $\param/\alpha$.

\end{lemm}

\proofthm{Lemma \ref{lem:L2boundphi}}{Let us fix $x$, such that $x_i=-n-1$ in the negative boundary. To make the Dirichlet form appear, we are going to force an occupied site  in a neighborhood of $x$, and  transform the particle creation into a particle transfer. This is the reason for smoothing out $\varphi_n$ and taking $\modphi_n$ instead. For the sake of clarity, any configuration $\confhat$ on $B_{3n}$ will be considered as the pair of an interior configuration $\confihat$ on $B_n$ (which is hence ${\mathcal F}_n$-measurable), and an exterior configuration $\confehat$ on $ B_{3n}\setminus B_n$. 

For any $y\in B_{3n}\setminus B_n$, we rewrite using the identity $(1-\alpha)^{-1}[1-{\confe}+{\confe}-\alpha]=1$
\[\Egcm(\left.\varphi_{3n}\right\vert{\mathcal F}_n)\pa{\confihat+\delta^{\theta}_{x+e_i}}=\frac{1}{1-\alpha}\Big(\Egcm\big(\left.(1-{\confe_y})\varphi_{3n}\;\right\vert\;{\mathcal F}_n\big)+\Egcm\big(\left.({\confe_y}-\alpha)\varphi_{3n}\;\right\vert\;{\mathcal F}_n\big)\Big)\pa{\confihat+\delta^{\theta}_{x+e_i}},\]
where $\confe_y$ is the occupation variable in $y$, and is either $1$ or $0$ depending on whether the site $y$ is empty or not.

The first part of this decomposition will be controlled by the Dirichlet form, as the existence of an empty site in $y$ (thanks to $1-\confe_y$) will allow us to reconstruct a particle transfer from $y$ to $x+e_i$. The second term will be estimated after a spatial averaging over a large microscopic box. This box must be measurable with respect to the sites in $B_{3n}\setminus B_n$, in order to be able to introduce it inside the expectation. For any $x$ in the negative boundary, consider the set \[B_{n-1,i}^x=x-ne_i+B_{n-1},\]
which is the box of radius $n-1$ centered in $x-ne_i$. Remark that the cardinal of $B_{n-1,i}^x$ is $(2n-1)^2$, so that averaging the previous identity over the $y$'s in  $B_{n-1,i}^x$ yields
\begin{multline}
\label{decompocond}\Egcm(\left.\varphi_{3n}\right\vert{\mathcal F}_n)\pa{\confihat+\delta^{\theta}_{x+e_i}}\\
=\frac{1}{(2n-1)^2}\sum_{y\in B_{n-1,i}^x}\pa{\Egcm\pa{\left.\frac{1-{\confe_y}}{1-\alpha}\varphi_{3n}\;\right\vert\;{\mathcal F}_n}+\Egcm\pa{\left.\frac{{\confe_y}-\alpha}{1-\alpha}\varphi_{3n}\;\right\vert\;{\mathcal F}_n}}\pa{\confihat+\delta^{\theta}_{x+e_i}}. 
\end{multline}

Let us consider the first term of the previous equality. For any $y$ in the boundary, thanks to the factor $1-\confe_y$ the site $y$ is empty. Performing the change of variable $\confehat\rightarrow \confehat-\delta_y$ where $\confehat-\delta_y$ is the configuration identical to $\confehat$ everywhere except in $y$ where the site is now empty, we obtain 
\begin{align*}\Egcm\Big(&\left.\frac{1-{\confe_y}}{1-\alpha}\;\varphi_{3n}\;\right\vert{\mathcal F}_n\Big)\pa{\confihat+\delta^{\theta}_{x+e_i}}\\
=&\Egcm\pa{\left. \frac{\confe_y}{\alpha}\;\varphi_{3n}\pa{\confehat-\delta_y}\;\right\vert\;{\mathcal F}_n}\pa{\confihat+\delta^{\theta}_{x+e_i}}\\
=&\Egcm\pa{\left.\frac{\confe_y}{\alpha}\cro{\varphi_{3n}\pa{\confihat+\delta^{\theta}_{x+e_i},\confehat-\delta_y}-\varphi_{3n}\pa{\confihat, \confehat}}\;\right\vert{\mathcal F}_n}+\Egcm\pa{\left.\frac{{\confe_y}}{\alpha}\;\varphi_{3n}\pa{\confihat,\confehat}\;\right\vert\;{\mathcal F}_n}.
\end{align*}
We deduce from the last identity and equation \eqref{decompocond} that we can write $\Egcm(\left.\varphi_{3n}\right\vert{\mathcal F}_n)\pa{\confihat+\delta^{\theta}_{x+e_i}}$ as 
\begin{multline*}
\frac{1}{(2n-1)^2}\sum_{y\in B_{n-1,i}^x}\Bigg[\Egcm\pa{\left.\frac{{\confe_y}}{\alpha}\cro{\varphi_{3n}\pa{\confihat+\delta^{\theta}_{x+e_i},\confehat-\delta_y}-\varphi_{3n}\pa{\confihat, \confehat}}\;\right\vert\;{\mathcal F}_n}\\
+\Egcm\pa{\left.\;\frac{\xi_y}{\am}\varphi_{3n}\pa{\confihat,\confehat}\;\right\vert\;{\mathcal F}_n}+\Egcm\pa{\left.\frac{{\confe_y}-\alpha}{1-\alpha}\varphi_{3n}\pa{\confihat+\delta^{\theta}_{x+e_i},\confehat}\;\right\vert\;{\mathcal F}_n}\Bigg],
\end{multline*}
and therefore
\begin{multline}
\label{secondeligne}
\Egcm(\left.\varphi_{3n}\right\vert{\mathcal F}_n)\pa{\confihat+\delta^{\theta}_{x+e_i}}-\Egcm(\left.\varphi_{3n}\right\vert{\mathcal F}_n)(\confihat)\\
=\frac{1}{(2n-1)^2}\sum_{y\in B_{n-1,i}^x}\Bigg[\Egcm\pa{\left.\frac{{\confe_y}}{\alpha}\cro{\varphi_{3n}\pa{\confihat+\delta^{\theta}_{x+e_i},\confehat-\delta_y}-\varphi_{3n}\pa{\confihat, \confehat}}\;\right\vert\;{\mathcal F}_n}\\
+\Egcm\pa{\left.\frac{{\confe_y}-\alpha}{\alpha}\;\varphi_{3n}\pa{\confihat,\confehat}\;\right\vert\;{\mathcal F}_n}+\Egcm\pa{\left.\frac{{\confe_y}-\alpha}{1-\alpha}\varphi_{3n}\pa{\confihat+\delta^{\theta}_{x+e_i},\confehat}\;\right\vert\;{\mathcal F}_n}\Bigg].
\end{multline}
Using $(\sum_{i=1}^k a_i)^2\leq k\sum_{i=1}^k a_i^2$ as well as Jensen's inequality yields 
\begin{align}\label{majospec}
\Egcm\Big((1-{\conf_{x+e_i}})&\big(\Egcm(\left.\varphi_{3n}\right\vert{\mathcal F}_n)(\confhat+\delta^{\theta}_{x+e_i})-\Egcm(\left.\varphi_{3n}\right\vert{\mathcal F}_n)(\confhat)\big)^2\Big)\nonumber\\
\leq&\frac{3}{(2n-1)^2}\cro{\sum_{y\in B_{n-1,i}^x}\Egcm\pa{\frac{{\conf_y}(1-\conf_{x+e_i})}{\alpha^2}\cro{\varphi_{3n}\pa{\confhat+\delta^{\theta}_{x+e_i}-\delta_y}-\varphi_{3n}\pa{\confhat}}^2}}\nonumber\\
&+3\Egcm\pa{\Egcm\pa{\left.\pa{\frac{1}{(2n-1)^2}\sum_{y\in B_{n-1,i}^x}\frac{{\conf_y}-\alpha}{\alpha}}\;\varphi_{3n}\;\right\vert\;{\mathcal F}_n }^2}\nonumber\\
&+3\Egcm\pa{\Egcm\pa{\left.\pa{\frac{(1-{\conf_{x+e_i}})}{(2n-1)^2}\sum_{y\in B_{n-1,i}^x}\frac{{\conf_y}-\alpha}{1-\alpha}}\varphi_{3n}\pa{\confhat+\delta^{\theta}_{x+e_i}}\;\right\vert\;{\mathcal F}_n}^2}.
\end{align}
From now on, the strategy to prove Lemma \ref{lem:L2boundphi} is straightforward. We are going to prove that each of the three terms in the right-hand side above is of order $n^2$~:
\begin{itemize}
\item The second and third line above are controlled thanks to the spatial averaging by the $L^2$ norm of the $\varphi_n$'s.
\item In the first line, the angle of the particle  deleted in $y$ is not necessarily the same as the one of the particle created in $x+e_i$, because the angle $\theta$ above is distributed according to $\param/\alpha$ and independent of the configuration. However, since the $\varphi_n$ are in ${T_n^{\omega}}$ their dependency in the angles can be sharply estimated. Once this difficulty is dealt with, the remaining quantity will be controlled by the Dirichlet form.
\end{itemize}
We first treat the first step above. Thanks to the Cauchy-Schwarz inequality, we can estimate the second line 
\begin{multline*}
\Egcm\pa{\Egcm\pa{\left.\pa{\frac{1}{(2n-1)^2}\sum_{y\in B_{n-1,i}^x}\frac{{\conf_y}-\alpha}{\alpha}}\;\varphi_{3n}\;\right\vert\;{\mathcal F}_n }^2}\\
\leq\;\frac{1}{\alpha^2}\Egcm\pa{\pa{\frac{1}{(2n-1)^2}\sum_{y\in B_{n-1,i}^x}\conf_y-\alpha}^2}\Egcm\pa{\varphi_{3n}^2}=\;\frac{(1-\alpha)}{\alpha (2n-1)^2} \Egcm\pa{\varphi_{3n}^2},
\end{multline*}
since under $\mesinv$, the $\conf_y$'s  are i.i.d. variables. We can now use the bound obtained in Lemma \ref{lem:normephiu}, which yields that for some constant $C_1=C_1(\param, \alpha', \uf)$, 
\begin{equation}\label{majo1}\Egcm\pa{\Egcm\pa{\left.\pa{\frac{1}{(2n-1)^2}\sum_{y\in B_{n-1,i}^x}{\conf_y}-\alpha}\;\varphi_{3n}\;\right\vert\;{\mathcal F}_n }^2}\leq C_1 n^2. \end{equation}
Similarly, since
\[\Egcm\pa{(1-{\conf_{x+e_i}})\varphi_{3n}\pa{\confhat+\delta^{\theta}_{x+e_i}}^2}=\frac{1-\alpha}{\alpha}\Egcm(\conf_{x+e_i}\varphi_{3n}^2)\leq Cn^2,\] 
we also have for some constant $C_2=C_2(\param, \alpha', \uf)$
\begin{equation}
\label{majo2}
\Egcm\pa{\Egcm\pa{\left.\pa{\frac{1}{(2n-1)^2}\sum_{y\in B_{n-1,i}^x}\frac{\conf_y-\alpha}{1-\alpha}}(1-{\conf_{x+e_i}})\varphi_{3n}\pa{\confhat+\delta^{\theta}_{x+e_i}}\;\right\vert\;{\mathcal F}_n}^2}\leq C_2 n^2. 
\end{equation}

\bigskip

We now estimate the first line of the right-hand side of \eqref{majospec}, namely 
\begin{equation}\label{firstlineGFF}\frac{1}{(2n-1)^2}\sum_{y\in B_{n-1,i}^x}\Egcm\pa{\frac{{\conf_y}(1-\conf_{x+e_i})}{\alpha^2}\cro{\varphi_{3n}\pa{\confhat+\delta^{\theta}_{x+e_i}-\delta_y}-\varphi_{3n}\pa{\confhat}}^2}.\end{equation}
We first deal with the fact that the deleted and created particles do not have the same angle. Recall that $\confhat^{y,\theta}$ is the configuration where the angle of the particle at the site $y$ has been set to $\theta$, we can thus write \[\confhat+\delta^{\theta}_{x+e_i}-\delta_y=\pa{\confhat^{y,\theta}}^{y,x+e_i},\]
therefore  
 \begin{align*}\pa{\varphi_{3n}\pa{\confhat+\delta^{\theta}_{x+e_i}-\delta_y}-\varphi_{3n}\pa{\confhat}}^2\leq2\cro{\varphi_{3n}\pa{\pa{\confhat^{y,\theta}}^{y,x+e_i}}-\varphi_{3n}\pa{\confhat^{y,\theta}}}^2+2\cro{\varphi_{3n}\pa{\confhat^{y,\theta}}-\varphi_{3n}\pa{\confhat}}^2.
 \end{align*}
 Since $\theta$ is distributed according to $\param/\alpha$, conditionally to $\conf_y=1$, $\confhat^{y,\theta}$ has the same distribution as $\confhat$ under $\mesinv$, and we can therefore control \eqref{firstlineGFF} by 
 \begin{equation}\label{MLGFF}\frac{2}{\alpha^2(2n-1)^2}\sum_{y\in B_{n-1,i}^x}\cro{\Egcm\pa{\conf_y(1-\conf_{x+e_i})\cro{\varphi_{3n}\pa{\confhat^{y,x+e_i}}-\varphi_{3n}\pa{\confhat}}^2}+\Egcm\pa{\conf_y\cro{\varphi_{3n}\pa{\confhat^{y,\theta}}-\varphi_{3n}\pa{\confhat}}^2}}.\end{equation}
Once again, we are going to prove that the contributions of both terms in the right-hand side above are of order $n^2$.

We first need to decompose, as in the proof of the  two-block estimate of Lemma \ref{lem:TBE}, the particle jumps appearing in the first term into nearest neighbor jumps. More precisely, there exists a finite family $x_0,\ldots ,x_p$  such that $x_0=y$, $x_p=x$ and for any $k\in\llbracket 0,p-1\rrbracket$, $\abss{x_k-x_{k+1}}=1$. Furthermore, we can safely assume that $p=\abss{y-x}$. With this construction, for any $y\in B_{n-1,i}^x$, we can write  
\begin{align}\label{dirboundphiy}\Egcm\bigg[\conf_y(1-{\conf_{x+e_i}})\Big(\varphi_{3n}(\confhat^{y,x+e_i})-&\varphi_{3n}(\confhat)\Big)^2\bigg]\nonumber \\
&\leq \abss{y-x}\sum_{k=1}^{\abss{y-x}}\Egcm\cro{\conf_{x_k}(1-{\conf_{x_{k+1}}})\Big(\varphi_{3n}(\confhat^{x_k,x_{k+1}})-\varphi_{3n}(\confhat)\Big)^2}\nonumber\\
&\leq \abss{y-x}\sum_{k=1}^{\abss{y-x}}\Egcm\pa{\cro{\nabla_{x_k, x_{k+1}}\varphi_{3n}}^2},
\end{align}
since $(\sum_{k=1}^p a_k)^2\leq p\sum_{k=1}^pa_k^2$. As in the proof of Lemma \ref{lem:normephiu}, one easily checks that, $x_k$ and $x_{k+1}$ being neighbors,    
\[\Egcm\pa{\cro{\nabla_{x_k, x_{k+1}}\varphi_{3n}(\confhat)}^2}\leq C(\uf).\]
therefore \eqref{dirboundphiy} yields  
\[\Egcm\cro{\conf_y(1-{\conf_{x+e_i}})\Big(\varphi_{3n}(\confhat^{y,x+e_i})-\varphi_{3n}(\confhat)\Big)^2}\leq \abss{y-x}^2C(\uf).\]
We now get back to the first term in \eqref{MLGFF}. It is not hard to see that $\sum_{y\in B_{n-1,i}^x}\abss{y-x}^2$ is of order $n^4$, and we obtain as wanted that for some constant $C_3=C_3(\param, \uf)$, 
\begin{equation}\label{majodir}\frac{2}{\alpha^2(2n-1)^2}\sum_{y\in B_{n-1,i}^x}\Egcm\pa{\conf_y(1-\conf_{x+e_i})\cro{\varphi_{3n}\pa{\confhat^{y,x+e_i}}-\varphi_{3n}\pa{\confhat}}^2}\leq C_3 n^2.\end{equation}

We now estimate the second contribution in \eqref{MLGFF}. The only difference between $\varphi_{3n}\pa{\confhat^{y,\theta}}$ and $\varphi_{3n}\pa{\confhat}$ is the angle of 
the particle at site $y$. Recall that for any $n$, $\varphi_n\in {T^\omega}$, therefore the variation of $\varphi_n$ when an angle is changed can be precisely estimated.  
{Fix $n\geq 0$, and recall that $\varphi_{3n}{\in T_{3n}^\omega}$}. Then, there exists angle-blind functions $ (\psi_{n,x})_{x\in B_{3n}}$, and $\psi_n$ in $\Sp$, such that 
\[\varphi_{3n}=\psi_n+\sum_{x\in B_{3n}}\com_x\psi_{n,x}.\]
Since the only difference between $\confhat^{y,\theta}$ and $\confhat$ is in the angle present at the site $y$, we can write
\[\varphi_{3n}\pa{\confhat^{y,\theta}}-\varphi_{3n}\pa{\confhat}=(\omega(\theta)-\omega(\theta_y))\conf_y\psi_{n,y}(\conf),\]
therefore the second contribution in \eqref{MLGFF} can be rewritten 
\begin{equation}\label{GFFchangement}
\frac{2}{\alpha^2(2n-1)^2}\sum_{y\in B_{n-1,i}^x}\Egcm\pa{\conf_y(\omega(\theta)-\omega(\theta_y))^2\psi_{n,y}^2 }= \frac{4V_{\param}(\omega)}{\alpha^2(2n-1)^2}\sum_{y\in B_{n-1,i}^x}\Egcm\pa{\conf_y\psi_{n,y}^2},
\end{equation}
where we shortened $V_{\param}( \omega)=Var_{\param}( \omega(\theta_0)\mid \conf_0=1)$, since the angles are independent of the configuration conditionally to the presence of a particle. 
Similarly to what we did in Section \ref{subsec:spectralgap} rewrite 
\[\varphi_{3n}=\varphi_n^{1}+\varphi_n^{b},\] where 
\[ \varphi_n^{1}= \sum_{x\in B_{3n}}(\omega(\theta_x)-\Egcm(\omega))\conf_x\psi_{n,x}\eqand \varphi_n^{b}=\psi_n+\Egcm(\omega)\sum_{x\in B_{3n}}\conf_x\psi_{n,x},\]
where $\Egcm(\omega)$ stands for $\Egcm(\omega(\theta_0)\mid \conf_0=1)$.
As in Section \ref{subsec:spectralgap}, 
\[\Egcm(\varphi_{3n}^2)=\Egcm((\varphi_n^1)^2)+\Egcm((\varphi^{b}_n)^2),\]
and
\[\Egcm((\varphi^{1}_n)^2)=V_{\param}(\omega)\sum_{x\in B_{3n}}\Egcm(\conf_x\psi_{n,x}^2).\]
The two previous identities finally yield that 
\[V_{\param}( \omega)\sum_{x\in B_{3n}}\Egcm(\conf_x\psi_{n,x}^2)\leq \Egcm(\varphi_{3n}^2).\]
We now use this bound as well as \eqref{GFFchangement} and Lemma \ref{lem:normephiu}  to obtain that for some constant $C_4=C_4(\confhat, \alpha', \uf)$
\begin{equation}
\label{majo3}
\frac{2}{\alpha^2(2n-1)^2}\sum_{y\in B_{n-1,i}^x}\Egcm\pa{\conf_y\cro{\varphi_{3n}\pa{\confhat^{y,\theta}}-\varphi_{3n}\pa{\confhat}}^2}\leq C_4n^2.
\end{equation}
This is the estimate we wanted for the second line of \eqref{MLGFF}.

 Letting $C=3(C_1+C_2+C_3+C_4)$, we now use the four bounds \eqref{majo1}, \eqref{majo2},  \eqref{majodir} and \eqref{majo3} in equation \eqref{majospec}, to obtain that 
\[\Egcm\Big((1-{\conf_{x+e_i}})\big(\Egcm(\left.\varphi_{3n}\right\vert{\mathcal F}_n)(\confhat+\delta^{\theta}_{x+e_i})-\Egcm(\left.\varphi_{3n}\right\vert{\mathcal F}_n)(\confhat)\big)^2\Big)\leq C n^2\]
as wanted, which concludes the proof of Lemma \ref{lem:L2boundphi}.}

We have now finished the second step, and proved that the sequences of boundary terms $(\widetilde{R}^+_{n,i})_{n\in \N} $ and $(\widetilde{R}^-_{n,i})_{n\in \N} $ are bounded in $L^2(\mesinv)$. To conclude the proof of Lemma \ref{lem:boundaryconvergence} we now prove that any weak limit point $\mathfrak{R}_i^-$  of $(\widetilde{R}^-_{n,i}) $ is in the linear span of the currents $\bff{J}^\omega$. The main difficulty is to prove that any limit point only depends on $\confhat_0$ and $\confhat_{e_i}$, which we state as a separate lemma. We will once again only consider the negative boundary terms, the positive boundary terms being treated in the same way.

\bigskip

\step{Third}{Proof that $\mathfrak{R}_i^-$ only depends on $\confhat$ through $\confhat_0$ and $\confhat_{e_i}$ }

Let us introduce \[\Z^2_{+,i}=\{x_i> 0\}\cap\Z^2\setminus\{e_i\}.\]
We first prove the following intermediate result.
\begin{lemm}\label{lem:weaklimitpointsGCF0}Any weak limit point  $\mathfrak{R}_i^-$ of the sequence $(\widetilde{R}^-_{n,i})$ is measurable w.r.t. the sites in \

\noindent$\Z^2\cap\{x_i> 0\}\cup \{0\}$. Furthermore, for any edge $(y, y+z)$ with both ends in the set $\Z^2_{+,i}$, the gradient  $\nabla_{y, y+z }\mathfrak{R}_i^-$ vanishes in $L^2(\mesinv)$.
\end{lemm}
\proofthm{Lemma \ref{lem:weaklimitpointsGCF0}}{
In order to avoid taking subsequences, let us also assume that $(\widetilde{R}^-_{n,i})$ weakly converges towards $\mathfrak{R}_i^-$. 
We first prove the first statement, which is elementary. For any $x$ in the negative boundary, $x_i=-n-1$, $\tau_{-x}\modphi_n$ is measurable with respect 
to the half plane $\{x_i >0\}$, therefore $\nabla_{0,e_i}\tau_{-x}\modphi$ is measurable with respect to the sites in $\{x_i>0\}\cup \{0\}$. We deduce from the last remark 
that for any $n$, $\widetilde{R}^-_{n,i}$ is measurable for any $n$ w.r.t. the sites in $\{x_i> 0\}\cup \{0\}$, therefore $\mathfrak{R}_i^-$ also is.

\bigskip

We now show that for any edge $\{y,y+z\}\subset\Z^2_{+,i}$, the gradient $\nabla_{y, y+z}\mathfrak{R}_i^-$ vanishes in $L^2(\mesinv)$. Fix an edge $(y,y+z)$ with both ends in $\Z^2_{+,i}$.
By definition, 
\begin{align*}\nabla_{y, y+z}\widetilde{R}^-_{n,i}&=\frac{1}{(2n)^2}\sum_{x_i=-n-1}\nabla_{y, y+z}\tau_{-x}\nabla_{x,x+e_i}\modphi_n\\
&=\frac{1}{(2n)^2}\sum_{x_i=-n-1}\nabla_{y, y+z}\nabla_{0,e_i}\tau_{-x}\modphi_n.
\end{align*}
Because $y,y+z$ are different from $0$ and $e_i$, the two gradients in the formula above commute, therefore using once again $(\sum_{i=1 }^na_i)^2\leq n\sum_{i=1}^n a_i^2$, as well as the crude bound $\Egcm((\grad f)^2)\leq 4\Egcm(f^2)$, yields
\begin{align}\label{normeboundarygrad}\Egcm\cro{\big|\nabla_{y, y+z}\widetilde{R}^-_{n,i}\big|^2}&\leq\frac{1}{(2n)^{3}}\sum_{x_i=-n-1}\Egcm\cro{\big(\nabla_{0,e_i}\nabla_{y, y+z}\tau_{-x}\modphi_n\big)^2}\nonumber\\
&= \frac{1}{(2n)^{3}}\sum_{x_i=-n-1}\Egcm\cro{\big(\nabla_{0,e_i}\tau_{-x}\nabla_{x+y, x+y+z}\modphi_n\big)^2}\nonumber\\
&\leq \frac{4}{(2n)^{3}}\sum_{x_i=-n-1}\Egcm\cro{\big(\nabla_{x+y, x+y+z}\modphi_n\big)^2}
.\end{align}
There are three cases to consider to estimate $\Egcm\cro{\big(\nabla_{x+y, x+y+z}\modphi_n\big)^2}$.
\begin{enumerate}
\item{
The first one is the case where both $x+y$ and $x+y+z$ are in $B_n^c$, the complementary set of $B_n$. In this case, 
\[\Egcm\cro{\big(\nabla_{x+y, x+y+z}\modphi_n\big)^2}=0,\] because $\modphi_n$ is $\F_n$- measurable.
}
\item{
The second case when both $x+y$ and $x+y+z$ are in $B_n$. In this case, using \eqref{eq:estBn} and Jensen's inequality we can write
\begin{equation}\label{majorg}\Egcm\pa{\big(\nabla_{x+y, x+y+z}\modphi_n\big)^2}\leq\Egcm\pa{\ind{\rho_{3n}\leq \alpha'}\big(\ufbar_{x+y, x+y+z}\big)^2}\leq C(\uf).\end{equation}
}
\item{The last case to consider is if  $x+y$ and $x+y+z$ link $B_n$ and $B_n^c$. Then, as in the proof of Lemma \ref{lem:L2boundaryterm} we obtain
\[\Egcm\cro{\big(\nabla_{x+y, x+y+z}\modphi_n\big)^2}\leq C(\param, \alpha', \uf)n^2.\]
}
\end{enumerate}

Fix an edge $(y,y+z)$ with both ends in $\Z^2_{+,i}$ and write $z$ as $\pm e_j$, we treat separately the two cases for $j$.
If $j=i$, for any $n$ large enough (more precisely as soon as $2n+2\geq y_i$), for any $x$ such that $x_1=-n-1$, either $x+y$ and $ x+y\pm e_i$ are both in $B_n$ or both are in its complementary set $B_n^c$. We are therefore either in the first or in the second case above, and since the number of terms in the sum is $O(n)$, equation \eqref{normeboundarygrad} yields 
\[\Egcm\cro{\big(\nabla_{y, y+z}\widetilde{R}^-_{n,i}\big)^2}\leq C'n^{-2} \underset{n\to\infty}{\to} 0,\]
for some constant $C'=C'(\param, \uf)$.

If now $j\neq i$, there can be only two terms in the sum over $x$ for which $x+y$ and $ x+y\pm e_i$ link $B_n$ and $B_n^c$ (third case above), whereas all the others are either in the first or the second case. In this case, equation \eqref{normeboundarygrad} yields
\[\Egcm\cro{\big(\nabla_{y, y+z}\widetilde{R}^-_{n,i}\big)^2}\leq  C'(\param, \uf)n^{-2}+C''(\param, \alpha', \uf)n^{-1} \norm{\uf}_{2,\param}^2\underset{n\to\infty}{\to} 0.\]
This proves that the sequence $\nabla_{y, y+z}\widetilde{R}^-_{n,i}$ vanishes as $n\to\infty $ in $L^2(\mesinv)$ for any edge $(y,y+z)$ with both ends in $\Z^2_{+,i}$. Since the gradient $\nabla_{y, y+z }$ is a (Lipschitz, and therefore) continuous functional in $L^2(\mesinv)$, $\nabla_{y, y+z }\mathfrak{R}_i^-$ vanishes for any edge $(y,y+z)$ with both ends in $\Z^2_{+,i}$.
This concludes the proof of Lemma \ref{lem:weaklimitpointsGCF0}.
}

\begin{lemm}\label{lem:weaklimitpointsGCF}Any weak limit point $\mathfrak{R}_i^-$ of the sequence $(\widetilde{R}^-_{n,i})_{n\in \N}$ only depends on the configuration through $\confhat_0$ and $\confhat_{e_i}$.
The same is true for the limit points of the positive boundary terms  $(\widetilde{R}^+_{n,i})_{n\in \N} $.
\end{lemm}
\proofthm{Lemma \ref{lem:weaklimitpointsGCF}}{This Lemma is a consequence of Lemma \ref{lem:weaklimitpointsGCF0}. Consider the localization $\mathfrak{R}_{i,n}^-=\Egcm(\mathfrak{R}_i^-\mid \F_n)$, then $\mathfrak{R}_{i,n}^-$ is measurable with respect to the sites in $\{x_i>0\}\cup \{0\}$ and for any edge $(y,y+z)$ with both ends in $\Z^2_{+,i}$ its gradient vanishes in $L^2(\mesinv)$. 
These two properties are immediate consequences of the properties of $\mathfrak{R}_i^-$ and Jensen's inequality.

Let 
\[B_{i,n}^+=B_n\cap\Z^2_{+,i},\]
since the gradients of $\mathfrak{R}^-_{i}$ vanish for any edge in $B_{i,n}^+$, on the event on which there are at least two empty sites in $B_{i,n}^+$, $\mathfrak{R}^-_{i}$ only depends on the  $\confhat_x,\; x\in B_{i,n}^+$ through the empirical measure on $B_{i,n}^+$
\[\dens_{B_{i,n}^+}:=\frac{1}{\abss{ B_{i,n}^+}}\sum_{B_{i,n}^+}\conf_x\delta_{\theta_x}.\]
 Indeed, for two configurations $\confhat$ and  $\confhat'$ with the same number of particles, and with the same angles in $B_{i,n}^+$, we can reach one from the other with a combination of the previous gradients, hence the difference $\mathfrak{R}^-_{i,n}(\confhat)-\mathfrak{R}^-_{i,n}(\confhat')$ vanishes. This is not true whenever there is one or less empty site in $B_{i,n}^+$, but since we are under the product measure, this happens with exponentially small probability and will not be an issue.
 
 \bigskip 
 
Let us denote by $E^*_{n}$ the event ''there are at least two empty sites in $B_{i,n}^+$'', the previous statement rewrites as 
\[\mathfrak{R}^-_{i,n}\1_{E^*_{n}}=\Egcm\pa{\mathfrak{R}^-_{i,n}\1_{E^*_{n}}\;\Bigg\vert \; \confhat_0, \confhat_{e_i}, \dens_{B_{i,n}^+}}.\]
For any cylinder function $f$, we are going to prove that $\Egcm(f.\mathfrak{R}_i^-)=\Egcm\cro{f.\;\E(\mathfrak{R}_i^-\mid \confhat_0, \confhat_{e_i})}$. 
Let \[f^+=\E\pa{f\mid\confhat_x,\;x\in\{x_i>0\}\cup \{0\}}\]
be the conditional expectation with respect to the sites in $\{x_i>0\}\cup \{0\}$. Since $f$ is a cylinder function, so is $f^+$, therefore for any sufficiently large integer $n$, we can write 
\begin{align}
\label{conditionnement}\Egcm(f.\mathfrak{R}^-_{i}\1_{E^*_{n}})&=\Egcm(f.\mathfrak{R}^-_{i,n}\1_{E^*_{n}})\nonumber\\
&=\Egcm\pa{\Egcm\pa{f.\mathfrak{R}^-_{i,n}\1_{E^*_{n} }\;\Bigg\vert\; \confhat_0, \confhat_{e_i}, \dens_{B_{i,n}^+}}}\nonumber\\
&=\Egcm\pa{\mathfrak{R}^-_{i,n}\1_{E^*_{n}} \Egcm\pa{f \;\Bigg\vert\;  \confhat_0, \confhat_{e_i},\dens_{B_{i,n}^+}}}\nonumber\\
&=\Egcm\pa{\mathfrak{R}^-_{i,n}\1_{E^*_{n}} \Egcm\pa{f^+ \;\Bigg\vert\;  \confhat_0, \confhat_{e_i}, \dens_{B_{i,n}^+}}}\nonumber\\
&=\Egcm\pa{\mathfrak{R}^-_{i,n}\Egcm\pa{f^+ \;\Bigg\vert\;  \confhat_0, \confhat_{e_i}, \dens_{B_{i,n}^+}}}+\Egcm\pa{\mathfrak{R}^-_{i,n}\1_{E^{*c}_{n}} \Egcm\pa{f^+ \;\Bigg\vert\;  \confhat_0, \confhat_{e_i}, \dens_{B_{i,n}^+}}}\nonumber\\
\nonumber\\
&{=}\Egcm\pa{\mathfrak{R}^-_{i} \Egcm\pa{f^+\;\Big\vert\; \confhat_0, \confhat_{e_i}}}+o_n(1),
\end{align}
since 
\[\Egcm\pa{f^+ \;\Bigg\vert\; \confhat_0, \confhat_{e_i},\dens_{B_{i,n}^+}}\xrightarrow[n\to \infty]{L^2(\mesinv)}\Egcm\pa{f^+ \;\Big\vert\; \confhat_0, \confhat_{e_i}},\]
because $\dens_{B_{i,n}^+}$ converges $\mesinv$ a.s. as $n\to \infty$ towards $\param$, and 
\[\Egcm\pa{\mathfrak{R}^-_{i,n}\1_{E^{*c}_{n}} \Egcm\pa{f^+ \;\Bigg\vert\;  \confhat_0, \confhat_{e_i}, \dens_{B_{i,n}^+}}}\xrightarrow[n\to \infty]{} 0,\]
because $f^+$ is a bounded function, and $\mathfrak{R}^-_{i,n}$ is in $L^2(\mesinv)$.
For the same reason, the left-hand side in  \eqref{conditionnement} converges as $n$ goes to $\infty$ towards $\Egcm(f.\mathfrak{R}^-_{i})$, and therefore for any cylinder function $f$
\[\Egcm\pa{\mathfrak{R}^-_{i} \Egcm\pa{f^+\;\Big\vert\; \confhat_0, \confhat_{e_i}}}=\Egcm(f.\mathfrak{R}^-_{i}),\]
so that 
\[\mathfrak{R}^-_{i}=\Egcm\pa{\mathfrak{R}^-_{i} \;\Big\vert\; \confhat_0, \confhat_{e_i}}.\]
This concludes the proof of Lemma \ref{lem:weaklimitpointsGCF}.
}

To complete the proof of Lemma \ref{lem:boundaryconvergence}, now that we have proved that all limit points of the boundary terms are function of $\confhat_0$ and $\confhat_{e_i}$, we still have to show that such limit points are in $\bff{J}^\omega$. First notice  that any limit point of the negative boundary $\mathfrak{R}^-_{i}$ verifies
\begin{equation}
\label{egalite31}
{\conf_{e_i}}\mathfrak{R}^-_{i}=(1-{\conf_{0}})\mathfrak{R}^-_{i}=0.
\end{equation}
Indeed,
\[{\conf_{e_i}}\mathfrak{R}^-_{i}=\lim_{n\to \infty}\frac{1}{(2n)^2}\sum_{\substack{x_i=-n-1\\
x\in B_n}}{\conf_{e_i}}\tau_{-x}\nabla_{x,x+e_i}\modphi_n=\lim_{n\to \infty}\frac{1}{(2n)^2}\sum_{\substack{x_i=-n-1\\
x\in B_n}}{\conf_{e_i}}\nabla_{0,e_i}\tau_{-x}\modphi_n,\]
since $\tau_x\grad f=\nabla_{\tau_x a}\tau_x f$. Now the latter obviously vanishes since ${\conf_{e_i}}\nabla_{0,e_i}=0$. The second identity is proved in the same way.

 Since the $\widetilde{\varphi}_n$'s are in ${T^\omega}$, so is  $\mathfrak{R}^-_{i}$. Since $\mathfrak{R}^-_{i}$ depends only on $\confhat_0$ and $\confhat_{e_i}$, using \eqref{egalite31} it can therefore be expressed as  
\[\mathfrak{R}^-_{i}(\confhat)=\conf_0(1-\conf_{e_i})\mathfrak{R}^-_{i}(\confhat_0, \confhat_{e_i})=\conf_0(1-\conf_{e_i})\cro{\psi(\conf_0, \conf_{e_i}) + \com_0\psi_0(\conf_0, \conf_{e_i})},\]
for some angle blind functions $\psi$, $\psi_0.$ In particular, letting $c_1=\psi_0(1,0), c_2=\psi(1,0)$,
\[\mathfrak{R}^-_{i}(\confhat)=(c_1\com_0+c_2\conf_0)(1- \conf_{e_i}).\]
Finally, any weak limit point of the boundary term is an element of  $\boldsymbol{\mathfrak{J}}^{\omega}$, which is what we wanted to show. 
The proof of Lemma \ref{lem:boundaryconvergence} is thus complete.
}

\subsection{An integration by parts formula}
\label{subsec:IPP}
\intro{Considering the symmetric exclusion generator $\gene$ as a discrete Laplacian, to prove  Theorem \ref{thm:limcovariance}, we are going to need an integration by parts formula in order to express the expectation of $\psi.h$ in terms of the gradient of $h$ and the ''integral'' $\nabla \gene^{-1}\psi$ of $\psi$. }

We first extend the definition of the canonical measures given in Definition  \ref{defi:CM} to any domain $B\subset \torus$. For that purpose, consider an integer $K\leq |B|$, and an orderless family $ \{\theta_1,\ldots ,\theta_K\}\in \ctoruspi^K$. 
Recall that we denote by $\K$ the pair $(K, \{\theta_1,\ldots ,\theta_K\} )$, and we let $\cm$ be the measure such that the $K$  particles with fixed angles $\theta_1,\ldots , \theta_K$ are uniformly distributed in the domain $B$. 
 If $B=B_l$ is the ball of radius $l$, this notation is shortened as $\mu_{l,\K}$ in accord with Definition \ref{defi:CM}. The expectation w.r.t both of these measures is respectively denoted $\Ecm$ and $\E_{l,\K}$. We will, in a similar fashion, write
\begin{equation*}\gene_B f(\confhat)=\sum_{\substack{x, x+z\in B\\
|z|=1}}\conf_x\left(1-\conf_{x+z}\right)\left(f(\confhat^{x,x+z})-f(\confhat)\right),\end{equation*}
\index{$ \gene_B$\dotfill part of $\gene $ with jumps inside $B$}
\index{$ \gene_l$\dotfill part of $\gene $ with jumps inside $B_l$}
for the generator of the symmetric exclusion process restricted to $B$, shortened as $\gene_l$ if $B=B_l$.

\index{$ I_a$\dotfill the ''inverse'' of $\grad$  }

Recall that we defined 
\begin{equation*}
\czero=\left\{\psi\in {\mathcal C} \; \;\Big| \;\; \; \E_{s_\psi,\K}(\psi)=0 \;\; \forall\K\in\Ksett_{s_\psi}\quad \mbox{ and  } \quad \psi_{\vert\Sigma^{\K}_{s_{\psi}}}\equiv 0\; \;\forall  \K\in\Kset_{s_\psi}\smallsetminus \Ksett_{s_\psi}\right\},
 \end{equation*}
and that $\grad$ is the gradient representing a particle jump along $a$. 
 \begin{lemm}[Integration by parts formula]
\label{lem:IPP} Let $\psi\in\czero$ be a cylinder function, and $a\subset{B_{s_{\psi}}}$ an oriented edge in its domain. Then, $\psi$ is in the range of the generator $\gene_{s_{\psi}}$, and we can define the "primitive" $I_{a}(\psi)$  of $\psi$ with respect to the gradient along the oriented edge $a$ as 
\[I_{a}(\psi)=\frac{1}{2}\grad (-\gene_{{s_{\psi}}})^{-1}\psi.\] 
Furthermore, for any $B\subset \torus$ containing $B_{s_{\psi}}$, any $\K=(K, (\theta_1,\ldots ,\theta_K))$ such that $K\leq |B|$ and $h\in \mathcal{C}$ measurable w.r.t. sites in $B$, we have 
\begin{equation}\label{IPPeq}\Ecm\pa{\psi.h  }=\sum_{a\subset B_{s_\psi}}\Ecm\pa{I_{a}(\psi).\grad h }.\end{equation}
This result is also true if $\cm$ is replaced by a grand-canonical measure $\gcm$. 
Note that if $K=|B|-1$ or $K=|B|$ the result is trivial because $\psi$ vanishes.
%
\end{lemm}
\begin{proof}[Proof of Lemma \ref{lem:IPP}]The proof of the previous result is quite elementary. Fix a function $\psi\in\czero$, to prove the integration by parts formula, we first show that $\psi $ is in the range of $\gene_{{s_{\psi}}}$, by building for any $\K$ a function $\varphi_{\K}$ on $\Sigma^{s_{\psi}}_{\K}$, verifying $\gene_{s_{\psi}}\varphi_{\K}=\psi_{\vert\Sigma^{s_{\psi}}_{\K}}$. This result is well-known for the color-blind exclusion process, but in our case where each particle has an angle, the canonical measures take an unusual form, and we prove it for the sake of exhaustivity.

For any $\varphi:\Sigma^{s_{\psi}}_{\K}\to\R$ such that $\gene_{s_{\psi}}\varphi=0$,  
\[\E_{B_{s_\psi},\K}(\varphi \gene_{s_{ \psi}}\varphi)=-\frac{1}{2}\E_{B_{s_\psi}, \K}\pa{\sum_{x,x+z\in B_{s_{\psi}}}\conf_x(1-\conf_{z})(\varphi(\confhat^{x,z})-\varphi(\confhat))^2}=0,\] 
therefore $\varphi$ is invariant under the allowed jump of a particle along any edge in $B_{s_\psi}$. For any $\K\in \Ksett_{s_\psi}$, the function $\varphi$ is constant on $\Sigma^{s_{\psi}}_{\K}$, because $\Sigma^{s_{\psi}}_{\K}$ is then irreducible w.r.t. the exclusion dynamics in $B_{s_{\psi}}$, according to Section \ref{subsec:irreducibility}. In particular $Ker_{\Sigma^{s_{\psi}}_{\K}}\gene_{s_{\psi}}$ is the set of constant functions, and 
\[\left\{\varphi:\Sigma^{s_{\psi}}_{\K}\to\R \; \;\big| \;\; \E_{B_{s_\psi},\K}(\varphi )=0 \right\}=\left\{\gene_{s_{\psi}}\psi,\quad  \psi:\Sigma^{s_{\psi}}_{\K}\to\R\right\}.\]
For any $\psi\in\czero$, any $\K\in \Ksett_{s_\psi}$, there exists a $\varphi_{\K}:\Sigma^{s_\psi}_{\K}\to \R$, such that
\[\gene_{s_{\psi}}\varphi_{\K}=\psi_{\vert\Sigma^{s_\psi}_{\K}}.\]
Since $\psi$ vanishes when $B_{s_\psi}$ has one or less empty site, we also let $\varphi_{\K}=0$ for any $ \K\in\Kset_{s_\psi}\setminus\Ksett_{s_\psi}$. We now define the local function $\varphi^*\in\mathcal{C}$ by $\varphi^*_{|\Sigma^{s_\psi}_{\K}}=\varphi_{\K}(\confhat)$, which verifies by construction 
\[\psi=\gene_{s_\psi}\varphi^*,\]
therefore $\psi\in \gene_{s_\psi}\mathcal{C}$.

Proving the integration by parts formula is now elementary~: since $\psi=\gene_{s_{\psi}}\gene_{s_{\psi}}^{-1}\psi$, 
 \begin{align*}
\Ecm(h.\psi )&=\Ecm\pa{h.\gene_{s_{\psi}}\gene^{-1}_{s_{\psi}}\psi}\\
&=-\frac{1}{2}\sum_{a\subset B_{\psi}}\Ecm\pa{\grad\gene_{s_{\psi}}^{-1}\psi.\grad h}\\
&=\sum_{a\subset B_{\psi}}\Ecm\pa{I_{a}(\psi).\grad h}
\end{align*}
which proves identity \eqref{IPPeq}. By conditioning to the canonical state in $B$, one easily obtains that the same is true when the canonical measure is replaced by a grand-canonical measure $\gcm$.
\end{proof}

We finish this section with a technical Lemma. Recall that for any cylinder function $\psi$, we denote by $s_\psi$ the size of its support and for any integer $l$, $l_\psi=l-s_\psi-1$.
\begin{lemm}
\label{lem:techps}
For any $\psi\in \czero+ J^*+\gene \mathcal{C}$, there exists a constant $C(\psi)$ such that for any $l$, $\K\in \Ksett_l$,  $h\in\mathcal{C}$ only depending on sites in $B_l$, $\gamma>0$, and $A\subset B_{l_\psi}$ 
\begin{equation*}
\Ecmlk\pa{h\sum_{x\in A}\tau_x\psi}\leq \gamma C(\psi)|A|+\frac{1}{2\gamma}\dir^{A_\psi}_{l, \K}(h),
\end{equation*}
where we shortened $A_\psi=\{x\in B_{l}, \; d(x,A)\leq s_\psi\}$,  $\dir^{A}_{l, \K}(h)=\Ecmlk(h(-\gene_{A}) h)$ and $\gene_{A}$ is the SSEP generator restricted to jumps with both ends in $A$.
\end{lemm}
\proofthm{Lemma \ref{lem:techps}}{
Since for some constant $C(s_\psi)$, $\sum_{x\in A} \dir^{B_{s_\psi}(x)}_{l, \K}(h) \leq C(s_\psi)\dir^{A_\psi}_{l, \K}(h)$ to establish this result, it is sufficient to prove that for any $x\in A$ and for any positive $\gamma'$,
\begin{equation}
\label{eq:claim1} 
\Ecmlk\pa{h\tau_x\psi}\leq \gamma' C'(\psi)+\frac{1}{2\gamma'}\dir^{B_{s_\psi}(x)}_{l, \K}(h).
\end{equation}
We now establish this last bound for any $\psi\in \czero\cup J^*\cup\gene \mathcal{C}$, which proves the Lemma.

\medskip

Assume first that $\psi=\cur_k^\Phi$ for $k\in\{1,2\}$, and $\Phi\in C^1(\ctoruspi)$. Then, $\Ecmlk\pa{h\tau_x\psi}=\Ecmlk\pa{h\cur_{x,x+e_k}^\Phi},$
where as before $\cur_{x,x+e_k}^\Phi=\Phi(\theta_x)\conf_x(1-\conf_{x+e_k})-\Phi(\theta_{x+e_k})\conf_{x+e_k}(1-\conf_x)$. Thanks to changes of variable $\confhat\mapsto\confhat^{x,x+e_k}$, in the second term, we obtain, using the elementary bound $ab\leq \gamma a^2/2+b^2/2\gamma$ which holds for any $\gamma$,
\begin{equation*}
\Ecmlk\pa{h\tau_x\psi}=-\Ecmlk\pa{\Phi(\theta_x)\nabla_{x,x+e_k}h}\leq \frac{\gamma \norm{\Phi}^2_{\infty}}{2}+\frac{1}{2\gamma}\Ecmlk\pa{(\nabla_{x,x+e_k}h)^2}
\end{equation*}
which proves \eqref{eq:claim1}.
\medskip

We now consider $\psi=\gene f\in \gene \mathcal{C}$. Since $f$ is a local function, fix $s_\psi$ such that $\gene f=\gene_{s_{\psi}}f$. We rewrite 
\begin{multline*}
\Ecmlk\pa{h\tau_x\psi}=\Ecmlk\pa{h\gene_{B_{s_\psi}(x)}(\tau_x f)}=\Ecmlk\pa{(\tau_x f)\gene_{B_{s_\psi}(x)}h}\\
=\sum_{y,y+z\in B_{s_\psi}(x)}\Ecmlk((\tau_xf)\nabla_{x,x+z}h)\leq\frac{\gamma C(s_\psi)\norm{f}_\infty^2}{2}+ \frac{1}{2\gamma}\dir_{l,\K}^{B_{s_\psi}(x)}(h), 
\end{multline*}
as wanted.

\medskip

Only remains the case $\psi\in \czero$, for which \eqref{eq:claim1} is a consequence of the integration by parts formula and is proved similarly 
to the  case $\psi = \gene f$. By definition of $I_a(\psi)$,
\[\sum_{y,y+z\in B_{s_\psi}(x)}\Ecmlk(I_{x,x+z}(\tau_x\psi)^2)=\frac{1}{2}\Ecmlk((\tau_x\psi)(-\gene_{B_{s_\psi}(x)}^{-1})(\tau_x\psi))=\frac{1}{2}\Ecmlk\pa{\psi(-\gene_{B_{s_\psi}}^{-1})\psi}\leq C(\psi),\]
where $C(\psi)$ can be chosen independently of $\K$. Using \eqref{IPPeq}, and this last bound, we obtain
\[\Ecmlk\pa{h\tau_x\psi}=\sum_{y,y+z\in B_{s_\psi}(x)}\Ecmlk\pa{I_{y,y+z}(\tau_x\psi).\nabla_{y,y+z} h }\leq \frac{\gamma C(\psi)}{2} +\frac{1}{2\gamma}\dir_{l,\K}^{B_{s_\psi}(x)}(h),\]
which proves \eqref{eq:claim1} and Lemma \ref{lem:techps}.}

\subsection{Heuristics on $\scal{\cdot}$ and Theorem  \ref{thm:limcovariance}}
\label{heuristicsNG}
The purpose of this section is to explain the variational formula for the limiting covariance $\scal{\psi}$ introduced in Definition \ref{defi:limitingcovariance2}. 
Given the generator $\gene$ of the SSEP on $\Z^2$, for any function $f$ with mean $0$ w.r.t. any canonical measure, consider the  linear application 
\index{$\cdf$\dotfill natural application from $\czero$ to $\textswab{C}_{\param}$}
\begin{equation}\label{cdfdef}\M~: f  \mapsto  \nabla \gene^{-1} \Sigma_f=\pa{\begin{matrix}\nabla_{0,e_1} \gene^{-1} \Sigma_f \\
\nabla_{0,e_2} \gene^{-1} \Sigma_f\end{matrix}}.\end{equation}
A priori, even if $f$ is a local function, $\gene^{-1}f$ is no longer local, and $\nabla \gene^{-1} \Sigma_f$ can therefore involve a infinite number of non-zero contribution, so that $\M$ is not a priori well defined.
However, assuming that $f$ is such that  $\nabla \gene^{-1} \Sigma_f$ is well-defined, the definition above indicates thanks to the translation invariance of $\Sigma_f$ and $\gene^{-1}$, that $\M(f)$ is the germ of a closed form as introduced in Section \ref{subsec:differentialforms}. To illustrate this last remark, we describe the effect  of this application on $\gene \mathcal{C}$ and $J^*$.

\bigskip 

Recall that for $\Phi\in C^1(\ctoruspi)$,  $\cur^{\Phi}_i=\conf^\Phi_0\left(1-\conf_{e_i}\right)-\conf^\Phi_{e_i}\left(1-\conf_0\right)$.
We first investigate the action of $\M$ on the currents $j_i^\Phi$. Consider an infinite configuration $\confhat$ with no particles outside of some large compact set $K$. For the sake of concision, we will call such a configuration \emph{bounded}.  Then, we can write
\[\gene \cro{\sum_{x\in \Z^2}x_i \conf^{\Phi}_x}=\sum_{x\in \Z^2}x_i\gene \conf^{\Phi}_x=\sum_{x\in \Z^2}\tau_x \cur^\Phi_{i}=\Sigma_{\cur^\Phi_i}.\]
Since the configuration was assumed bounded,  both of the sums above are finite, and the identity above is well posed. Coming back to our application $\M$, the previous identity yields 
\[\M(\cur^\Phi_i)=\pa{\begin{matrix}\nabla_{0,e_1} \gene^{-1} \Sigma_{\cur^\Phi_i} \\
\nabla_{0,e_2} \gene^{-1} \Sigma_{\cur^\Phi_i}\end{matrix}}=\pa{\begin{matrix}\nabla_{0,e_1}\sum_{x\in \Z^2}x_i \conf^\Phi_x \\
\nabla_{0,e_2} \sum_{x\in \Z^2}x_i \conf^\Phi_x\end{matrix}}.\]
Since the only positive contribution in the right-hand side above is for $x=e_i$, elementary calculations yield 
\[\M(\cur^\Phi_i)=\curg^{i,\Phi},\]
where the $\curg^{i,\Phi}$'s are the germs of closed forms introduced in equation \eqref{eq:Defcurg}. 
The application $\M$ therefore maps $J^*$ (cf. \eqref{eq:DefJstar}) into 
\[\boldsymbol{\mathfrak{J}}^*:=\left\{\curg^{1,\Phi_1}+\curg^{2,\Phi_2}, \quad \Phi_1, \Phi_2\in C^1(\ctoruspi)\right\}.\]
Since one can also write  $\M(f)=\nabla  \Sigma_{\gene^{-1}f}$,  we can define $\M$ on $ \gene \mathcal{C}$ as
\begin{equation*}\M(\gene f)=\nabla  \sum_{x\in \Z^2}\tau_x \gene^{-1}\gene  f=\boldsymbol{\nabla}\Sigma_f,\end{equation*}
which is the germ of an exact form associated with $f$.

\bigskip

Denote by $\mathfrak{E}^*$ the set of germs of exact forms associated with functions in $\mathcal{C}$, the construction above allow us to define  the bijective application 
\begin{equation*}\func{\cdf}{J^*+\gene \mathcal{C}}{\boldsymbol{\mathfrak{J}}^*+ \mathfrak{E}^*}{\cur_1^{\Phi_1}+\cur_2^{\Phi_2}+\gene f}{\curg^{1,\Phi_1}+\curg^{2,\Phi_2}+\boldsymbol{\nabla}\Sigma_f}.\end{equation*}
Recall that we defined the $L^2$-norm of any closed form $\uf$ as 
\[\norm{\uf}_{2,\param}=\cro{\Egcm\pa{\uf_{1}^2+\uf_{2}^2}}^{1/2}.\] 
According to Proposition \ref{prop:GCF}, we can rewrite for any $\uf\in \bff{T}^{\omega}$,
\begin{equation}\label{norme}\norm{\uf}_{2,\param}^2= \sup_{\substack{g\in \tcal\\
a, b\in \R^2}}\left\{2\Egcm\pa{\uf\cdot(\boldsymbol{\nabla}\Sigma_g+\curg^{a,b})}-\norm{\boldsymbol{\nabla}\Sigma_g+\curg^{a,b}}_{2,\param}^2\right\}.\end{equation}
Define $Ker_{\param}(\cdf)$ the kernel of $\M$ w.r.t $\norm{\;.\;}_{2,\param}$, we can equip $\tzero/Ker_{\param}(\cdf)$  with the norm $\scal{\cdot}^{1/2}$ induced by the mapping $\M$, defined as 
\[\scal{f}=\norm{\cdf(f)}^2_{2,\param}= \sup_{\substack{g\in \tcal\\
a, b\in \R^2}}\left\{2\Egcm\pa{\cdf(f)\cdot(\boldsymbol{\nabla}\Sigma_g+\curg^{a,b})}-\norm{\boldsymbol{\nabla}\Sigma_g+\curg^{a,b}}_{2,\param}^2\right\}.\]
By generalizing the integration by parts formula in the previous section, this formula is strictly analogous to Definition \ref{defi:limitingcovariance1}, and $\cdf$ is therefore an isomorphism
\begin{equation*}\cdf:\pa{\tzero/Ker_{\param}(\cdf)\;,\; \scal{\cdot}}\longrightarrow \pa{\bff{T}^{\omega}=\boldsymbol{\mathfrak{J}}^\omega+ \mathfrak{E}^\omega\;,\;\norm{\cdot}_{2,\param}^2},\end{equation*}
which gives $\tzero/Ker_{\param}(\cdf)$, as stated in Proposition \ref{prop:structureHalpha}, the same structure as $J^{\omega}+ \overline{\gene \tcal}/Ker_{\param}(\cdf) .$

\bigskip

We now  briefly carry on with our heuristics and explain why Theorem \ref{thm:limcovariance} holds, which is rigorously proved in Section \ref{sec:C}. The proof is based on the integration by parts obtained in Subsection \ref{subsec:IPP}. Applying it to $ \sum_{x\in B_{l_{\psi}}}\tau_x\psi$ yields that the quantity in the right-hand side of \eqref{limitcovQtty} can be rewritten  
\begin{equation*}
\lim_{l\to \infty}\frac{1}{(2l+1)^2}\Ecml\pa{\frac{1}{2}\sum_{x, x+z\in B_l}\cro{\nabla_{x,x+z}\gene_{l}^{-1} \sum_{x\in B_{l_{\psi}}}\tau_x\psi}^2}.
\end{equation*}
Assuming that one is able to replace $\mu_{l, \K_l}$ by the translation invariant grand-canonical measure $\mesinv$, and all quantities being ultimately translation invariant, this limit should be the same as
\begin{align*}
\lim_{l\to \infty}\frac{1}{(2l+1)^2}\Egcm\pa{\frac{1}{2}\sum_{x, x+z\in B_l}\cro{\nabla_{x,x+z}\gene_{l}^{-1} \sum_{x\in B_{l_{\psi}}}\tau_x\psi}^2}&=\lim_{l\to \infty}\Egcm\pa{\sum_{i=1,2}\cro{\nabla_{0,e_i}\gene_{l}^{-1} \sum_{x\in B_{l_{\psi}}}\tau_x\psi}^2}\\
&=\norm{\M(\psi)}^2_{2,\param}\\
&=\scal{\psi}.\end{align*}
The rigorous proof of this result, given in the next section, is technical due to the delicate nature of $\gene^{-1}$. 

\subsection{Proof of Theorem \ref{thm:limcovariance}}
\label{sec:C}
In order to prove Theorem \ref{thm:limcovariance}, we need to prove that 
\begin{equation}
\label{eq:aprouvers8}
\lim_{l\to \infty}\frac{1}{(2l+1)^2}\Ecml\pa{(-\gene_{l})^{-1} \sum_{x\in B_{l_{\psi}}}\tau_x \psi\; .\sum_{x\in B_{l_{\varphi}}}\tau_x \varphi}=\;\scal{\psi, \varphi}
\end{equation}
in three cases~:
\begin{enumerate}
\item $\varphi=\psi$ and $\psi\in \gene\mathcal{C}+J^*$,
\item $\varphi\in \tzero$ and $\psi\in \gene\mathcal{C}+J^*$,
\item $\varphi=\psi$ and $\psi\in \tzero$.
\end{enumerate}
The first two cases correspond to Definition \ref{defi:limitingcovariance1}, whereas the last one corresponds to Definition \ref{defi:limitingcovariance2}. The first two cases are easier, we treat them first as a separate Lemma. The uniformity of the convergence will be proved at the end of the section as in \cite{KLB1999}.
\begin{lemm}
\label{lemsup}
Fix $\varphi\in \tzero$ and  $\psi=\gene g+j_1^{\Phi_1}+j_2^{\Phi_2}\in \gene\mathcal{C}+J^*$. For any sequence $(\K_l)$ such that $\param_{\K_l}\to\param$, 
\begin{equation*}
\lim_{l\to\infty}\frac{1}{(2l+1)^2}\Ecml\pa{(-\gene_l^{-1})\sum_{x\in B_{l_\psi}}\tau_x\psi  \cdot\sum_{x\in B_{l_\psi}}\tau_x\psi }=\sum_{i=1}^2\Egcm\pa{\conf_0(1-\conf_{e_i})\cro{\Phi_i(\theta_0)+\Sigma_g(\confhat^{0,e_i})-\Sigma_g }^2},
\end{equation*}
and 
\begin{equation}
\label{eq:lemsup2}
\lim_{l\to\infty}\frac{1}{(2l+1)^2}\Ecml\pa{(-\gene_l^{-1})\sum_{x\in B_{l_\psi}}\tau_x\psi  \cdot\sum_{x\in B_{l_\varphi}}\tau_x \varphi }=-\Egcm\pa{\varphi\cro{\Sigma_g+\sum_{x\in \Z^2}\pa{x_1\conf^{\Phi_1}_x+x_1\conf^{\Phi_1}_x}}}.
\end{equation}	
\end{lemm}
\proofthm{Lemma \ref{lemsup}}{
Fix $\psi=\gene g+j_1^{\Phi_1}+j_2^{\Phi_2}\in \gene\mathcal{C}+J^*$, and shorten $\widetilde{B}^i_l=\{x\in B_l, x_i\leq l-1\}$ one easily obtains the identity
\[\sum_{x\in \widetilde{B}^i_l}\tau_xj_i^{\Phi_i}=\gene_{l}\sum_{x\in B_l}x_i\conf_x^{\Phi_i}.\]
Shorten 
\[F=F^{g,\Phi_1, \Phi_2}_l:=\sum_{x\in B_{l_\psi}}\tau_x g +\sum_{\substack{i=1,2,\\x\in B_l}}x_i\conf_x^{\Phi_i} \eqand G=-\sum_{\substack{i=1,2,\\x\in \widetilde{B}^i_l\setminus B_{l_\psi}}}\tau_{x}j_i^{\Phi_i},\]
we can then rewrite $\sum_{x\in B_{l_{\psi}}}\tau_x \psi=\gene_lF+G,$ and therefore 
\begin{equation}
\label{eq1sec8}
\Ecmlk\pa{(-\gene_{l}^{-1}) \sum_{x\in B_{l_{\psi}}}\tau_x \psi\; .\sum_{x\in B_{l_{\psi}}}\tau_x \psi}=\Ecmlk\pa{F (-\gene_l) F}-2\Ecmlk\pa{FG}+\Ecmlk\pa{G(-\gene_{l})^{-1}G}. 
\end{equation}
Writing \[\Ecmlk\pa{G(-\gene_{l})^{-1}G}=\sup_h\{\Ecmlk(Gh)-\dir_{l,\K}(h)\},\]
and using Lemma \ref{lem:techps}, we obtain that the last term in \eqref{eq1sec8} is less than $C(\Phi_1, \Phi_2)|\widetilde{B}^i_l\setminus B_{l_\psi}|=O(l)$, 
and therefore the corresponding contribution vanishes in the limit \eqref{eq:aprouvers8}. 
Regarding the second term, elementary computations yield
\[\Ecml(\conf_y^{\Phi_i}\tau_{x}j_k^{\Phi_k})=C(\1_{\{y=x\}}-\1_{\{y=x+e_k\}}),\]
where we shortened $C=\Ecml(\Phi_i\Phi_k(\theta_0)\conf_0(1-\conf_{e_k}))$, which yields after elementary computations that 
\[\Ecmlk\pa{\sum_{\substack{i=1,2,\\y\in B_l}}y_i\conf_y^{\Phi_i}\sum_{\substack{k=1,2,\\x\in \widetilde{B}^k_l\setminus B_{l_\psi}}}\tau_{x}j_k^{\Phi_k}}=O(l).\]
Similarly, for any $y$ such that $\{x, x+e_k\} \cap B_{s_g}(y)=\emptyset$, we have $\Ecml(\tau_y g\tau_{x}j_k^{\Phi_k})=0$, so that 
\[\Ecml\pa{FG}\leq C(g,\Phi_1,\Phi_2)|\widetilde{B}^i_l\setminus B_{l_\psi}|=O(l)\] 
and thus vanishes as well in the limit \eqref{eq:aprouvers8}. 

\medskip

Finally, the last two contributions in \eqref{eq1sec8} vanish in the limit, and we now only need to compute $\Ecml\pa{F (-\gene_l) F}$, that we split into three parts. We rewrite the first one 
\[\Ecmlk\pa{(-\gene_l)\sum_{x\in B_{l_\psi}}\tau_x g \cdot\sum_{x\in B_{l_\psi}}\tau_x g }=\frac{1}{2}\sum_{y,y+z\in B_l}\Ecmlk\pa{\cro{\nabla_{y,y+z}\sum_{x\in B_{l_\psi}}\tau_x g }^2}.\]
Since $f$ only depends on sites in $ B_{s_g}$, for any $y\in B_{l-2s_g-2}$, we can write $\nabla_{y,y+z}\sum_{x\in B_{l_\psi}}\tau_x g=\nabla_{y,y+z}\Sigma_g,$ where as before $\Sigma_g$ is the formal sum $\sum_{x\in \Z^2}\tau_xg$. Furthermore, for any $y\notin B_{l-2s_g-2}$ 
\[\cro{\nabla_{y,y+z}\sum_{x\in B_{l_\psi}}\tau_x g}^2=\cro{\nabla_{y,y+z}\sum_{\abs{x-y}\leq s_g+2}\tau_x g}^2\leq C(s_g)\norm{g}_\infty^2.\]
Since all the $\nabla_{y,y+z}\Sigma_g$ have the same distribution under $\mu_{l,\K}$ for $y\in B_{l-2s_g-2}$, we can therefore write using the two bounds above
\begin{multline}
\frac{1}{(2l+1)^2}\Ecmlk\pa{(-\gene_l)\sum_{x\in B_{l_\psi}}\tau_x g \cdot\sum_{x\in B_{l_\psi}}\tau_x g }\\
=\frac{|B_{l-2s_g-2}|}{2(2l+1)^2}\sum_{|z|=1}\Ecmlk\pa{\cro{\nabla_{0,z}\Sigma_g }^2}+C(f)O\pa{\frac{|B_l\setminus B_{l-2s_g-2}|}{(2l+1)^2}}=\sum_{i=1}^2\Ecmlk\pa{\cro{\nabla_{0,e_i}\Sigma_g }^2}+C(f)O(1/l). 
\end{multline}
Since $\nabla_{0,e_i}\Sigma_g$ is a local function, the equivalence of ensembles (cf. Proposition \eqref{prop:equivalenceofensembles}) finally yields for any sequence $\K_l$ such that $\param_{\K_l}\to\param$
\[\lim_{l\to\infty}\frac{1}{(2l+1)^2}\Ecml\pa{(-\gene_l)\sum_{x\in B_{l_\psi}}\tau_x g \cdot\sum_{x\in B_{l_\psi}}\tau_x g }
=\sum_{i=1}^2\Egcm\pa{\cro{\nabla_{0,e_i}\Sigma_g }^2} \]
as wanted.

\medskip

Similarly, one obtains straightforwardly after elementary computations
\[\Ecmlk\pa{(-\gene_l)\sum_{\substack{i=1,2,\\x\in B_l}}x_i\conf_x^{\Phi_i} \cdot\sum_{\substack{i=1,2,\\x\in B_l}}x_i\conf_x^{\Phi_i} }=\frac{1}{2}\sum_{y,y+z\in B_l}\Ecmlk\pa{\cro{\nabla_{y,y+z}\conf_y^{\Phi_{i_z}} }^2},\]
where $i_z=k$ iff $z=\pm e_k$. Once again, under $\mu_{l,\K}$, all the terms have the same distribution, and we can rewrite 
\[\frac{1}{(2l+1)^2}\Ecmlk\pa{(-\gene_l)\sum_{\substack{i=1,2,\\x\in B_l}}x_i\conf_x^{\Phi_i} \cdot\sum_{\substack{i=1,2,\\x\in B_l}}x_i\conf_x^{\Phi_i} }=\sum_{i=1}^2\Ecmlk\pa{\cro{\Phi_i(\theta_0)\conf_0(1-\conf_{e_i}) }^2}+C(\Phi_1, \Phi_2)O(1/l),\]
therefore using once again the equivalence of ensembles also yields 
\[\lim_{l\to\infty}\frac{1}{(2l+1)^2}\Ecml\pa{(-\gene_l)\sum_{\substack{i=1,2,\\x\in B_l}}x_i\conf_x^{\Phi_i} \cdot\sum_{\substack{i=1,2,\\x\in B_l}}x_i\conf_x^{\Phi_i} }=\sum_{i=1}^2\Egcm\pa{\cro{\Phi_i(\theta_0)\conf_0(1-\conf_{e_i}) }^2}.\]
Using the fact that $\Ecmlk(f\gene_l g)=-\sum_{y,y+z\in B_l}\Ecmlk([\nabla_{y,y+z}f][\nabla_{y,y+z}g]),$ is is straightforward to adapt the previous estimates to the cross term, and obtain
\[\lim_{l\to\infty}\frac{1}{(2l+1)^2}\Ecml\pa{(-\gene_l)\sum_{\substack{i=1,2,\\x\in B_l}}x_i\conf_x^{\Phi_i} \cdot \sum_{x\in B_{l_\psi}}\tau_x g}=\sum_{i=1}^2\Egcm\pa{\Phi_i(\theta_0)\nabla_{0,e_i}\Sigma_g }.\]
These three estimates finally yield as wanted 
\begin{equation}
\label{eq:FLF}
\lim_{l\to\infty}\frac{1}{(2l+1)^2}\Ecmlk\pa{F (-\gene_l) F}=\sum_{i=1}^2\Egcm\pa{\conf_0(1-\conf_{e_i})[\Phi_i(\theta_0)+\Sigma_g(\confhat^{0,e_i})-\Sigma_g]^2 },
\end{equation}
which proves the first statement of the Lemma.

\bigskip

The second identity in Lemma \ref{lemsup} is proved in a similar way. Using the same notations as for the first identity, we have $\sum_{x\in B_{l_\psi}}\tau_x\psi=\gene_lF +G$, and given $ f\in \tzero$, we rewrite the left-hand side in \eqref{eq:lemsup2}
\[\Ecml\pa{(F+(-\gene_l^{-1})G) \cdot\sum_{x\in B_{l_f}}\tau_x f }.\]
Using once again the equivalence of ensembles, it is easy to prove that 
\begin{equation}
\label{eq:scalFabg}
\lim_{l\to\infty}\frac{1}{(2l+1)^2}\Ecml\pa{F\sum_{x\in B_{l_f}}\tau_xf }=-\Egcm\pa{f\cro{\Sigma_g+\sum_{x\in \Z^2}\pa{x_1\conf^{\Phi_1}_x+x_1\conf^{\Phi_1}_x}}}, 
\end{equation}
therefore we only need to prove that the contribution of $G$ vanishes. This is straightforward, since the contribution of $G$ can be rewritten
\begin{multline*}\frac{1}{(2l+1)^2}\Ecml\pa{(-\gene_l^{-1})G \cdot (-\gene_l)(-\gene_l^{-1})\sum_{x\in B_{l_f}}\tau_xf }\\
=\frac{1}{(2l+1)^2}\cro{\frac{1}{2}\sum_{x,x+z\in B_l}\Ecml\pa{\nabla_{x,x+z}(-\gene_l^{-1})G \cdot \nabla_{x,x+z}(-\gene_l^{-1})\sum_{x\in B_{l_f}}\tau_xf }}.
\end{multline*}
We now use Holder's inequality, and that for any positive $\gamma$, $|ab|\leq \gamma a^2/2+b^2/2\gamma$, to obtain that the absolute value of the left-hand side above is less than
\begin{multline*}
\abs{\frac{1}{(2l+1)^2}\Ecml\pa{(-\gene_l^{-1})G \cdot \sum_{x\in B_{l_f}}\tau_xf }}\\
\leq \frac{\gamma}{2(2l+1)^2}\Ecml\pa{G(-\gene_l^{-1})G} +\frac{1}{2\gamma(2l+1)^2}\Ecml\pa{(-\gene_l^{-1})\sum_{x\in B_{l_f}}\tau_xf \cdot\sum_{x\in B_{l_f}}\tau_xf}. 
\end{multline*}
We already proved that the first term in the right-hand side is $O(\gamma l^{-1})$, whereas  in the limit $l\to\infty$ the second is bounded by $\scal{f}/\gamma$ according to Lemma \ref{lem:bornesup} below. We can therefore choose $\gamma=\sqrt{l}$, to obtain that both terms vanish as $l\to\infty$,  thus concluding the proof of Lemma \ref{lemsup}.
}

We now consider the case $\psi\in \tzero$, which is the main result of this section, and conclude by proving that the convergence is uniform 
and that \eqref{covunif} holds. Thanks to the decomposition of the germs of closed forms obtained in Proposition \ref{prop:GCF} and Lemma \ref{lemsup} above, 
these two steps follow closely Section 7.4 of \cite{KLB1999}, we repeat the proof here for the sake of exhaustivity.
Recall that we denoted for any $\psi\in \tzero$
\begin{equation*}\scal{\psi}=\sup_{\substack{g\in \tcal\\
a, b\in \R^2}}\left\{2\Egcm\pa{\psi.\cro{ \Sigma_g+\sum_{y\in\Z^2} (y.a)\com_y+(y.b)\conf_y}}-\scal{\gene g+\cur^{a,b}}\right\}.\end{equation*}
We split the proof of the third  case $\psi\in\tzero$ in two Lemmas, namely an upper and a lower bound. 
Using the identities obtained in Lemma \ref{lem:techps}, the lower bound is easy to prove.
\begin{lemm}\label{lem:borneinf}Under the assumption of Theorem \ref{thm:limcovariance}, 
\begin{equation}\label{secondemaj}\liml\frac{1}{(2l+1)^2}\Ecml\pa{(-\gene_{l}^{-1}) \sum_{x\in B_{l_{\psi}}}\tau_x \psi\; .\sum_{x\in B_{l_{\psi}}}\tau_x \psi}\geq\; \scal{\psi}.\end{equation}
\end{lemm}
\proofthm{Lemma \ref{lem:borneinf}}{{Denote by $\mathcal{C}_l$ the set of local functions measurable w.r.t. sites in $B_l$.} We start by writing the variational formula
\begin{align}\label{FormuleVar}\Ecml\pa{(-\gene_{l}^{-1}) \sum_{x\in B_{l_{\psi}}}\tau_x \psi\; .\sum_{x\in B_{l_{\psi}}}\tau_x \psi}&=\sup_{h\in{\mathcal{C}_l}}\left\{2\Ecml\pa{ h \sum_{x\in B_{l_{\psi}}}\tau_x \psi\; }-\dir_{l,\K_l}(h)\right\}\nonumber\\
&\geq \sup_{h\in \widetilde{\mathcal{T}}^\omega_l}\left\{2\Ecml\pa{h \sum_{x\in B_{l_{\psi}}}\tau_x \psi\; }-\dir_{l,\K_l}(h)\right\},\end{align}
where $\widetilde{\mathcal{T}}^\omega_l$ is the subspace of ${\mathcal{C}_l}$ 
\[\widetilde{\mathcal{T}}^\omega_l=\left\{F_l^{g,a,b}=\sum_{x\in B_{l_g}}\tau_x g+\sum_{x\in B_l} ((a.x)\com_x+(b.x)\conf_x),\quad g\in \tcal, a, b\in \R^2\right\}.\]
As stated in \eqref{eq:scalFabg} the contribution of the first term in \eqref{FormuleVar} is
\[\lim_{l\to \infty}\frac{1}{(2l+1)^2}\Ecml\pa{ \sum_{x\in B_{l_{\psi}}}\tau_x \psi\; .F_l^{g,a,b}}=-\Egcm\pa{\psi \sum_{y\in \Z^2}\cro{\tau_y g+\sum_{i=1}^2 ((a.x)\com_y+(b.y)\conf_y)}}.\]
and we proved in \eqref{eq:FLF} that 
\[\lim_{l\to\infty}\frac{1}{(2l+1)^2}\dir_{l,\K_l}(F_l^{g,a,b})=\scal{\gene g +j^{a,b}}.\]
These two identities prove \eqref{FormuleVar}, and concludes the proof of the Lemma.
}

We now state and prove the upper bound, which is more difficult.
\begin{lemm}\label{lem:bornesup}Under the assumptions of Theorem \ref{thm:limcovariance}, for any $\psi\in \tzero$,
\begin{equation}
\label{eq:bornesupscal}
\liml \frac{1}{(2l+1)^2}\Ecml\pa{(-\gene_{l})^{-1} \sum_{x\in B_{l_{\psi}}}\tau_x \psi\; .\sum_{x\in B_{l_{\psi}}}\tau_x \psi}\leq\;\scal{\psi}. 
\end{equation}
\end{lemm}
\proofthm{Lemma \ref{lem:bornesup}}{We start by replacing the canonical measure $\mu_{\K_l,l}$ by the grand-canonical measure $\mu_{\param}$ thanks to the equivalence of ensembles stated in Proposition \ref{prop:equivalenceofensembles}. The main obstacle in doing so is that the support of the function whose expectation we want to estimate grows with $l$. 

By the variational formula for the variance, we can write for any $\K\in \Ksett_l$
\[\Ecmlk\pa{(-\gene_{l})^{-1} \sum_{x\in B_{l_{\psi}}}\tau_x \psi\; .\sum_{x\in B_{l_{\psi}}}\tau_x \psi}=\sup_{h\in \tcall}\left\{2\Ecmlk\pa{ \sum_{x\in B_{l_{\psi}}}\tau_x \psi\; .h}-\dir_{l, \K}(h)\right\}\]
where as before, {$\tcall=\mathcal{C}_l\cap T^\omega$  and} $\dir_{l,\K}(h)=\Ecml\pa{h. (-\gene_l h)}$.
As in the proof of the one-block-estimate, let $k$ be an integer that will go to $\infty$ after $l$, and let us partition $B_l$ into disjoint boxes $\widetilde{\Lambda}_0,\ldots ,\widetilde{\Lambda}_p$, where $p=\lfloor\frac{2l_\psi+1}{2k+1}\rfloor^2$, $\widetilde{\Lambda}_j=B_{2k+1}(x_j)$ for any $1\leq j\leq p$ and some family of sites $x_1,\dots,x_p$, and where we let $\widetilde{\Lambda}_0=B_{l_\psi}\setminus(\cup_{j=1}^p \widetilde{\Lambda}_j)$.
Recall that $s_{\psi}$ is the smallest integer such that $\psi$ is  measurable with respect to the sites in $B_{s_{\psi}}$, we now define 
\[\Lambda_j=\{x\in \widetilde{\Lambda}_j,\; d(x,\widetilde{\Lambda}_j^c)> s_{\psi}\} \eqand \Lambda_0=B_{l_\psi}\setminus(\cup_{j=1}^p \Lambda_j).\]
One easily obtains that for some universal constant $C$, $|\Lambda_0|\leq Cs_{\psi}(l^2/k+lk)$.

Let $h$ be a function in  $\tcall$, we can split
\begin{equation}\label{decompVar}\sum_{x\in B_{l_{\psi}}}\Ecml\pa{ \tau_x \psi\; .h}=\sum_{\substack{j=1,\dots,p\\
x\in \Lambda_j}}\Ecml\pa{ \tau_x \psi\; .h}+\sum_{x\in \Lambda_0}\Ecml\pa{ \tau_x \psi\; .h}.\end{equation}

Letting $\gamma=\sqrt{k}/2$ in Lemma $\ref{lem:techps}$,  for any $l\geq k^2$, the second term is less   than $k^{-1/2}\big[C(\psi)l^2+\dir_{l,\K}(h)\big]$. 
Letting $c_k=1-k^{-1/2}$, for some constant $C(\psi)$, and for any $l\geq k^2,$ the left-hand side of \eqref{eq:bornesupscal} is therefore less than
\[\frac{c_k}{(2l+1)^2}\sup_{h\in \tcall}\left\{\sum_{\substack{j=1,\dots,p\\
x\in \Lambda_j}} \frac{2}{c_k}\Ecmlk\pa{\tau_x \psi\; .h}-\dir_{l, \K}(h)\right\}+\frac{C(\psi)}{\sqrt{k}}.\]
For any $h\in \tcall$, $1\leq j\leq p$ define $h_j=\Ecmlk(h\mid \confhat_y, y\in \widetilde{\Lambda}_j)$, by convexity of the Dirichlet form, we have 
\[\dir_{l,\K}(h)\geq \sum_{j=1}^p\dir^{\widetilde{\Lambda}_j}_{l,\K}(h)\geq \sum_{j=1}^p\dir_{l,\K}^{\widetilde{\Lambda}_j}(h_j),\]
where as before $\dir^{A}_{l,\K}(h)$ is the contribution to the Dirichlet form of edges in $A$. Denoting ${T^\omega_{k,j}}$ the set of functions in ${T^\omega}$ measurable w.r.t. sites in 
$\widetilde{\Lambda}_j$, we can therefore finally bound from above the left-hand side of \eqref{eq:bornesupscal} by
\[\frac{c_k}{(2l+1)^2}\sum_{j=1}^p\sup_{h\in {T^\omega_{k,j}}}\left\{\sum_{x\in \Lambda_j} \frac{2}{c_k}\Ecml\pa{\tau_x \psi\; .h}-\dir^{\widetilde{\Lambda}_j}_{l, \K_l}(h)\right\}+\frac{C(\psi)}{\sqrt{k}}.\]
All the terms in the sum over $j$ are identically distributed, the quantity above is thus less than
\begin{multline*}
\frac{c_k}{(2k+1)^2}\sup_{h\in {T^\omega_k}}\left\{\sum_{x\in B_{k_\psi}} \frac{2}{c_k}\Ecml\pa{\tau_x \psi\; .h}-\dir^{B_{k}}_{l, \K_l}(h)\right\}+\frac{C(\psi)}{\sqrt{k}}\\
=\frac{1}{c_k(2k+1)^2}\Ecml\pa{(-\gene_k^{-1})\sum_{x\in B_{k_\psi}}\tau_x\psi\cdot\sum_{x\in B_{k_\psi}}\tau_x\psi}+\frac{C(\psi)}{\sqrt{k}}. 
\end{multline*}
The quantity inside the expectation is now a local function w.r.t. $l$, we can now let $l\to\infty$ and as $\param_{\K_l}\to \param$, replace $\cml$ by $\gcm$ by the equivalence of ensembles stated in Proposition \ref{prop:equivalenceofensembles}. Letting then $k\to\infty$, we finally obtain
\begin{multline}
\label{eq:bornesupscal2}
\liml \frac{1}{(2l+1)^2}\Ecml\pa{(-\gene_{l})^{-1} \sum_{x\in B_{l_{\psi}}}\tau_x \psi\; .\sum_{x\in B_{l_{\psi}}}\tau_x \psi}\\
\leq\;\limsup_{k\to\infty}\frac{1}{(2k+1)^2}\Egcm\pa{(-\gene_{k})^{-1} \sum_{x\in B_{k_{\psi}}}\tau_x \psi\; .\sum_{x\in B_{k_{\psi}}}\tau_x \psi}. 
\end{multline}

\medskip

By the variational formula for the variance, to prove the Lemma it is enough to show
\begin{equation}
\label{ConvGCM}
\limsup_{k\to\infty}\frac{1}{(2k+1)^2}\sup_{h\in {T^\omega_k}}\left\{2\Egcm\pa{h\sum_{x\in B_{k_{\psi}}}\tau_x \psi}-\dir_{\param,k}(h)\right\}\leq \scal{\psi},
\end{equation}
where we shortened $\dir_{\param,k}(h)=\Egcm(h(-\gene_k) h)$.  According to Lemma \ref{lem:techps}, there exists a constant $C(\psi)$ such that the first term $2\Egcm\pa{ \sum_{x\in B_{k_{\psi}}}\tau_x \psi\; .h}$ is less than $C(\psi)(2k+1)^2+\dir_{\param,k}(h)/2$. For any $h$ such that $\dir_{\param,k}(h)\geq 2C(\psi)(2k+1)^2$, the right-hand side above is therefore negative, and since it vanishes for $h=0$, we can therefore safely assume that the supremum is taken w.r.t. functions  $h\in {T^\omega_k}$ satisfying $\dir_{\param,k}(h)\leq 2C(\psi)(2k+1)^2$.
Using the integration by parts formula of Lemma \ref{lem:IPP} yields 
\[\Egcm\pa{ \tau_x \psi\; .h}=\sum_{x\in B_{\psi}(x)}\Egcm(I_{a}(\tau_x\psi)\grad h),\] 
where $I_{a}(\psi)=(1/2)\grad (-\gene_{s_{\psi}})^{-1}\psi$.
For any edge $a$,  let us denote by $B^{\psi}(a)$ the set of sites $x\in \Z^2$ such that $a $ is in $B_{\psi}(x)$, and $\widetilde{B}_k^{\psi}(a)=B^{\psi}(a)\cap B_{k_{\psi}}$. Note that for any edge $a\in B_{k_{\psi}-s_{\psi}}$, these two sets coincide. The integration by parts formula then yields
\begin{align*}\sum_{x\in B_{k_{\psi}}}\Egcm\pa{ h \tau_x \psi}=&\sum_{a\in B_k}\sum_{x\in \widetilde{B}_k^{\psi}(a)} \Egcm(I_{a}(\tau_x\psi)\grad h)\\
=&\sum_{a\in B_{k}}\sum_{x\in B^{\psi}(a)} \Egcm(I_{a}(\tau_x\psi)\grad h)-\sum_{a\in B_k}\sum_{x\in B^{\psi}\setminus \widetilde{B}_k^{\psi}(a)} \Egcm(I_{a}(\tau_x\psi)\grad h)\\
=&\sum_{a\in B_{k}}\sum_{x\in B^{\psi}(a)} \Egcm(I_{a}(\tau_x\psi)\grad h)-\sum_{a\in B_k\setminus B_{k_{\psi}-s_{\psi}}}\sum_{x\in B^{\psi}\setminus \widetilde{B}_k^{\psi}(a)} \Egcm(I_{a}(\tau_x\psi)\grad h).\end{align*}
For any positive $\gamma$,   
\[\Egcm(I_{a}(\tau_x\psi)\grad h)\leq \frac{1}{2\gamma}\Egcm(I_{a}(\tau_x\psi)^2)+\frac{\gamma }{2}\Egcm((\grad h)^2),\]
since $|B_k\setminus B_{k_{\psi}-s_{\psi}}|\leq C(\psi)k$,  and thanks to the bound on $\dir_{\param,k}(h)$, letting $\gamma=1/\sqrt{k}$, it is then straightforward to obtain 
\[\sum_{a\in B_k\setminus B_{k_{\psi}-s_{\psi}}}\sum_{x\in B^{\psi}\setminus \widetilde{B}_k^{\psi}(a)} \Egcm(I_{a}(\tau_x\psi)\grad h)\leq C(\psi)k^{3/2}.\]
therefore its contribution to the left-hand side of \eqref{ConvGCM} vanishes in the limit $k\to\infty$. Letting $\overline{I}_{a}(\psi)=\sum_{x\in B^{\psi}(a)}I_{a}(\tau_x\psi)$, the left-hand side of equation \eqref{ConvGCM} is therefore less than 
\begin{multline}\label{limite}\limsup_{k\to \infty}\frac{1}{(2k+1)^2} \sup_{h\in {T^\omega_k}} \left\{2\sum_{a\in B_{k}} \Egcm(\overline{I}_{a}(\psi)\grad h) -\dir_{\param,k}(h) \right\}\\
=\lim_{k\to \infty}\frac{1}{(2k+1)^2} \left\{2\sum_{a\in B_{k}} \Egcm(\overline{I}_{a}(\psi)\grad h_k) -\dir_{\param,k}(h_k) \right\}.\end{multline}
for some sequence of functions $h_k\in {T^\omega_k}$ ultimately realizing the limit $k\to\infty$ of the left-hand side.

\bigskip

Thanks to the translation invariance of $\mu_{\param}$, and since $\tau_y \overline{I}_{a }(\psi)= \overline{I}_{\tau_ya }(\psi)$, letting $y=a_1$ be the first site of the edge $a=(a_1,a_2)$, we have
\[\Egcm(\overline{I}_{a}(\psi)\grad h_k)=\Egcm\pa{\overline{I}_{(0, a_2-a_1)}(\psi)\nabla_{(0, a_2-a_1)}\tau_{-a_1} h_k}.\]
A seen before, a simple change of variable yields that $\Egcm\pa{\grad f.\grad g}=\Egcm\pa{\nabla_{-a} f.\nabla_{-a} g}$, from which we deduce 
\[2\sum_{a\in B_{k}} \Egcm(\overline{I}_{a}(\psi)\grad h_k)=4\sum_{i=1,2} \Egcm\pa{\overline{I}_{(0,e_i)}(\psi).\nabla_{(0,e_i)}\sum_{\substack{x \\
x, x+e_i\in B_k}}\tau_{-x} h_k}.\]
Define
\[\uf_{i}^k=\frac{1}{(2k+1)^2}\nabla_{(0,e_i)}\sum_{\substack{x \\
x, x+e_i\in B_k}}\tau_{-x} h_k\in {T^\omega}. \]  
The elementary bound $(\sum_{i=1}^n a_i)^2\leq n\sum_{i=1}^n a_i^2$ yields
\begin{align*}\sum_{i=1,2}\Egcm((\uf_{i}^k)^2)\leq &\frac{2k(2k+1)}{(2k+1)^4}\sum_{\substack{x \\
x, x+e_i\in B_k}}\Egcm\pa{\pa{\nabla_{(x,x+e_i)} h_k}^2}\\
\leq &\frac{1}{(2k+1)^2}\dir_{\param,k}(h_k)
\end{align*}
Thanks to this bound, equation \eqref{limite} yields 
\[\frac{1}{(2k+1)^2}\Egcm\pa{(-\gene_{k})^{-1} \sum_{x\in B_{k_{\psi}}}\tau_x \psi\; .\sum_{x\in B_{k_{\psi}}}\tau_x \psi}\leq \lim_{k\to \infty} \left\{4\sum_{i=1,2} \Egcm(\overline{I}_{(0,e_i)}(\psi).\uf_{i}^k) -\sum_{i=1,2} \Egcm((\uf_{i}^k)^2) \right\},\]
and since we already assumed that for some constant $C(\psi)$, $\dir_{\param,k}(h_k)\leq C(\psi)(2k+1)^2$, the sequence of differential forms $(\uf^{k})_{k\in \N}$ 
is bounded in $L^2(\mu_{\param})$. It is straightforward to check that any of its limit point $\uf=(\uf_1, \uf_2)$ is the germ of a closed form in $\bff{T}^{\omega}$ {in the sense of  Definition \ref{defi:GCF}).

Indeed, given a limit point $\uf$ and a finite path $\gamma$ defined by jumps $x_i,$ $x_i+z_i$, $0\leq i\leq q_\gamma -1$, we can write for the closed form $\ufbar$ associated with $\uf$
\[\Egcm(\1_{\gamma\in \Gamma_c(\confhat)}|I_{\gamma, \ufbar}(\confhat)|)=\lim_{k\to\infty}\Egcm(\1_{\gamma\in \Gamma_c(\confhat)}|I_{\gamma, \ufbar^k}(\confhat)|),\]
where $\ufbar^k$ is the (non-closed) differential form 
\[\ufbar^k_{x,x+z}=\frac{1}{(2k+1)^2}\nabla_{(x,x+z)}\sum_{\substack{y \\
y, y+z\in B_k(x)}}\tau_{-y} h_k.\]
Since $\gamma$ is a finite path, it depends on edges in a finite box $B_n$, with $n$ fixed. In particular, for any $y\in B_{k-n}$, when computing $I_{\gamma, \ufbar^k}(\confhat)$, the contribution of $\tau_{-y}h_k$ vanishes since it involves the complete path. We can therefore write for some constant $C_\gamma$ and any $k>n$, 
\[\Egcm(\1_{\gamma\in \Gamma_c(\confhat)}|I_{\gamma, \ufbar^k}(\confhat)|)\leq \frac{q_\gamma}{(2k+1)^2}\sum_{\substack{y, \; y+e_i\in B_k\\
y\; or \; y+e_i\notin B_{k-n}}}\Egcm\pa{|\nabla_{0,e_i}\tau_{-y}h_k|}\leq \frac{q_\gamma}{(2k+1)^2} \pa{C_{n,k}\dir_{\param,k}(h_k)}^{1/2},\]
where $C_{n,k}\leq c nk$ is the cardinal of the $y$'s such that $y$ and  $y+e_i$ are in $B_k$ and either $y$ or $ y+e_i$ are not in $B_{k-n}$. Since $\dir_{\param,k}(h_k)\leq C(\psi)(2k+1)^2$, the right-hand side above vanishes as $k\to\infty$ for any path $\gamma$. This proves that $ \Egcm(\1_{\gamma\in \Gamma_c(\confhat)}|I_{\gamma, \ufbar}(\confhat)|)=0$ for any path $\gamma$, and any limit point $\uf$ of $(\uf^k)_k$, and in particular $\1_{\gamma\in \Gamma_c(\confhat)}|I_{\gamma, \ufbar}(\confhat)|$ vanishes $\mesinv$-a.s. for any finite path $\gamma$.

}We can therefore write 
\[\frac{1}{(2k+1)^2}\Egcm\pa{(-\gene_{k})^{-1} \sum_{x\in B_{k_{\psi}}}\tau_x \psi\; .\sum_{x\in B_{k_{\psi}}}\tau_x \psi}\leq \sup_{\uf\in \bff{T}^{\omega}} \left\{4\sum_{i=1,2} \Egcm(\overline{I}_{(0,e_i)}(\psi).\uf_{i}) -\sum_{i=1,2} \Egcm(\uf_{i}^2) \right\},\]
{where $\bff{T}^{\omega}$ is the set of germs of closed forms introduced in Definition \ref{defi:GCF}.}

According to Proposition \ref{prop:GCF}, the estimate above becomes
\begin{align*}\frac{1}{(2k+1)^2}\Egcm\Bigg((-\gene_{k})^{-1} &\sum_{x\in B_{k_{\psi}}}\tau_x \psi\; .\sum_{x\in B_{k_{\psi}}}\tau_x \psi\Bigg)\\
&\leq\sup_{\substack{g\in \tcal\\
a,b\in \R^2}}\left\{4\sum_{i=1,2} \Egcm(\overline{I}_{(0,e_i)}(\psi).(\curg^{a,b}_i+\nabla_{(0,e_i)}\Sigma_g)) -\sum_{i=1,2} \Egcm((\curg^{a,b}+\boldsymbol{\nabla}\Sigma_g)^2) \right\}\\
&=\sup_{\substack{g\in \tcal\\
a, b\in \R^2}}\left\{2\Egcm\pa{\psi.\cro{ \Sigma_g+\sum_{y\in\Z^2} (y.a)\com_y+(y.b)\conf_y}}-\scal{\gene g+\cur^{a,b}}\right\}.\end{align*}
The last identity is derived as in the proof of Lemma \ref{lemsup}. The right-hand-side above is $\scal{\cdot}$ as defined in Definition \ref{defi:limitingcovariance2}, which concludes the proof of the upper bound.
}

\bigskip

In order to complete the proof of Theorem \ref{thm:limcovariance}, we still need to prove that the convergence is uniform in $\param$, 
to prove \eqref{covunif}. Let us denote 
\[V_{l,\psi, \varphi}(\param_{\K_l})=\frac{1}{(2l+1)^2}\Ecml\pa{-\gene_{l}^{-1} \sum_{x\in B_{l_{\psi}}}\tau_x \psi\; .\sum_{x\in B_{l_{\varphi}}}\tau_x \varphi},\] 
and let us extend smoothly the domain of definition of $V_{l,\psi, \varphi}$ to $ \pset$. The three previous Lemmas yield that
$V_{l,\psi, \varphi}(\K_l(2l+1)^{-2}))$ converges as $l$ goes to $\infty$ to $\scal{\psi, \varphi}$ as soon as $\K_l$ converges towards the profile $\param$, 
hence in particular, $V_{l,\psi, \varphi}(\param_l)$ converges as $l$ goes to $\infty$ towards $\scal{\psi, \varphi}$ as soon as $\param_l$ goes to 
$\param$. For that reason, $ \scal{\cdot}$ is continuous, and $V_{l,\psi, \varphi}(\param)$ converges uniformly in $\param$ towards $\scal{\psi, \varphi}$ 
as $l$ goes to $\infty$.
This, combined with the three lemmas \ref{lemsup}, \ref{lem:bornesup} and \ref{lem:borneinf}, completes the proof of Theorem \ref{thm:limcovariance}.

\appendix

\section{Possible application~: Coarsening and global order in active Matter}
\label{sec:contexte}
\intro{We give some context on the modeling of collective dynamics and  the rich phenomenology of active matter.}
\subsection{Collective motion among biological organisms}

Collective motion is a widespread phenomenon in nature, and has motivated in the last decades a fruitful and interdisciplinary field of study \cite{PEK1999}.  Such behavior can be observed among many animal species, across many scales of the living spectrum, and in a broad range of environments. Animal swarming usually needs to balance out the benefits of collective behavior (defense against predation, protection of the young ones, increased vigilance) against the drawback of large groups (food hardships, predator multiplication, etc.).

Despite the numerous forms of interaction between individuals, all of these self-organization phenomenons present spontaneous emergence of density fluctuations and long range correlations. This similarity suggests some universality of collective dynamics models \cite{GC2004}, \cite{VZ2012}.
Even though the biological reasons for collective behavior are now well known, the underlying microscopic and macroscopic mechanisms are not yet fully understood. To unveil these mechanisms, numerous aggregation models have been put forward.

These models can be built on two distinct principles. The first approach specifies the macroscopic partial differential equation which rules the evolution of the local density of individuals. The main upside is that one can use the numerous tools developed for solving PDE's. Several examples of such models are presented in Okubo and Levin's book, \cite{OLB2001}. Since it represents an average behavior, this approach to collective dynamics is, however, mainly fitted to describe systems with large number of individuals, and does not take into account the fluctuations to which smaller systems are subject.

The second approach, called \emph{Individual-Based Models} (IBM),  specifies the motion of each individual organism. If the motion of each individual was described realistically (from a biological standpoint), the theoretical study of these models with large number of degrees of freedom would be extremely difficult. For this reason, it is usually preferred to simplify the rules for the motion of each individual, as well as its interaction with the group. A classical simplification is to consider that the interaction of each individual with the group is averaged out   over a large number of its neighbors. This so-called \emph{local field} simplification often allows to obtain explicit results, at the expense however of their biological accuracy (cf. below). 


\subsection{Microscopic active matter models}

In order to represent the direction of the motion of each individual, as well as spatial constraints (e.g. volume of each organism), collective dynamics are often modeled by individual-based \emph{active matter} models. Active matter is characterized by an energy dissipation taking place at the level of each individual particle, which allows it to self-propel, thus yielding an extra degree of freedom representing the direction of its motion. One can therefore obtain a phase transition towards collective motion when these directions align on lengths large with respect to the size of the particles. Active matter models exhibit various behaviors, and in the context of collective motion, two phenomena are particularly important~:

\begin{itemize}
\item when each particle tends to align the direction of its motion to that of its neighbors, one can observe a phase transition between order and disorder depending on the strength of the alignment. This \emph{alignment phase transition} was first observed in an influential model for collective dynamics introduced by Vicsek et al. \cite{Vicsek1995}
\item When the particle's velocity decreases with the local density, congestion effects appear~: particles spend more time where their speed is lower, and therefore tend to accumulate there. This phenomenon, called Motility-Induced Phase Separation (MIPS), was extensively studied in the recent years \cite{CT2008}, \cite{FM2012}, \cite{CT2015}.
\end{itemize}

\subsubsection*{Vicsek model and phase transition in alignment models}

Interest for self-organization phenomenons have grown significantly in statistical physics, where the diversity of such behaviors opens numerous modeling perspectives, and raises new questions regarding out-of-equilibrium systems. Many stochastic models have been introduced to represent specific biological behavior using statistical physics methods and have revealed a phase transition between high density collective motion, and disordered behavior with short range correlations at low densities. 

\medskip 

A pioneering model was proposed in 1995 by Vicsek et al.  They introduce in \cite{Vicsek1995} a general IBM (cf. previous paragraph) to model collective dynamics. In the latter, a large number of particles move in discrete time, and update the direction of their motion to the average direction of the particles in a small neighborhood. The direction of their motion is also submitted to a small noise, which makes the dynamics stochastic.

Despite its relative simplicity, the original model described in \cite{Vicsek1995} is extremely rich, and has given rise to a considerable literature (cf. the review by Viczek and Zafeiris, \cite{VZ2012}). The first article on this model unveiled a phase transition between a high-noise, low-density disordered phase and a low-noise high-density ordered phase. Initially thought to be critical, this transition was later shown to be discontinuous \cite{Chate2004}, with an intermediate region in which an ordered band cruises in a disordered background. It was recently shown that this transition can be understood as a liquid-gas phase separation in which the coexistence phase is organized in a smectic arrangement of finite-width bands traveling collectively \cite{SChT2015}. Numerous extensions and variations on Vicsek's model have been put forward, usually by considering a continuous time dynamics, more pertinent to represent biological organisms.

\medskip

\begin{figure}
        \centering
	\begin{subfigure}{0.40\textwidth}
		\includegraphics[width=6cm]{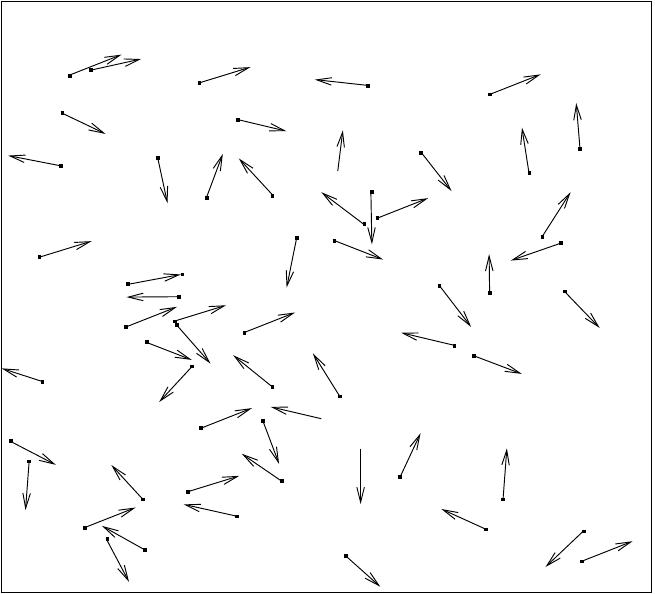}
       		\caption{}
                	\label{Vicsekrholow}
        \end{subfigure}
	\begin{subfigure}{0.4\textwidth}
		\includegraphics[width=6.4cm]{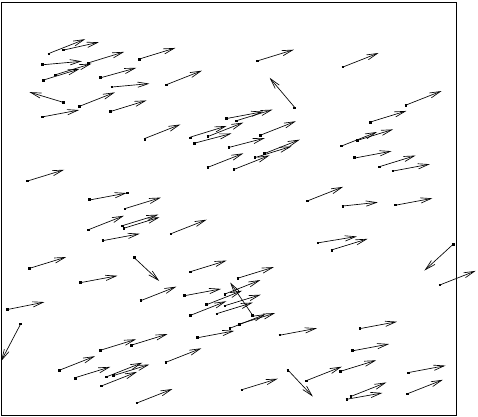}
         		\caption{}
                	\label{Vicsekrhobig}
        \end{subfigure}
\caption{Schematic representation of the phase transition in Vicsek's model. \\
({\sc a}) low density and high noise intensity, \\
({\sc b}) high density and low noise intensity.}
\label{VicsekFig}\end{figure}

Phase transitions are central to the study of collective dynamics, where coherent behavior arise when the alignment becomes strong enough. This notion of phase transition for alignment dynamics is reminiscent of the Ising  and $XY$ models, two classical statistical physics  models. The Ising model is known to have a symmetry breaking phase transition leading to the emergence of a spontaneous magnetization. Unlike the Ising model,  the \emph{$XY$ model} (for which the spins are two-dimensional unit vectors parametrized by angles $\theta\in [0,2\pi[$) does not present in two dimensions this type of symmetry  breaking phase transition, according to the Mermin-Wagner Theorem. This is one of the reasons for the popularity of the  Vicsek model \cite{Vicsek1995}, whose alignment dynamics is reminiscent of the $XY$ model, but unlike the latter presents a phase transition of the magnetization due to the particle motility \cite{TT1995}. Both the Ising and $XY$ models are now well understood. These are \emph{equilibrium models} and they  fall within the formalism of Gibbs measures, which relates to the thermodynamical parameters of the system.

\medskip

Active matter models like Vicsek's are out of equilibrium, and in the case of Vicsek's model, the phase transition is a dynamical phenomenon. The concepts developed for equilibrium models, namely Gibbs measures and free energy, can therefore no longer be used, and despite ample numerical evidence of spontaneous magnetization, (cf. \cite{SCB2015}) mathematically proving a phase transition becomes significantly harder. 

Despite these issues, several exact results have been obtained for systems closely related to Vicsek's model. In 2007, Degond and Motsch notably introduced a continuous time version of Vicsek's model, and derived the macroscopic scaling limit of the system \cite{DM2007}, as well as its microscopic corrections \cite{DY2010}. Their model, which was directly inspired by that of Vicsek et al., is a \emph{locally mean-field model}, where particles interact with all other particles present in a small macroscopic neighborhood. This approximation  simplifies a number of difficulties of out-of-equilibrium systems. In their initial article \cite{DM2007}, Degond and Motsch assume that a law of large number holds for the microscopic system. This was later rigorously proved in \cite{BCC2011}. The phase transition as a function of the noise level, between disordered system and global alignment, was shown in \cite{DFL2014} for this model. Similar results have since been extended to more general forms of alignment, (e.g. \cite{BCC2010}, \cite{CDW2011}, \cite{DLM2011}) and to density dependent parameters \cite{Frouvelle2011}. The evolution of the macroscopic density was also obtained in the particular case where the interaction between individuals is driven by a Morse potential, \cite{CHM2014}, where previously the shape of animal aggregates (e.g. fish schools mills) was only known empirically.

The \emph{Active Ising Model} (AIM) is another alignment model, phenomenologically close to Vicsek's model \cite{SCB2015}, put forward to better understand  collective dynamics. It is less demanding from a computational standpoint, and is extensively studied both numerically and theoretically by Solon and Tailleur in \cite{ST2015}. This model does not rely on the mean-field approximation of the Vicsek's model. The particles (with either ''+'' or ''-'' spins) move independently in a discrete space domain, performing an asymmetric random walk with drift directed according to the particle's spin. In addition to the displacement dynamics, the particles align their spins with the other particles on the same site as in a fully connected Ising model.

It was numerically shown in \cite{ST2015} that the AIM presents, as does Vicsek's, a phase transition depending both on the temperature and the particle density. At low temperature and density, one observes a magnetically neutral gas, whereas at strong temperature and densities, one obtains a strongly polarized liquid. In an intermediary domain, these two phases coexist. The AIM being an out-of-equilibrium model as well, its mathematical study is difficult, mainly because of the lack of mean-field approximation present in Vicsek's model. To our knowledge, there exists to this day no mathematical proof of the phase transition of the AIM. The model considered in this paper is closely related to both the Vicsek and the active Ising models.

\subsubsection*{Motility-Induced Phase Transition (MIPS)} As previously emphasized, a second interesting phenomenon can occur in active matter~: when the motility of the particles decreases as the local particle density increases, one can observe a phase separation between a low density gaseous phase, and condensed clusters. This separation is a direct consequence of particles slowing down in dense areas~: since they spend more time there, they tend to accumulate. This creates the congestion phenomenon called \emph{Motility Induced Phase Transition}, or \emph{MIPS}, which was thoroughly studied in recent years (cf. the review by Cates and Tailleur, \cite{CT2015}). 

This congestion phenomenon can be observed across several types of dynamics, under the condition that the particle's velocities and diffusion constants depend on the local density. One of the most studied  is the \emph{run-and-tumble dynamics} \cite{Berg2004}, which models the behavior of bacteria~: each individual goes in a straight line for a while, and then reorients in another random direction. However, MIPS is not specific to run and tumble dynamics~: it is shown numerically in \cite{CT2013}, \cite{SCT2015} that MIPS also occurs for active Brownian particles, for which each particles motion's direction diffuses, instead of updating at discrete times like in the run-and-tumble dynamics. MIPS can also be observed in lattice models \cite{TTC2011}, or in models with repulsive forces \cite{FM2012}, for which the kinetic slowdown is a consequence of repulsive forces.

\medskip

As already pointed out, one can expect that the active exclusion process investigated in this article may exhibit both MIPS and alignment phase transition. However, mathematically proving this statement is a difficult task, and this claim is left as a conjecture at this point.

\section{General tools}
This appendix regroups a general definitions and results that have been used throughout the proof.

\subsection{Topological setup}
\label{subsec:topo}
This paragraph defines the topological setup we endow the trajectories space for our process with. Denoting by $\meas(\ctorus\times\ctoruspi)$ the space of {non-negative} measures on the continuous configuration space, and 
\[\Mhat=D\pa{[0,T],\meas(\ctorus\times\ctoruspi)}\] 
the space of right-continuous and left-limited trajectories of measures on $\ctorus\times \ctoruspi$. Each trajectory $\confhat^{[0,T]}$ of our process admits a natural image in $\Mhat$ through its empirical measure 
\begin{equation}\label{mesemp}\pi_t^N\pa{\confhat^{[0,T]}}=\frac{1}{N^2}\sum_{x\in \torus}\conf_x(t)\delta_{x/N, \theta_x(t)}.\end{equation}
\index{$\meas(\ctorus\times \ctoruspi)$ \dotfill space of measures on $\ctorus\times\ctoruspi$}
\index{$\meas^{[0,T]}$ \dotfill space of c\`adl\`ag traj. on $\meas(\ctorus\times \ctoruspi)$}
\index{$\pi^N_t$ \dotfill empirical measure at time $s$}
Let $(f_k)_{k\in \N}$ be a dense family of functions in $C^{\infty}(\ctorus\times\ctoruspi)$, and assume that $f_0\equiv1$. The weak topology on $\meas(\ctorus\times\ctoruspi)$ is metrizable, by letting 
\begin{equation}\label{metricM}\delta(\pi_0, \pi'_0)=\sum_{k=0}^{\infty}\frac{1}{2^k}\frac{\abs{<\pi_0,f_k>-<\pi'_0,f_k>}}{1+\abs{<\pi_0,f_k>-<\pi'_0,f_k>}}.\end{equation}
Given this metric, $\Mhat$ is endowed with Skorohod's metric, defined as 
\begin{equation}\label{metricMhat}d(\pi,\pi')=\inf_{\kappa\in F}\max\left\{\norm{\kappa},\sup_{[0,T]}\delta(\pi_t,\pi'_{\kappa_t})\right\},\end{equation}
where $F$ is the set of strictly increasing continuous functions from $[0,T]$ into itself, such that $\kappa_0=0$ and $\kappa_T=T$, equipped with the norm
\[\norm{\kappa}=\sup_{s,t\in [0,T]}\left\{\log\cro{\frac{\kappa_s-\kappa_t}{s-t}}\right\}.\]
Now,  $(\Mhat,d)$ is a metric space, and we endow the set $\mathcal{P}(\Mhat)$ of probability measures on $\Mhat$ with the weak topology.

Given the empirical measure  $\pi^N_t$ of the process at time $t$, defined in equation \eqref{mesemp}, define the application \[\func{\pi^N}{\statespace^{[0,T]}}{\Mhat}{\confhat^{[0,T]}}{\pa{\pi^N_t\pa{\confhat^{[0,T]}}}_{t\in[0,T]}},\]  we define 
\begin{equation}\label{QNdef}Q^N=\Prob^{\lambda,\beta}_{\mu^N}\circ \pa{\pi^N}^{-1}\in \mathcal{P}(\Mhat)\end{equation} the pushforward of $\Prob^{\lambda,\beta}_{\mu^N}$ by $\pi^N$. 
\index{$Q^N$ \dotfill distribution of $(\pi^N_t)_{t\in[0,T]}$ for the active exclusion process}

\subsection{Self-diffusion coefficient}
\label{subsec:sdc}
We regroup in this paragraph some useful results regarding the self-diffusion coefficient. Consider on $ \Z^2$, an initial configuration where each site is initially occupied w.p. $\rho\in[0,1]$, and with a tagged particle at the origin. Each particle then follows a symmetric exclusion process with finite range transition matrix $p(\cdot)$, verifying $\sum_z zp(z)=0$, and $p(z)=0$ outside of a finite set of vertices $B$. 
\begin{defi}[Self-Diffusion Coefficient]
Given $\boldsymbol X_t=(X_t^1,\ldots ,X_t^d)$ the position at time $t$ of the tagged particle, the $d$-dimensional \emph{self-diffusion matrix} $D_s=D_s(\rho)$ is defined as 
\begin{equation}\label{matsdApp}x^\dagger D_s x=\lim_{t\to\infty} \frac{\E((x\cdot\boldsymbol X_t)^2)}{t}\quad \forall y\in \R^d,\end{equation}
where $x^\dagger$ is the transposed vector of $x$ and $(\;.\;)$ is the usual inner product in $\R^d$.
\end{defi}
This result follows from \cite{KV1986}. 
The following Lemma gives a variational formula for $D_s$ and was obtained in Spohn \cite{Spohn1990}.
\begin{prop}[Variational formula for the self-diffusion coefficient]
\label{prop:VariationnalSDCh}
The self-diffusion matrix $D_s=D_s(\rho)$ is characterized by the variational formula 
\begin{equation*}x^\dagger D_s x=\inf_{f\in \Sp}\Bigg\{\sum_{i=1,2}\cro{\Egcm\pa{(1-{\conf_{e_i}})\cro{x_i+\tau_{e_i}f\pa{{\conf^{0,e_i}}}-f}^2}+\sum_{y\neq 0, e_i}\Egcm\pa{\cro{\nabla_{0,e_i}\tau_y f}^2}}\Bigg\}.
\end{equation*}
\end{prop}
 Our system being invariant through coordinates inversions, it is shown in \cite{DMFG1989} that the matrix $D_s$ is diagonal, and can therefore be written \[D_s(\rho)=\sdc(\rho)I.\]
Finally, the regularity of the self-diffusion coefficient follows from  \cite{LOV2001}, and a lower and upper bound was derived by Varadhan in all dimensions by Varadhan in  \cite{Varadhan1994}.
\begin{prop}[Regularity of the self-diffusion coefficient]
\label{prop:regularitySDC}
In any dimension $d\geq 1$, the self-diffusion coefficient $\sdc$ is $C^{\infty}([0,1])$, and for some constant $C>0$, we can write 
\[\frac{1}{C}(1-\rho)\leq \sdc(\rho)\leq C(1-\rho).\]
\end{prop}
Finally, we prove a result that we postponed in during the proof of Proposition \ref{prop:currentsdecomposition}.
\begin{prop}[Conductivity matrix]
\label{prop:scalvar}
Fix $\param\in \pset$, let $\curohat=(\curohat_1, \curohat_2)$, where as before
\[\curohat_i=[\omega(\theta_0)-\Egcm(\omega)]\conf_0(1-\conf_{e_i})-[\omega(\theta_{e_i})-\Egcm(\omega)]\conf_{e_i}(1-\conf_0).\] 
Recall that we defined the conductivity matrix $Q=Q^\omega$ as
\begin{equation*}
x^\dagger Q x=\inf_{g\in \tcal}\scal{x\cdot\curohat+\gene g },
\end{equation*}
then, we have the identity
\begin{equation}
\label{eq:Qid}
Q=\alpha V_{\param}(\omega)D_s(\alpha)=\alpha V_{\param}(\omega) d_s(\alpha)I.
\end{equation}
\end{prop}
\proofthm{Proposition {\color{red}B.4}}{The proof is analogous to that of Theorem 3.2 in \cite{Quastel1992}. We first consider the trivial case $\am=0,1$. Since $d_s(1)=0$, if $\alpha=0,1,$ Proposition \ref{prop:scalvar}  is trivially true, because both sides of the identity vanish. 
Furthermore, assuming that $V_{\param}(\omega)=0$, we then have $\curohat=0$, therefore both sides vanish as well.
We now assume that $\am\in]0,1[$ and $V_{\param}(\omega)>0$. 
By definition \ref{defi:limitingcovariance1},  
\begin{equation*}\scal{x\cdot\curohat+\gene g }\egal \Egcm\pa{\sum_{i=1,2}\cro{x_i\comz_0(1-{\conf_{e_i}})+\nabla_{0,e_i}\sum_{y\in\Z^2}\tau_y g}^2}.\end{equation*}
Since $g\in {T^\omega}$, it can be rewritten $g= \varphi(\conf)+\sum_y\comz_y\psi_y(\conf)$ for some angle-blind functions $\varphi, \psi_y\in\Sp$. As we saw in the proof of the spectral gap, any angle-blind function is orthogonal to any function $\comz_y\psi(\conf)$, therefore
\begin{equation*}\scal{x\cdot\curohat+\gene g }={\sum_{i=1,2}}\Egcm\bigg(\Big[x_i\comz_0(1-{\conf_{e_i}})+\nabla_{0,e_i}\sum_{y,y'\in\Z^2}\tau_{y'} [ \comz_{y}\psi_y]\Big]^2\bigg)+\Egcm\bigg(\Big[\nabla_{0,e_i}\sum_{y\in \Z^2}\varphi\Big]^2\bigg).
\end{equation*}
To minimize the left-hand side, we can choose $\varphi=0$, so that $g$ must take the form $g=\sum_y\comz_y\psi_y$.
Since $g$ is a local function, $\psi'=\sum_y\tau_{-y}\psi_y$ is well defined, and satisfies
$\sum_{y,y'\in\Z^2}\tau_{y'}[ \comz_y\psi_y]=\sum_{y\in \Z^2}\comz_y\tau_y\psi'$, therefore 
\[\scal{x\cdot\curomz+\gene g}={\sum_{i=1,2}}\Egcm\pa{\cro{x_i\comz_0(1-{\conf_{e_i}})+\nabla_{0,e_i}\sum_{y\in\Z^2}\comz_y\tau_y \psi'}^2}.\]

\bigskip

Elementary computations yield $\nabla_{0,e_i}\comz_0 \psi'=-\comz_0(1-\conf_{e_i})\psi'$, $\nabla_{0,e_i}\comz_{e_i} \tau_{e_i}\psi'=\comz_{0}(1-\conf_{e_i}) \tau_{e_i}\psi'\pa{\conf^{0,e_i}}$, and for any $y\neq 0,e_i$, $\comz_y\nabla_{0,e_i}\tau_y \psi'$, therefore
\[\scal{x\cdot\curomz+\gene g}={\sum_{i=1,2}}\Egcm\pa{\cro{\comz_0(1-{\conf_{e_i}})\cro{x_i+\tau_{e_i}\psi'\pa{\conf^{0,e_i}}-\psi'}+\sum_{y\neq 0,e_i}\comz_y\nabla_{0,e_i}\tau_y \psi'}^2}.\]
For any angle-blind function $\psi\in \Sp$, we have already established in Section \ref{subsec:spectralgap} that
\begin{equation*}\Egcm(\comz_y\comz_{y'}\psi({\conf}))=\1_{\{y=y'\}V_{\param}(\omega)\Egcm(\conf_y\psi({\conf}))}.\end{equation*}
The previous quantity now rewrites
\begin{multline*}\scal{x\cdot\curomz+\gene g}\\
=V_{\param}(\omega)\sum_{i=1,2}\cro{\Egcm\pa{\conf_0(1-{\conf_{e_i}})\cro{x_i+\tau_{e_i}\psi'\pa{{\conf^{0,e_i}}}-\psi'}^2}+\sum_{y\neq 0, e_i}\Egcm\pa{\conf_y\cro{\nabla_{0,e_i}\tau_y \psi'}^2}}.\end{multline*}
Denote by $f=\Egcm(\psi'|\conf_0=1){=\am^{-1}\Egcm(\conf_0\psi')}$, where the expectation is taken only w.r.t. $\conf_0$ ($f$ is therefore a function of the configuration $(\eta_x)_{x\neq 0}$), we have 
\[\Egcm(\tau_{e_i}\psi'\pa{{\conf^{0,e_i}}}|\conf_0=1)=\tau_{e_i}f\pa{{\conf^{0,e_i}}}\eqand \Egcm(\nabla_{0,e_i}\tau_y \psi'|\conf_y=1)=\nabla_{0,e_i}\tau_y f,\]
so that 
\begin{multline*}
\scal{x\cdot\curohat+\gene g }={\am} V_{\param}(\omega)\sum_{i=1,2}\cro{\Egcm\pa{(1-{\conf_{e_i}})\cro{x_i+\tau_{e_i}f\pa{{\conf^{0,e_i}}}-f}^2}+\sum_{y\neq 0, e_i}\Egcm\pa{\cro{\nabla_{0,e_i}\tau_y f}^2}}.
\end{multline*}
Taking the infimum over $g\in \tcal$, $f$ spans $\Sp$ which yields as wanted, according to Proposition \ref{prop:VariationnalSDCh}
\[x^\dagger Qx=\scal{x\cdot\curohat+\gene g }=\am V_{\param}(\omega)x^\dagger D_sx,\]
thus concluding the proof.
}

\subsection{Entropy}
Given two measures on a space $E$, let us denote \[H(\mu\mid \nu)=\E_{\nu}\left(\frac{d\mu}{d\nu}\log\frac{d\mu}{d\nu}\right)\] the relative entropy of $\mu$ w.r.t $\nu$.
\begin{prop}[Entropy inequality]
\label{prop:entropyinequality}
 Let $\pi$ be a reference measure on some probability space $E$. Let $f$ be a function $E\to \R$, and $\gamma\in \R^+$. Then, for any {non-negative} measure $\mu$ on $E$, we have 
\[\int f d\mu\leq \frac{1}{\gamma}\left[\log\left(\int e^{\gamma f}d\pi\right) + H(\mu \vert \pi)\right],\]
where $H(\mu\vert\pi)$ is the relative entropy of $\mu$ with respect to $\pi$.  
\end{prop}
 \proofthm{Proposition \ref{prop:entropyinequality}}{The proof is omitted, it can be found in Appendix 1.8 of \cite{KLB1999}.
}
\begin{rema}[Utilization throughout the proof]
This inequality is used throughout this proof with $\mu_s^N$ the marginal at time $s$ of the measure of the process started from an initial profile $\mu^N$, and with $\pi=\mesinv$ the equilibrium measure of a symmetric simple exclusion process with grand-canonical parameter $\param$. Then, for any fixed time $s$ and for any function $f$ and any positive $ \gamma$
\[\E_{\mu_s^N}( f) \leq \frac{1}{\gamma}\cro{\log\Egcm\pa{ e^{\gamma f} } + H(\mu_s^N \vert \mesinv)}.\]
This inequality will be our main tool to bound expectation w.r.t the measure of our process of vanishing quantities . 
\end{rema}

\subsection{Bound on the largest eigenvalue of a perturbed Markov generator}
\begin{prop}[Largest eigenvalue for a small perturbation of a Markov generator]
\label{prop:markoveigenvalue}
Let us consider a Markov Generator $L$ with positive spectral gap $\gamma$ and a bounded function $V$ with mean $0$ with respect to the equilibrium measure $\mu_{\param}$ of the Markov process. Then, for any small $\varepsilon >0$,  the Largest eigenvalue of the operator $L+\varepsilon V$ can be bounded from above by 
\[ \sup_{f}\left\{\varepsilon\Egcm(Vf^2)+\Egcm(fLf)\right\}\leq \frac{\varepsilon^2}{A-2\varepsilon \gamma\norm{V}_{\infty}}\Egcm\pa{V(-L)^{-1}V},\]
where the supremum in the variational formula is taken among the probability densities $f$ w.r.t $\mu_{\param}$.
\end{prop}
The proof of this result is omitted, it is  given in Theorem A3.1.1, p.375 in \cite{KLB1999}.

\section{Space of grand-canonical parameters}
\label{sec:B}
In this appendix, we prove some useful results regarding the space of parameters $(\pset, \normm{\cdot})$ introduced in Section \ref{subsec:canonicalmeasures}.
\subsection{Equivalence of ensembles}
\label{subsec:equivalenceensembles}
\begin{prop}[Equivalence of ensembles]
\label{prop:equivalenceofensembles}
Let $f$ be a cylinder function (in the sense of Definition \ref{def:conf}),  we have \[\liml \sup_{\K\in \Kset_l}\abs{\E_{l, \K}(f)-\E_{\param_{\K}}(f)}\to 0,\]
where the first measure is the projection along sets with $\K$ particles in $B_l$, whereas the second is the grand-canonical measure with parameter $\param=\param_{\K}$ introduced in Definition \ref{defi:paramK}.
\end{prop}
\proofthm{Proposition \ref{prop:equivalenceofensembles}}{The proof of this result is quite elementary, and is a matter of carefully writing expectations for a random sampling with (grand-canonical measures) and without (canonical measures) replacement. 

The proof of this problem can be reduced to the following~: Consider two samplings of $M$ occupation variables, chosen among $L$ fixed  possible values 
\[\{\confhat^1,\dots,\confhat^L\}\in\Sigma_1^L:=\{(\delta,\theta)\in \{0,1\}\times \ctoruspi,\; \theta=0 \mbox{ if }\delta=0\}^L.\]
The first sampling is made without replacement to represent the canonical measure $\mu_{l, \K}$, and the sampled items will be denoted $X_1,\dots,X_M$, where each $X_i$ is of the form $(\delta,\theta)$. The second sampling is made with replacement to represent the grand-canonical measure $\mu_{\param_{\K}}$, and will be denoted $Y_1,\ldots ,Y_M$. let us denote by $\xi^L$ the set 
\[\xi^L=\{\confhat^1,\dots,\confhat^L\},\]
and denote by $\E_{\xi^L}$ the expectation w.r.t. the two samplings $(X_i)$ and $(Y_i)$ given $\xi^L$. Further denote by $I_{L,M}=\{1,\dots,L\}^M$, ${\bf{i}}=(i_1,\dots,i_M)$ the elements of $I_{L,M}$, and $D_{L,M}$ and $C_{L,M}$ its two subsets
\[D_{L,M}=\{(i_1,\dots,i_M)\in I_{L,M} \; \;\big| \;\; i_1\neq \cdots \neq i_M\}, \quad \mbox{ and }\quad C_{L,M}=I_{L,M}\setminus D_{L,M}\]
Then, for any function 
\[g:\Sigma_1^M\to \R,\]
we have 
\begin{multline*}\abs{\E_{\xi^L}(g(X_1,\ldots ,X_M))-\E_{\xi^L}(g(Y_1,\ldots ,Y_M))}\\
\leq \norm{g}_{\infty}\sum_{{\bf i}\in I_{L,M}}\abs{\Prob_{\xi^L}\Big[(X_1,\dots,X_M)=(\confhat^{i_1},\dots,\confhat^{i_M})\Big]-\Prob_{\xi^L}\Big[(Y_1,\dots,Y_M)=(\confhat^{i_1},\dots,\confhat^{i_M})\Big]}\\
= \norm{g}_{\infty}\sum_{{\bf i}\in D_{L,M}}\abs{\Prob_{\xi^L}\Big[(X_1,\dots,X_M)=(\confhat^{i_1},\dots,\confhat^{i_M})\Big]-\Prob_{\xi^L}\Big[(Y_1,\dots,Y_M)=(\confhat^{i_1},\dots,\confhat^{i_M})\Big]}\\
+\norm{g}_{\infty}\sum_{{\bf i}\in C_{L,M}}\Prob_{\xi^L}\Big[(Y_1,\dots,Y_M)=(\confhat^{i_1},\dots,\confhat^{i_M})\Big].\end{multline*}
The sum on the last line is the probability that at least two indexes among the $M$ we chosen uniformly in $\{1,\dots,L\}$ are equal. This probability is 
\[\sum_{{\bf i}\in C_{L,M}}\Prob_{\xi^L}\Big[(Y_1,\dots,Y_M)=(\confhat^{i_1},\dots,\confhat^{i_M})\Big]=1-\frac{L(L-1)\cdots (L-M+1)}{L^M},\]
which for $M$ fixed vanishes uniformly in $\xi^L$ as $L\to\infty$. We now take a look at the other term, for which we write
\begin{align*}
\sum_{{\bf i}\in D_{L,M}}\bigg|\Prob_{\xi^L}\Big[(X_1,\dots,X_M)=(\confhat^{i_1},\dots,\confhat^{i_M})\Big]&-\Prob_{\xi^L}\Big[(Y_1,\dots,Y_M)=(\confhat^{i_1},\dots,\confhat^{i_M})\Big]\bigg|\\
&=\sum_{{\bf i}\in D_{L,M}}\abs{\frac{1}{L(L-1)\cdots (L-M+1)}-\frac{1}{L^M}}\\
&=1-\frac{L(L-1)\cdots (L-M+1)}{L^M},
 \end{align*}
which also vanishes uniformly in $\xi^L$ as $L\to\infty$. We can therefore write for any bounded function $g$ depending on $M$ sites
\[\sup_{\xi^L\in \Sigma_1^L}\abs{\E_{\xi^L}(g(X_1,\ldots ,X_M))-\E_{\xi^L}(g(Y_1,\ldots ,Y_M))}\leq \norm{g}_{\infty}C(M)o_L(1),\]
thus proving Proposition \ref{prop:equivalenceofensembles}.
}

\subsection{Regularity of the grand-canonical measures in their parameter}
\label{subsec:lipshitzcontinuity}
\begin{prop}\label{prop:Lipschitzcontinuity}
Consider the set of local profiles $\pset$ equipped with the norm $\normm{\cdot}$ defined in Definition \ref{defi:convparam}. Then, given a function $g\in {\mathcal C}$, the application
\[\func{\Psi_g}{(\pset, \normm{\cdot})}{\R}{\param}{\Egcm(g)}\]
is Lipschitz-continuous with Lipschitz constant depending on the function $g$.
\end{prop}
\proofthm{Proposition \ref{prop:Lipschitzcontinuity}}{Let us consider a cylinder function $g$ depending only on vertices $x_1,\ldots ,x_M$, and let us start by assuming that $g$ vanishes as soon as one of the sites $x_1,\ldots ,x_M$ is empty. We can then rewrite $g(\confhat)$ as $\conf_{x_1}\ldots \conf_{x_M}g(\theta_{x_1},\ldots ,\theta_{x_M})$, and 
\[\Egcm(g)=\int_{\theta_1}\ldots \int_{\theta_M}g(\theta_{x_1},\ldots ,\theta_{x_M})d\param(\theta_{x_1})\ldots d\param(\theta_{x_M}).\]
We can now proceed by recurrence on $M$. Given a function $g$ depending only on a site $x_1$, and for any two grand-canonical parameters $\param$ and $\param'$ we can write  
\begin{align*}\Egcm(g)-\E_{\param'}(g)=\norm{g}^*\int_{\theta_{x_1}}\frac{g(\theta_{x_1})}{\norm{g}^*}d(\param-\param')(\theta_{x_1})\leq \norm{g}^*\normm{\param-\param'}\end{align*}
Assuming now that the proposition is true for any function depending on $M-1$ sites, and considering a function $g$ depending on $M$ vertices, we can write 
\begin{align}\label{rec}\Egcm(g)-\E_{\param'}(g)= \Egcm\pa{\Egcm(g\mid \confhat_{x_2},\ldots ,\confhat_{x_M})}-\E_{\param'}\pa{\E_{\param'}(g\mid \confhat_{x_2},\ldots ,\confhat_{x_M})}.\end{align}
Fix any angle $\theta$, and let $g^{\theta}$ be the function $g^{\theta}(\confhat)=g(\theta, \theta_{x_2,\ldots ,\theta_{x_M}})$, we can write thanks to the recurrence hypothesis that 
\[\abs{\Egcm(g^{\theta})-\E_{\param'}(g^{\theta})}\leq C_{\theta}\normm{\param-\param'},\] which, integrated in $\theta$ against $\param'$, yields
\[\abs{\E_{\param'}\pa{\E_{\param'}(g\mid \confhat_{x_2},\ldots ,\confhat_{x_M})}-\E_{\param'}\pa{\E_{\param}(g\mid \confhat_{x_2},\ldots ,\confhat_{x_M})}}\leq C_1\normm{\param-\param'},\] 
On the other hand, we can also write 
\[\abs{\E_{\param}\pa{\E_{\param}(g\mid \confhat_{x_2},\ldots ,\confhat_{x_M})}-\E_{\param'}\pa{\E_{\param}(g\mid \confhat_{x_2},\ldots ,\confhat_{x_M})}}\leq C_2\normm{\param-\param'},\] 
therefore \eqref{rec} yields that 
\[\abs{\Egcm(g)-\E_{\param'}(g)}\leq (C^1+C^2)\normm{\param-\param'} ,\]
which is what we wanted to show.

To complete the proof of Proposition \ref{prop:Lipschitzcontinuity}, we now only need to extend the result to functions $g$ which do not necessarily vanish when one site in their domain is empty. This case is easily derived, since any function $g$ depending on vertices $x_1$,\ldots ,$,x_M$ can be rewritten
\begin{equation}\label{decg}g(\confhat_{x_1},\ldots ,\confhat_{x_M})=\sum_{B\subset\{1,\ldots ,M\}}g_B(\theta_{x_i}, i\in B),\end{equation}
where $g_B(\theta_{x_i}, i\in B)$ is defined in the following fashion~: recall that $\confhat_x=(\conf_x, \theta_x)$, with  $\theta_x=0$ if $\conf_x=0$, and let us assume that $B$ is the set of increasing indexes $i_1,\ldots , i_p$, then $g_B$ is defined as 
\[g_B(\theta_{x_{i_1}},\ldots ,\theta_{x_{i_p}})=\conf_{x_{i_1}}\ldots \conf_{x_{i_p}}g((0,0),\ldots ,(0,0), (1,\theta_{x_{i_1}}),(0,0),\ldots ,(0,0),(1,\theta_{x_{i_p}}),(0,0),\ldots ,(0,0)).\]
These functions all vanish whenever one of their depending sites is empty, therefore  according to the beginning of the proof, there exists a family of constants $C_B$ such that for any $B\subset\{1,\ldots ,M\}$ we have \[\abs{\Egcm(g_B)-\E_{\param'}(g_B)}\leq C_B\normm{\param-\param'}.\]
We now only need to let $C=\sum_{B\subset\{1,\ldots ,M\}}C_B$ to obtain thanks to the decomposition \eqref{decg} that \[\abs{\Egcm(g)-\E_{\param'}(g)}\leq C\normm{\param-\param'}\]
as intended. This completes the proof of Proposition \ref{prop:Lipschitzcontinuity}.
}

\subsection{Compactness of the set of grand-canonical parameters}
\label{subsec:compaciteparam}
\begin{prop}[Compactness of $(\pset,\normm{\cdot})$]
\label{prop:psetcompact}
The metric space $(\pset,\normm{\cdot})$ introduced in Definition \ref{defi:convparam} is totally bounded and Cauchy complete, and is therefore compact.
\end{prop}
\proofthm{Proposition \ref{prop:psetcompact}}{The proof of the Cauchy-completeness is almost immediate, we treat it first. Consider a Cauchy sequence $(\param_k)_{k\in\N}\in\pset^{\N}$, then by definition of $\normm{\cdot}$, for any $g\in B^*$, the sequence $(\int_{\ctoruspi}g(\theta)\param_k(d\theta))_{k}$ is a real Cauchy sequence and therefore converges, and we can let \[\int_{\ctoruspi}g(\theta)\param^*(d\theta)=\lim_{k\to\infty}\int_{\ctoruspi}g(\theta)\param_k(d\theta).\]
This definition can be extended to any $C^1(\ctoruspi)$ function $g$ by letting 
\[\int_{\ctoruspi}g(\theta)\param^*(d\theta)=\max(\norm{g}_{\infty},\norm{g'}_{\infty})\lim_{k\to\infty}\int_{\ctoruspi}\frac{g(\theta)}{\max(\norm{g}_{\infty},\norm{g'}_{\infty})}\param_k(d\theta).\]
This defines a measure $\param^*$ on $\ctoruspi$, whose total mass is given by \[\int_{\ctorus}\param^*(d\theta)=\lim_{k\to\infty}\int_{\ctorus}\param_k(d\theta)\in [0,1],\]
which proves the Cauchy completeness of $(\pset, \normm{\cdot})$.

We now prove that $(\pset,\normm{\cdot})$ is totally bounded. For any integer $n$, we are going to construct a \emph{finite set} $\mathcal{M}_{1,n}\subset \pset$ such that \[\sup_{\param\in \pset}\inf_{\param'\in \mathcal{M}_{1,n}}\normm{\param- \param'}\leq \frac{1}{n}.\]
For any $n\in \N$ and any $j\in  \llbracket 0,n-1\rrbracket$, we shorten $\theta_{j,n}=2\pi j/n$, and $\theta_{n,n}=\theta_{0,n}=0$. We can now define 
\[\mathcal{M}_{1,n}=\left\{\sum_{j=0}^{n-1} \frac{k_{j}}{n^2}\delta_{\theta_{j,n}} \quad \left| \quad k_j\in \llbracket 0,n^2\rrbracket, \quad \sum_{j}k_j\leq n^2\right.\right\}.\]
The inclusion $\mathcal{M}_{1,n}\subset \pset$ is trivial thanks to the condition $\sum_{j}k_j\leq n^2$, and $\mathcal{M}_{1,n}$ is finite since the $k_j$'s can each take only a finite number of values. we now prove that any $\param\in\pset$ is at distance at most $1/n$ of an element $\param_n\in \mathcal{M}_{1,n}$. 

Fix $\param\in\pset$, and let \[k_j=\lfloor n^2\param([\theta_{j,n},\theta_{j+1,n}[)\rfloor.\]
Since $\param\in \pset$, its total mass is in $[0,1]$, and the conditions $k_j\in\llbracket 0,n^2\rrbracket $ and $\sum_j k_j\leq n^2$ are trivially verified. We now let \[\param_n=\sum_{j=0}^{n-1} \frac{k_{j}}{n^2}\delta_{\theta_{j,n}},\]
and prove that $\normm{\param-\param_n}\leq 2/n$. Fix a function $g\in C^1(\ctoruspi)$ such that  $\max(\norm{g}_{\infty},\norm{g'}_{\infty})\leq 1$, we can write 
\begin{align*}\int_{\ctoruspi}g(\theta)(\param-\param_n)(d\theta)=&\sum_{j=0}^{n-1} {\int_{[\theta_{j,n}\theta_{j+1,n}[}}g(\theta)\param(d\theta)- \frac{k_{j}}{n^2}g(\theta_{j,n})\\
=&\sum_{j=0}^{n-1} \param([\theta_{j,n},\theta_{j+1,n}[)g(\theta_{j,n})- \frac{k_{j}}{n^2}g(\theta_{j,n})+\sum_{j=0}^{n-1} {\int_{[\theta_{j,n}\theta_{j+1,n}[}}(g(\theta)-g(\theta_{j,n}))\param(d\theta)\\
\leq &\sum_{j=0}^{n-1}\norm{g}_{\infty}\underset{\leq 1/n^2}{\underbrace{\abs{\param([\theta_{j,n},\theta_{j+1,n}[)- \frac{k_j}{n^2}}}}+\sum_{j=0}^{n-1}\norm{g'}_{\infty}\underset{\leq 1/n}{\underbrace{\abs{\theta_{j+1,n}-\theta_{j+1,n}}}} {\int_{[\theta_{j,n}\theta_{j+1,n}[}}\param(d\theta)\\
\leq& \frac{\norm{g}_{\infty}+\norm{g'}_{\infty}}{n}\leq 2/n
.\end{align*}
Finally, we have proved that \[\normm{\param-\param_n}\leq 2/n,\]
which proves that $\pset$ is totally bounded. This, together with the Cauchy completeness, immediately yields the compactness, and concludes the proof of Proposition \ref{prop:psetcompact}.

}

\backmatter
\printindex

\bibliographystyle{plain}

\bibliography{Bibliographie}

\end{document}

%% file: Irred2trous.pdftex_t
\begin{picture}(0,0)%
\includegraphics{Irred2trous.pdf}%
\end{picture}%
%
%
\setlength{\unitlength}{6315sp}%
\begingroup\makeatletter\ifx\SetFigFont\undefined%
\gdef\SetFigFont#1#2#3#4#5{%
  \reset@font\fontsize{#1}{#2pt}%
  \fontfamily{#3}\fontseries{#4}\fontshape{#5}%
  \selectfont}%
\fi\endgroup%
\begin{picture}(1910,2192)(1410,-2844)
\put(1553,-751){\makebox(0,0)[lb]{\smash{{\SetFigFont{10}{12.0}{\rmdefault}{\mddefault}{\updefault}{\color[rgb]{0,0,0}New position of the two empty sites}%
}}}}
\put(2656,-1242){\makebox(0,0)[lb]{\smash{{\SetFigFont{10}{12.0}{\rmdefault}{\mddefault}{\updefault}{\color[rgb]{0,.56,0}$a_2$}%
}}}}
\put(2212,-1237){\makebox(0,0)[lb]{\smash{{\SetFigFont{10}{12.0}{\rmdefault}{\mddefault}{\updefault}{\color[rgb]{.69,0,0}$a_1$}%
}}}}
\end{picture}%

%% file: LHPtest.pdftex_t
\begin{picture}(0,0)%
\includegraphics{LHPtest.pdf}%
\end{picture}%
%
%
\setlength{\unitlength}{2763sp}%
\begingroup\makeatletter\ifx\SetFigFont\undefined%
\gdef\SetFigFont#1#2#3#4#5{%
  \reset@font\fontsize{#1}{#2pt}%
  \fontfamily{#3}\fontseries{#4}\fontshape{#5}%
  \selectfont}%
\fi\endgroup%
\begin{picture}(4748,4749)(1643,-5079)
\put(4816,-4816){\makebox(0,0)[lb]{\smash{{\SetFigFont{8}{9.6}{\rmdefault}{\mddefault}{\updefault}{\color[rgb]{0,0,0}$l$}%
}}}}
\put(2806,-1726){\makebox(0,0)[lb]{\smash{{\SetFigFont{8}{9.6}{\rmdefault}{\mddefault}{\updefault}{\color[rgb]{0,0,0}$B^1$}%
}}}}
\put(4921,-3826){\makebox(0,0)[lb]{\smash{{\SetFigFont{8}{9.6}{\rmdefault}{\mddefault}{\updefault}{\color[rgb]{0,0,0}$B^p$}%
}}}}
\put(2341,-4951){\makebox(0,0)[lb]{\smash{{\SetFigFont{8}{9.6}{\rmdefault}{\mddefault}{\updefault}{\color[rgb]{0,0,0}$B^0$}%
}}}}
\put(6143,-3885){\makebox(0,0)[lb]{\smash{{\SetFigFont{8}{9.6}{\rmdefault}{\mddefault}{\updefault}{\color[rgb]{0,0,0}$B_l$}%
}}}}
\put(4291,-781){\makebox(0,0)[lb]{\smash{{\SetFigFont{8}{9.6}{\rmdefault}{\mddefault}{\updefault}{\color[rgb]{0,0,0}$2k$}%
}}}}
\end{picture}%

%% file: 2BE.pdftex_t
\begin{picture}(0,0)%
\includegraphics{2BE.pdf}%
\end{picture}%
%
%
\setlength{\unitlength}{3158sp}%
\begingroup\makeatletter\ifx\SetFigFont\undefined%
\gdef\SetFigFont#1#2#3#4#5{%
  \reset@font\fontsize{#1}{#2pt}%
  \fontfamily{#3}\fontseries{#4}\fontshape{#5}%
  \selectfont}%
\fi\endgroup%
\begin{picture}(3331,2860)(3296,-3426)
\put(3661,-2971){\makebox(0,0)[lb]{\smash{{\SetFigFont{10}{12.0}{\rmdefault}{\mddefault}{\updefault}{\color[rgb]{0,0,0}$B_l,\; \confrond_1$}%
}}}}
\put(5311,-1831){\makebox(0,0)[lb]{\smash{{\SetFigFont{10}{12.0}{\rmdefault}{\mddefault}{\updefault}{\color[rgb]{0,0,0}$\tau_yB_l,\; \confrond_2$}%
}}}}
\put(6076,-1261){\makebox(0,0)[lb]{\smash{{\SetFigFont{10}{12.0}{\rmdefault}{\mddefault}{\updefault}{\color[rgb]{0,0,0}$y$}%
}}}}
\put(4336,-2701){\makebox(0,0)[lb]{\smash{{\SetFigFont{10}{12.0}{\rmdefault}{\mddefault}{\updefault}{\color[rgb]{0,0,0}$0$}%
}}}}
\put(5851,-2341){\makebox(0,0)[lb]{\smash{{\SetFigFont{10}{12.0}{\rmdefault}{\mddefault}{\updefault}{\color[rgb]{0,0,0}$2l$}%
}}}}
\end{picture}%

%% file: yz.pdftex_t
\begin{picture}(0,0)%
\includegraphics{yz.pdf}%
\end{picture}%
%
%
\setlength{\unitlength}{3947sp}%
\begingroup\makeatletter\ifx\SetFigFont\undefined%
\gdef\SetFigFont#1#2#3#4#5{%
  \reset@font\fontsize{#1}{#2pt}%
  \fontfamily{#3}\fontseries{#4}\fontshape{#5}%
  \selectfont}%
\fi\endgroup%
\begin{picture}(4600,3958)(399,-3571)
\put(4569,-1644){\makebox(0,0)[lb]{\smash{{\SetFigFont{12}{14.4}{\rmdefault}{\mddefault}{\updefault}{\color[rgb]{0,0,0}$z_0$}%
}}}}
\put(2836,-1759){\makebox(0,0)[lb]{\smash{{\SetFigFont{12}{14.4}{\rmdefault}{\mddefault}{\updefault}{\color[rgb]{0,0,0}$0$}%
}}}}
\put(4962,-1954){\makebox(0,0)[lb]{\smash{{\SetFigFont{12}{14.4}{\rmdefault}{\mddefault}{\updefault}{\color[rgb]{0,0,.82}$B_{\varepsilon N}$}%
}}}}
\put(961,-3496){\makebox(0,0)[lb]{\smash{{\SetFigFont{12}{14.4}{\rmdefault}{\mddefault}{\updefault}{\color[rgb]{.69,0,0}$B_{\varepsilon N}(-e_i)$}%
}}}}
\put(594,-129){\makebox(0,0)[lb]{\smash{{\SetFigFont{12}{14.4}{\rmdefault}{\mddefault}{\updefault}{\color[rgb]{0,0,0}$y_{\varepsilon N}$}%
}}}}
\put(421,-489){\makebox(0,0)[lb]{\smash{{\SetFigFont{12}{14.4}{\rmdefault}{\mddefault}{\updefault}{\color[rgb]{0,0,0}$y_{\varepsilon N-1}$}%
}}}}
\put(459,-2964){\makebox(0,0)[lb]{\smash{{\SetFigFont{12}{14.4}{\rmdefault}{\mddefault}{\updefault}{\color[rgb]{0,0,0}$y_{-\varepsilon N}$}%
}}}}
\put(414,-2619){\makebox(0,0)[lb]{\smash{{\SetFigFont{12}{14.4}{\rmdefault}{\mddefault}{\updefault}{\color[rgb]{0,0,0}$y_{1-\varepsilon N}$}%
}}}}
\put(4569,-181){\makebox(0,0)[lb]{\smash{{\SetFigFont{12}{14.4}{\rmdefault}{\mddefault}{\updefault}{\color[rgb]{0,0,0}$z_{\varepsilon N}$}%
}}}}
\put(4561,-519){\makebox(0,0)[lb]{\smash{{\SetFigFont{12}{14.4}{\rmdefault}{\mddefault}{\updefault}{\color[rgb]{0,0,0}$z_{\varepsilon N-1}$}%
}}}}
\put(4561,-2641){\makebox(0,0)[lb]{\smash{{\SetFigFont{12}{14.4}{\rmdefault}{\mddefault}{\updefault}{\color[rgb]{0,0,0}$z_{1-\varepsilon N}$}%
}}}}
\put(4569,-3009){\makebox(0,0)[lb]{\smash{{\SetFigFont{12}{14.4}{\rmdefault}{\mddefault}{\updefault}{\color[rgb]{0,0,0}$z_{-\varepsilon N}$}%
}}}}
\put(414,-826){\makebox(0,0)[lb]{\smash{{\SetFigFont{12}{14.4}{\rmdefault}{\mddefault}{\updefault}{\color[rgb]{0,0,0}$y_{\varepsilon N-2}$}%
}}}}
\put(639,-1546){\makebox(0,0)[lb]{\smash{{\SetFigFont{12}{14.4}{\rmdefault}{\mddefault}{\updefault}{\color[rgb]{0,0,0}$y_{0}$}%
}}}}
\put(414,-2251){\makebox(0,0)[lb]{\smash{{\SetFigFont{12}{14.4}{\rmdefault}{\mddefault}{\updefault}{\color[rgb]{0,0,0}$y_{2-\varepsilon N}$}%
}}}}
\put(4569,-871){\makebox(0,0)[lb]{\smash{{\SetFigFont{12}{14.4}{\rmdefault}{\mddefault}{\updefault}{\color[rgb]{0,0,0}$z_{\varepsilon N-2}$}%
}}}}
\put(4569,-2296){\makebox(0,0)[lb]{\smash{{\SetFigFont{12}{14.4}{\rmdefault}{\mddefault}{\updefault}{\color[rgb]{0,0,0}$z_{2-\varepsilon N}$}%
}}}}
\end{picture}%

%% file: phiep.pdftex_t
\begin{picture}(0,0)%
\includegraphics{phiep.pdf}%
\end{picture}%
%
%
\setlength{\unitlength}{2368sp}%
\begingroup\makeatletter\ifx\SetFigFont\undefined%
\gdef\SetFigFont#1#2#3#4#5{%
  \reset@font\fontsize{#1}{#2pt}%
  \fontfamily{#3}\fontseries{#4}\fontshape{#5}%
  \selectfont}%
\fi\endgroup%
\begin{picture}(4805,3563)(877,-3720)
\put(3398,-852){\makebox(0,0)[lb]{\smash{{\SetFigFont{7}{8.4}{\rmdefault}{\mddefault}{\updefault}{\color[rgb]{0,0,0}$\widetilde{\varphi}_{\varepsilon}(.,\upsilon)$}%
}}}}
\put(1247,-609){\makebox(0,0)[lb]{\smash{{\SetFigFont{7}{8.4}{\rmdefault}{\mddefault}{\updefault}{\color[rgb]{0,0,0}$\abs{\upsilon}>\varepsilon+\varepsilon^3$}%
}}}}
\put(1246,-331){\makebox(0,0)[lb]{\smash{{\SetFigFont{7}{8.4}{\rmdefault}{\mddefault}{\updefault}{\color[rgb]{0,0,0}$\abs{\upsilon}\leq \varepsilon$}%
}}}}
\put(2397,-1756){\makebox(0,0)[lb]{\smash{{\SetFigFont{7}{8.4}{\rmdefault}{\mddefault}{\updefault}{\color[rgb]{0,0,0}$\varphi_{\varepsilon}$}%
}}}}
\put(4422,-3646){\makebox(0,0)[lb]{\smash{{\SetFigFont{7}{8.4}{\rmdefault}{\mddefault}{\updefault}{\color[rgb]{0,0,0}$\varepsilon$}%
}}}}
\put(5571,-3323){\makebox(0,0)[lb]{\smash{{\SetFigFont{7}{8.4}{\rmdefault}{\mddefault}{\updefault}{\color[rgb]{0,0,0}$u_i$}%
}}}}
\put(5328,-3644){\makebox(0,0)[lb]{\smash{{\SetFigFont{7}{8.4}{\rmdefault}{\mddefault}{\updefault}{\color[rgb]{0,0,0}$\varepsilon+\varepsilon^3$}%
}}}}
\put(3332,-1246){\makebox(0,0)[lb]{\smash{{\SetFigFont{7}{8.4}{\rmdefault}{\mddefault}{\updefault}{\color[rgb]{0,0,0}$1/4\varepsilon^2$}%
}}}}
\put(2037,-3631){\makebox(0,0)[lb]{\smash{{\SetFigFont{7}{8.4}{\rmdefault}{\mddefault}{\updefault}{\color[rgb]{0,0,0}$-\varepsilon$}%
}}}}
\put(892,-3632){\makebox(0,0)[lb]{\smash{{\SetFigFont{7}{8.4}{\rmdefault}{\mddefault}{\updefault}{\color[rgb]{0,0,0}$-(\varepsilon+\varepsilon^3)$}%
}}}}
\end{picture}%

%% file: gradphiep.pdftex_t
\begin{picture}(0,0)%
\includegraphics{gradphiep.pdf}%
\end{picture}%
%
%
\setlength{\unitlength}{2368sp}%
\begingroup\makeatletter\ifx\SetFigFont\undefined%
\gdef\SetFigFont#1#2#3#4#5{%
  \reset@font\fontsize{#1}{#2pt}%
  \fontfamily{#3}\fontseries{#4}\fontshape{#5}%
  \selectfont}%
\fi\endgroup%
\begin{picture}(6598,5221)(1242,-7913)
\put(2624,-4225){\makebox(0,0)[lb]{\smash{{\SetFigFont{7}{8.4}{\rmdefault}{\mddefault}{\updefault}{\color[rgb]{0,0,0}$\nabla^{\varepsilon}_i{\varphi}_{\varepsilon}$}%
}}}}
\put(4846,-2896){\makebox(0,0)[lb]{\smash{{\SetFigFont{7}{8.4}{\rmdefault}{\mddefault}{\updefault}{\color[rgb]{0,0,0}$h_{\varepsilon}=\nabla_{i}^{\varepsilon}\widetilde{\varphi}_{\varepsilon}(.,v)$}%
}}}}
\put(3780,-5770){\makebox(0,0)[lb]{\smash{{\SetFigFont{7}{8.4}{\rmdefault}{\mddefault}{\updefault}{\color[rgb]{0,0,0}$-\varepsilon^3$}%
}}}}
\put(4808,-3499){\makebox(0,0)[lb]{\smash{{\SetFigFont{7}{8.4}{\rmdefault}{\mddefault}{\updefault}{\color[rgb]{0,0,0}$1/4\varepsilon^3$}%
}}}}
\put(7087,-5491){\makebox(0,0)[lb]{\smash{{\SetFigFont{7}{8.4}{\rmdefault}{\mddefault}{\updefault}{\color[rgb]{0,0,0}$\varepsilon$}%
}}}}
\put(7796,-5491){\makebox(0,0)[lb]{\smash{{\SetFigFont{7}{8.4}{\rmdefault}{\mddefault}{\updefault}{\color[rgb]{0,0,0}$\varepsilon+\varepsilon^3$}%
}}}}
\put(1257,-5770){\makebox(0,0)[lb]{\smash{{\SetFigFont{7}{8.4}{\rmdefault}{\mddefault}{\updefault}{\color[rgb]{0,0,0}$-(\varepsilon+\varepsilon^3)$}%
}}}}
\put(2227,-5770){\makebox(0,0)[lb]{\smash{{\SetFigFont{7}{8.4}{\rmdefault}{\mddefault}{\updefault}{\color[rgb]{0,0,0}$-\varepsilon$}%
}}}}
\put(7796,-5734){\makebox(0,0)[lb]{\smash{{\SetFigFont{7}{8.4}{\rmdefault}{\mddefault}{\updefault}{\color[rgb]{0,0,0}$u_i$}%
}}}}
\put(5434,-5491){\makebox(0,0)[lb]{\smash{{\SetFigFont{7}{8.4}{\rmdefault}{\mddefault}{\updefault}{\color[rgb]{0,0,0}$\varepsilon^3$}%
}}}}
\end{picture}%

%% file: phifig.pdftex_t
\begin{picture}(0,0)%
\includegraphics{phifig.pdf}%
\end{picture}%
%
%
\setlength{\unitlength}{2763sp}%
\begingroup\makeatletter\ifx\SetFigFont\undefined%
\gdef\SetFigFont#1#2#3#4#5{%
  \reset@font\fontsize{#1}{#2pt}%
  \fontfamily{#3}\fontseries{#4}\fontshape{#5}%
  \selectfont}%
\fi\endgroup%
\begin{picture}(4234,4671)(2629,-5787)
\put(5296,-2086){\makebox(0,0)[lb]{\smash{{\SetFigFont{8}{9.6}{\rmdefault}{\mddefault}{\updefault}{\color[rgb]{0,0,0}$1/4\varepsilon^2+0_{\varepsilon}(1)$}%
}}}}
\put(5198,-1287){\makebox(0,0)[lb]{\smash{{\SetFigFont{8}{9.6}{\rmdefault}{\mddefault}{\updefault}{\color[rgb]{0,0,0}$\Phi_{\varepsilon,i}(.,\upsilon)$}%
}}}}
\put(3009,-1781){\makebox(0,0)[lb]{\smash{{\SetFigFont{8}{9.6}{\rmdefault}{\mddefault}{\updefault}{\color[rgb]{0,0,0}$\abs{\upsilon}>\varepsilon+\varepsilon^3$}%
}}}}
\put(3016,-1516){\makebox(0,0)[lb]{\smash{{\SetFigFont{8}{9.6}{\rmdefault}{\mddefault}{\updefault}{\color[rgb]{0,0,0}$\abs{\upsilon}\leq \varepsilon$}%
}}}}
\put(2644,-5533){\makebox(0,0)[lb]{\smash{{\SetFigFont{8}{9.6}{\rmdefault}{\mddefault}{\updefault}{\color[rgb]{0,0,0}$-(\varepsilon+\varepsilon^3)$}%
}}}}
\put(6442,-5533){\makebox(0,0)[lb]{\smash{{\SetFigFont{8}{9.6}{\rmdefault}{\mddefault}{\updefault}{\color[rgb]{0,0,0}$\varepsilon+\varepsilon^3$}%
}}}}
\put(3676,-5554){\makebox(0,0)[lb]{\smash{{\SetFigFont{8}{9.6}{\rmdefault}{\mddefault}{\updefault}{\color[rgb]{0,0,0}$-\varepsilon$}%
}}}}
\put(6007,-5548){\makebox(0,0)[lb]{\smash{{\SetFigFont{8}{9.6}{\rmdefault}{\mddefault}{\updefault}{\color[rgb]{0,0,0}$\varepsilon$}%
}}}}
\put(5311,-5527){\makebox(0,0)[lb]{\smash{{\SetFigFont{8}{9.6}{\rmdefault}{\mddefault}{\updefault}{\color[rgb]{0,0,0}$\varepsilon^3$}%
}}}}
\put(4402,-5551){\makebox(0,0)[lb]{\smash{{\SetFigFont{8}{9.6}{\rmdefault}{\mddefault}{\updefault}{\color[rgb]{0,0,0}$-\varepsilon^3$}%
}}}}
\end{picture}%

%% file: Tip.pdftex_t
\begin{picture}(0,0)%
\includegraphics{Tip.pdf}%
\end{picture}%
%
%
\setlength{\unitlength}{2763sp}%
\begingroup\makeatletter\ifx\SetFigFont\undefined%
\gdef\SetFigFont#1#2#3#4#5{%
  \reset@font\fontsize{#1}{#2pt}%
  \fontfamily{#3}\fontseries{#4}\fontshape{#5}%
  \selectfont}%
\fi\endgroup%
\begin{picture}(4534,3824)(451,-3435)
\put(3121,-1786){\makebox(0,0)[lb]{\smash{{\SetFigFont{8}{9.6}{\rmdefault}{\mddefault}{\updefault}{\color[rgb]{0,0,0}$e_i$}%
}}}}
\put(4771,-2566){\makebox(0,0)[lb]{\smash{{\SetFigFont{8}{9.6}{\rmdefault}{\mddefault}{\updefault}{\color[rgb]{0,0,.69}$\tau_{e_i} B_p$}%
}}}}
\put(592,-1169){\makebox(0,0)[lb]{\smash{{\SetFigFont{8}{9.6}{\rmdefault}{\mddefault}{\updefault}{\color[rgb]{.69,0,0}$B_p$}%
}}}}
\put(2206,-1771){\makebox(0,0)[lb]{\smash{{\SetFigFont{8}{9.6}{\rmdefault}{\mddefault}{\updefault}{\color[rgb]{0,0,0}$0$}%
}}}}
\put(676, 14){\makebox(0,0)[lb]{\smash{{\SetFigFont{8}{9.6}{\rmdefault}{\mddefault}{\updefault}{\color[rgb]{0,0,0}$y_p$}%
}}}}
\put(706,-691){\makebox(0,0)[lb]{\smash{{\SetFigFont{8}{9.6}{\rmdefault}{\mddefault}{\updefault}{\color[rgb]{0,0,0}$y_{p-1}$}%
}}}}
\put(496,-3151){\makebox(0,0)[lb]{\smash{{\SetFigFont{8}{9.6}{\rmdefault}{\mddefault}{\updefault}{\color[rgb]{0,0,0}$y_{-p}$}%
}}}}
\put(4591,-31){\makebox(0,0)[lb]{\smash{{\SetFigFont{8}{9.6}{\rmdefault}{\mddefault}{\updefault}{\color[rgb]{0,0,0}$z_p$}%
}}}}
\put(4576,-751){\makebox(0,0)[lb]{\smash{{\SetFigFont{8}{9.6}{\rmdefault}{\mddefault}{\updefault}{\color[rgb]{0,0,0}$z_{p-1}$}%
}}}}
\put(4561,-3181){\makebox(0,0)[lb]{\smash{{\SetFigFont{8}{9.6}{\rmdefault}{\mddefault}{\updefault}{\color[rgb]{0,0,0}$z_{-p}$}%
}}}}
\end{picture}%